\documentclass[reqno]{conm-p-l}

\usepackage{amsmath}
\usepackage{amssymb}
\usepackage{graphicx}
\usepackage[linktocpage]{hyperref}
\usepackage{bm}

\usepackage[linktocpage]{hyperref}

\copyrightinfo{2017}{American Mathematical Society}

\newcommand{\mca}{\mathcal{A}}
\newcommand{\mfa}{\mathfrak{a}}

\newcommand{\mcb}{\mathcal{B}}
\newcommand{\mcc}{\mathcal{C}}
\newcommand{\mbc}{\mathbb{C}}
\newcommand{\mbd}{\mathbb{D}}
\newcommand{\mcd}{\mathcal{D}}

\newcommand{\mfD}{\mathfrak{D}}

\newcommand{\mcf}{\mathcal{F}}
\newcommand{\mbf}{\mathbb{F}}

\newcommand{\mbh}{\mathbb{H}}
\newcommand{\bfh}{{\bf H}}
\newcommand{\mch}{\mathcal{H}}
\newcommand{\mcl}{\mathcal{L}}

\newcommand{\mcm}{\mathcal{M}}

\newcommand{\mbn}{\mathbb{N}}
\newcommand{\mcp}{\mathcal{P}}

\newcommand{\mbq}{\mathbb{Q}}

\newcommand{\mbr}{\mathbb{R}}

\newcommand{\mfs}{\mathfrak{s}}

\newcommand{\mbz}{\mathbb{Z}}
\newcommand{\mcz}{\mathcal{Z}}

\newcommand{\ao}{A, \Omega}
\newcommand{\zao}{\zeta_{A, \Omega}}
\newcommand{\tzao}{\widetilde{\zeta}_{A, \Omega}}
\newcommand{\tzad}{\widetilde{\zeta}_{A, D}}
\newcommand{\ptoo}{\partial \Omega, \Omega}
\newcommand{\res}{\operatorname{res}}

\newcommand{\vc}[3]{\overset{#2}{\underset{#3}{#1}}}

\makeatletter
\newcommand*{\hyperlinkcite}[1]{\hyper@link{cite}{cite.#1}}
\makeatother

\newtheorem{theorem}{Theorem}[section]

\newtheorem{corollary}[theorem]{Corollary}
\newtheorem{conjecture}[theorem]{Conjecture}

\theoremstyle{definition}
\newtheorem{definition}[theorem]{Definition}
\newtheorem{example}[theorem]{Example}
\newtheorem{exercise}[theorem]{Exercise}
\newtheorem{op}[theorem]{Open Problem}

\theoremstyle{remark}
\newtheorem{remark}[theorem]{Remark}

\numberwithin{equation}{section}

\begin{document}

\title[An Overview of Complex Fractal Dimensions]{An Overview of Complex Fractal Dimensions:\\ From Fractal Strings to Fractal Drums, and Back}



\author{Michel L. Lapidus}
\address{Department of Mathematics \\ University of California, Riverside \\ CA 92521\\ USA}
\curraddr{}
\email{lapidus@math.ucr.edu}
\thanks{Work partially supported by the US National Science Foundation under the research grants DMS-0707524 and DMS-1107750, as well as by the Burton Jones Endowed Chair Fund while he has been the holder of the F. Burton Jones Endowed Chair in Pure Mathematics at the University of California, Riverside, since July 2017.}


\date{}

\begin{abstract}
Our main goal in this long survey article is to provide an overview of the theory of complex fractal dimensions and of the associated geometric or fractal zeta functions, first in the case of fractal strings (one-dimensional drums with fractal boundary), in \S\ref{Sec:2}, and then in the higher-dimensional case of relative fractal drums and, in particular, of arbitrary bounded subsets of Euclidean space of $\mbr^N$, for any integer $N \geq 1$, in \S\ref{Sec:3}.

Special attention is paid to discussing a variety of examples illustrating the general theory rather than to providing complete statements of the results and their proofs, for which we refer to the author's previous (joint) books mentioned in the paper.

Finally, in an epilogue (\S\ref{Sec:4}), entitled ``From quantized number theory to fractal cohomology'', we briefly survey aspects of related work (motivated in part by the theory of complex fractal dimensions) of the author with H. Herichi (in the real case) \cite{HerLap1}, along with \cite{Lap8}, and with T. Cobler (in the complex case) \cite{CobLap1}, respectively, as well as in the latter part of a book in preparation by the author, \cite{Lap10}.
\end{abstract}

\maketitle

{\em To my daughter, Julie Anne Myriam Lapidus, who has recently embarked on a journey fraught with perilous and frightening dangers, the ultimate struggle for the continuation of her own beautiful and young life.

You are my scintillating treasure and will always remain my untarnished jewel. 

With all my love.\\}

{\em 2010 AMS Subject Classification}:\newline
{\parindent=0pt\small Primary -- 11M06, 11M41, 28A12, 28A80, 30D10, 30B50, 35P20, 42B20, 45Q05, 81Q20.\newline
Secondary -- 30D30, 37C30, 37C45, 40A10, 42B35, 42B37, 44A10, 45Q05, 81R40.
}
\newline

\smallskip
{\em Key words and phrases}:\newline
{\parindent=0pt\small Fractal strings and drums, relative fractal drums (RFDs), complex dimensions, fractal zeta functions, distance and tube zeta functions, relative fractal zeta functions, fractal tube formulas, Minkowski dimension and content, Minkowski measurability criteria, fractality, unreality, Riemann zeta function, Riemann hypothesis, direct and inverse problems for fractal strings, quantized number theory, infinitesimal shift, spectral operator, regularized determinants, fractal cohomology.

}

\tableofcontents

\section{Introduction}\label{Sec:1}

Our main goal in this research expository article is to provide an overview of the theory of complex dimensions and of the associated geometric or fractal zeta functions, first in the case of fractal strings (one-dimensional drums with fractal boundary) or essentially equivalently, for compact subsets of the real line (see \S\ref{Sec:2}), and then (see \S\ref{Sec:3}), in the higher-dimensional case of relative fractal drums and, in particular, of bounded subsets in Euclidean space $\mbr^N$, for any integer $N \geq 1$.

Special attention is paid to providing a variety of examples (especially, in higher dimensions) illustrating the general theory rather than to stating precise theorems in their greatest generality or providing their proofs (or even a sketch thereof). For a rigorous and quite detailed account of the theory, we refer instead the interested reader to the books by the author and Machiel van Frankenhuijsen, \cite{Lap-vF4} for the case of fractal strings (i.e., $N=1$) and the recent book by the author, Goran Radunovi\'c and Darko \v{Z}ubrini\'c, \cite{LapRaZu1}, when $N \geq 1$ is arbitrary.

In the epilogue (see \S\ref{Sec:4}), we survey related work, motivated in part by the theory of complex dimensions of fractal strings, on quantized number theory (in the real setting) and the Riemann hypothesis (joint with Hafedh Herichi; see the papers [\hyperlinkcite{HerLap2}{HerLap2--5}] and the forthcoming book \cite{HerLap1}, along with \cite{Lap8}), as well as on quantized number theory (in the complex setting), regularized determinants and ``fractal cohomology'' (see [\hyperlinkcite{CobLap1}{CobLap1--2}], joint with Tim Cobler, along with the author's book in preparation, \cite{Lap10}, of which the present paper is both the seed and a significantly condensed version).  

Except at the very end of this introduction, we now focus on \S\ref{Sec:2} and \S\ref{Sec:3} in the rest of \S\ref{Sec:1}. For a brief overview of the contents of \S\ref{Sec:4}, we refer the interested reader to the overall introduction of \S\ref{Sec:4} and to the beginning of each of the subsections of \S\ref{Sec:4} (namely, \S\S \ref{Sec:4.1}--\ref{Sec:4.4}), along with part ($iii$) towards the end of this section.

Complex dimensions provide a natural way to extract the information about the oscillatory nature of fractal objects. This is done via generalized explicit formulas (extending to this setting Riemann's original explicit formula for the prime number counting function, see \S \ref{Sec:2.3}).

In this paper, we have chosen to emphasize a type of explicit formulas called {\em fractal tube formulas} and which enable us to express the volume, $V(\varepsilon)$, of the \linebreak $\varepsilon$-neighborhoods of the given fractal set $A$ or, more generally, of the given relative fractal drum $(\ao)$ (RFD, in short)  in $\mbr^N$,\footnote{An RFD in $\mbr^N$ is a pair $(\ao)$, with $A \subseteq \mbr^N$, $\Omega$ open in $\mbr^N$ and $\Omega \subseteq A_{\delta_1}$, for some $\delta_1 > 0$, where for any $\varepsilon > 0$,
\begin{equation}\label{1.1/4}
A_\varepsilon := \{ x \in \mbr^N :d(x, A) < \varepsilon \}
\end{equation}
is the $\varepsilon$-neighborhood of $A$ and $d(x, A)$ denotes the Euclidean distance from $x \in \mbr^N$ to $A$. Also, we let $V(\varepsilon) = V_A (\varepsilon) := |A_\varepsilon|$ (respectively, $V(\varepsilon) = V_{\ao} (\varepsilon) := |A_\varepsilon \cap \Omega|$) in the case of a bounded set $A$ (respectively, RFD $(\ao)$) in $\mbr^N$. Here and thereafter, $|\cdot| = |\cdot|_N$ denotes the $N$-dimensional volume (or Lebesgue measure in $\mbr^N$).}  
as an extended power series in $\varepsilon$ with exponents the underlying complex codimensions and (normalized) coefficients (in the case of simple poles) the residues of the associated geometric or fractal zeta function.

For example, under suitable assumptions (and still in the case of simple poles), we have the following (pointwise or distributional) exact fractal tube formula for a given bounded set $A$ (or, more generally, RFD $(\ao)$) in $\mbr^N$ (up to a possible error term, which can be estimated explicitly):
\begin{equation}\label{1.1}
V(\varepsilon) = \sum_{\omega \in \mcd} c_\omega \frac{\varepsilon^{N-\omega}}{N- \omega},
\end{equation}
where $c_\omega := \res (\zeta_A, \omega)$ (or, more generally, $c_\omega := \res (\zeta_{\ao}, \omega)$) for each $\omega \in \mcd$, and where $\mcd = \mcd_A$ (respectively, $\mcd = \mcd_{\ao}$) denotes the set of complex dimensions of $A$ (respectively, of the RFD  $(\ao)$), viewed as a multiset (i.e., a set with finite integer multiplicities). Here, $\zeta_A$ (respectively, $\zeta_{\ao}$) is the {\em distance zeta function} of $A$ (respectively, of $(\ao)$) defined (in the case of the bounded set $A$, for example) for all $s \in \mbc$ with $Re(s)$ sufficiently large by the Lebesgue (and hence, absolutely convergent) integral
\begin{equation}\label{1.2}
\zeta_A (s) = \int_{A_\delta} d(x, A)^{s-N} dx,
\end{equation}
for some fixed $\delta > 0$ whose specific value is unimportant from the point of view of the theory of complex dimensions.

More specifically, \eqref{1.2} holds for all $s \in \mbc$ with $Re(s) > \overline{D}$, where $\overline{D}$ is the (upper) Minkowski dimension of $A$, and this lower bound is optimal. In other words, the abscissa of convergence of $\zeta_A$ coincides with $\overline{D}$; this is one of the first basic results of the theory (see part ($a$) of \S \ref{Sec:3.3.1}).

The (visible) {\em complex dimensions} of $A$ are defined as the poles of the meromorphic continuation (if it exists) of $\zeta_A$ to some given connected open neighborhood of the vertical line $\{ Re(s) = \overline{D} \}$ (or, equivalently, of the closed right half-plane  $\{ Re(s)  \geq \overline{D} \}$, since $\zeta_A$ is holomorphic on the open right half-plane $\{ Re(s) > \overline{D} \}$; see part ($b$) of \S \ref{Sec:3.3.1}).\footnote{Throughout this paper, we use the following short-hand notation: given $\alpha \in \mbr$, we let
\begin{equation}\label{1.3}
\{ Re(s) \geq \alpha \} := \{ s \in \mbc: Re(s) \geq \alpha \}
\end{equation}
denote the closed right half-plane with abscissa $\alpha$; and analogously for the vertical line $\{Re(s) = \alpha \}$ or the open right half-plane $\{Re(s) > \alpha \}$ with abscissa $\alpha$, say. (If $\alpha = \pm \infty$, we adopt the obvious conventions $\{Re(s) \geq +\infty \} = \emptyset$ and $\{Re(s) \geq -\infty \} = \mbc$, for example.)}

In particular, if $\overline{D}$ itself is a pole of $\zeta_A$ (under mild conditions, it is always a nonremovable singularity of $\zeta_A$), then it is a complex dimension having the largest possible real part.

Provided $\overline{D} < N$ (i.e., if $\overline{D} \neq N$ since we always have $\overline{D} \leq N$), all of the above results and definitions  extend to another useful fractal zeta function, called the {\em tube zeta function} of $A$ and denoted by $\widetilde{\zeta}_A$.\footnote{For $Re(s)$ sufficiently large (in fact, precisely for $Re(s) > \overline{D}$, provided $\overline{D} < N), \widetilde{\zeta}_A$ is given by the Lebesgue (and hence, absolutely convergent) integral
\begin{equation}\label{1.3.1/2}
\widetilde{\zeta}_A (s) := \int_0^\delta V(\varepsilon) \varepsilon^{s-N} \frac{d \varepsilon}{\varepsilon},
\end{equation}
for some arbitrary but fixed $\delta > 0$, the value of which is unimportant from the point of view of the definition (and the values) of the complex dimensions.}
The fractal zeta functions $\zeta_A$ and $\widetilde{\zeta}_A$ are connected via a functional equation (see \eqref{3.11}), which implies that the (visible) complex dimensions of $A$ can be defined indifferently via either $\zeta_A$ or $\widetilde{\zeta}_A$. Furthermore, the fractal tube formula \eqref{1.1} has a simple counterpart expressed in terms of the residues of $\widetilde{\zeta}_A$ (instead of those of $\zeta_A$) evaluated at the complex dimensions of $A$. Namely, up to a possible error term which can be estimated explicitly,
\begin{equation}\label{1.4}
V (\varepsilon) = \sum_{\omega \in \mcd} d_\omega \varepsilon^{N - \omega},
\end{equation}
where $d_\omega:= \res (\widetilde{\zeta}_A, \omega)$ for each $\omega \in \mcd$.

All of the above results (including the fractal tube formulas \eqref{1.1} and \eqref{1.4}) extend to relative fractal drums (RFDs) in $\mbr^N$ (with $\zeta_A, \widetilde{\zeta}_A$ and $\mcd = \mcd_A$ replaced by $\zao, \tzao$ and $\mcd = \mcd_{\ao}$, respectively), which are very useful tools in their own right and enable us, in particular, to compute (by using appropriate decompositions and symmetry considerations) the fractal zeta functions and complex dimensions of many fractal compact subsets of $\mbr^N$.

At this stage, it is helpful to point out that many key results of the theory of complex dimensions of fractal strings \cite{Lap-vF4} (briefly discussed in \S\ref{Sec:2}), including the fractal tube formulas (of which \eqref{2.10} is a typical example), can be recovered by specializing the higher-dimensional theory of complex dimensions of RFDs in $\mbr^N$ to the $N=1$ case and by viewing fractal strings as RFDs in $\mbr$. In the process, a simple functional equation connecting the so-called geometric zeta function of a fractal string (described in the beginning of \S\ref{Sec:2}) and the distance zeta function of the associated RFD plays a key role; see \S \ref{Sec:3.2.2} and \S \ref{Sec:3.5.1}.

Intuitively, a fractal, viewed as a geometric object, is like a musical instrument tuned to play certain notes with frequencies (respectively, amplitudes) essentially equal to the real parts (respectively, the imaginary parts) of the underlying complex dimensions. Alternatively, one can think of a ``geometric wave'' propagating through the fractal and with the aforementioned frequencies and amplitudes. This ``physical'' intuition is corroborated, for example, by the fractal tube formulas expressed via the distance (respectively, tube) zeta function, as in \eqref{1.1} (respectively, \eqref{1.4}).

As was mentioned just above, the theory of complex dimensions of fractal strings can be viewed essentially as the one-dimensional special case of the general theory of complex dimensions (of RFDs in $\mbr^N$) developed in \cite{LapRaZu1}. Conversely, fractal string theory has provided the author and his collaborators with a broad and rich collection of examples with which to test various conjectures and formulate various definitions as well as elaborate tools that could eventually be used in more complicated higher-dimensional situations. Also, several of the key steps towards the proof of the higher-dimensional fractal tube formulas (such as in \eqref{1.1} and \eqref{1.4}) rely, in part, on techniques developed for dealing with the case of (generalized) fractal strings \cite{Lap-vF2, Lap-vF3, Lap-vF4}.

In addition, ``fractality'' is characterized (or rather, {\em defined}) in our general theory by the presence of {\em nonreal} complex dimensions.\footnote{It is also very useful to extend the notion of ``complex dimensions'' by allowing more general (nonremovable) singularities than poles of the associated fractal zeta functions; see \cite{LapRaZu1} and [\hyperlinkcite{LapRaZu6}{LapRaZu6--7, 10}], along with \S \ref{Sec:2.5}, \S \ref{Sec:3.5.2}, \S \ref{Sec:3.6} and \S \ref{Sec:4.4}.}
This extends to any dimension $N \geq 1$ the definition of fractality given earlier in [\hyperlinkcite{Lap-vF2}{Lap-vF2--4}], thanks to the fact that we now have to our disposal a general definition of fractal zeta functions valid for arbitrary bounded (or, equivalently, compact) subsets of $\mbr^N$ (as well as, more generally, for all RFDs in $\mbr^N$). 

We will also discuss (in \S \ref{Sec:3.5.2} when $N \geq 1$ is arbitrary, and in Theorem \ref{Thm:2.2} when $N=1$) a general Minkowski measurability criterion expressed in terms of complex dimensions. Namely, under certain mild conditions (which imply that the Minkowski dimension $D$ exists and is a complex dimension), a bounded set $A$ (or, more generally, an RFD $(\ao)$) in $\mbr^N$ is Minkowski measurable\footnote{Intuitively, Minkowski measurability is some kind of ``fractal regularity'' of the underlying geometry; for a precise definition, see \S \ref{Sec:3.2} when $N \geq 1$ is arbitrary (or \S \ref{Sec:2.1} when $N=1$).}
if and only if its only complex dimension with real part $D$ (i.e., its only {\em principal complex dimension}) is $D$ itself and $D$ is simple. In other words, the existence of nonreal complex dimensions (i.e., the ``{\em critical fractality}'' of $A$ or of $(\ao)$; see \S \ref{Sec:3.6}), along with the simplicity of $D$ (as a pole of $\zeta_A$ or equivalently, of $\zeta_{\ao}$), characterizes the Minkowski {\em non}measurability of $A$ (or of $(\ao)$). As a simple illustration, the Cantor set, the Cantor string, the Sierpinski gasket and the Sierpinski carpet, along with lattice self-similar strings (and more generally, in higher dimensions, lattice self-similar sprays with ``sufficiently nice'' generators), are all Minkowski nonmeasurable but are Minkowski nondegenerate; see \S \ref{Sec:3.5.3}. On the other hand, a ``generic'' (i.e., nonlattice) self-similar Cantor-type set (or string) or a ``generic'' self-similar carpet is Minkowski measurable (because it does not have any nonreal complex dimensions other than $D$, which is simple).\\

Beside this introduction (i.e., \S\ref{Sec:1}), this paper is divided into three parts:\\

{\bf (\emph{i})} \S\ref{Sec:2}, a brief account of the theory of complex dimensions for fractal strings ($N=1$) \cite{Lap-vF4} and its prehistory, including a discussion (in \S \ref{Sec:2.6}) of natural direct and inverse spectral problems for fractal strings along with their intimate connections with the Riemann zeta function \cite{LapPo2} and the Riemann hypothesis \cite{LapMa2}.\\

{\bf (\emph{ii})} \S\ref{Sec:3}, an introduction to the higher-dimensional theory of complex dimensions and the associated fractal zeta functions (namely, the distance and tube zeta functions), based on \cite{LapRaZu1} (and aspects of [\hyperlinkcite{LapRaZu2}{LapRaZu2--9}]), with emphasis on several key examples of bounded sets and relative fractal drums in $\mbr^N$ (with $N=2, N =3$ or $N \geq 1$ arbitrary) illustrating the key concepts of fractal zeta functions and their poles or, more generally, nonremovable singularities (i.e., the complex dimensions), as well as the associated fractal tube formulas. As was alluded to earlier, the latter explicit formulas provide a concrete justification of the use of the phrase ``complex fractal dimensions'' and help explain why both intuitively and in actuality, the theory of complex dimensions is a theory of oscillations that are intrinsic to fractal geometries. 

It is noteworthy that even though we will mostly stress the aforementioned geometric oscillations in \S\ref{Sec:3} (and in much of \S\ref{Sec:2}), the broad definition of ``fractality'' proposed in \S \ref{Sec:3.6} and expressed in terms of the presence of nonreal complex dimensions encompasses oscillations that are intrinsic to number theories (via Riemann-type explicit formulas expressed in terms of the poles and the zeros of attached $L$-functions, or equivalently, in terms of the poles of the logarithmic derivatives of those $L$-functions; see \S \ref{Sec:2.3} for the original example), or to dynamical systems (e.g., via explicit formulas for the counting functions of primitive periodic orbits; see \cite[Ch. 7]{Lap-vF4} for a class of examples), as well as to the spectra of fractal drums [both ``drums with fractal boundary'' (as, e.g., in [\hyperlinkcite{Lap1}{Lap1--3}] and parts of \cite{Lap-vF4}) and ``drums with fractal membrane'' (as, e.g., in \cite{Lap3}, \cite{KiLap1} and \cite{Lap7}) and other classical or quantum physical systems (via detailed spectral asymptotics or, essentially equivalently, via explicit formulas for the associated frequency or eigenvalue counting functions). 

Much remains to be done in all of these directions for a variety of specific classes of dynamical systems and of fractal drums, for example. We point out, however, that the deep analogy between many aspects of fractal geometry and number theory (see, e.g., [\hyperlinkcite{Lap-vF1}{Lap-vF1--5}], \cite{Lap7}, \cite{HerLap1} and \cite{LapRaZu1}) was a key motivation for the author to want to develop (since the mid-1990s) a theory of ``fractal cohomology'', itself an important motivation for many aspects of the work described in \S\ref{Sec:4}. \\

{\bf (\emph{iii})} \S\ref{Sec:4}, the epilogue, a very brief account (compared to the size of the corresponding material to be described) of ``quantized number theory'', both in the ``real case'' (\S \ref{Sec:4.2}, based on \cite{HerLap1},  [\hyperlinkcite{HerLap2}{HerLap2--5}] and \cite{Lap8}) and in the ``complex case'' (\S \ref{Sec:4.3}, based principally on [\hyperlinkcite{CobLap1}{CobLap1--2}] and  on aspects of \cite{Lap10}) and the associated fractal cohomology (\S \ref{Sec:4.1} and, especially, \S \ref{Sec:4.4}, as expanded upon in \cite{Lap10}), with applications to several reformulations of the Riemann hypothesis (\S \ref{Sec:4.2}) expressed in terms of the ``quasi-invertibility'' \cite{HerLap1} or the invertibility \cite{Lap8} of so-called ``spectral operators'', in particular, as well as to the representations (\S \ref{Sec:4.3}) of various arithmetic (or number-theoretic) $L$-functions and other meromorphic functions (including the completed Riemann zeta function and the Weil zeta functions attached to varieties over finite fields [\hyperlinkcite{Wei1}{Wei1--6}, \hyperlinkcite{Gro1}{Gro1--4}, \hyperlinkcite{Den1}{Den1--6}], and, e.g., \cite{Mani}, \cite{Kah}, [\hyperlinkcite{Tha1}{Tha1--2}]) via (graded  or supersymmetric) regularized (typically infinite dimensional) determinants of suitable unbounded linear operators (the so-called ``generalized Polya--Hilbert operators'') restricted to their eigenspaces (which are the proposed ``fractal cohomology spaces'').\\

These developments open-up a vast and very rich new domain of research, extending in a variety of directions and located at the intersection of many fields of mathematics, including fractal geometry, number theory and arithmetic geometry, mathematical physics, dynamical systems, harmonic analysis and spectral theory, complex analysis and geometry, geometric measure theory, as well as algebraic geometry and topology, to name a few. 

We hope that the reader will be stimulated by the reading of this expository article (and eventually, of its much expanded sequel, the author's book in preparation \cite{Lap10}) to explore the various ramifications and consequences of the theory,  many of which are yet to be discovered. In other words, instead of offering here a complete and closed theory, we prefer to (and, in fact, must) offer here (especially, in \S \ref{Sec:4}) only glimpses of a possible future unifying and ``universal'' theory, resting on the contributions and conjectures or dreams of many past and contemporary mathematicians and physicists. 

\section{Fractal Strings and Their Complex Dimensions}\label{Sec:2}
A ({\em bounded}) {\em fractal string} can be viewed either as a bounded open set $\Omega \subseteq \mbr$ or else as a nonincreasing sequence of lengths (or positive numbers) $\mcl = (\ell_j)_{j=1}^\infty$ such that $\ell_j \downarrow 0$. (The latter condition is not needed if the sequence $(\ell_j)$ is finite.)

Let us briefly explain the connection between these two points of view. If $\Omega$ is a bounded open subset of $\mbr$, we can write $\Omega = \cup_{j \geq 1} \, I_j$ as an at most countable disjoint union of bounded open intervals $I_j$, of length $\ell_j > 0$. These intervals are nothing but the connected components of $\Omega$. Since $|\Omega|_1 = \sum_{j \geq 1} \, \ell_j < \infty$ (i.e., $\Omega$ has finite total length), without loss of generality, we may assume (possibly after having reshuffled the intervals $I_j$, that $\ell_1 \geq \ell_2 \geq \cdots$ (counting multiplicities) and (provided the sequence $(\ell_j)_{j \geq 1}$ is infinite) $\ell_j \downarrow 0$.

Slightly more generally, in the definition of a {\em bounded fractal string}, one can assume that instead of being bounded, the open set $\Omega \subseteq \mbr$ has finite volume (i.e., length): $|\Omega|_1 < \infty.$

Unbounded fractal strings (i.e., strings  $\mcl = (\ell_j)_{j \geq 1}$ such that $\sum_{j \geq 1} \ell_j = +\infty$) also play an important role in the theory (see, e.g., \cite[Ch. 3 and parts of Chs. 9--11 along with \S 13.1]{Lap-vF4}) but from now on, unless explicitly mentioned otherwise, we will assume that all of the (geometric) fractal strings under consideration are bounded. As a result, we will often drop the adjective ``bounded'' when referring to fractal strings.

From a physical point of view, the `lengths' $\ell_j$ can also be thought of as being the underlying {\em scales} of the system. This is especially useful in the case of unbounded fractal strings but should also be kept in mind in the geometric situation of fractal strings.

A {\em geometric realization} of a (bounded) fractal string $\mcl = (\ell_j)_{j \geq 1}$ is any bounded open set $\Omega$ in $\mbr$ (or, more generally, any open  set $\Omega$ in $\mbr$ of finite length) with associated length sequence $\mcl$.

The {\em geometric zeta function} $\zeta_\mcl$ of a fractal string $\mcl = (\ell_j)_{j \geq 1}$ is defined by 
\begin{equation}\label{2.1}
\zeta_\mcl (s) = \sum_{j \geq 1} \ell_j^s,
\end{equation}
for all $s \in \mbc$ with $Re(s)$ sufficiently large. (Here and henceforth, we let $\ell_j^s := (\ell_j)^s$, for each $j \geq 1$.)

A simple example of a fractal string is the {\em Cantor string}, denoted by $\Omega_{CS}$ (or $\mcl_{CS}$) and defined by $\Omega_{CS} = [0,1] \backslash C$, the complement of the (classic ternary) Cantor set $C$ in the unit interval. (Observe that the boundary of the Cantor string is the Cantor set itself: $\partial \mcl_{CS}:= \partial \Omega_{CS} = C$.) 
Then, $\Omega_{CS}$ consists of the disjoint union of the deleted (open) intervals, in the usual construction of the Cantor set:
\begin{equation}\label{2.2}
\Omega_{CS} = (0, 1/3) \cup (1/9, 2/9) \cup (7/9, 8/9) \cup \cdots.
\end{equation}
Hence, the associated sequence of lengths $\mcl_{CS}$ is given by
\begin{equation}\label{2.3}
1/3, 1/9, 1/9, 1/27, 1/27, 1/27, 1/27, \cdots;
\end{equation}
equivalently, $\mcl_{CS}$ consists of the lengths $1/3^n$ counted with multiplicity $2^{n-1}$, for $n=1,2,3, \cdots$. It follows that $\zeta_{CS}$ can be computed by simply evaluating the following geometric series:
\begin{align}\notag
\zeta_{CS} (s) &= \sum_{n=1}^\infty 2^{n-1} (3^{-n})^s
= 3^{-s} \sum_{n=0}^\infty (2 \cdot 3^{-s})^n \\ &= \frac{3^{-s}}{1-2 \cdot 3^{-s}} = \frac{1}{3^s - 2}.\notag
\end{align}
This calculation is valid for $Re(s) > \log_3 2$ (i.e., $|2 \cdot 3^{-s}| < 1$) but upon analytic continuation, we see that $\zeta_{CS}$ admits a (necessarily unique) meromorphic continuation to all of $\mbc$ (still denoted by $\zeta_{CS}$, as usual) and that 
\begin{equation}\label{2.4}
\zeta_{CS} (s) = \frac{1}{3^s -2}, \text{ for all } s \in \mbc.
\end{equation}

The {\em complex dimensions} of the Cantor string are the poles of $\zeta_{CS}$; that is, here, the complex solutions of the equation $3^s -2 = 0$. Thus, the set $\mcd_{CS}$ of complex dimensions of $\mcl_{CS}$ is given by a single (discrete) vertical line,
\begin{equation}\label{2.5}
\mcd_{CS} = \{ D + in{\bf p}: n \in \mbz \},
\end{equation}
where $D:= D_{CS} = \log_3 2$ is the Minkowski dimension of the Cantor string (or of the Cantor set) and ${\bf p} := 2\pi/\log 3$ is its oscillatory period. (The definition of $D$ is recalled in \eqref{2.7}.)\footnote{In the case of the Cantor string (or set), the Minkowski dimension exists and hence, there is not need to talk about (upper) Minkowski dimension; see \S \ref{Sec:3.2} for the precise definitions.}

For an arbitrary fractal string $\mcl$, the {\em complex dimensions} of $\mcl$ (relative to a given domain $U \subseteq \mbc$ to which $\zeta_\mcl$ admits a meromorphic continuation), also called the {\em visible complex dimensions} of $\mcl$, are simply the poles of $\zeta_\mcl$ which lie in $U$. Thus, for the Cantor string, $\mcd_{CS} = \mcd_{CS} (\mbc)$ is given by \eqref{2.5}.

Recall that the {\em abscissa of convergence}\label{gls:ac10} $\alpha = \alpha_\mcl$ of the Dirichlet series defining $\zeta_\mcl$ in \eqref{2.1} is given by
\begin{equation}\label{2.6}
\alpha := \inf \Big\{ \beta \in \mbr: \sum\nolimits_{j \geq 1} \ell_j^\beta < \infty \Big\};
\end{equation}
so that $\alpha$ is the unique real number such that  $\sum_{j \geq 1} \ell_j^s$ converges absolutely for $Re(s) > \alpha$ but diverges for $Re(s) < \alpha$.

\begin{theorem}[\protect{Abscissa of convergence and Minkowski dimension; \cite{Lap2, Lap3}, \cite[Thm. 1.10]{Lap-vF4}}]\label{Thm:2.1}
Let $\mcl$ be an arbitrary bounded fractal string $\mcl$ having infinitely many lengths. $($When $\mcl$ has finitely many lengths, it is immediate to check that $\zeta_\mcl$ is entire and hence, $\alpha = -\infty$ while $D = 0.)$ Then $\alpha = D$, the $($upper$)$ Minkowski dimension of $\mcl$ $($i.e., of $\partial \Omega$, where the bounded open set $\Omega$ is any geometric realization of $\mcl);$ see, respectively, \eqref{2.6} and \eqref{2.7} for the definition of $\alpha$ and $D$. In other words, the abscissa of convergence of $\zeta_\mcl$ and the Minkowski dimension of $\mcl$ coincide.\footnote{In \cite{Lap2, Lap3}, the proof of this equality relied on a result obtained in \cite{BesTa}. Then, several direct proofs were given in [\hyperlinkcite{Lap-vF2}{Lap-vF2--4}]. See, especially, \cite[Thm. 1.10 and Thm. 13.111]{Lap-vF4}; see also \cite{LapLu-vF2}) and most recently, in \cite[\S 2.1.4, esp. Prop. 2.1.59 and Cor. 2.1.61]{LapRaZu1}, via the higher-dimensional theory of complex dimensions (to be discussed in \S\ref{Sec:3}).}
\end{theorem}

More precisely, here, the ({\em upper}) {\em Minkowski dimension} $D = D_\mcl$ of $\mcl$ is the nonnegative real number given by\footnote{The Minkowski dimension is also called the Minkowski--Bouligand dimension \cite{Bou}, the box dimension or the capacity dimension in the literature on fractal geometry; see, e.g., \cite{Man}, \cite{Fa1}, \cite{MartVuo}, \cite{Mat}, [\hyperlinkcite{Tri1}{Tri1--3}], [\hyperlinkcite{Lap1}{Lap1--3}], \cite{Lap-vF4}, \cite{LapRaZu1} and \cite{LapRaRo}.}
\begin{equation}\label{2.7}
D:= \inf \{ \beta \geq 0: V(\varepsilon) = O (\varepsilon^{1-\beta}) \quad \text{as } \varepsilon \rightarrow 0^+ \},
\end{equation}
where
\begin{equation}\label{2.8}
V (\varepsilon) = V_\mcl (x) := | \{x \in \Omega : d (x, \partial \Omega) < \varepsilon \} |_1
\end{equation}
is the volume (or length) of the $\varepsilon$-neighborhood of the boundary $\partial \Omega$ (relative to $\Omega$) and $d (x, \partial \Omega)$ denotes the distance (in $\mbr$) from $x$ to $\partial \Omega$.\footnote{In the present section (i.e., \S\ref{Sec:2}), for the simplicity of exposition, we will mostly ignore the distinction between upper Minkowski dimension and Minkowski dimension of $\mcl$. By contrast, in \S\ref{Sec:3}, we will denote, respectively, by $\overline{D}$ and $D$ these two dimensions (when the latter exists); see \S \ref{Sec:3.2} for the precise definitions. Note that in the terminology of \S\ref{Sec:3}, the notion introduced in \eqref{2.7} is that of upper Minkowski dimension of the bounded fractal string $\mcl = (\ell_j)_{j \geq 1}$, viewed as the relative fractal drum (or RFD) $(\partial \Omega, \Omega)$ in $\mbr$, where $\Omega$ is any geometric realization of $\mcl$.}

It follows at once from Theorem \ref{Thm:2.1}, along with the definition of $\alpha$ and $D$ respectively given in \eqref{2.6} and \eqref{2.7}, that for a fractal string, we have $0 \leq D \leq 1$.

For example, for the Cantor string, the computation leading to \eqref{2.4} shows that $\alpha = \log_3 2$ and it is well known that $D = \log_3 2$, in agreement with Theorem \ref{Thm:2.1}.

It is clear that the set $\mcd_\mcl$ of complex dimensions of a fractal string forms a discrete (and hence, at most countable) subset of $\mbc$ and (in light of Theorem \ref{Thm:2.1}, since $\zeta_\mcl$ is holomorphic for $Re(s) > D$)
\begin{equation}\notag
\mcd_\mcl \subseteq \{ Re(s) \leq D \},
\end{equation} 
where we use the short-hand notation 
\begin{equation}\notag
\{ Re(s) \leq D \} := \{ s \in \mbc : Re(s) \leq D \},
\end{equation}
here and henceforth. (Similarly, for example, the notation $\{ Re(s) = D \}$ stands for the vertical line $\{s \in \mbc: Re(s) = D \}.$)

The set of {\em principal complex dimensions} of $\mcl$, denoted by $\dim_{PC} \mcl$, is the set of complex dimensions with maximal real part:
\begin{equation}\label{2.9}
\dim_{PC} \mcl = \{ \omega \in \mcd_\mcl : Re(\omega) = D \}.
\end{equation}
This set (or rather, multiset) plays an important role in the general theory of complex fractal dimensions. The same is true for its counterpart in the higher-dimensional theory, to be discussed in \S\ref{Sec:3}.

For the Cantor string, in light of \eqref{2.5}, we clearly have $\mcd_\mcl = \dim_{PC} \mcl$ but in general, $\dim_{PC} \mcl$ is often a strict subset of $\mcd_\mcl = \mcd_\mcl (U)$. (We implicitly assume here and in \eqref{2.9} that the connected open set $U$ is a neighborhood of the vertical line $\{ Re(s) = D \}$, or equivalently, of the closed half-plane $\{ Re(s) \geq D \}$; observe, however, that the set $\dim_{PC} \mcl$ itself is independent of such a choice of $U$.) 

We note for later use that since $\zeta_\mcl$ is a Dirichlet series with positive coefficients, $\zeta_\mcl (s) \rightarrow + \infty$ as $s \rightarrow D^+, s \in \mbr$ (or, more generally, as $s \in \mbc$ tends to $D$ from the right within a sector of half-angle $< \pi/2$ and symmetric with respect to the real axis); see, e.g., \cite{Ser} or \cite[\S 1.2]{Lap-vF4}. It follows that for a fractal string (with infinitely many lengths), the half-plane $\{Re(s) > D \}$ of absolute convergence of $\zeta_\mcl$ always coincides with the half-plane of holomorphic continuation of $\zeta_\mcl,$ i.e., the maximal open right half-plane to which $\zeta_\mcl$ can be holomorphically continued. (See \cite{Lap-vF4}.) Hence, in the terminology and with the notation of \cite{LapRaZu1} (to be introduced in \S \ref{Sec:3.3}), we have that $D = D_{hol} (\zeta_\mcl)$, the abscissa of holomorphic continuation of $\zeta_\mcl$. 

Observe that since $D$  is always a singularity of $\zeta_\mcl$, then, provided $\zeta_\mcl$ can be meromorphically continued to a neighborhood of $D$, $D$ must necessarily be a pole of $\zeta_\mcl$ (i.e., a complex dimension of $\mcl)$.

\subsection{Fractal tube formulas}\label{Sec:2.1}
Given a fractal string $\mcl$, under suitable hypotheses,\footnote{Namely, we assume that $\mcl$ is {\em languid} in a suitable connected open neighborhood  $U$ of $\{ Re(s) \geq D \}$; i.e., roughly speaking, $\zeta_\mcl$ can be meromorphically continued to $U$ and satisfies a suitable polynomial growth condition for a screen $S$ bounding $U$ (in the sense of \cite[\S 5.3]{Lap-vF4}).\label{Fn:8}} 
we can express its {\em tube function} $V (\varepsilon) = V_\mcl(\varepsilon)$ (or rather $\varepsilon \mapsto V(\varepsilon)$), as given by \eqref{2.8}, in terms of its complex dimensions and the associated residues, as follows:
\begin{equation}\label{2.10}
V(\varepsilon) = \sum_{\omega \in \mcd_\mcl} c_\omega \, \frac{(2 \varepsilon)^{1-\omega}}{\omega (1-\omega)} + R(\varepsilon),
\end{equation}
where $c_\omega := \res (\zeta_\mcl, \omega)$ is the residue of $\zeta_\mcl$ at $\omega \in \mcd_\mcl$ and $R(\varepsilon)$ is an error term which can be explicitly estimated.\footnote{In this discussion, for clarity, we assume implicitly that all of the complex dimensions are simple (i.e., are simple poles of $\zeta_\mcl$). In the general case, \eqref{2.10} should be replaced by
\begin{equation}\label{2.10.1/2}
V(\varepsilon) = \sum_{\omega \in \mcd_\mcl} \res \bigg( \frac{(2 \varepsilon)^{1-s}}{s (1-s)} \, \zeta_\mcl (s), \omega \bigg) + R(\varepsilon).
\end{equation}
}
If $R(\varepsilon) \equiv 0$ (which occurs, for example, for any self-similar string if we choose $U := \mbc$), the corresponding {\em fractal tube formula} \eqref{2.10} is said to be {\em exact}.\footnote{More generally, we obtain an exact tube formula whenever $\mcl$ (i.e., $\zeta_\mcl$) is {\em strongly languid} (which implies that $U := \mbc$), in the sense of \cite[\S 5.3]{Lap-vF4}.}

In \cite[Ch. 8]{Lap-vF4}, the interested reader can find the precise statement and hypotheses of the fractal tube formula. In fact, depending, in particular, on the growth assumptions made on the geometric zeta function $\zeta_\mcl$, there are a variety of fractal tube formulas, with or without error term (the latter ones being called {\em exact}), as well as interpreted {\em pointwise} or {\em distributionally}; see \cite[\S 8.1]{Lap-vF4}.

Furthermore, in the important special case of self-similar strings (of which the Cantor string is an example), even more precise (pointwise) fractal tube formulas (exact or else with an {\em error term}, depending on the goal being pursued) are obtained in \cite[\S 8.4]{Lap-vF4}.

For the example of the Cantor string (which is a self-similar string because its boundary, the ternary Cantor set, is itself a self-similar set in $\mbr$), we have the following {\em exact fractal tube formula}, valid {\em pointwise} for all $\varepsilon \in (0, 1/2)$:
\begin{equation}\label{2.11}
V_{CS} (\varepsilon) = \frac{1}{2 \log 3} \, \sum_{n \in \mbz} \, \frac{(2 \varepsilon)^{1 - D - in{\bf p}}}{(D + in{\bf p})(1-D - in{\bf p})} - 2 \varepsilon,
\end{equation} 
with $D:= \log_3 2$ and ${\bf p} := 2\pi/ \log 3$. 

Observe that we can rewrite \eqref{2.11} in the following form:
\begin{equation}\label{2.11.1/2}
V_{CS} (\varepsilon) = \varepsilon^{1-D} \, G (\log_3 \varepsilon^{-1}) -2 \varepsilon, 
\end{equation}
where $G$ is a  nonconstant, positive $1$-periodic function on $\mbr$ which is bounded away from zero and infinity. In fact,
\begin{equation}\notag
0 < \mcm_* = \min_{u \in \mbr} \, G (u) \text{ and } \mcm^* = \max_{u \in \mbr} \, G(u) < \infty,
\end{equation}
where $\mcm_*$ and $\mcm^*$ denote, respectively, the {\em lower} and {\em upper Minkowski content} of $\mcl$, defined by\footnote{An entirely analogous definition of $\mcm_*$ and $\mcm^*$ can be given for any fractal string $\mcl$; simply replace $V_{CS} (\varepsilon)$ by $V(\varepsilon) = V_\mcl (\varepsilon)$ in \eqref{2.12} and \eqref{2.13}, respectively.} 
\begin{equation}\label{2.12}
\mcm_* := \liminf_{\varepsilon \rightarrow 0^+} \varepsilon^{-(1-D)} V_{CS} (\varepsilon)
\end{equation}
and
\begin{equation}\label{2.13}
\mcm^* := \limsup_{\varepsilon \rightarrow 0^+} \varepsilon^{-(1-D)} V_{CS} (\varepsilon).
\end{equation}
(Clearly, we have that $0 \leq \mcm_* \leq \mcm^* \leq \infty$.) For the Cantor string,
\begin{equation}\notag
\mcm_* = 2^{1-D} \, D^{-D} \approx 2.4950 \text{ and } \mcm^* = 2^{2-D} \approx 2.5830.
\end{equation}
Hence, $\mcm_* < \mcm^*$ and thus, the limit of $V(\varepsilon)/\varepsilon^{1-D}$ as $\varepsilon \rightarrow0^+$ does not exist; i.e., the Cantor string (and hence, also the Cantor set) is {\em not} Minkowski measurable.\footnote{This fact was first established in [\hyperlinkcite{LapPo1}{LapPo1--2}], by using a direct computation and Theorem \ref{Thm:2.3}, and then extended in [\hyperlinkcite{Lap-vF2}{Lap-vF2--4}] to a whole class of examples (including lattice self-similar strings and generalized Cantor strings; see \cite[\S 8.4.2 and \S 10.1]{Lap-vF4}). Another, more conceptual, proof was given in \cite[Ch. 8]{Lap-vF4} by using the existence of nonreal principal complex dimensions of the Cantor string; see Theorem \ref{Thm:2.2} and the comments following it.}

Recall that a fractal string $\mcl$ (or its boundary $\partial \Omega$) is said to be {\em Minkowski measurable} if the above limit exists in $(0, +\infty)$ and then, 
\begin{equation}\label{2.14}
\mcm := \lim_{\varepsilon \rightarrow 0^+} \varepsilon^{-(1-D)} \, V(\varepsilon)
\end{equation}
is called the {\em Minkowski content} of $\mcl$ (or of $\partial \Omega$). In other words, $\mcl$ is Minkowski measurable if $\mcm_* = \mcm^*$, and this common value, denoted by $\mcm$, lies in $(0, +\infty).$

There is another way to show that $\mcl$ (in the present case, $\mcl =\mcl_{CS}$, the Cantor string) is not Minkowski measurable. This can be seen by using the principal complex dimensions of $\mcl$, as defined by \eqref{2.9}; in other words, the complex dimensions with maximal real part $D$. Indeed, the following useful Minkowski measurability criterion was obtained in [\hyperlinkcite{Lap-vF1}{\bf{Lap-vF1--4}}].

\begin{theorem}[\protect{Minkowski measurability and complex dimensions; \cite[Thm. 8.15]{Lap-vF4}}]\label{Thm:2.2} 
Under suitable hypotheses,\footnote{In essence, we assume that $\mcl$ is languid (in the sense of footnote \ref{Fn:8}) for a screen $S$ passing between the vertical line $\{ Re(s) = D \}$ and all of the complex dimensions of $\mcl$ with real part $ < D$.} the following statements are equivalent{\em :}
\begin{itemize}
\item [(i)] $\mcl$ is Minkowski measurable $($with Minkowski dimension $D \in (0,1)).$
\item [(ii)] The only principal complex dimension of $\mcl$ is $D$ itself, and it is simple. 
\end{itemize}
\end{theorem}

Observe that for the Cantor string $\mcl_{CS}$, there are infinitely many complex conjugate nonreal complex dimensions with real part $D$. Furthermore, $D = \log_3 2$ (like each of the complex dimensions of $\mcl_{CS}$ in \eqref{2.5}) is simple; i.e., it is a simple pole of $\zeta_{CS}$. Therefore, this yields another proof of the fact that $\mcl_{CS}$ (or, equivalently, the Cantor set $C$) is not Minkowski measurable.

There is another, very useful, {\em characterization of Minkowski measurability}, obtained in \cite{LapPo2} and announced in \cite{LapPo1}.

\begin{theorem}[Minkowski measurability and fractal strings; \cite{LapPo2}]\label{Thm:2.3}
Let $\mcl= (\ell_j)_{j=1}^\infty$ be an arbitrary fractal string $($of Minkowski dimension $D \in (0,1))$. Then, the following statements are equivalent{\em :}
\begin{itemize}
\item [(i)] $\mcl$ is Minkowski measurable.
\item [(ii)] $\ell_j \sim Lj^{-1/D}$ \textnormal{as} $j \rightarrow \infty$, for some constant $L \in (0, +\infty)$.\footnote{Here, $\ell_j \sim m_j$ as $j \rightarrow \infty$ means that $\ell_j/m_j \rightarrow 1$ as $j \rightarrow \infty.$}
\end{itemize}
In this case, the Minkowski content $\mcm$ of $\mcl$ is given by 
\begin{equation}\label{2.16}
\mcm = \frac{2^{1-D}}{1-D} L^D.
\end{equation}
\end{theorem}

\begin{remark}\label{Rem:2.4}
($a$) The proof of Theorem \ref{Thm:2.3} given in \cite{LapPo2} is analytical. Later, a different approach to a part of that proof was taken by Kenneth Falconer in \cite{Fa2}, based on a suitable dynamical system, and more recently, by Jan Rataj and Steffen Winter in \cite{RatWi}, based on aspects of geometric measure theory.

($b$) If $\zeta_\mcl$ has a meromorphic continuation to a neighborhood of $D$ and either condition (i) or (ii) of Theorem \ref{Thm:2.3} is satisfied (or certainly, if the hypotheses and either condition ($i$) or ($ii$) of Theorem \ref{Thm:2.2} hold), then 
\begin{equation}\label{2.17}
\mcm = \frac{2^{1-D}}{D(1-D)} \, \res (\zeta_\mcl, D).
\end{equation}

($c$) Even though the Cantor string $\mcl_{CS}$ is not Minkowski measurable, it is the case that its {\em average Minkowski content}, $\mcm_{av}$, defined as a suitable Cesaro average of $V_{CS} (\varepsilon) \varepsilon^{-(1-D)}$ (see the $N=1$ case of footnote \ref{Fn:64}), exists and can be explicitly computed in terms of $\res (\zeta_\mcl, D)$; see \cite[\S 8.4.3]{Lap-vF4}. The same is true for any lattice self-similar string; see \cite[Thm. 8.23]{Lap-vF4}.\footnote{The precise definition of (bounded) self-similar strings is given in \cite[Ch. 2]{Lap-vF4}. Here, we simply recall that a self-similar string is said to be {\em lattice} if its distinct scaling ratios generate a (multiplicative) group of rank 1. It is said to be {\em nonlattice}, otherwise. The detailed structure of the complex dimensions of self-similar strings is discussed in \cite[Chs. 2 and 3]{Lap-vF4}.} More specifically, a lattice self-similar string is not Minkowski measurable but its average Minkowski content, $\mcm_{av}$, exists in $(0, +\infty)$ and is also given by the right-hand side of \eqref{2.17}; see \cite[Thm. 8.30]{Lap-vF4}.

(d) More generally, a self-similar string is Minkowski measurable if and only if it is nonlattice. In this case, its Minkowski content, $\mcm$, is given by either \eqref{2.16} or \eqref{2.17}; see \cite[Thms. 8.23 and 8.36]{Lap-vF4}. Further, we have $\mcm_{av} = \mcm$, since there is no need to take any averaging anymore. 
\end{remark}
 
\subsection{Other examples of fractal explicit formulas}\label{Sec:2.2}
Let $\mcl = (\ell_j)_{j \geq 1}$ be a fractal string. Then, it is a vibrating object and its (normalized frequency) spectrum consists of the numbers $f_{j,n} = n \cdot \ell_j^{-1}$, where $n, j \in \mbn = \{ 1, 2, \cdots \}.$ One can think of $\mcl$ as being composed of infinitely many ordinary Sturm--Liouville strings, with lengths $\ell_j$, vibrating independently of one another and with their endpoints fixed (i.e., corresponding to homogeneous Dirichlet boundary conditions for the one-dimensional Laplacian $-d^2/dx^2$ on the open set $\Omega \subseteq \mbr$). 

One of the major themes of fractal string theory is the study of the interplay between the geometry and the spectra of fractal strings.

Let $N_\mcl$ be the {\em geometric counting function} of $\mcl$, given by (here, $\#A$ denotes the cardinality of a finite set $A$)
\begin{equation}\label{2.18}
N_\mcl (x) = \# \{ j \geq 1: \ell_j^{-1} \leq x \}, \text{ for } x > 0.
\end{equation}  
Similarly, let $N_\nu$ denote the ({\em frequency} or) {\em spectral counting function} of $\mcl$:
\begin{equation}\label{2.19}
N_\nu (x) = \# \{ f: f \text{ is a frequency of } \mcl , \text{ with } f \leq x\}, \text{ for } x > 0.
\end{equation}
Then, $N_\mcl$ and $N_\nu$ are connected via the following identity, for all $x > 0$:\footnote{Note that for each fixed $x >0$, the sum in \eqref{2.20} contains only finitely many nonzero terms. However, as $x \rightarrow +\infty,$ the number of these terms tends to $+\infty.$}
\begin{equation}\label{2.20}
N_\nu (x) = \sum_{j=1}^\infty N_\mcl \bigg( \frac{x}{j} \bigg).
\end{equation}
Essentially equivalently, the geometric and spectral zeta functions $\zeta_\mcl$ and $\zeta_\nu$ of $\mcl$ are connected by the following key identity (first observed in \cite{Lap2}, \cite{Lap3}):\footnote{Here, $\zeta_\nu (s)$ is given for $Re(s) > 1$ by $\zeta_\nu (s) := \sum_f f^{-s}$, where $f$ ranges through all the (normalized) frequencies of $\mcl$, and is then meromorphically continued wherever possible.}
\begin{equation}\label{2.21}
\zeta_\nu (s) = \zeta (s) \cdot \zeta_\mcl (s),
\end{equation}
where $\zeta$ denotes the classic Riemann zeta function, initially defined by $\zeta (s):= \sum_{n=1}^\infty n^{-s}$ for $Re(s) > 1$ and then meromorphically continued to all of $\mbc$ (see, e.g., \cite{Edw, Pat, Tit}).  

Note that in order to apprehend the principal complex dimensions of $\mcl$ and their effect on the spectrum of $\mcl$, one must work in the {\em closed critical strip} $\{ 0 \leq Re(s) \leq 1 \}$ of $\zeta$ or, if one excludes the extreme cases when $D=0$ or $D =1$, in its open counterpart, $\{0 < Re(s) < 1 \}$, henceforth referred to as the {\em critical strip}.\\

Now, let us give a few examples of {\em fractal explicit formulas}, analogous to the fractal tube formulas discussed in \S \ref{Sec:2.1} above. In the spirit of this overview, we will not strive here for either mathematical precision or for the most general statements but instead refer to \cite[Chs. 5 and 6]{Lap-vF4} for all of the details and a much broader perspective.

Assume, for clarity, that all of the complex dimensions of $\mcl = (\ell_j)_{j=1}^\infty$ are simple. Then, under appropriate hypotheses, we obtain the following pointwise or distributional explicit formulas with error terms:\footnote{Under somewhat stronger assumptions, we obtain exact formulas; namely, either $R_\mcl (x) \equiv 0$ or (more rarely) $R_\nu (x) \equiv 0.$}
\begin{equation}\label{2.22}
N_\mcl (x) = \sum_{\omega \in \mcd_\mcl} c_w \frac{x^\omega}{\omega} + R_\mcl (x)
\end{equation}
and
\begin{equation}\label{2.23}
N_\nu (x) = \zeta_\mcl (1) x + \sum_{\omega \in \mcd_\mcl} c_\omega \, \zeta (\omega) \frac{x^\omega}{\omega} + R_\nu (x),  
\end{equation}
where, as before, $c_\omega := \res (\zeta_\mcl, \omega)$ for every $ \omega \in \mcd_\mcl$ and $R_\mcl$ and $R_\nu$ are error terms which can be suitably estimated (either pointwise or distributionally); see \cite[\S 6.2]{Lap-vF4}. (Note that $\zeta_\mcl (1) = |\Omega|_1 = \sum_{j=1}^\infty \ell_j$, the {\em total length} of the fractal string $\mcl$.)

Analogous formulas, now necessarily interpreted distributionally rather than pointwise, can be obtained for the positive measures $\eta$ and $\nu$, respectively defined as $d N_\mcl/ dx$ and $d N_\nu/dx$ (the distributional derivatives of $N_\mcl$ and $N_\nu$) and referred to as the {\em geometric} and {\em spectral densities of states}; see \cite[\S 6.3.1]{Lap-vF4}. Alternatively,
\[ \eta ([0,x]) := \frac{ N_\mcl ([0,x]) + N_\mcl ([0, x))}{2}, \text{ for all } x > 0, \]
and similarly for $\nu$ and $N_\nu$.

\begin{remark}$(${\em Fractal string theory and its ramifications.}$)$\label{Rem:2.5}
Fractal string theory and the associated theory of complex dimensions has been developed in many directions and applied to a variety of fields, including harmonic analysis, fractal geometry, number theory and arithmetic geometry, complex analysis, spectral geometry, geometric measure theory, probability theory, nonarchimedean analysis, operator algebras and noncommutative geometry, as well as dynamical systems and mathematical physics.

In particular, in \cite[Ch. 13]{Lap-vF4}, are discussed a variety of extensions or applications of fractal string theory in diverse settings (prior to the development of the higher-dimensional theory of complex dimensions and of fractal zeta functions in [\hyperlinkcite{LapRaZu1}{LapRaZu1--10}], to be partly surveyed in \S\ref{Sec:3}), including fractal tube formulas for fractal sprays (especially, self-similar sprays and tilings), in \cite[\S 12.1]{Lap-vF4} (based on [\hyperlinkcite{LapPe2}{LapPe2--3}, \hyperlinkcite{LapPeWi1}{LapPeWi1--2}], \cite{Pe, PeWi}), complex dimensions and fractal tube formulas for $p$-adic fractal (and self-similar) strings, in \cite[\S 12.2]{Lap-vF4} (based on [\hyperlinkcite{LapLu1}{LapLu1--3}, \hyperlinkcite{LapLu-vF1}{LapLu-vF1--2}]), multifractal zeta functions and strings, in \cite[\S 12.3]{Lap-vF4} (based on \cite{LapRo, LapLevyRo, ElLapMcRo}), random fractal strings and zeta functions, in \cite[\S 12.4]{Lap-vF4} (based on \cite{HamLap}), as well as fractal membranes (or `quantized fractal strings') and their associated moduli space, in \cite[\S 12.5]{Lap-vF4} (based on \cite{Lap7, LapNes}). See also \cite[\S 12.2.1]{Lap-vF4} for a brief description of a first direct attempt at a higher-dimensional theory, in \cite{LapPe1}, where a fractal tube formula was obtained for the Koch snowflake curve via a direct computation.

In addition to the aforementioned articles, we refer the interested reader to the research books [\hyperlinkcite{Lap-vF2}{Lap-vF2--4}], \cite{Lap7}, \cite{LapRaZu1}, \cite{HerLap1}, \cite{Lap10}, as well as \cite{Lap-vF6}, \cite{CaLapPe-vF} and \cite{LapRaRo}, along with the research papers and survey articles \cite{Cae}, [\hyperlinkcite{CobLap1}{CobLap1--2}], \cite{CranMH}, \cite{DemDenKoU}, \cite{DemKoOU}, \cite{DenKoOU}, \cite{DenKoORU}, \cite{deSLapRRo}, \cite{DubSep}, [\hyperlinkcite{Es1}{Es1--2}], [\hyperlinkcite{EsLi1}{EsLi1--2}], \cite{Fa2}, \cite{Fr}, \cite{FreiKom}, \cite{Gat}, \cite{Ger}, [\hyperlinkcite{GerScm1}{GerScm1--2}], \cite{HeLap}, [\hyperlinkcite{HerLap2}{HerLap2--5}], \cite{KeKom}, \cite{Kom}, \cite{KomPeWi}, \cite{KoRati}, [\hyperlinkcite{LalLap1}{LalLap1--2}], [\hyperlinkcite{Lap1}{Lap1--6}], [\hyperlinkcite{Lap8}{Lap8--9}], [\hyperlinkcite{LapMa1}{LapMa1--2}], [\hyperlinkcite{LapPo1}{LapPo1--3}], [\hyperlinkcite{LapRaZu2}{LapRaZu2--10}], \cite{LapRoZu}, \cite{Lap-vF1}, \cite{Lap-vF5}, \cite{Lap-vF7}, \linebreak \cite{LapWat}, \cite{LevyMen}, \cite{LiRadz}, \cite{MorSep}, [\hyperlinkcite{MorSepVi1}{MorSepVi1--2}], [\hyperlinkcite{Ol1}{Ol1--2}], [\hyperlinkcite{Ra1}{Ra1--2}], \cite{RatWi}, [\hyperlinkcite{Tep1}{Tep1--2}], [\hyperlinkcite{vF1}{vF1--2}], \cite{Wat} and [\hyperlinkcite{Zu1}{Zu1--2}], for various aspects of fractal string theory and its applications.
\end{remark}

\subsection{Analogy with Riemann's explicit formula}\label{Sec:2.3}
One of the most beautiful formulas in mathematics, in the author's opinion, is Riemann's explicit formula. The latter connects the {\em prime number counting function}
\begin{equation}\label{2.24}
\Pi_\mcp (x) := \# \{ p \in \mcp : p \leq x \}, \text{ for } x>0
\end{equation}
(where $\mcp$ denotes the set of prime numbers) and the zeros of the Riemann zeta function. Because it is simpler to state (as well as to justify), although a precise proof was provided only about forty years after the publication of Riemann's celebrated 1858 paper \cite{Rie}, we will state a modern form of this formula. Namely, consider the {\em weighted prime powers counting function}
\begin{equation}\label{2.25}
\varphi (x) := \sum_{p^n \leq x} \frac{1}{n}, \text{ for } x >0,
\end{equation} 
where the sum ranges over all prime powers $p^n$ (counted with a weight $1/n$, for every $n \in \mbn$). Then
\begin{equation}\label{2.26}
\varphi (x) = Li(x) - \sum_\rho Li(x^\rho) + \int_x^{+\infty} \frac{1}{t^2 -1} \frac{dt}{t \log t} - \log 2,
\end{equation}
where $Li(x) := \int_0^x \frac{dt}{\log t}$ is the {\em integral logarithm} and the infinite sum in \eqref{2.26} is taken over all of the critical zeros $\rho$ of $\zeta$ (in $\{ 0 < Re(s) < 1\}$), while the negative of the integral in \eqref{2.26} corresponds to the same sum but now taken over the trivial zeros $-2, -4, -6, \cdots$ of $\zeta$. Furthermore, the leading term, $Li (x)$, corresponds to the (simple) pole of $\zeta = \zeta (s)$ at $s=1$.

It follows from this formula (combined with an appropriate analysis, for instance based on a Tauberian theorem) that\footnote{The expression $f(x) \sim g(x)$ as $x \rightarrow +\infty$ means that $f(x)/g(x) \rightarrow 1$ as $x \rightarrow +\infty.$}
\begin{equation}\label{2.27}
\Pi_\mcp (x) \sim Li (x) \quad \text{as } x \rightarrow + \infty,
\end{equation}
or equivalently, that
\begin{equation}\label{2.28}
\Pi_\mcp (x) \sim \frac{x}{\log x} \quad \text{as } x \rightarrow + \infty,
\end{equation}
which is the statement of the famous Prime Number Theorem (PNT), as conjectured independently by Gauss (1792) and Legendre (1797) and proved independently by Hadamard (\cite{Had2}) and de la Vall\' ee Poussin (\cite{dV1}) in the same year (1896), but a century later. Note, however, that it also took about forty years after the publication of \cite{Rie} in 1858 in order to prove PNT, in the form of \eqref{2.27}, about the same amount of time it took to rigorously justify Riemann's original explicit formula; see van Mangoldt's work [\hyperlinkcite{vM1}{vM1--2}], along with Ingham's book \cite{Ing}.\footnote{There are more direct (but less insightful) ways to prove PNT. They also typically require to know that $\zeta (s)$ does not have any zero on the vertical line $\{ Re(s) = 1 \}$ (Hadamard, \cite{Had1}, 1893). However, in order to obtain a version of \eqref{2.27} with error term (PNT with error term), the Riemann--von Mangoldt explicit formula \eqref{2.26} (or one of its counterparts) is the most reliable tool (combined, for example, with an appropriate Tauberian theorem); see \cite{dV2} along, e.g., with \cite{Edw} (for a detailed history and analysis of Riemann's paper, \cite{Rie}) and, especially, \cite{Ing, Ivi, KarVo, Pat, Tit}.\label{Fn:19}}
The latter formula is obtained by first proving \eqref{2.26} and then by using  M\" obius inversion \cite{Edw, Ove} as follows (with $\varphi$ now given by the explicit formula \eqref{2.26}):
\begin{equation}\label{2.29}
\Pi_\mcp (x) = \sum_{n=1}^\infty \frac{\mu (n)}{n} \, \varphi \, (x^{1/n}),
\end{equation}
where $\mu$ denotes the {\em M\" obius function} defined on $\mbn$ by $\mu (n) = (-1)^k$ if $n \geq 2$ is a product of $k$ distinct primes, $\mu (1) = 1$, and $\mu (n) = 0$ otherwise. Riemann's original explicit formula is then deduced by substituting into \eqref{2.29} the expression of $\varphi$ given by \eqref{2.26}.

The analogy between Riemann's explicit formula (in any of its various disguises) and the fractal explicit formulas discussed in \S \ref{Sec:2.1} and \S \ref{Sec:2.2} is now apparent. The (critical) zeros of $\zeta$ correspond to the (nonreal) complex dimensions of $\mcl$, while the prime counting function $\Pi_\mcp$ in \eqref{2.24} (or the weighted prime powers counting function $\varphi$ in \eqref{2.25}) corresponds to the geometric or spectral counting function $N_\mcl$ or $N_\nu$ (e.g., in \eqref{2.22} or \eqref{2.23}, respectively), or else (in a more sophisticated but also more geometric form) to the tube function (or distribution) $V (\varepsilon)$ in \eqref{2.10}.

In particular, the {\em oscillations} (in the counting function of the primes) associated with the (critical) zeros of $\zeta$ in the infinite sum appearing in \eqref{2.26} or in \eqref{2.29} correspond to the {\em geometric oscillations} (in the geometric counting function $N_\mcl$ in \eqref{2.22} or in the tube function $V(\varepsilon)$ in \eqref{2.10}) or to the {\em spectral oscillations} (in the frequency counting function $N_\nu$ in \eqref{2.23}). We will further discuss these oscillations in \S \ref{Sec:2.4}.

At this stage, it is natural for the reader to be troubled by the presence of zeros in \eqref{2.26} or \eqref{2.29}, as opposed to just poles (or ``complex dimensions'') in \eqref{2.22}, \eqref{2.23} and \eqref{2.10}. However, this apparent discrepancy is quickly resolved by noting that the (necessarily simple) poles of (minus) the logarithmic derivative $-\zeta' (s)/\zeta (s)$ of $\zeta (s)$ correspond precisely to the zeros of $\zeta (s)$ and to its only pole (at $s=1$, which accounts for the leading term $Li (x)$ in \eqref{2.26}). In addition, the residue of $-\zeta' (s)/\zeta (s)$ at a pole $s = \omega$ is a nonzero integer whose sign tells us whether it corresponds to a zero or a pole (here, $s=1$) of $\zeta (s)$, and whose absolute value is the multiplicity of the zero or pole. As a simple exercise, the reader may wish to verify this statement and determine which sign of the residue corresponds to a zero or a pole.

In closing this subsection, we point out that the (pointwise or distributional) explicit formulas obtained in \cite[Ch. 5]{Lap-vF4}, with or without error term, and used throughout \cite[esp., in Chs. 6--11]{Lap-vF4}, extend Riemann's explicit formula (and its known number-theoretic counterparts) to a fractal, geometric, spectral, or dynamical setting in which the corresponding zeta functions do not necessarily have an Euler product or satisfy a functional equation. Furthermore, the general framework within which these explicit formulas are developed is sufficiently broad and flexible in order for the resulting formulas to be applied to a variety of situations (including arithmetic ones) and to help unify, in the process, aspects of fractal geometry and number theory, both technically and conceptually.\footnote{The interested reader can find in \cite[\S 5.1.1 and \S 5.6]{Lap-vF4} many references about number-theoretic and analytic explicit formulas in a variety of contexts, including [\hyperlinkcite{Wei4}{Wei4--5}, \hyperlinkcite{Den1}{Den1--3}, \hyperlinkcite{DenSchr}{DenSchr}, \hyperlinkcite{Har1}{Har1--3}].} 

\subsection{The meaning of complex dimensions}\label{Sec:2.4}
In light of the fractal tube formulas and the other fractal explicit formulas discussed in \S \ref{Sec:2.2} and \S \ref{Sec:2.3}, the following intuition of the notion of complex dimensions can easily be justified, mathematically. The {\em real parts} of the complex dimensions correspond to the {\em amplitudes} of `{\em geometric waves}' (propagating through the `space of scales'), while the {\em imaginary parts} of the complex dimensions correspond to the {\em frequencies} of those waves.

An analogous interpretation can be given in the spectral setting and in the dynamical setting. A common thread to these interpretations is provided by the (generalized) explicit formulas of \cite{Lap-vF4} mentioned at the end of \S \ref{Sec:2.3}. Associated key words are {\em oscillations}, vibrations, and wave-like phenomena, which could also be applied to the number-theoretic setting corresponding to Riemann's explicit formula for the prime number counting function and discussed in \S \ref{Sec:2.3}.

The author has conjectured since the early 1990s that (possibly generalized or even virtual) fractal geometries and arithmetic geometries pertained to the same mathematical realm. Consequently, there should exist a fractal-like geometry whose complex dimensions are the Riemann zeros (the critical zeros of $\zeta = \zeta (s)$); see, especially, the author's book \cite{Lap7}. There, in particular, an extension (and `quantization') of the notion of fractal string is introduced and coined `fractal membrane'. It turns out to be a noncommutative space, in the sense of \cite{Con1}. The associated moduli space of fractal membranes (which can be thought of physically as a quantization of the moduli space of fractal strings) plays a fundamental role in \cite{Lap7} in order to provide a conjectural explanation of why the Riemann hypothesis should be true, both for the classic Riemann zeta function and for all number-theoretic zeta functions (or $L$-functions, [\hyperlinkcite{ParsSh1}{ParsSh1--2}], \cite{Sarn}, \cite[App. A]{Lap-vF4}, \cite[Apps. B, C \& E]{Lap7}) occurring in arithmetic geometry. It is expressed in terms of a (still conjectural) noncommutative dynamical system on the moduli space of fractal membranes, as well as of its counterparts on the associated moduli spaces of zeta functions (or `partition functions') and of divisors (i.e., zeros and poles) on the Riemann sphere (which is the natural realm of the Riemann zeros and more generally, of the complex fractal dimensions). See, especially, \cite[Ch. 5]{Lap7}.

\subsection{Fractality, complex dimensions and irreality}\label{Sec:2.5}
Since, as we have seen, the imaginary parts of the complex dimensions give rise to oscillations in the intrinsic geometry (or in the spectra) of fractal strings, it is natural to wonder whether one could not define the elusive notion of fractality in terms of complex dimensions.

In [\hyperlinkcite{Lap-vF1}{Lap-vF1--4}] (as well as later, in higher dimensions, in \cite{LapRaZu1} to be discussed in \S\ref{Sec:3} below), an object is said to be `{\em fractal}' if it has at least one {\em nonreal} complex dimension\footnote{Since nonreal complex dimensions come in complex conjugate pairs, a fractal-like object must have at least two complex conjugate nonreal complex dimensions. In fact, in the geometric setting, it typically has infinitely many nonreal complex conjugate pairs of them.} 
and hence, in other words, according to the explicit formulas discussed in \S \S \ref{Sec:2.1}--\ref{Sec:2.4} (when $N=1$) and in \S \ref{Sec:3.5} (when $N \geq 1$ is arbitrary, where $N$ is the dimension of the embedding space), if it has {\em intrinsic geometric, spectral, dynamical or arithmetic oscillations}.\footnote{For various examples for which the source of the oscillations is of a dynamical (respectively, spectral) nature, see \cite[Ch. 7 and \S 12.5.3]{Lap-vF4} (respectively, \cite[Chs. 6 and 9--11]{Lap-vF4}), while for the case when it is of a geometric (respectively, arithmetic) nature, see \cite[Chs. 6 and 9--13]{Lap-vF4} (respectively, \cite[Chs. 9 and 11]{Lap-vF4}).}

In the case of a fractal string, the complex dimensions are the poles of the associated geometric zeta functions, whereas (anticipating on the discussion of \cite{LapRaZu1} given in \S\ref{Sec:3}), in higher dimensions (i.e., for bounded subsets of $\mbr^N$ or, more generally, for relative fractal drums in $\mbr^N$, for any $N \geq 1$), the complex dimensions are the poles of the associated fractal zeta functions. Furthermore, as was alluded to near the end of \S \ref{Sec:2.3}, in the arithmetic setting, the role played by the complex dimensions in fractal geometry is now essentially played by the Riemann zeros, or by their more general number-theoretic analogs (e.g., the critical zeros of automorphic $L$-functions or of zeta functions of varieties over finite fields).

\begin{remark}$(${\em Reality principle.}$)$\label{Rem:2.6}
Geometrically, the fact that the nonreal complex dimensions come in complex conjugate pairs is significant. This is what enables us, for instance, to obtain a real-valued (and even positive) expression for the tube function $V(\varepsilon)$ in the fractal tube formulas of \S \ref{Sec:2.1}.\footnote{An analogous comment can be made about the geometric and spectral counting functions of \S \ref{Sec:2.2} or even the prime numbers counting functions of \S \ref{Sec:2.3} (provided the ``complex dimensions'' are interpreted there as the zeros and the pole(s) of $\zeta = \zeta (s)$ or of its counterpart; that is, as the poles of (minus) the logarithmic derivative of the (arithmetic) zeta function under consideration.} For example, the fractal tube formula for the Cantor string in \eqref{2.11} can be written as follows (with $D := \log_3 2$ and ${\bf p} := 2\pi/\log 3$):
\begin{equation}\label{2.40}
V_{CS} (\varepsilon) = \frac{2^{-D} \varepsilon^{1-D}}{D(1-D) \log 3} + (2 \varepsilon)^{1-D} \Bigg(\sum_{n=1}^\infty Re \Bigg(\frac{(2\varepsilon)^{-in{\bf p}}}{(D + in{\bf p})(D- in{\bf p})} \Bigg) \Bigg) - 2 \varepsilon,
\end{equation}
\\
\noindent which in turn can be further expressed in terms of real-valued functions involving sine and cosine, by using Euler's identity
\begin{equation}\label{2.41}
(2 \varepsilon)^{-in{\bf p}} = \cos (n{\bf p} \log 2 \varepsilon) - i \sin (n{\bf p} \log 2 \varepsilon), \text{ for all } n \in \mbz.
\end{equation}
\end{remark}

It is obvious that the Cantor string $\mcl_{CS}$ (and hence, also the Cantor set) is fractal according to the above definition. Indeed, in light of \eqref{2.5}, it has infinitely many complex conjugate pairs of nonreal (principal) complex dimensions, $D \pm in{\bf p}$, with $n \in \mbn$ (as well as with $D$ and ${\bf p}$ as in Remark \ref{Rem:2.6}). This is also apparent in the fractal tube formula \eqref{2.11} (or its ``real'' form \eqref{2.40} in Remark \ref{Rem:2.6}), as well as in the following pointwise explicit formulas for $N_{CS}$ and $N_{\nu, CS}$, the geometric and spectral counting functions of $\mcl_{CS}$, respectively:\footnote{See \cite[Eqns. (1.31) and (6.57)]{Lap-vF4}, where in the latter equations, we have slightly adapted the formula because our Cantor string (unlike in \cite[Ch. 2]{Lap-vF4}) has total length $1$ (rather than $3$). Note that \eqref{2.42} is an exact explicit formula whereas \eqref{2.43} has an error term.}
\begin{equation}\label{2.42}
N_{CS} (x) = \frac{1}{2 \log 3} \sum_{n \in \mbz} \frac{x^{D + in{\bf p}}}{D + in{\bf p}} - 1
\end{equation}
and
\begin{equation}\label{2.43}
N_{\nu, CS} (x) = x + \frac{1}{2 \log 3} \sum_{n \in \mbz} \zeta (D + in{\bf p}) \frac{x^{D + in{\bf p}}}{D + in{\bf p}} + O(1).
\end{equation}

It is shown in \cite[Ch. 11]{Lap-vF4} that the Riemann zeta function $\zeta = \zeta(s)$ and many other arithmetic zeta functions (as well as other Dirichlet series) cannot have infinite vertical arithmetic progressions of critical zeros.\footnote{Unknown to the authors of [\hyperlinkcite{Lap-vF1}{Lap-vF1--2}] at the time, C.~R. Putnam [\hyperlinkcite{Put1}{Put1--2}] had established a similar result in the case of $\zeta = \zeta(s)$, by completely different methods which only extended to a few arithmetic zeta functions.} 
It follows that the spectral oscillations in \eqref{2.43}, just like the geometric oscillations in \eqref{2.42}, subsist. In fact, this result is established by reasoning by contradiction and using the counterparts of the explicit formulas \eqref{2.42} and \eqref{2.43}, as well as by proving that (virtual) generalized Cantor strings always have both geometric and spectral oscillations of a suitable kind (namely, of leading order $x^D$ as $x \rightarrow +\infty$); see \cite[Ch. 10]{Lap-vF4}. The sought for contradiction is then reached by making the nonreal complex dimensions $D + in {\bf p}$ (with $n \in \mbz \backslash  \{0 \}$) of the (virtual, generalized) Cantor string coincide with the presumed zeros of $\zeta$ in infinite arithmetic progression along the vertical line $\{ Re(s) = D \}$. (Here, $D \in (0,1)$ and the period ${\bf p} > 0$ can be chosen to be arbitrary.)\footnote{This is why the corresponding Cantor strings are not geometric, in general, but instead generalized (and virtual) fractal strings, in the sense of \cite[Chs. 4 and 10--11]{Lap-vF4}.} Then, in light of the counterpart of \eqref{2.43} in this context, we deduce that $N_\nu$ does not have any oscillations of leading order $x^D$, in contradiction to what is claimed just above. Observe the analogy with the method of proof (outlined in the latter part of \S \ref{Sec:2.6.2} below) of a key result of \cite{LapMa2} connecting the Riemann hypothesis and inverse spectral problems for fractal strings.

\begin{remark}$(${\em Multiplicative oscillations.}$)$\label{Rem:2.7}
In order to better understand the nature of the (multiplicative) oscillations intrinsic to `fractality', as expressed by fractal explicit formulas and the presence of nonreal complex dimensions (necessarily in complex conjugate pairs), it is helpful to first consider the following simple situation. Think, for instance, that each term of the form $x^z$ (with $x>0$ and $z \in \mbc$ such that $z = d + i \tau$, where $d \in \mbr$ and $\tau > 0$, say) arising in a given fractal (or arithmetic) explicit formula can be written as follows:
\begin{equation}\label{2.44}
x^z = x^d x^{i \tau} = x^d (\cos (\tau \log x) + i \sin (\tau \log x)).
\end{equation}
Now, clearly, the real part $d$ of $z$ governs the amplitude of the oscillations while the imaginary $\tau$ of $z$ governs the frequency of the oscillations. Observe that physically, the term $x^z$ can be viewed as a multiplicative analog of a plane wave (or a standing wave). The different terms of the form $x^z$ (or $\varepsilon^{-z}$) occurring in the infinite sum ranging over all of the visible complex dimensions and appearing in a given fractal explicit formula (or fractal tube formula), such as \eqref{2.22}, \eqref{2.23}, \eqref{2.42}, \eqref{2.43} (or \eqref{2.10}, \eqref{2.11}, \eqref{2.40}) provide a whole spectrum of amplitudes and frequencies associated with the corresponding superposition of `standing waves'. The fact that these waves arise in {\em scale space} (rather than in ordinary frequency or momentum space, in physicists' terminology), explains why the corresponding oscillations are viewed multiplicatively (rather than additively) here.

This is very analogous to what happens for Fourier series. Observe, however, that unlike for Fourier series, the frequencies are no longer, in general, multiple integers of a given fundamental frequency.\footnote{In that sense, the Cantor string and more generally, all lattice self-similar strings, are rather exceptional. In contrast, the complex dimensions of (bounded or unbounded) nonlattice self-similar strings have a much richer quasiperiodic structure; see \cite[Ch. 3]{Lap-vF4}.}
In addition, the varying amplitudes of the oscillations do not have a counterpart for ordinary Fourier series. If one replaces the classic theory of Fourier series by Harald Bohr's less familiar but more general theory of almost periodic functions \cite{Boh} (usually associated with purely imaginary rather than with arbitrary complex numbers $z$), one is getting closer to improving one's understanding of the situation. Nevertheless, the fact that typically, the complex dimensions form a countably infinite and discrete subset of $\mbc$ (rather than of $\mbr$) adds a lot of complexity and richness to the corresponding generalized `almost periodic' functions (or distributions, \cite{Katzn}).
\end{remark}

\begin{example}$(${\em The $a$-string.}$)$\label{Ex:2.8}
For a simple example of a string that is not fractal, in the above sense, consider the $a$-string, where $a >0$ is arbitrary.\footnote{The $a$-string, then viewed in \cite{Lap1} as a one-dimensional fractal drum, was the first example of fractal string (before that notion was formalized in \cite{LapPo2}) and was used in \cite[Exple. 5.1 and Exple. 5.1']{Lap1} to show that the remainder estimates obtained in \cite{Lap1} for the spectral asymptotics of fractal drums are sharp, in general (and in every possible dimension).}
Thus,
\begin{equation}\label{2.45}
\Omega = \Omega_a := [0,1] \backslash \{ j^{-a} : j \in \mbn \} = \bigcup_{j=1}^\infty ((j+1)^{-a}, j^{-a})
\end{equation}
and hence,
\begin{equation}\label{2.46}
\partial \Omega = \partial \Omega_a = \{ j^{-a} : j \in \mbn \} \cup \{ 0\},
\end{equation}
while
\begin{equation}\label{2.47}
\mcl = \mcl_a = (j^{-a} - (j+1)^{-a} )_{j=1}^\infty.
\end{equation}
It is shown in \cite[Thm. 6.21]{Lap-vF4} that the geometric zeta function $\zeta_{\mcl_a}$ of $\mcl_a$ admits a meromorphic continuation to all of $\mbc$ and that the complex dimensions of $\mcl_a$ are all simple with
\begin{equation}\label{2.48}
\mcd_{\mcl_a} = \{ D, -D, -2D, -3D, \cdots \},
\end{equation}
where $D= 1/(a+1)$ is the Minkowski dimension of $\mcl_a$ (or equivalently, of $\partial \Omega_a$).\footnote{It is possible that some of the numbers $-nD$, with $n \in \mbn$, do not truly appear in \eqref{2.48}, depending on the value of $a$. However, $D$ is always a complex dimension and ``typically'', we have an equality (rather than an inclusion) in \eqref{2.48}.}, 
Note that $D \in (0,1)$ whereas $H=0$ for any $a >0$, where $H$ is the Hausdorff dimension of the compact set $\partial \Omega_a \subseteq \mbr$. This illustrates the fact that the Minkowski dimension, and not the Hausdorff dimension, is the proper notion of real dimension pertaining to the theory of complex fractal dimensions.

Since the $a$-string $\mcl_a$ (or $\partial \Omega_a$) does not have any nonreal complex dimension, it is {\em not} fractal, in the above sense. This conclusion is entirely compatible with our intuition according to which $\mcl_a$ (or the associated compact set $\partial \Omega_a \subseteq \mbr$ in \eqref{2.46}) does not have much complexity.
\end{example}

\begin{example}$(${\em Self-similar strings.}$)$\label{Ex:2.9}
Next, we discuss a more geometrically interesting family of fractal strings, namely, {\em self-similar fractal strings} (as introduced and studied in detail in [\hyperlinkcite{Lap-vF1}{Lap-vF1--4}], see \cite[Chs. 2 and 3]{Lap-vF4}). It is shown in the just mentioned references that self-similar strings have infinitely many nonreal complex dimensions. In fact, their only real dimension is the Minkowski dimension $D$ (and it is simple); as a result, $D$ is also the maximal real part of the complex dimensions of such strings. Thus, as expected (since their boundaries are nontrival self-similar sets in $\mbr$), self-similar strings are always fractal.

Now, as we may recall from our earlier discussion, there are two kinds of self-similar strings, {\em lattice strings} and {\em nonlattice} strings. For both types of (i.e., for all) self-similar strings, the geometric zeta function admits a meromorphic continuation to all of $\mbc.$ In the {\em lattice case} (i.e., when $G=r^\mbz$, for some $r \in (0,1),$ where $G$ is the multiplicative group generated by the distinct scaling ratios and the `gaps' of the self-similar string), the complex dimensions are periodically distributed (with the same vertical period ${\bf p} := 2 \pi/ \log (r^{-1}) > 0$, called the {\em oscillatory period} of the given lattice string) along finitely many vertical lines.\footnote{The number of vertical lines (counted according to multiplicity) is equal to the degree of the polynomial obtained after making the change of variable $w := r^s$ in the denominator of the expression for $\zeta_\mcl (s) = (\sum_{k=1}^K g_k^s)/(1 - \sum_{j=1}^J r_j^s)$, where $(r_j)_{j=1}^J$ are the (not necessarily distinct) scaling ratios and $(g_k)_{k=1}^K$ are the gaps of $\mcl$, with $J \geq 2$.}
Furthermore, on each of these vertical lines, the multiplicity of the complex dimensions is the same, while the right  most vertical line is $\{Re(s) = D \}$, on which lie the principal complex dimensions (which are necessarily all simple).

By contrast, in the {\em nonlattice case} (i.e., when the rank of the group $G$ is strictly greater than 1), the complex dimensions are no longer periodically distributed. In fact, typically, on a given vertical line $\{ Re(s) = \alpha \}$, with $\alpha \in \mbr$, there is either zero, or one complex dimension, necessarily $D$ itself (only if $\alpha =D$), or else two complex conjugate nonreal complex dimensions $\omega, \overline{\omega}$ (with $Re(\omega) = \alpha < D$). (Note that there can be at most countably vertical lines containing at least one complex dimension.) Furthermore, on the right most vertical line, the only complex dimension is $D$ itself, and it is simple. (Hence, according to a version of Theorem \ref{Thm:2.2} obtained in \cite[\S 8.4]{Lap-vF4}, a nonlattice string is always Minkowski measurable, in contrast to a lattice string, which is never so.)

Moreover, it is shown in \cite[\S 3.4]{Lap-vF4} by using Diophantine approximation that any nonlattice string $\mcl$ can be approximated by a sequence of lattice strings $\{\mcl_n \}_{n=1}^\infty$ with oscillatory periods ${\bf p}_n$ increasing exponentially fast to $+\infty$, as $n \rightarrow \infty$. Hence, the complex dimensions of $\mcl$ are themselves approximated by the complex dimensions of $\mcl_n$. As a result, the complex dimensions of a nonlattice string exhibit a {\em quasiperiodic structure}. (See {\em ibid} for more precision and for many examples of quasiperiodic patterns of complex dimensions of nonlattice strings.)

Finally, we refine (as in \cite{LapRaZu1}) the notion of fractality by saying that a geometric object is `{\em fractal in dimension} $d$', where $d \in \mbr$, if it has a nonreal complex dimension with real part $d$. Hence, `fractality' in our sense is equivalent to fractality in dimension $d$, for some $d \in \mbr$. (Clearly, $d \leq D$.) Also, fractality in dimension $D$ amounts to the existence of nonreal principal complex dimensions. According to \cite[Thms. 2.16 and 3.6]{Lap-vF4}, this is always the case for lattice strings (such as the Cantor string) but is never the case for nonlattice strings.

In light of the above discussion, a lattice string is fractal in dimension $d$ for at least one value but at most finitely many values of $d$, whereas by contrast, nonlattice strings are fractal in dimension $d$ for infinitely (but countably) many values of $d$. In fact, in the generic nonlattice case, a significantly stronger statement is true;\footnote{A nonlattice string is said to be {\em generic} if the group $G$ generated by its $M$ distinct scaling ratios is of rank $M$ and $M \geq 2$. It is {\em nongeneric}, otherwise.} 
namely, the countable set of such $d$'s is dense in a single compact interval of the form $[d_{\min}, D]$, with $-\infty < d_{\min} < D$. This latter density result was obtained by the authors of \cite{MorSepVi1} who thereby proved a conjecture made in the generic nonlattice case in \cite{Lap-vF5} (see also [\hyperlinkcite{Lap-vF3}{Lap-vF3--4}]).\footnote{In general, in the nongeneric nonlattice case, the set of such $d$'s should be dense in at most finitely (but at least one, ending at $D$) compact and pairwise disjoint intervals; this is noted, by means of examples, independently in \cite{Lap-vF4} and in \cite{DubSep}.} 
\end{example}

In \S \ref{Sec:3.6}, we will further discuss and broaden the notion of fractality (and the related notion of hyperfractality), but now by focusing on higher-dimensional examples, such as the devil's staircase (i.e., the Cantor graph), rather than on fractal strings.

\begin{remark}\label{Rem:2.10.1/2}
The lattice/nonlattice dichotomy arose in probabilistic renewal theory \cite{Fel}, where it is usually referred to as the arithmetic/nonarithmetic dichotomy. In fractal geometry, it was used in [\hyperlinkcite{Lall1}{Lall1--3}] in connection with self-similar sets and generalizations thereof, and then, e.g., in [\hyperlinkcite{Lap2}{Lap2--7}], \cite{KiLap1}, \cite{Gat}, \cite{LeviVa},  [\hyperlinkcite{Lap-vF1}{Lap-vF1--7}], \cite{Fr}, [\hyperlinkcite{Sab1}{Sab1--3}], \cite{HamLap},  [\hyperlinkcite{LapPe2}{LapPe2--3}], [\hyperlinkcite{LapPeWi1}{LapPeWi1--2}], \newline [\hyperlinkcite{LapLa-vF2}{LapLa-vF2--3}], [\hyperlinkcite{LapLu-vF1}{LapLu-vF1--2}], \cite{KeKom}, \cite{Kom}, \cite{MorSep}, [\hyperlinkcite{MorSepVi1}{MorSepVi1--2}], \linebreak \cite{DubSep}, [\hyperlinkcite{LapRaZu1}{LapRaZu1--9}] and \cite{Lap10}. The notion of {\em generic} nonlattice self-similar string (or, more generally, spray or even set) was introduced in [\hyperlinkcite{Lap-vF3}{Lap-vF3--5}].
\end{remark}

\subsection{Inverse spectral problems and the Riemann hypothesis}\label{Sec:2.6}

In this subsection, we briefly discuss the intimate connection between direct (respectively, inverse) spectral problems for fractal strings and the Riemann zeta function (respectively, the Riemann hypothesis). For a recent survey of this subject, we refer the interested reader to \cite{Lap9}.\\

\subsubsection{Direct spectral problems for fractal strings}\label{Sec:2.6.1}

Let $\mcl = (\ell_j)_{j=1}^\infty$ (or any of its geometric realizations $\Omega \subseteq \mbr$) be a fractal string of dimension $D \in (0,1)$,\footnote{Recall that for a fractal string, we always have that $D \in [0,1]$.}
and let $N_\nu$ denote the associated spectral counting function of $\mcl$, as defined in \S \ref{Sec:2.2}. It turns out that the leading spectral asymptotics of $\mcl$ are given by the so-called {\em Weyl term}, $W$ (named after the well-known $N$-dimensional result in [\hyperlinkcite{Wey1}{Wey1--2}]).\footnote{Here, $\mcl$ is viewed as the RFD $(\partial \Omega, \Omega)$, in the terminology of \S \ref{Sec:3.2}, with (upper) Minkowski dimension $D \in (0,1)$.}
Namely,
\begin{equation}\label{2.49}
N_\nu (x) \sim W(x) \quad \text{as } x \rightarrow +\infty,
\end{equation}
where $W = W(x)$ is the {\em Weyl term} given by
\begin{equation}\label{2.50}
W(x) := |\Omega|_1 x, \text{ for }x >0,
\end{equation}
with $|\Omega|_1 = \sum_{j=1}^\infty \ell_j$ being the total length of the string $\mcl$ (or the `volume' of $\Omega$).

This is a very special (one-dimensional) case of the spectral asymptotics with error term obtained in \cite{Lap1} for fractal drums in $\mbr^N$, with $N \geq 1$ arbitrary. (In fact, $N_\nu (x) = W(x) + R(x)$, with $R(x) = O(x^D)$ if $\mcm^* < \infty$ and $R(x) =O(x^{D + \varepsilon})$ for any $\varepsilon >0$, otherwise, where $D \in (N-1, N)$ is the (upper) Minkowski dimension of $\Omega$, relative to $\partial \Omega$ or in the terminology of \S \ref{Sec:3.2}, of the RFD $(\partial \Omega, \Omega)$. Here, the Weyl term $W=W(x)$ is proportional to $|\Omega|_N x^N$, where $|\Omega|_N$ is the $N$-dimensional volume of $\Omega$.) It also follows, for example, from a result of \cite{LapPo2}.  

The following theorem (joint with C. Pomerance), obtained in \cite{LapPo2} (and first announced in \cite{LapPo1}), resolved in the affirmative the one-dimensional case of the modified Weyl--Berry (MWB) conjecture stated in \cite{Lap1} for fractal drums.\footnote{In higher dimensions, the situation is not as clear cut and the MWB conjecture itself needs to be further modified; see, e.g., \cite{FlVa, LapPo3, MolVai}.\label{Fn:24}}

We refer to \cite{Lap1, Lap3, LapPo2, Lap9} and \cite[\S 12.5]{Lap-vF4} for physical and mathematical motivations, as well as for many further references (including \cite{BirSol}, \cite{BroCar}, \cite{CouHil}, \cite{FlLap}, \cite{Ger}, [\hyperlinkcite{GerScm1}{GerScm1--2}], \cite{Gi}, [\hyperlinkcite{Ho1}{Ho1--3}], [\hyperlinkcite{Ivr1}{Ivr1--3}], \cite{LeviVa}, [\hyperlinkcite{Mel1}{Mel1--2}], [\hyperlinkcite{Met1}{Met1--2}], \cite{Ph}, [\hyperlinkcite{ReSi1}{ReSi1--3}],  [\hyperlinkcite{See1}{See1--3}], [\hyperlinkcite{Lap2}{Lap2--3}] and \cite[App. B]{Lap-vF4}, along with those cited in footnote \ref{Fn:24}), concerning the original conjectures of Weyl [\hyperlinkcite{Wey1}{Wey1--2}] for `smooth drums', Berry [\hyperlinkcite{Berr1}{Berr1--2}] for `fractal drums' and their later modifications and extensions in [\hyperlinkcite{Lap1}{Lap1--3}], in particular.

\begin{theorem}[\cite{LapPo2}]\label{Thm:2.6}
Let $\mcl$ be a fractal string which is Minkowski measurable and of Minkowski dimension $D \in (0,1)$. Then, $N_\nu$, the spectral counting function of $\mcl$, admits a {\em monotonic} asymptotic second term, proportional to $x^D$. More specifically,
\begin{equation}\label{2.51}
N_\nu (x) = W (x) - c_D \mcm x^D + o(x^D) \quad \textnormal{as} \rightarrow +\infty,
\end{equation}
where $W(x) = |\Omega|_1 x$ is the Weyl term given by \eqref{2.50} and $\mcm$ is the Minkowski content of $\mcl$. Furthermore, the constant $c_D$ is positive, depends only on $D$ and is given explicitly by 
\begin{equation}\label{2.52}
c_D := (1-D) 2^{-(1-D)} (-\zeta (D)),
\end{equation}
where $\zeta = \zeta (s)$ denotes the Riemann zeta function.
\end{theorem}

Theorem \ref{Thm:2.6} relies, in particular, on (the easier direction of) the characterization of Minkowski measurability, given in \cite{LapPo2} and recalled in Theorem \ref{Thm:2.3}. (Namely, $\mcl$ is Minkowski measurable iff $N_\mcl (x) \sim Q \, x^D$, for some constant $Q \in (0, +\infty)$, where $N_\mcl$ is the geometric counting function of $\mcl$.) It also makes use of the identity \eqref{2.20} connecting $N_\nu (x)$ and $N_\mcl (x)$ for any $x >0$, as well as of a direct computation involving the {\em analytic continuation} of $\zeta (s)$ to the open right half-plane $\{ Re(s) >0 \}$ and hence, to the (open) {\em critical strip} $\{0 < Re(s) < 1 \}$ since $D \in (0,1)$).

\begin{remark}\label{Rem:2.10}
($a$) ({\em Drums with fractal boundary.}) The aforementioned spectral error estimates obtained in \cite{Lap1} are valid for Laplacians (or, suitably adapted, for more general elliptic operators with variable, nonsmooth coefficients and of order $2m$, with $m \in \mbn$) on bounded open sets $\Omega \subseteq \mbr^N$ (where $N \geq 1$ is arbitrary) with (possibly) fractal boundary (``drums with fractal boundary'', in the sense of \cite{Lap3}) and with Dirichlet or Neumann boundary conditions. For the Dirichlet problem, $\Omega$ is allowed to have finite volume (i.e., $|\Omega|_N < \infty$) rather than to be bounded. Furthermore, for the Neumann problem, one must assume, for example, that $\Omega$ satisfies the {\em extension property} for the Sobolev space $H^m = W^{m,2}$ \cite{Br, Maz}, where $m=1$ in the present case of the Laplacian; see, e.g., \cite{Maz}, \cite{Jon}, [\hyperlinkcite{HajKosTu1}{HajKosTu1--2}], \cite{Lap1}, \cite{Vel-San}. For instance, for a simply connected planar domain, $\Omega$ is an extension domain if and only if it is a quasidisk (i.e., $\partial \Omega$ is a quasicircle, \cite{Pomm}); that is, $\Omega$ is the homeomorphic quasiconformal image of the open unit disk in $\mbc$.\footnote{Another characterization of a quasidisk is that it is a planar domain which is both a John domain (see \cite{John} and, e.g., \cite{CarlJonYoc}, \cite{McMul}, \cite{DieRuzSchu}, \cite{Dur-LopGar}, \cite{AcoDur-LopGar} or [\hyperlinkcite{LopGar1}{LopGar1--2}]) and linearly connected; see, e.g., \cite{ChuOsgPomm}.\label{Fn:38}} Quasidisk and quasicircles (as well as John domains) are of frequent use in harmonic analysis, partial differential equations, complex dynamics and conformal dynamics; see, e.g., \cite{Be}, \cite{Maz}, \cite{Pomm} and \cite{BedKS}.

($b$) ({\em Drums with fractal membrane.}) An analog of the leading term (Weyl's asymptotic formula) and the associated error term in \cite{Lap1} discussed in part ($a$) was obtained in \cite{KiLap1} for Laplacians {\em on} fractals, rather than on bounded open sets of $\mbr^N$ with fractal boundary; that is, for ``drums with fractal membrane'' (in the sense of \cite{Lap3}) rather than for ``drums with fractal boundary'' (as, e.g., in \cite{BirSol}, \cite{BroCar}, [\hyperlinkcite{Lap1}{Lap1--4}], [\hyperlinkcite{LapPo1}{LapPo1--3}], [\hyperlinkcite{LapMa1}{LapMa1--2}], \cite{Ger}, [\hyperlinkcite{GerScm1}{GerScm1--2}], \cite{FlLap}, \cite{FlVa}, \cite{EdmEv}, \cite{Dav}, \cite{LeviVa}, \cite{HeLap}, \cite{MolVai}, \cite{vB-Gi}, \linebreak
 \cite{HamLap}, \cite{Lap9}, \cite{LapPa} and \cite{LapNRG}). Examples of physics and mathematics references on Laplacians on fractals include the books \cite{Ki} and \cite{Str} along with the papers \cite{Ram}, \cite{RamTo}, \cite{Shi}, \cite{FukShi}, [\hyperlinkcite{KiLap1}{KiLap1--2}], \cite{Barl}, [\hyperlinkcite{Ham1}{Ham1--2}], [\hyperlinkcite{Sab1}{Sab1--3}], [\hyperlinkcite{Tep1}{Tep1--2}], \cite{DerGrVo}, [\hyperlinkcite{Lap5}{Lap5--6}], \cite{ChrIvLap}, \cite{CipGIS}, \cite{LapSar} and [\hyperlinkcite{LalLap1}{LalLap1--2}], as well as the many relevant references therein.
\end{remark}

\subsubsection{Inverse spectral problems for fractal strings}\label{Sec:2.6.2}

Now that we have explicitly and fully solved the above direct spectral problem for fractal strings, it is natural to consider its converse, which is called an inverse spectral problem, (ISP). In fact, we have a one-parameter family of such problems, (ISP)$_D$, parametrized by the (Minkowski) dimension $D$. Recall that $D \in (0,1)$ is arbitrary, so that the parameter $D$ sweeps out the entire `{\em critical interval}' $(0,1)$ for the Riemann zeta function $\zeta = \zeta(s)$.\\

(ISP)$_D$ {\em Let $\mcl$ be a fractal string of Minkowski dimension $D \in (0,1)$ such that its associated spectral counting function $N_\nu$ admits a} monotonic {\em asymptotic second term, proportional to $x^D$. Namely, with the Weyl term $W$ given by \eqref{2.50}, assume that}
\begin{equation}\label{2.53}
N_\nu (x) = W (x) - \mcc x^D + o(x^D) \quad \text{as } x \rightarrow +\infty,
\end{equation}
{\em for some nonzero constant $\mcc$} ({\em depending only on $\mcl$}). {\em Does it then follow that $\mcl$ is Minkowski measurable?}\\

The above question \` a la Marc Kac, \cite{Kac}, could be coined {\em Can one hear the shape of a fractal string} ({\em of dimension $D \in (0,1)$})? Note, however, that this question (or equivalently, the corresponding inverse spectral problem $(ISP)_D$), is of a very different nature from the original one, raised in \cite{Kac}.

\begin{remark}\label{Rem:2.11} 
Equation \eqref{2.53} alone with $\mcc \neq 0$ and $D \in (0,1)$ implies that $\mcl$ has Minkowski dimension $D$. Furthermore, if (ISP)$_D$ has an affirmative answer, then it follows from Theorem \ref{Thm:2.6} that $\mcc >0$ and $\mcc = c_D \mcm$, where $c_D > 0$ is the constant (depending only on $D$) given by \eqref{2.52} and $\mcm$ is the Minkowski content of $\mcl$.
\end{remark}

We can finally give the precise statement of the main result (Theorem \ref{Thm:2.12} and Corollary \ref{Cor:2.13}) connecting the family of inverse spectral problems (ISP)$_D$, with $D \in (0,1)$, and the Riemann hypothesis (RH); see [\hyperlinkcite{LapMa1}{LapMa1--2}], joint with H. Maier.

\begin{theorem}[Critical zeros of $\zeta$ and inverse spectral problems; \cite{LapMa2}]\label{Thm:2.12}
Fix $D \in (0,1)$. Then, the inverse spectral problem $($ISP$)_D$ has an affirmative answer if and only if $\zeta = \zeta(s)$ does not have any zeros on the vertical line $\{Re(s) = D \};$ i.e., if and only if the `{\em partial RH}' $($abbreviated $(RH)_D)$ holds for this value of $D$. In short $($and with the obvious notation$)$, we have that
\begin{equation}\label{2.54}
(RH)_D \Leftrightarrow (ISP)_D, \textnormal{ for any } D \in (0,1).
\end{equation} 
\end{theorem}

\begin{corollary}[Spectral reformulation of RH; \cite{LapMa2}]\label{Cor:2.13}
The Riemann hypothesis is true if and only the inverse spectral problem $($ISP$)_D$ has an affirmative answer for all $D \in (0,1)$, except in the `midfractal case' when $D = 1/2$.
\end{corollary}
 
\begin{proof}({\em Proof of Corollary \ref{Cor:2.13}.})
Since it is known that $\zeta$ has a nonreal zero (and even infinitely many zeros, by a theorem of G.~H. Hardy; see, e.g., \cite{Edw, Tit}) on the critical line $\{Re(s) = 1/2 \},$ this is a consequence of Theorem \ref{Thm:2.12}.
\end{proof} 
 
\begin{remark}\label{Rem:2.13}
Observe that we could reformulate Corollary \ref{Cor:2.13} by stating that RH is equivalent to (ISP)$_D$ having an affirmative answer for all $D \in (0, 1/2)$ (or equivalently, for all $D \in (1/2, 1)$). This follows from the well-known functional equation for $\zeta$ connecting $\zeta (s)$ and $\zeta (1-s)$, for all $s \in \mbc$; see, e.g., \cite{Edw} or \cite{Tit}.
\end{remark}

\begin{proof}({\em Sketch of the proof of Theorem \ref{Thm:2.12}.})
The proof of one implication in Theorem \ref{Thm:2.12} relies on the  Wiener--Ikehara Tauberian theorem (see, e.g., \cite{Pos}) or one of its later improvements (see, e.g., \cite{PitWie}, \cite{Pit} and \cite{Kor}). 

The converse (i.e., the reverse implication $(ISP)_D \Rightarrow (RH)_D$ in \eqref{2.54}) is proved by contraposition. That is, we assume that $(RH)_D$ fails and we therefore want to show that the inverse spectral problem $(ISP)_D$ does not have an affirmative answer. Hence, suppose that there exists $\omega \in \mbc$, with $\omega = D + i \tau \, (D \in (0,1), \tau > 0)$, such that $\zeta (\omega) = 0$. Then, clearly, $\overline{\omega} = D - i \tau$ also satisfies $\zeta (\overline{\omega}) = 0.$

Next, we use the intuition of the notion of complex dimension\footnote{At the time, in the early 1990s, the rigorous definition of complex dimension did not yet exist, even for fractal strings. This only came later, in the mid-1990s; see \cite{Lap-vF1, Lap-vF2}.}
in order to construct a fractal string $\mcl = (\ell_j)_{j=1}^\infty$ which has oscillations of leading order $x^D$ in its {\em geometry} but which (because $0 =\zeta (\omega) = \overline{\zeta (\omega)} = \zeta (\overline{\omega})$) does {\em not} have oscillations of order $D$ in its {\em spectrum} (so that $N_\nu (x)$ has a {\em monotonic} asymptotic second term proportional to $x^D$).

More specifically, we have that (for some positive constant $\beta$ sufficiently small and with $[y]$ denoting the integer part of $y \in \mbr$)
\begin{equation}\label{2.55}
N_\mcl (x) = [V(x)], \text{ for any } x > 0, 
\end{equation}
where
\begin{align}\label{2.56}
V(x) &:= x^D + \beta (x^\omega + x^{\overline{\omega}})\\
&= x^D (1 + 2 \beta \cos (\tau \log x)),\notag
\end{align}
for any $x >0$. Note that it suffices to choose a positive number $\beta < 1/2$ so that $V(x) > 0$ for all $x > 0$ and $\beta < D/2 (D + \tau)$ so that $V' (x) >0$ for all $x > 0$.\footnote{Indeed, an elementary computation shows that 
\[V' (x) = x^{D-1} (D+ 2 \beta D \cos ( \omega \log x ) - 2 \beta \tau \sin (\omega \log x)), \text{ for all } x > 0. \]}
Hence, we can simply choose $\beta \in (0, D/2(D + \tau)$). Since $V$ is now strictly increasing from $(0, +\infty)$ to itself, we can uniquely define $\ell_j > 0$  so that $V(\ell_j^{-1}) = j$; i.e., in light of \eqref{2.55}, $N_\mcl (\ell_j^{-1}) = j$, for each integer $j \geq 1$. This defines the sought for (bounded) fractal string $\mcl = (\ell_j)_{j=1}^\infty$;\footnote{Note that $\mcl$ is bounded because, in light of \eqref{2.57}, $N_\mcl (x)$ is of the order of $x^D$ as $x \rightarrow +\infty$ and hence, $\ell_j$ is of the order of $j^{-1/D}$ as $j \rightarrow \infty$. Since $D \in (0,1),$ it follows that $\sum_{j \geq 1} \ell_j < \infty.$} 
one can then choose any geometric realization $\Omega$ of $\mcl$ as a subset of $\mbr$ with finite length.

Since by construction,
\begin{equation}\label{2.57}
N_\mcl (x) \sim x^D (1 + 2\beta \cos (\tau \log x)) \quad \text{as } x \rightarrow +\infty, 
\end{equation}
we deduce from (the difficult part of) Theorem \ref{Thm:2.3} (the characterization of Minkowski measurability) that the fractal string $\mcl$ is {\em not} Minkowski measurable.

Indeed, because $x^{-D} N_\mcl (x)$ oscillates asymptotically between the positive constants $1-2 \beta$ and $1+2 \beta$ (in light of \eqref{2.57}), $x^{-D} N_\mcl (x)$ cannot have a limit as $x \rightarrow +\infty$; equivalently, $j^{1/D} \ell_j$ cannot have a limit as $j \rightarrow \infty$, which violates condition (ii) of Theorem \ref{Thm:2.3} and therefore shows that $\mcl$ is {\em not} Minkowski measurable, as desired.\footnote{It follows from another theorem in \cite{LapPo2} that $\mcl$ is Minkowski nondegenerate (i.e., $0 < \mcm_* (\leq) \mcm^* < \infty$) and has Minkowski dimension $D$; so that $0 < \mcm_* < \mcm^* < \infty$, by combining these two results.}

Next, a direct (but relatively involved) computation in \cite{LapMa2} (based, in particular, on the identity \eqref{2.21} and several key properties of $\zeta = \zeta (s)$ and its meromorphic continuation to the critical strip $\{ 0 < Re(s) < 1 \}$) shows that there exist nonzero constants $E_D$ and $E_\omega$ such that as $x \rightarrow +\infty$,
\begin{equation}\label{2.58}
N_\nu (x) = W (x) + E_D \zeta (D) x^D + E_\omega \zeta(\omega) x^\omega  + \overline{E}_\omega \zeta (\overline{\omega}) + o(x^D),
\end{equation} 
where $W(x) = |\Omega|_1 x$ is the Weyl term (as given by \eqref{2.50}).

Now, since $\zeta (\omega) = \zeta (\overline{\omega}) = 0$, we see that $N_\nu (x)$ has a {\em monotonic} asymptotic second term, proportional to $x^D$, as claimed:
\begin{equation}\label{2.59}
N_\nu (x) = W(x) +E_D \zeta (D) x^D + o(x^D) \quad \text{as } x \rightarrow +\infty.
\end{equation}
This shows that the fractal string $\mcl$ which we have constructed is not Minkowski measurable but that its spectral counting function has a monotonic asymptotic second term, of the order of $x^D$. Therefore, the inverse spectral problem (ISP)$_D$ does not have an affirmative answer for this value of $D \in (0,1).$ Thus the implication $(ISP)_D \Rightarrow (RH)_D$ is now established. This concludes the sketch of the proof of Theorem \ref{Thm:2.12} (and  hence also of Corollary \ref{Cor:2.13}).
\end{proof}

\begin{remark}\label{Rem:2.15} ($a$) The proof of Theorem \ref{Thm:2.12} given in \cite{LapMa2} actually shows that one can replace $o(x^D)$ by $O(1)$ in \eqref{2.58}, and hence also in \eqref{2.59}.

($b$) In the language of the mathematical theory of complex dimensions, the fractal string $\mcl$ constructed just above has three (simple) complex dimensions; namely, the Minkowski dimension $D$ and the complex conjugate pair of nonreal (principal) complex dimensions $(\omega, \overline{\omega}) = (D +i \tau, D - i \tau)$. Hence,
\begin{equation}\label{2.60}
\mcd_\mcl = \{ D, \omega, \overline{\omega} \}.
\end{equation}
Accordingly, $\mcl$ (or, equivalently, any of its geometric realizations $\Omega$ as an open subset of $\mbr$ with finite length) is fractal (in the sense of \S \ref{Sec:2.5}) and even `{\em critically fractal}' (since it is fractal in the maximal possible dimension, $D$).

($c$) Theorem \ref{Thm:2.12} (suitably interpreted), along with Corollary \ref{Cor:2.13}, has been extended to a large class of arithmetic zeta functions (and other Dirichlet series) in \cite{Lap-vF2, Lap-vF3, Lap-vF4}; see \cite[Ch. 9]{Lap-vF4}. This broad generalization relies on the mathematical theory of complex dimensions developed in those references, and especially, on the fractal explicit formulas obtained therein (see \cite[Chs. 5--6]{Lap-vF4}), of which a few examples were provided in \S \ref{Sec:2.2}.

($d$) As was alluded to earlier, the second part of the proof of Theorem \ref{Thm:2.12} (namely, the proof of the implication $(ISP)_D \Rightarrow (RH)_D$) was motivated by the intuition of the notion of complex dimensions. At the same time, it served as a powerful incentive for developing the mathematical theory of complex dimensions. Along with the statements of Theorem \ref{Thm:2.12} and Corollary \ref{Cor:2.13}, it also provided a {\em natural geometric interpretation of the} (closed) {\em critical strip} $\{0 \leq Re(s) \leq 1 \}$. In particular, the {\em midfractal case} when $D = 1/2$ corresponds to the critical line $\{Re(s) = 1/2 \}$, while the {\em least} (respectively, {\em most}) {\em fractal case} when $D=0$ (respectively, $D=1$) corresponds to the left (respectively, right) most vertical line in the closed critical strip, $\{ Re(s) = 0 \}$ (respectively, $\{ Re(s) = 1 \}$). This observation has played a key role in later work, including \cite{Lap-vF4}, [\hyperlinkcite{Lap7}{Lap7--9}] and \cite{HerLap1}. 

It is also in agreement with the author's conjecture according to which there should exist a natural fractal-like geometry whose complex dimensions coincide with the union of $\{ 1/2 \}$ and the critical zeros of $\zeta.$ Consequently, it would have Minkowski dimension $1/2$ and apart from $1/2$, the critical zeros $\rho$ of $\zeta$ would be its principal complex dimensions. It is possible that instead, this ``geometry'' would only be {\em critically subfractal} (in dimension $\frac{1}{2}$), with the critical zeros $\rho$ of $\zeta$ being precisely its complex dimensions with real part $1/2$ (assuming RH, for clarity). Namely, in that case, it would have Minkowski dimension $1$ (corresponding to the only pole of $\zeta$, which is simple), midfractal complex dimensions the critical zeros of $\zeta$, and possibly one other (simple) complex dimension at $0$ (provided it has the perfect symmetry of the completed Riemann zeta function $\xi$, as described in the next paragraph). Accordingly (and still assuming RH, for clarity), in either the former case or the latter case, this geometry would be fractal in dimension $1/2$, but not fractal in any other dimension $d \in \mbr$, with $d \neq 1/2$ (since both $0$ and $1$ are real). 

Hence (in the latter case), in algebraic geometric terms, the set $\mcd_\zeta$ of complex dimensions of this elusive fractal geometry would coincide with the divisor of the completed (or global) Riemann zeta function $\xi (s) := \pi^{-s/2} \Gamma (s/2) \zeta (s)$ (or, equivalently, with the poles of minus the logarithmic derivative, $-\xi' (s)/ \xi (s)$), leaving apart the multiplicities, which can be explicitly recovered by considering the corresponding residues. Namely, since the zeros of $\xi$ coincide with the critical zeros of $\zeta$ and $\xi$ satisfies the celebrated functional equation $\xi(s) = \xi (1-s)$ for all $s \in \mbc$, from which it follows that $\xi$ has a (simple) pole at $s=0$ and $s=1$, we would have 
\begin{equation}\label{2.61}
\mcd_\zeta = \{0,1 \} \cup \{ \text{critical zeros of } \zeta \},
\end{equation}
where each complex dimension is counted with multiplicity one. Accordingly, the sought for fractal-like geometry would be ``{\em self-dual}'' (in the sense of \cite{Lap7}), reflecting the perfect symmetry of the functional equation with respect to the critical line $\{Re(s) = 1/2 \}$. (See also \S \ref{Sec:4.4} and \cite{Lap10}.)
\end{remark}
 
\section{A Taste of the Higher-Dimensional Theory: Complex Dimensions and Relative Fractal Drums (RFDs)}\label{Sec:3}

In this section, which by necessity of concision, will be significantly shorter than would be warranted, we limit ourselves to giving a brief overview of some of the key definitions and results of the higher-dimensional theory of complex dimensions, with emphasis on several key examples illustrating them. [We refer to \cite[Ch. 3]{Lap10} for a more extensive overview, and to the book \cite{LapRaZu1} along with the accompanying papers [\hyperlinkcite{LapRaZu2}{LapRaZu2--10}] (including the two survey articles [\hyperlinkcite{LapRaZu8}{LapRaZu8--9}]) for a much more detailed account of the theory, with complete proofs.] First, however, we begin by providing a short history of one aspect of the subject.

\subsection{Brief history}\label{Sec:3.1}

For a long time, the theory of complex dimensions was restricted to the one-dimensional case (corresponding to fractal strings and arbitrary compact subsets of $\mbr$) or in the higher-dimensional case, to very special (although useful) geometries; namely, to fractal sprays \cite{Lap3}, \cite{LapPo3}, obtained as countable disjoint unions of scaled copies of one or finitely many generators, and particularly, to self-similar sprays.

An approximate tube formula was first obtained in [\hyperlinkcite{Lap-vF2}{Lap-vF2--4}] for the devil's staircase (i.e., the graph of the well-known Cantor function), from which the complex dimensions of the Cantor graph could be deduced, by analogy with the fractal tube formula for fractal strings (discussed in \S \ref{Sec:2.1}). Then, for the important example of the snowflake curve (or, equivalently, of the Koch curve), an exact fractal tube formula was obtained by the author and Erin Pearse in \cite{LapPe1} via a direct computation, based in part on symmetry considerations.\footnote{In that context, an interesting open problem remains to explicitly determine the Fourier coefficients of a nonlinear (and periodic) analog of the Cantor function arising naturally in the computation leading to the corresponding fractal tube formula.}
(See \cite[\S 12.2.1]{Lap-vF4} for an exposition.) Again, the complex dimensions of the Koch curve could be deduced by analogy with the case of fractal strings. However, no analog of the geometric zeta functions of fractal strings was used in the process, and therefore, the complex dimensions of the Koch curve (like those of the Cantor graph at that stage) could not yet be precisely defined.

Another important step was carried out by the author and Erin Pearse in [\hyperlinkcite{LapPe2}{LapPe2--3}] and significantly more generally, in joint work of those two authors and Steffen Winter in \cite{LapPeWi1}, where fractal tube formulas were obtained for a large class of fractal sprays \cite{LapPo3}, and especially, of self-similar sprays or self-similar tilings (as in [\hyperlinkcite{Lap2}{Lap2--3}], \cite{Pe}, [\hyperlinkcite{LapPe1}{LapPe1--2}], \cite{PeWi} and [\hyperlinkcite{LapPeWi1}{LapPeWi1--2}]).\footnote{See \cite[\S 13.1]{Lap-vF4} for an exposition of part of these results.}
This time, the resulting fractal tube formulas made use of certain ``tubular zeta functions'', but those zeta functions were of a rather ad hoc nature and could not be extended to more general types of fractal geometries. Also, of course, fractal sprays are rather special cases of fractals.\footnote{In one dimension, however, fractal sprays reduce to fractal strings, while the latter enables us to deal with the general case of arbitrary bounded (or, equivalently, compact) subsets of the real line $\mbr$.}

Therefore, there still remained to find appropriate fractal-type zeta functions which enabled one, in particular, to both define the complex dimensions of arbitrary bounded (or, equivalently, compact) subsets of $\mbr^N$ (for any $N \geq 1$) and to obtain fractal tube formulas valid in this general higher-dimensional setting.

Finally, in 2009, this significant hurdle was overcome when the author introduced a fractal zeta function, now called the distance zeta function, which could be used to achieve the aforementioned goals. This was only the beginning of what turned out to be a very fruitful and creative period, extending from 2009 through 2017, during which large parts of the higher-dimensional theory of complex dimensions and associated fractal tube formulas were developed by the author, Goran Radunovi\' c and Darko \v Zubrini\' c in the series of papers  [\hyperlinkcite{LapRaZu2}{LapRaZu2--9}] as well as in the nearly 700-page research book \cite{LapRaZu1}.

Our goal here is not to give a detailed account of the theory. Instead, we simply wish to highlight in the rest of this section a few definitions, results and useful examples. We refer to [\hyperlinkcite{LapRaZu2}{LapRaZu2--10}], \cite[Ch. 3]{Lap10} and, especially, to the research monograph \cite{LapRaZu1}, for a more detailed account, as well as for precise statements (and complete proofs) of the main results.

\subsection{Fractal zeta functions and relative fractal drums (FZFs and RFDs)}\label{Sec:3.2}

The theory of complex dimensions of [\hyperlinkcite{LapRaZu1}{LapRaZu1--10}] is developed for arbitrary bounded subsets $A$ of $\mbr^N$ and, more generally, for arbitrary relative fractal drums (RFDs) $(A, \Omega)$ in $\mbr^N$, for any integer $N \geq 1$, as we now briefly explain.

Given $A$ a (nonempty) subset of $\mbr^N$ and $\Omega$ a possibly unbounded open subset of $\mbr^N$ with finite volume and such that $\Omega \subseteq A_{\delta_1}$, for some $\delta_1 > 0$, the pair $(A, \Omega)$ is called a {\em relative fractal drum} (or {\em RFD}, in short) in $\mbr^N$.

Two important special cases of RFDs are ($i$) when $N=1, \Omega$ is a bounded open subset of $\mbr^N$, and $A:= \partial \Omega$, and ($ii$) when $N \geq 1$ is arbitrary, $A$ is any bounded subset of $\mbr^N$, and $\Omega := A_{\delta_1},$ for some fixed (but arbitrary) $\delta_1 > 0$. Case ($i$) just above corresponds to the (ordinary or bounded) fractal strings discussed in \S\ref{Sec:2} (but now viewed as the RFDs $(\partial \Omega, \Omega)$), while case ($ii$) corresponds to arbitrary bounded subsets $A$ of $\mbr^N$ (for any integer $N \geq 1$); see also \S \ref{Sec:3.2.2} or \S \ref{Sec:3.2.1}, respectively.

It is a simple matter to extend the definition of the Minkowski dimension and of the Minkowski content to RFDs, as we now explain. (In the sequel, given a measurable subset $B$ of $\mbr^N$, we let $|B|_N = |B|$ denote the $N$-dimensional volume or Lebesgue measure of $B$.)

Given an RFD $(A, \Omega)$ in $\mbr^N$, we define its {\em tube function} 
\begin{equation}\label{3.1}
\varepsilon \mapsto V(\varepsilon) = V_{A, \Omega} (\varepsilon) := |A_\varepsilon \cap \Omega|, \text{ for } \varepsilon >0, 
\end{equation}
and then,\footnote{If $A$ is a bounded subset of $\mbr^N$ (viewed as an RFD, see \S \ref{Sec:3.2.1}), we then simply write $V(\varepsilon) = V_A (\varepsilon)$. Note that we then have $V_A (\varepsilon) = |A_\varepsilon|$.}
its {\em upper Minkowski dimension}
\begin{align}\label{3.2}
\overline{D} &= \overline{\dim}_B (A, \Omega)\\ \notag
&= \inf \Big\{ \beta \in \mbr : V (\varepsilon) = O(\varepsilon^{N - \beta}) \quad \textnormal{as } \varepsilon \rightarrow 0^+ \Big\}\\ \notag
&= \inf \Big\{ \beta \in \mbr : \limsup_{\varepsilon \rightarrow 0^+} \frac{V(\varepsilon)}{\varepsilon^{N - \beta}} < \infty \Big\}. \notag
\end{align} 
We define similarly $\underline{D} = \underline{\dim}_B (A, \Omega)$, the {\em lower Minkowski dimension} of $(A, \Omega)$, by simply replacing the upper limit by a lower limit on the right-hand side of the last equality of \eqref{3.2}. In general, we have $\underline{D} \leq \overline{D}$ but if $\underline{D} = \overline{D}$ (which is the case for most classic fractals), we denote by $D$ this common value and call it the {\em Minkowski dimension} of $(\ao)$; the latter is then said to exist.

The {\em upper} (respectively, {\em lower}) {\em Minkowski content} of $(A, \Omega)$ is then defined (still with $V(\varepsilon) = V_{A, \Omega} (\varepsilon)$) by
\begin{equation}\label{3.3}
\mcm^* = \limsup_{\varepsilon \rightarrow 0^+} \frac{V_{A, \Omega} (\varepsilon)}{\varepsilon^{N-D}}
\end{equation}
and
\begin{equation}\label{3.4}
\mcm_* = \liminf_{\varepsilon \rightarrow 0^+} \frac{V_{A, \Omega} (\varepsilon)}{\varepsilon^{N-D}}.
\end{equation}
We clearly have $0 \leq \mcm_* \leq \mcm^* \leq \infty$. If $\mcm_* > 0$ and $\mcm^* < \infty$, we say that $(A, \Omega)$ is {\em Minkowski nondegenerate.} If, in addition, $\mcm^* = \mcm_*$ (i.e., if the limit in \eqref{3.3} and \eqref{3.4} exists and is in $(0, +\infty)$), then we denote by $\mcm$, and call the  {\em Minkowski content} of $(A, \Omega)$, this common value and say that the RFD $(A, \Omega)$ is {\em Minkowski measurable.} It is easy to check that if $(\ao)$ is Minkowski nondegenerate (and, in particular, if it is Minkowski measurable), then its Minkowski dimension $D$ exists. 

We note that the upper Minkowski dimension of an RFD can be negative or even take the value $-\infty$. For instance, let $A := \{(0,0) \}$ and for $\alpha > 1$,  let
\begin{equation}\label{3.4.1/4}
\Omega := \{ (x,y) \in \mbr^2: 0 < y < x^\alpha, x \in (0,1) \}.
\end{equation}
Then the RFD $(A, \Omega)$ in $\mbr^2$ has (relative) Minkowski dimension $D = 1 - \alpha < 0$, which takes arbitrary large negative values as $\alpha \rightarrow + \infty$. Furthermore, $(A, \Omega)$ is Minkowski measurable with (relative) Minkowski content $\mcm = 1/(1 + \alpha).$ (See \cite[Prop. 4.1.35]{LapRaZu1}.)

Further, another RFD $(A, \Omega)$ in $\mbr^2$, defined by $A:= \{(0,0) \}$ and
\begin{equation}\label{3.4.1/2}
\Omega := \{ (x,y) \in \mbr^2 : 0 < y < e^{-1/x}, 0 < x < 1 \},
\end{equation} 
is such that its Minkowski dimension $D$ exists (as in the previous example) but is no longer finite; namely, $D = -\infty$. (See \cite[Cor. 4.1.38]{LapRaZu1}.)

For a fractal string or for an arbitrary bounded subset $A$ of $\mbr^N$, however, we always have that $\underline{D} \geq 0$ and hence also, that $\overline{D} \geq 0$ (so that when the Minkowski dimension $D$ exists, then $D \geq 0$).\footnote{To see why this is true in the latter case, simply note that $A \cap A_{\delta} = A$ for every $\delta > 0.$}

RFDs extend the notion of a bounded set in $\mbr^N$ (see \S \ref{Sec:3.2.1} below), of a fractal string (when $N=1$, see \S \ref{Sec:3.2.2}) and more generally, of a fractal drum $(\partial \Omega, \Omega)$, with $\Omega$ of finite volume in $\mbr^N$ (used, e.g., in [\hyperlinkcite{Lap1}{Lap1--3}] in the case of Dirichlet boundary conditions). They are also very useful tools in the computation of the fractal zeta functions and complex dimensions of fractals (especially when some kind of self-similarity is present), by means of scaling and symmetry considerations; see \cite{LapRaZu1} and \cite{LapRaZu4}.

\begin{remark}$(${\em Relative box dimension.}$)$\label{Rem:3.1}
It is natural to wonder if there exists a notion of box dimension which is also valid for RFDs and thus, for which the corresponding values can be negative. It turns out to be the case, as is explained in a work in preparation by the authors of \cite{LapRaZu1}. This (upper) ``{\em relative box dimension}'' not only exists but also coincides with the (upper) relative Minkowski dimension of the given RFD $(\ao)$ in $\mbr^N$, as in the usual case of bounded subsets $A$ of $\mbr^N$ (see, e.g., \cite{Fa1}).\footnote{It then follows that in the case of a bounded open set with an external cusp $p$ (such as in \eqref{3.4.1/4} and \eqref{3.4.1/2}, respectively, and the text surrounding them), the relative box dimension of the RFD $(\{p\}, \Omega)$ is negative or even equal to $-\infty$.} However, its definition now requires a {\em fractional counting} of the boxes involved (corresponding essentially to the relative volume of those boxes).

Furthermore, for RFDs of the form $(\partial \Omega, \Omega)$, with $\Omega$ a John domain in $\mbr^N$ $(N \geq 1)$,\footnote{See footnote \ref{Fn:38} and the references therein, including \cite{John}, \cite{CarlJonYoc}, \cite{McMul} and \cite{ChuOsgPomm}, for example.} one can show that as in the usual case of bounded subsets of $\mbr^N$, the relative box dimension of the RFD exists and not only coincides with the relative Minkowski dimension, $\dim_B (\partial \Omega, \Omega)$, but is also nonnegative. 

We note that this notion of (possibly negative) relative box dimension (and that of relative Minkowski dimension, with which it coincides) could perhaps be used to make sense of the heuristic and elusive notion of ``negative dimension'' and ``degree of emptiness'' used (or sought for) in \cite{ManFra}; see also \cite{Tri4}, in the context of the Hausdorff (rather than of the Minkowski) dimension.
\end{remark}

We now define as follows the {\em distance zeta function} of an RFD $(A, \Omega)$ in $\mbr^N$:
\begin{equation}\label{3.5}
\zeta_{A, \Omega} (s) = \int_\Omega d(x, A)^{s-N} dx,
\end{equation}
for all $s \in \mbc$ with $Re(s)$ sufficiently large.\footnote{More specifically, according to property ($b$) of \S \ref{Sec:3.3}, this means that \eqref{3.5} holds for all $s \in \mbc$ with $Re(s) > \overline{D}$, where $\overline{D} := \overline{\dim}_B (\ao)$, and that $\overline{D}$ is best possible.\label{Fn:48}}

Another very useful fractal zeta function is the closely related {\em tube zeta function} of the RFD $(A, \Omega)$,
\begin{equation}\label{3.6}
\widetilde{\zeta}_{A, \Omega} (s) = \int_0^\delta V_{A, \Omega} (\varepsilon) \varepsilon^{s-N} \frac{d \varepsilon}{\varepsilon},
\end{equation}
for all $s \in \mbc$ with $Re(s)$ sufficiently large\footnote{This can be interpreted exactly as in footnote \ref{Fn:48} just above, provided $\overline{D} < N$.} and for $\delta > 0$ fixed but arbitrary.\footnote{It turns out that the poles of $\tzao$ (i.e., the complex dimensions of $(A, \Omega)$) are independent of the choice of $\delta > 0$ since changing the value of $\delta$ in \eqref{3.6} amounts to adding an entire function to $\widetilde{\zeta}_{A, \Omega}.$}

One can show that we have the following {\em functional equation}:
\begin{equation}\label{3.7}
\zeta_{A, A_\delta \cap \Omega} (s) = \delta^{s-N} |A_\delta \cap \Omega| + (N-s) \, \widetilde{\zeta}_{A, \Omega} {(s)},
\end{equation}
from which one deduces that $\zeta_{A, \Omega}$ and $\widetilde{\zeta}_{A,\Omega}$ contain essentially the same information, provided $\overline{D} < N$ and $\delta >0$.\footnote{Note that it follows at once from \eqref{3.2} that we always have that $\overline{D} \leq N$, since $|\Omega| < \infty$ according to the definition of an RFD.} 
Indeed, under this very mild condition, $\zeta_{A, \Omega}$ has a meromorphic extension in a given domain $U \subseteq \mbc$ if and only if $\widetilde{\zeta}_{A, \Omega}$ does, and in this case, $\zao$ and $\tzao$ have the same poles (i.e., the same visible complex dimensions) in $U$, with the same multiplicities. Furthermore, still in this case, if $\omega \in U$ is a simple pole of $\widetilde{\zeta}_{A, \Omega}$, then it is also a simple pole of $\zeta_{A, A_\delta \cap \Omega}$ and 
\begin{equation}\label{3.8}
\res (\widetilde{\zeta}_{A, \Omega}, \omega) = \frac{1}{N- \omega} \, \res (\zeta_{A, A_\delta \cap \Omega}, \omega).
\end{equation}
We stress that since according to the definition of an RFD, $\Omega \subseteq A_{\delta_1}$ for some $\delta_1 > 0$, we have $A_\delta \cap \Omega = \Omega$ for any $\delta \geq \delta_1$. Hence, under that condition, we can replace $\zeta_{A, A_\delta \cap \Omega}$ by $\zeta_{A, \Omega}$ in both \eqref{3.7} and \eqref{3.8}, as well as in the above statements. The latter equations \eqref{3.7} and \eqref{3.8} then become (for every $\delta \geq \delta_1$), respectively,
\begin{equation}\label{3.8.1/4}
\zao (s) = \delta^{s-N} |\Omega| +(N-s) \tzao (s)
\end{equation}
and 
\begin{equation}\label{3.8.1/2}
\res (\tzao, \omega) = \frac{1}{N-\omega} \res (\zao, \omega).
\end{equation}
In the sequel, we will assume, most of the time implicitly, that $\delta \geq \delta_1$ (and $\overline{D} < N$); so that \eqref{3.7} can be written in the simpler form of the functional equation \eqref{3.8.1/4}.

Given a domain $U \subseteq \mbc$ containing the ``critical line'' $\{ Re(s) = \overline{D} \}$ and assuming that $\zeta_{A, \Omega}$ has a (necessarily unique) meromorphic continuation to $U$, the poles of $\zeta_{A, \Omega}$ are called the {\em visible complex dimensions} of the RFD $(A, \Omega)$. If $U = \mbc$ (or if there is no ambiguity as to the choice of $U$), we simply call them the {\em complex dimensions} of $(A, \Omega)$.

The set (really, multiset) of (visible) {\em complex dimensions} of $(A, \Omega)$ is denoted by $\mcd (\zeta_{A, \Omega})$ (or $\mcd (\zeta_{A, \Omega}; U)$ if we want to specify $U$) and when $U = \mbc$, we also use the notation $\dim_\mbc (A, \Omega) := \mcd (\zeta_{A, \Omega}; \mbc$).

We adopt a similar notation when $\zao$ is replaced by $\tzao$. In fact, as we alluded to earlier, we will always assume that $\overline{D} < N$; so that $\mcd (\tzao; U) = \mcd (\zao; U),$ in which case we also denote by $\mcd_{\ao} (= \mcd_{\ao} (U))$ the set (or really, multiset) of (visible) complex dimensions of $(\ao)$. \\
 
\subsubsection{The special case of bounded sets in $\mbr^N$}\label{Sec:3.2.1}

If $A$ is a bounded subset of $\mbr^N$, then for any fixed $\delta >0$ and for all $s \in \mbc$ with $Re(s)$ large enough (really, for all $s \in \mbc$ with $Re(s) > \overline{D}$),
\begin{equation}\label{3.9}
\zeta_A (s) := \zeta_{A,A_\delta} (s) = \int_{A_\delta} d(x, A)^{s-N} dx
\end{equation}
and 
\begin{equation}\label{3.10}
\widetilde{\zeta}_A (s) := \zeta_{A,A_\delta} (s) = \int_0^\delta V_A (\varepsilon) \varepsilon^{s-N} \, \frac{d \varepsilon}{\varepsilon},
\end{equation}
where $V_A (\varepsilon) := |A_\varepsilon|$ (as mentioned earlier). Then, \eqref{3.8.1/4} reduces to the following simpler functional equation, valid for any $\delta > 0$:
\begin{equation}\label{3.11}
\zeta_A (s) = \delta^{s-N} |A_\delta| + (N-s) \, \widetilde{\zeta}_A (s).
\end{equation}
Furthermore, \eqref{3.8.1/2} becomes (provided $\overline{D} < N$ and assuming that $\omega \in U$ is a simple pole of $\widetilde{\zeta}_A$ and hence also, of $\zeta_A$):
\begin{equation}\label{3.12}
\res (\widetilde{\zeta}_A, \omega) = \frac{1}{N-\omega} \, \res (\zeta_A, \omega).
\end{equation}

Moreover, still provided $\overline{D} < N$, $\zeta_A$ has a meromorphic continuation to a given domain $U \subseteq \mbc$ if and only if $\widetilde{\zeta}_A$ does. In this case, $\zeta_A$ and $\widetilde{\zeta}_A$ have the same poles in $U$, called the {\em visible complex dimensions} of $A$ and denoted by $\mcd_A (U) = \mcd (\zeta_A; U) = \mcd (\widetilde{\zeta}_A; U)$ or, when $U = \mbc$ or when no ambiguity may arise, by $\mcd_A = \dim_\mbc A = \mcd (\zeta_A) = \mcd (\widetilde{\zeta}_A)$.  Indeed, the complex dimensions of the RFD $(A, A_\delta)$ (i.e., the complex dimensions of $A$) can be defined indifferently as the poles of $\zeta_A$ or of $\widetilde{\zeta}_A$. In addition, they are independent of the choice of $\delta > 0$ because changing the value of $\delta$ in \eqref{3.9} or \eqref{3.10} amounts to adding an entire function to the original distance zeta function $\zeta_A$ or tube zeta function $\widetilde{\zeta}_A$, respectively.

Finally, we note that since $d (\cdot, A) = d(\cdot, \overline{A}),$ we have that $A_\delta = (\overline{A})_\delta$ and hence, $V_A = V_{\overline{A}}$; so that $\zeta_A = \zeta_{\overline{A}}$ and $\widetilde{\zeta}_A = \widetilde{\zeta}_{\overline{A}}$. As a result, by simply replacing $A$ by its closure $\overline{A}$ in $\mbr^N$, we may as well assume without loss of generality that $A \subseteq \mbr^N$ is compact instead of just being bounded.\\

\subsubsection{The special case of fractal strings}\label{Sec:3.2.2}

Let $N=1$ and $\Omega$ be a bounded open subset of $\mbr$ (or more generally, an open subset of finite length in $\mbr$). Then, we view the RFD $(\partial \Omega, \Omega)$ in $\mbr$ as a geometric realization of the fractal string $\mcl = (\ell_j)_{j=1}^\infty$, the sequence of lengths of the connected components (i.e., open intervals) of $\Omega$, as in \S\ref{Sec:2}.

Remarkably, the distance zeta function $\zeta_{\partial \Omega, \Omega}$ of the RFD $(\partial \Omega, \Omega)$ and the geometric zeta function $\zeta_\mcl$ of the fractal string $\mcl = (\ell_j)_{j \geq 1}$ (see \S\ref{Sec:2}, Equation \eqref{2.1}) are related via the following very simple functional equation:
\begin{equation}\label{3.13}
\zeta_{\partial \Omega, \Omega} (s) = 2^{1-s} \frac{\zeta_\mcl (s)}{s},
\end{equation}
valid initially within the open half-plane $\{Re(s) > \overline{D} \},$ where
\begin{equation}\label{3.13.1/2}
\overline{D} = \overline{\dim}_B (\partial \Omega, \Omega) = \overline{\dim}_B \mcl = D(\zeta_{\partial \Omega, \Omega}) = D(\zeta_\mcl),
\end{equation}
and then (upon analytic continuation), within any domain $U \subseteq \mbc$ containing the closed half-plane $\{Re(s) \geq \overline{D} \}$ to which $\zeta_\mcl$ (and hence also, $\zeta_{\partial \Omega, \Omega}$) can be meromorphically continued.

One therefore deduces from \eqref{3.13} that the visible poles of $\zeta_{\partial \Omega, \Omega}$ in $U$ (i.e., the visible complex dimensions of the RFD $(\partial \Omega, \Omega)$) are the same as those of $\zeta_\mcl$, except for the fact that (provided $0 \in U$) $0$ is always a visible pole of $\zeta_{\partial \Omega, \Omega}$ (or more precisely, has multiplicity $m+1$, where $m$ is possibly equal to zero and is defined as the multiplicity of $0$ as a pole of $\zeta_\mcl$). Furthermore, $\zeta_\mcl$ has a meromorphic continuation to $U$ if and only if $\zeta_{\partial \Omega, \Omega}$ does, and in this case, the corresponding residues (at a simple pole $\omega \in U$, $\omega \neq 0$, of either $\zeta_\mcl$ or $\zeta_{\partial \Omega, \Omega}$) are related by the following identity:\footnote{For the simplicity of exposition, we assume throughout this discussion in \S \ref{Sec:3.2.1} that $0$ is not a zero of $\zeta_\mcl: \zeta_\mcl (0) \neq 0$. Otherwise, we would have to make use of the notion of a divisor (in essence, the multiset of zeros and poles) of a meromorphic function, to be discussed and used later in \S\ref{Sec:3} and \S\ref{Sec:4}.}
\begin{equation}\label{3.14}
\res(\zeta_{\partial \Omega, \Omega}, \omega) = \frac{2^{1-\omega}}{\omega} \, \res (\zeta_\mcl, \omega).
\end{equation}
If one starts instead with a given (bounded) fractal string $\mcl = (\ell_j)_{j \geq 1}$ such that $\sum_{j \geq 1} \ell_j < \infty$, then all the above facts and identities are independent of the choice of the geometric realization $\Omega$ of $\mcl$ as a bounded open subset of $\mbr$ and hence, of the choice of the RFD $(\partial \Omega, \Omega)$ associated with $\mcl$.

The above simple observations enable one, in particular, to view the theory of complex dimensions and the associated fractal tube formulas developed in \cite{LapRaZu1} and [\hyperlinkcite{LapRaZu4}{LapRaZu4,6}] as containing (as a very special case) the corresponding theory of fractal strings (as developed in \cite{Lap-vF4}). It also enables us to significantly simplify the statement of the fractal tube formulas for fractal strings\footnote{These formulas were briefly discussed in \S \ref{Sec:2.1} and will be revisited in the relevant part of \S \ref{Sec:3.5.3}.} 
in the case when $0$ happens to be a visible pole of $\zeta_\mcl$. Suffices to say here that, as a result, $\zeta_{\partial \Omega, \Omega}$ and therefore $\mcd_{\partial \Omega, \Omega} = \mcd (\zeta_{\partial \Omega, \Omega})$ should be considered as the fractal zeta function and the complex dimensions of a fractal string $\mcl$ or $\Omega$ (rather than $\zeta_\mcl$ and $\mcd_\mcl = \mcd (\zeta_\mcl)$, respectively).

It follows from the functional equation \eqref{3.13} and the above discussion that $\mcd_\mcl = \mcd (\zeta_\mcl)$ and $\mcd_{\partial \Omega, \Omega} = \mcd (\zeta_{\partial \Omega, \Omega}$), respectively the set of complex dimensions of the fractal string $\mcl$ (in the sense of \cite{Lap-vF4} and \S\ref{Sec:2} above) and the set of complex dimensions of the associated RFD $(\partial \Omega, \Omega)$ (in the sense of \cite{LapRaZu1} and of the present \S \ref{Sec:3.2}), for any geometric realization $\Omega \subseteq \mbr$ of $\mcl$, are connected via the following relation (between multisets):
\begin{equation}\label{3.14.1/2}
\mcd_{\partial \Omega, \Omega} = \mcd_\mcl \cup \{ 0\},
\end{equation}
in the sense that if $0$ is a pole of $\zeta_\mcl$ (i.e., a complex dimension of $\mcl$, in the sense of \S\ref{Sec:2}) of multiplicity $m \geq 0$, then it is also a pole of $\zeta_{\partial \Omega, \Omega}$ (i.e., a complex dimension of the RFD $(\partial \Omega, \Omega)$, in the sense of \S \ref{Sec:3.2}) of multiplicity $m + 1$. All the other (i.e., nonzero) complex dimensions of the fractal string $\mcl$ and of the RFD $(\partial \Omega, \Omega)$ coincide, and have the same multiplicities.
 
\subsection{A few key properties of fractal zeta functions}\label{Sec:3.3}

In this section, we discuss several important properties of the fractal zeta functions (that is, of the distance and tube zeta functions) of relative fractal drums (RFDs) and, in particular, of bounded sets in $\mbr^N$. Throughout, $N \geq 1$ is a fixed (but arbitrary) positive integer. We will stress the case of distance zeta functions since (except for a few small differences) fractal tube zeta functions have entirely analogous properties.\\

\subsubsection{Abscissa of convergence and holomorphicity}\label{Sec:3.3.1}

Let $(A, \Omega)$ be an RFD in $\mbr^N$ and consider its distance and tube zeta functions, $\zeta_{A, \Omega}$ and $\widetilde{\zeta}_{A, \Omega}$. In order to simplify the exposition, we still assume that $\overline{D} < N$.\footnote{We also implicitly assume throughout \S \ref{Sec:3.3} (and beyond) that $\delta \geq \delta_1$, where $\Omega \subseteq A_{\delta_1}$.}\\

{\bf (\emph{a})} ({\em Abscissa of convergence}).\label{gls:ac36} It is shown in \cite{LapRaZu1} (see also [\hyperlinkcite{LapRaZu2}{LapRaZu2,4}]) that the {\em abscissa of convergence} $\sigma = D(\zao) = D(\tzao)$ of $\zeta_{A, \Omega}$ and $\widetilde{\zeta}_{A, \Omega}$ is well defined and coincides with the upper Minkowski dimension $\overline{D} = \overline{\dim}_B (A, \Omega)$ of $(\ao)$. Namely, the largest right half-plane $\{Re(s) > \alpha \}$ (with $\alpha \in (-\infty, N])$ in which the Lebesgue integral defining $\zao (s)$ (respectively, $\tzao (s)$) in \eqref{3.5} (respectively, in \eqref{3.6}) converges is $\{Re(s) > \overline{D} \}$, called the {\em half-plane of} ({\em absolute}) {\em convergence} of $\zao$ (respectively, of $\tzao$); for example, 
\begin{align}\label{3.15}
\sigma :&= \inf \Bigg\{ \alpha \in \mbr: \int_\Omega d(x, A)^{\alpha - N} dx < \infty \Bigg\}\\
&= \overline{D} := \overline{\dim}_B (A, \Omega).\notag
\end{align}

A part of the proof of this key fact relies on a suitable extension of a result obtained in \cite{HarvPol} for a completely different purpose (the study of the singularities of the solutions of certain linear partial differential equations).\\

{\bf (b)} ({\em Holomorphicity}). It is also shown in the aforementioned references that $\zao$ (and hence, in light of the functional equation \eqref{3.8.1/4}, also $\tzao$) is holomorphic in the open right half-plane $\{Re(s) > \overline{D} \}$, where $\overline{D} := \overline{\dim}_B (A, \Omega)$, as before.

In addition, under mild assumptions (namely, if $D$, the Minkowski dimension of $(A, \Omega)$, exists and $\mcm_* > 0$), then this right half-plane $\{Re(s) > \overline{D} \}$ is also optimal from the point of view of holomorphicity because one can show that $\zao (s) \rightarrow + \infty,$ as $s \rightarrow \overline{D}$, with $s \in \mbr$ and $s > \overline{D}$. Therefore, the half-plane of (absolute) convergence and the {\em half-plane of holomorphic continuation} coincide and are equal to $\{ Re(s) >D \}$, in this case. One then says that $D$, the (relative) Minkowski dimension of $(A, \Omega)$, also coincides with $D_{hol} (\zao) = D_{hol} (\tzao)$, the common {\em abscissa of holomorphic continuation} of $\zao$ and of $\tzao$.\footnote{In general, however, we only have that
\begin{equation}\label{3.16}
D_{hol} (\zao) = D_{hol} (\tzao) \leq \overline{D}.
\end{equation}\label{Fn:58}}

Finally, one can always  (complex) differentiate $\zao$ or $\tzao$ under the integral sign and as many times as one wants. For example, provided $Re(s) > \overline{D}$,
\begin{equation}\label{3.17}
\zeta'_{A, \Omega} (s) = \int_\Omega d(x,A)^{s-N} \log (d(x, A)) dx.
\end{equation}

Naturally, all of the above properties (in part ($a$) or ($b$)) hold without change for $\zeta_A$ and $\widetilde{\zeta}_A$, the distance and tube zeta functions of bounded sets $A$ in $\mbr^N$ (initially defined for $Re(s) > \overline{D}$ by \eqref{3.9} and \eqref{3.10}, respectively) by considering  the associated RFDs $(A, A_\delta)$, for a fixed but arbitrary $\delta >0$.

It is not known whether we can always have an equality in \eqref{3.16} or, equivalently, if $D_{hol} (\zao) = \overline{D}$. By contrast, for fractal strings, the analog of this property always holds;\footnote{More specifically, provided the (bounded) fractal string $\mcl = (\ell_j)_{j \geq 1}$ is infinite (i.e., has an infinite number of lengths), then $D_{hol} (\zeta_\mcl) = D (\zeta_\mcl) = \overline{D}$, the (upper) Minkowski dimension of $\mcl$.}
see \cite{Lap-vF4} or \cite[\S 2.1.4]{LapRaZu1}. Note that the property stated in part ($a$) above is the exact counterpart of Theorem \ref{Thm:2.1}.

Since $\zao$ and $\tzao$ admit a holomorphic continuation (necessarily unique) to $\{Re(s) > \overline{D} \}$), they cannot have any pole there. Consequently, the set of complex dimensions of $(A,\Omega)$ must be contained in $\{ Re(s) \leq \overline{D}\}$. More specifically, if the domain $U \subseteq \mbc$ contains the vertical line $\{ Re(s) = \overline{D} \}$ (or, equivalently, contains the closed half-plane $\{Re(s) \geq \overline{D} \}$),\footnote{{\em Caution}: There exist RFDs for which such domains $U$ do not exist; in fact, in \cite[\S 4.6]{LapRaZu1}, one constructs RFDs in $\mbr^N$ (and also compact sets in $\mbr^N$, along with fractal strings) such that {\em every} point of the vertical line $\{Re(s) = \overline{D} \}$ is a (nonremovable) singularity of $\zao$ and of $\tzao$.} 
then $\zao$ (and hence also, $\tzao$) can be meromorphically continued to a domain $U \subseteq \mbc$, and 
\begin{equation}\label{3.18}
\mcd (\zao ; U) = \mcd (\tzao; U) \subseteq \{ Re(s) \leq \overline{D} \}.
\end{equation}
In addition, if $\overline{D} \in U$, which is certainly the case if $U$ also contains the vertical line $\{Re(s) = \overline{D} \}$, then 
\begin{equation}\label{3.19}
\overline{D} := \overline{\dim}_B (A, \Omega) = \max \{ Re(\omega): \omega \in \mcd (\zao ; U \}.
\end{equation}
In other words, {\em the} ({\em upper}) {\em Minkowski dimension of $(A, \Omega)$ is equal to the maximal real part of the} ({\em visible}) {\em complex dimensions of $(A, \Omega)$.}

In the literature on this topic (see [\hyperlinkcite{LapRaZu1}{LapRaZu1--10}]), the vertical line $\{ Re(s) = \overline{D} \}$ is often called the ``{\em critical line}'' (by analogy with the terminology adopted for the Riemann zeta function, but clearly, with a different meaning). Assume that $\zao$ (or, equivalently, $\tzao$) can be meromorphically extended (necessarily uniquely) to a domain $U$ containing the critical line $\{ Re(s) = \overline{D} \}$, then
\begin{align}\label{3.20}
\dim_{PC} (\ao) = \mcd_{PC} (\ao) &= \mcd_{PC} (\zao)(= \mcd_{PC} (\tzao))\notag \\ 
&:= \{ \omega \in \mcd (\zao; U) : Re (\omega) = \overline{D} \}
\end{align}
is called the set (really, multiset) of {\em principal complex dimensions} of the RFD $(A, \Omega)$. Clearly, $\mcd_{PC} (\ao) = \mcd_{PC} (\zao) = \mcd_{PC} (\tzao)$ is independent of the choice of the domain $U$ satisfying the above assumption. (An entirely analogous notation and terminology is used for a bounded set $A$ in $\mbr^N$; namely, $\dim_{PC} (A) = \mcd_{PC} (A) = \mcd_{PC} (\zeta_A) = \mcd_{PC} (\widetilde{\zeta}_A)$ denotes the set of {\em principal complex dimensions} of $A$.)\\

\subsubsection{Meromorphic continuation and Minkowski content}\label{Sec:3.3.2}

In this subsection, we state a few results concerning the existence of a meromorphic continuation in a suitable region of the distance or tube zeta function of an RFD $(A, \Omega)$, along with related results concerning the (upper, lower) Minkowski content of $(A, \Omega)$. We will consider both the Minkowski measurable case and the (log-periodic) Minkowski nonmeasurable case, for which the residue evaluated at $D$ (the Minkowski dimension of $(\ao)$) of the given fractal zeta function is directly connected to the Minkowski content or, respectively, the average Minkowski content of $(\ao)$. For complete proofs of those results and of related results about the existence of a meromorphic extension for various classes of RFDs (or of bounded sets) in $\mbr^N$, we refer to \S \ref{Sec:2.2}, \S \ref{Sec:2.3}, \S \ref{Sec:3.5}, \S \ref{Sec:3.6} and \S 4.5 of \cite{LapRaZu1}, as well as to [\hyperlinkcite{LapRaZu3}{LapRaZu3--4}].

We shall state the results for RFDs in $\mbr^N$; they, of course, specialize to the case of bounded sets in $\mbr^N$, as explained in \S \ref{Sec:3.2.1}. (The proofs are the same in the special case of bounded sets as in the general situation of RFDs.) Throughout, as in \S \ref{Sec:3.2}, $(\ao)$ is an RFD such that $\overline{D} < N$, even though this is only needed to easily state simultaneously the results both for the distance and the tube zeta functions.\footnote{We also suppose, as before, that $\delta >0$ is such that $\delta \geq \delta_1$, where $\Omega \subseteq A_{\delta_1}$, which can always be assumed without loss of generality.}\\

{\bf (\emph{a})} ({\em Minkowski content and residue}). We begin by stating a simple result according to which, if $\zeta_{\ao}$ (and hence also, $\tzao$, in light of \eqref{3.8.1/4}) has a meromorphic extension to a connected open neighborhood of $D$ (where the Minkowski dimension $D = D_{(\ao)}$ of $(\ao)$ is assumed to exist), and if $\mcm^* < \infty$, then\footnote{Clearly, if $\mcm_* = 0$ or $\mcm^* = +\infty$, then the corresponding inequality in \eqref{3.20.1/2} or in \eqref{3.21} is trivial. Hence, we may as well assume that $(\ao)$ is Minkowski nondegenerate in order to obtain the full strength of the result.}
\begin{equation}\label{3.20.1/2}
(N-D) \mcm_* \leq \res(\zao, D) \leq (N-D) \mcm^*,
\end{equation}
$s=D$ is a simple pole of $\zao$ (and thus also of $\tzao$) and\footnote{The value of the residue $\res(\tzao, D)$ is independent of the choice of $\delta >0$ in the definition \eqref{3.6} of $\tzao$, as well as of the choice of $\delta_1$ implicit in \eqref{3.5} and \eqref{3.6}. And likewise for the values of the residues $\res (\zeta_A, D)$ and $\res (\tilde{\zeta}_A, D)$ in the definition of $\zeta_A$ and $\tilde{\zeta}_A$ in \eqref{3.9} and \eqref{3.10}, in the counterpart of this result for a bounded subset $A$ of $\mbr^N$; recall from \eqref{3.9} that $\zeta_A := \zeta_{A, A_\delta}$.}
\begin{equation}\label{3.21}
\mcm_* \leq \res(\tzao, D) \leq \mcm^*.
\end{equation}

In particular, if $(\ao)$ is assumed to be Minkowski measurable, then
\begin{equation}\label{3.22}
\res (\zao, D) = (N-D) \mcm
\end{equation}
and
\begin{equation}\label{3.23}
\res (\tzao, D) = \mcm,
\end{equation}
the Minkowski content of $(A, \Omega)$.

\begin{example}({\em Cantor sets}.)\label{Ex:3.1}
As an illustration of the above result (as well as of the result stated in part ($c$) below), we consider the generalized Cantor set $A = C^{(a)} \subseteq [0,1]$, defined much like the usual ternary Cantor $C= C^{(1/3)}$, where the parameter $a$ lies in $(0, 1/2)$. Then, $D$ exists and $A$ is Minkowski nondegenerate but is not Minkowski measurable. More precisely, $D = \dim_B A = \log_{a^{-1}} 2$. Furthermore,
\begin{equation}\label{3.24}
\mcm_* = \frac{1}{D} \Big( \frac{2D}{1-D} \Big)^{1-D}, \quad \mcm^* = 2(1-a) \Big(\frac{1}{2} - a \Big)^{D-1}
\end{equation}
and
\begin{equation}\label{3.25}
\res (\widetilde{\zeta}_A, D) = \frac{2}{\log 2} \Big(\frac{1}{2} - a \Big)^D.
\end{equation}
Moreover, we have strict inequalities in \eqref{3.21} and \eqref{3.20.1/2} for this example:
\begin{equation}\label{3.26}
0 < \mcm_* < \res (\tzad) < \mcm^* < \infty,
\end{equation}
and analogously for $\res(\zeta_A, D)$. Also, since $D < 1$, we have that
\begin{equation}\label{3.27}
\res (\zeta_A, D) = (1-D) \res (\widetilde{\zeta}_A, D),
\end{equation}
with $\res (\widetilde{\zeta}_A, D)$ given by \eqref{3.25}.
\end{example}

We could discuss analogously the examples of the Sierpinski gasket and the Sierpinski carpet. In particular, we would find that the inequalities in \eqref{3.20.1/2} and \eqref{3.21} are also strict in those two cases, because the  classic Sierpinski gasket and  carpet are both Minkowski nondegenerate and not Minkowski measurable, as can be checked via a direct computation.

\begin{example}({\em $a$-string}.)\label{Ex:3.2}
Another simple illustration of the above result stated in part ($a$) (as well as of the result stated in part ($b$) below) is the $a$-string discussed in \S\ref{Sec:2} (see Example \ref{Ex:2.8} in \S \ref{Sec:2.5}). Recall that for any $a >0,$
\begin{equation}\label{3.27.1/4}
\Omega = \Omega_a = \bigcup_{j=1}^\infty ((j+1)^{-a}, j^{-a})
\end{equation}
and so (with $\partial \Omega_a := \partial (\Omega_a)$)
\begin{equation}\label{3.27.1/2}
\partial \Omega_a = \{ j^{-a}: j \geq 1 \} \cup \{ 0 \}
\end{equation}
and $\mcl_a = (\ell_j)_{j \geq 1}$, with $\ell_j := j^{-a} - (j+1)^{-a}$
for each $j \geq 1$. Furthermore, the RFD $(\partial \Omega_a, \Omega_a)$ (or, equivalently, the fractal string $\mcl_a$) is Minkowski measurable with Minkowski content
\begin{equation}\label{3.27.3/4}
\mcm =\frac{2^{1-D} a^D}{1-D},
\end{equation}
where 
\begin{equation}\label{3.27.4/5}
D = \dim_B (\partial \Omega_a, \Omega_a) = D_{\mcl_a} = \frac{1}{a + 1}. 
\end{equation}
Moreover, it follows from a direct computation, left as an exercise for the interested reader, that
\begin{equation}\label{3.27.5/6}
\res(\zeta_{\mcl_a}, D) = D a^D,
\end{equation} 
which since (in light of \eqref{3.14})
\begin{equation}\label{3.27.6/7}
\res(\zeta_{\partial \Omega_a, \Omega_a}, D) = \frac{2^{1-D}}{D} \res(\zeta_{\mcl_a},D),
\end{equation}
is in agreement with the exact counterpart for RFDs of \eqref{3.27}; namely, 
\begin{equation}\label{3.34.1/2}
\res (\zeta_{\partial \Omega_a, \Omega_a}) = (1-D) \res (\widetilde{\zeta}_{\partial \Omega_a, \Omega_a}, D). 
\end{equation}
Indeed, by combining \eqref{3.27.3/4}--\eqref{3.27.6/7}, we obtain that 
\begin{align}\notag
(1-D) \mcm &= 2^{1-D} a^D = \frac{2^{1-D}}{D} (D a^D) \\ \notag
&= \frac{2^{1-D}}{D} \res (\zeta_{\mcl_a}, D) = \res (\zeta_{\partial \Omega_a, \Omega_a}, D),
\end{align}
as desired.
\end{example}
\smallskip
{\bf (\emph{b})} ({\em Existence of a meromorphic extension}: {\em Minkowski measurable case}). Let $(A, \Omega)$ be an RFD in $\mbr^N$ such that there exists $\alpha > 0$, $\mcm \in (0, +\infty)$ and $D \geq 0$ such that\\
\begin{equation}\label{3.28}
V_{A, \Omega} (\varepsilon) := |A_\varepsilon \cap \Omega| = \varepsilon^{N-D} (\mcm + O (\varepsilon^\alpha)) \quad \text{as } \varepsilon \rightarrow 0^+.
\end{equation}
\\
Then, $\dim_B (\ao)$ exists and $\dim_B (\ao) = D$. Furthermore, $(\ao)$ is Minkowski measurable with Minkowski content equal to $\mcm$. Moreover, the distance zeta function $\zao$ has for abscissa of convergence $D$ and possesses a (necessarily unique) meromorphic continuation (still denoted by $\zeta_{\ao}$, as usual) to (at least) the open right half-plane $\{ Re(s) > D - \alpha \}$; that is, $D_{mer} (\zeta_{\ao}) \leq D - \alpha$, where $D_{mer} (\zeta_{\ao})$ is the {\em abscissa of meromorphic continuation} of $\zeta_{\ao}$ (defined much as the abscissa of holomorphic continuation, except for the adjective ``holomorphic'' replaced by ``meromorphic''). The only pole of $\zeta_{A, \Omega}$ (i.e., the only visible complex dimension of $(\ao)$) in this half-plane is $s=D$; it is simple and $\res(\zeta_{\ao}, D) = (N-D) \mcm.$

Clearly, the same result holds for $\widetilde{\zeta}_{\ao}$, the tube (instead of the distance) zeta function of $(\ao)$, except for the fact that in that case, $\res (\tzao, D) = \mcm$.

\begin{exercise}\label{Exc:3.3}
($i$) Show that the hypotheses of part ($b$) are satisfied for the $a$-string of Example \ref{Ex:3.2}, viewed as the RFD $(\partial \Omega_a, \Omega_a)$.

($ii$) Prove via a direct computation that $\zeta_{\mcl_a}$ has a meromorphic continuation to all of $\mbc$ with (simple) poles at $s = D = 1/(a + 1)$ and (at possibly a subset of) $\{ -D, -2D, \cdots, -nD, \cdots: n \geq 1 \}$. (For a complete answer, see \cite[Thm. 6.21 and its proof]{Lap-vF4}.) Deduce that 
\begin{equation}\notag
\mcd_{\mcl_{a}} = \mcd (\zeta_{\partial \Omega_a, \Omega_a}) \subseteq \{ D, 0, -D, -2D, \cdots, -nD, \cdots: n \geq 1 \}, 
\end{equation}
where all the complex dimensions are simple and $D$ and $0$ are always complex dimensions of the fractal string $\mcl_a$ or equivalently, of the RFD $(\partial \Omega_a, \Omega_a)$.
\end{exercise}

{\bf (\emph{c})} ({\em Existence of a meromorphic extension}: {\em Minkowski nonmeasurable case}). Let $(A, \Omega)$ be an RFD in $\mbr^N$ such that there exists $\alpha > 0, D \in \mbr$ and a nonconstant periodic function $G$ with minimal period $T>0$, satisfying \\
\begin{equation}\label{3.29}
V_{A, \Omega} (\varepsilon) := |A_\varepsilon \cap \Omega| = \varepsilon^{N-D} (G(\log \varepsilon^{-1}) + O (\varepsilon^\alpha)) \quad \text{as } \varepsilon \rightarrow 0^+.  
\end{equation}
\\
Then, $\dim_B (A, \Omega)$ exists and $\dim_B (A, \Omega) = D$. Furthermore, $G$ is continuous and $A$ is Minkowski nondegenerate with lower and upper Minkowski contents respectively given by 
\begin{equation}\label{3.30}
\mcm_* = \min G \, \text{ and } \,  \mcm^* = \max G.
\end{equation}
Moreover, the tube zeta function $\tzao$ also has for abscissa of convergence $D(\tzao) = D$ and possesses a (necessarily unique) meromorphic extension (still denoted by $\tzao$) to (at least) the open right half-plane $\{ Re(s) > D - \alpha \}$; that is, $D_{mer} (\tzao) \leq D - \alpha$.\footnote{We first state the results for $\tzao$ because they are more elegantly written in this situation; we will then mention the few changes needed for the corresponding statements about $\zao$.}

In addition, the set of all the poles of $\tzao$ (i.e., the set of visible complex dimensions of the RFD $(A, \Omega)$) in this half-plane is given by
\begin{equation}\label{3.31} 
\mcd_{PC} (\ao) = \dim_{PC} (\tzao) = \mcd (\tzao) = \bigg\{ s_k := D + i \frac{2 \pi}{T} k: \widehat{G} \Big( \frac{k}{T} \Big) \neq 0, k \in \mbz \bigg\},
\end{equation}
where
\begin{equation}\label{3.32}
\widehat{G}_0 (t) := \int_0^T e^{-2 \pi it \tau} G(\tau) d \tau, \text{ for all } t \in \mbr. 
\end{equation}
We have, for every $k \in \mbz,$
\begin{equation}\label{3.33}
\res (\tzao, s_k) = \frac{1}{T} \widehat{G}_0 \Big( \frac{k}{T} \Big),
\end{equation}
and hence,
\begin{equation}\label{3.34}
|\res (\tzao, s_k)| \leq \frac{1}{T} \int_0^{T} G(\tau) d \tau.
\end{equation}
Also,
\begin{equation}\label{3.35}
\lim_{k \rightarrow \pm \infty} \res (\tzao, s_k) = 0.
\end{equation}

Finally, the residue of $\tzao$ at $s =D$ coincides both with the mean value of $G$ and with $\widetilde{\mcm}$, the {\em average Minkowski content} of $(A, \Omega)$:\footnote{Here, $\widetilde{\mcm}$ is defined as the limit of a suitable Cesaro logarithmic average of $V_{A, \Omega}(t)/t^{N-D}$, much as in \cite[Thm. 8.30]{Lap-vF4} where $N=1$. More specifically, 
\begin{equation}\label{3.42.1/2}
\widetilde{\mcm} := \lim_{\tau \rightarrow +\infty} \frac{1}{\log \tau} \int_{1/\tau}^{1} \frac{V_{\ao} (t)}{t^{N-D}} \, \frac{dt}{t},
\end{equation} 
where the indicated limit is assumed to exist in $(0, +\infty).$\label{Fn:64}}
\begin{equation}\label{3.34.1/2}
\res (\tzao, D) = \frac{1}{T} \int_0^T G (\tau) d \tau = \widetilde{\mcm}.
\end{equation}
In particular, the RFD $(A, \Omega)$ is Minkowski nondegenerate but is not Minkowski measurable and, in fact, we have that
\begin{equation}\label{3.35.1/2}
\mcm_* < \res (\tzao, D) < \mcm^* < \infty.
\end{equation}

For the distance (instead of the tube) zeta function $\zao$, entirely analogous results hold, except for the fact that in light of \eqref{3.12}, 
\begin{equation}\label{3.36}
\res (\zao, s_k) = (N - s_k) \frac{1}{T} \widehat{G}_0 \Big( \frac{k}{T} \Big), \text{ for all } k \in \mbz,
\end{equation}
and hence,
\begin{equation}\label{3.37}
\res( \zao, s_k) = o(|k|) \quad \text{as } |k| \rightarrow \infty.
\end{equation}

\begin{exercise}\label{Exc:3.4}
For the generalized Cantor sets $A = C^{(a)}$ of Example \ref{Ex:3.1} above, verify that both the hypotheses and the conclusions (of the counterpart for bounded sets) of part ($c$) just above are satisfied for an arbitrary $\alpha >0$. Deduce that for any $a \in (0, 1/2)$, both $\zeta_A$ and $\widetilde{\zeta}_A$ have a meromorphic extension to all of $\mbc$ and 
\begin{equation}\notag
\mcd_A  = \mcd (\zeta_A) = \mcd(\widetilde{\zeta}_A) = \{ D + ik{\bf p}: k \in \mbz \},
\end{equation}
where $D:= \log 2/\log (1/a)$ and ${\bf p} := 2 \pi/\log (1/a)$ are respectively the Minkowski dimension and the oscillatory period of $A = C^{(a)}$. Also, calculate the average Minkowski content $\widetilde{\mcm}$ of $A$ both via a direct computation and by using one of the main results of part ($c$).
\end{exercise}

The next ``exercise'' is significantly more difficult than the previous one.\footnote{In fact, it essentially corresponds to Problem 6.2.35 in \cite{LapRaZu1} and would also help solve one part of the much broader Problem 6.2.36 in \cite{LapRaZu1}.}

\begin{exercise}\label{Exc:3.5}
Let $A$ be a self-similar set satisfying the open set condition. Find geometric conditions on $A$ so that in the lattice case (respectively, in the nonlattice case), the hypotheses and hence also the conclusions of the main result (of the counterpart for bounded sets) of part ($c$) (respectively, part ($b$)) of this subsection (i.e., \S \ref{Sec:3.3.2}) are satisfied.\footnote{Recall that a self-similar set is said to be lattice if the multiplicative group generated by its distinct scaling ratios is of rank $1$, and is called nonlattice otherwise.}
\end{exercise}

A variety of significantly more complicated behaviors for the asymptotics (as $\varepsilon \rightarrow 0^+$) of the tube function $\varepsilon \mapsto |A_\varepsilon \cap \Omega| = V_{\ao} (\varepsilon)$ of an RFD are considered in \cite{LapRaZu1}. These include, most notably, transcendentally $n$-quasiperiodic behavior, for any $n \in \mbn \cup \{ \infty \}$; see \cite[\S 3.1 and \S 4.6]{LapRaZu1}. Instead of giving the precise (but somewhat involved) definitions and results here, we limit ourselves for now to the following simple example. (Further information will be provided later in the paper, especially in \S \ref{Sec:3.6}.)

\begin{example}\label{Ex:3.6}
Suppose that the tube function of the bounded set $A$ satisfies \eqref{3.29}, where the function $G$ is no longer assumed to be periodic (of period $T$) but is given instead by $G = G_1 + G_2$, where the nonconstant functions $G_1$ and $G_2$ are periodic of (minimal) periods $T_1$ and $T_2$, respectively, with $T_1/T_2$ irrational. Then, $\zeta_A$ has a meromorphic continuation to (at least) $\{ Re(s) > D - \alpha \}$ and the set of principal complex dimensions of $A$ consists of simple (nonremovable) singularities of $\zeta_A$ (and of $\widetilde{\zeta}_A$) and is given by
\begin{equation}\label{3.37.1/3}
\mcd_{PC} (A) = \dim_{PC} A = \bigcup_{j=1}^2 \Big(D + i \frac{2 \pi}{T_j} \mbz\Big) = D + i \Big( \bigcup_{j=1}^2 \frac{2 \pi}{T_j} \mbz \Big),
\end{equation}
where $D = D (\zeta_A) = D (\widetilde{\zeta}_A) = \dim_B A$. We note that since $T_1$ and $T_2$ are incommensurable, the imaginary parts of the principal complex dimensions of $A$ have a rather different structure than in part ($c$) above. In particular, they are no longer in arithmetic progression.

Finally, we point out that the assumptions of this example are realized by the compact subset of $\mbr$ obtained by taking the disjoint union of two distinct and suitably chosen (two-parameter) generalized Cantor sets; see \cite[Thm. 3.1.12]{LapRaZu1} for the details and \cite[Thm. 3.1.15]{LapRaZu1} for the generalization to $n$ such Cantor sets, corresponding to the case when $A$ is (transcendentally) 2-quasiperiodic or more generally, $n$-quasiperiodic, respectively. The important (and highly nontrivial) extension to the case when $n = \infty$ is dealt with in \cite[\S 4.6]{LapRaZu1}.
\end{example}

\subsubsection{Scaling property and invariance under isometries}\label{Sec:3.3.3}

We first state the scaling invariance property of the distance zeta function $\zao$ of an RFD $(A, \Omega)$.\footnote{The properties stated for RFDs in \S \ref{Sec:3.3.3} and in fact, in all of \S \ref{Sec:3.3.1}--\S \ref{Sec:3.3.4}, have natural counterparts for bounded sets $A$ in $\mbr^N$. In particular, recall from \eqref{3.15} and the discussion surrounding it that $D(\zeta_{A, A_\delta})$ denotes the abscissa of convergence of the RFD $(A, A_\delta).$} (We leave it as an exercise for the interested reader to state its counterpart for the tube zeta function  $\tzao$; alternatively, see \cite[\S 4.1.3]{LapRaZu1}.)

For any $\lambda > 0$, we have $D (\zeta_{\lambda A, \lambda \Omega}) = D (\zao) = \overline{D} := \overline{\dim}_B (A, \Omega)$ and 
\begin{equation}\label{3.38}
\zeta_{\lambda A, \lambda \Omega} (s) = \lambda^s \zao (s),
\end{equation}
for all $s \in \mbc$ with $Re(s) > \overline{\dim}_B A$ or, more generally, for all $s \in U$, where $U$ is any domain containing the closed right half-plane $\{Re(s) \geq \overline{D}\}$ to which one, and hence both, of these fractal zeta functions has a meromorphic continuation. 

Furthermore, if $\omega \in \mbc$ is a simple pole of the meromorphic extension of $\zao$ to some open connected neighborhood of the critical line $\{ Re(s) =\overline{D}\}$ (or, equivalently, of $\{Re(s) \geq \overline{D}\}$), then
\begin{equation}\label{3.39}
\res (\zeta_{\lambda A, \lambda \Omega}, \omega) = \lambda^\omega \res (\zao, \omega). 
\end{equation}
It is noteworthy that the scaling property of the residues of $\zao$, as stated in \eqref{3.39}, is very analogous to the scaling property of Hausdorff measure; the latter, however, is restricted to a single exponent, namely, the Hausdorff dimension, whereas \eqref{3.39} is valid for any (visible) complex dimension of $(A, \Omega)$. We do not wish to elaborate on this point here but simply mention that under appropriate hypotheses, a suitable version of these residues should give rise to a family of complex measures, defined by the maps $\Omega \mapsto \res (\zao, \omega)$ and indexed by the (visible) complex dimensions $\omega$ of the bounded set $A$ and where $\Omega$ is allowed to run through the Borel (and not necessarily open) subsets of $\mbr^N.$  For more information, see \cite[App. B]{LapRaZu1}.

Next, we simply mention that the distance and tube zeta functions of an RFD $(A, \Omega)$ are clearly invariant under the group of displacements of $\mbr^N$ (that is, under the group of isometries of the affine space $\mbr^N$, generated by the rotations and translations). More specifically, if $f$ is such a displacement of $\mbr^N$, then
\begin{equation}\label{3.40}
\zeta_{f (A), f(\Omega)} = \zao
\end{equation}
(and analogously for $\tzao$, as well as for the corresponding upper, lower Minkowski dimensions and contents, and the visible complex dimensions, in particular).

The scaling and invariance properties stated in the present subsection (i.e., \S \ref{Sec:3.3.3}), along with other ``covariance'' properties of the fractal zeta functions, are very useful in the concrete computation of the distance and tube zeta functions of a variety of examples (including many of those discussed in \S \ref{Sec:3.4}), as well as of the corresponding complex dimensions. They also play an important role in the direct computation of fractal tube formulas for many concrete examples (including several of those discussed in \S \ref{Sec:3.5.3}). (See \cite[esp., Chs. 3--5]{LapRaZu1}.)\\
 
\subsubsection{Invariance of the complex dimensions under embedding into higher-\\ dimensional spaces}\label{Sec:3.3.4}

Let $(A, \Omega)$ be an RFD in $\mbr^N$ and let $M \geq 1$ be an arbitrary integer. Denote by $(A, \Omega)_M$ the natural embedding of $(A, \Omega)$ into $\mbr^{N + M}$, where $(A,\Omega)_M := (A_M, \Omega_M)$, with
\begin{equation}\label{3.37.1/32}
A_M:= A \times \{0 \} \times \cdots \times \{0 \} \subseteq \mbr^{N + M}
\end{equation}
and\footnote{Here, $\{ 0\} \times \cdots \times \{ 0\}$ is the $M$-fold Cartesian product of $\{ 0\}$ by itself, viewed as a subset of $\mbr^M$.}
\begin{equation}\label{3.37.1/16}
\Omega_M := \Omega \times (-1, 1)^M \subseteq \mbr^{N + M}.
\end{equation}
Then, it is shown in \cite[\S 4.7.2]{LapRaZu1} that given any connected open neighborhood $U$ of the critical line $\{Re(s) = \overline{D} \}$, where (as usual)
\begin{equation}\label{3.37.1/8}
\overline{D} := \overline{\dim}_B A = D (\widetilde{\zeta}_A) = D (\zeta_A),
\end{equation}
with $\overline{D} < N$, the tube zeta function $\tzao$ has a (necessarily unique) meromorphic extension to $U$ if and only if $\widetilde{\zeta}_{(A, \Omega)_M}$ does, and in that case, the visible complex dimensions of the RFD $(A, \Omega)$ in $\mbr^N$ and of the RFD $(A, \Omega)_M$ in $\mbr^{N +M}$ coincide (and similarly for the distance zeta functions $\zao$ and $\zeta_{(A, \Omega)_M}$):
\begin{align}\label{3.37.1/4}
\mcd_{\ao} (U) &= \mcd (\tzao; U) = \mcd (\widetilde{\zeta}_{(A, \Omega)_M}; U)\notag \\
&=\mcd (\zao; U) = \mcd(\zeta_{(A, \Omega)_M}; U) = \mcd_{(A, \Omega)_M} (U),
\end{align}
as equalities between multisets. Moreover, 
\begin{align}\label{3.37.1/2}
\mcd_{PC} (\ao) &= \dim_{PC} (A, \Omega) = \mcd_{PC} (\tzao) = \mcd_{PC} (\widetilde{\zeta}_{(A, \Omega)_M})   \notag \\
&= \mcd_{PC} (\zao) = \mcd_{PC} (\zeta_{(A, \Omega)_M}) = \dim_{PC} (\zeta_{(A, \Omega)_M}) = \mcd_{PC} (A, \Omega)_M
\end{align} 
and
\begin{align}\label{3.37.3/4}
\overline{D} :&= \overline{\dim}_B (A, \Omega) = D (\tzao) = D(\widetilde{\zeta}_{(A, \Omega)_M})\notag\\ 
&= D(\zeta_{A, \Omega}) = D(\zeta_{(A, \Omega)_M}) = \overline{\dim}_B (A, \Omega)_M.
\end{align}
{\em Consequently, neither the values nor the multiplicities of the} ({\em visible}) {\em complex dimensions} ({\em and, in particular, of the principal complex dimensions}) {\em of the RFD $(A, \Omega)$ depend on the dimension of the ambient space.}\footnote{It is significantly simpler to check that the values of $\underline{\dim}_B (A, \Omega),$ and also the statements according to which $(\ao)$ is Minkowski nondegenerate or is Minkowski measurable, are independent of the dimension of the ambient space; see \cite[\S 4.7.2]{LapRaZu1}. In the Minkowski measurable case, the corresponding Minkowski content can be suitably normalized (much as in \cite{Fed2}) so as to also be independent of the embedding dimension; see \cite{Res}. (And similarly for the normalized values of $\mcm_*$ and $\mcm^*.$)}

Moreover, we point out that since $\overline{\dim}_B (\ao) < N$ (and hence also, $\overline{\dim}_B (\ao)_M$ $< N + M$), the exact same results hold for the tube zeta functions $\tzao$ and $\widetilde{\zeta}_{(A, \Omega)_M}$ replaced, respectively, by the distance zeta functions $\zao$ and $\zeta_{(A, \Omega)_M}$.

Finally, we note that as usual, the exact analog of the results stated in this subsection (i.e., \S \ref{Sec:3.3.4}) hold for the special case of bounded subsets $A$ (instead of RFDs $(\ao)$) in $\mbr^N.$ It suffices to replace the RFDs $(A, \Omega)$ and $(\ao)_M$ by the bounded sets $A$ and $A_M$ (as given by \eqref{3.37.1/32}) in $\mbr^N$ and $\mbr^{N + M}$, respectively, in all of the corresponding statements; see \cite[\S 4.7.1]{LapRaZu1} for the details.

\subsection{Examples of fractal zeta functions and complex dimensions}\label{Sec:3.4}

In this subsection, we give a variety of examples of bounded sets and of RFDs for which the associated distance zeta function (or, equivalently, in light of the functional equation \eqref{3.8.1/4} or \eqref{3.11}, the tube zeta function) can be calculated explicitly and the corresponding poles (i.e., the complex dimensions) can be determined. We will limit ourselves to the simplest examples and omit the computation involved, often based in part on symmetry and scaling considerations, but refer instead to \cite{LapRaZu1} or to [\hyperlinkcite{LapRaZu2}{LapRaZu2--9}] for the details and many further examples.\\

\subsubsection{The Sierpinski gasket}\label{Sec:3.4.1}

Let $A \subseteq \mbr^2$ be the classic Sierpinski gasket. It is a self-similar set in $\mbr^2$ with three equal scaling ratios $r_1 = r_2 = r_3 = 1/2$ and hence, with Minkowski dimension $D = \dim_B A = \log 3/\log 2 = \log_2 3$ (coinciding with the similarity dimension of $A$). As is well-known, $A$ is the unique (nonempty) compact subset of $\mbr^2$ satisfying the fixed point equation
\begin{equation}\label{3.38.2}
A = \bigcup_{j=1}^3 S_j (A),
\end{equation} 
where $S_1, S_2, S_3$ are contractive similarity transformations of $\mbr^2$ (with scaling ratios $r_1, r_2, r_3$, as above) defined in  a simple way and with respective fixed points $v_1, v_2, v_3$, the vertices of the unit equilateral triangle, which is the generator of $A.$

Then, one can show (see \cite[\S 5.5.3]{LapRaZu1}) that for $\delta > 1/4 \sqrt{3}$ (so that $A_\delta$ be connected),\footnote{Recall from \S \ref{Sec:3.2} that the poles of $\zao$ (and of $\tzao$) are independent of the choice of $\delta >0$;  i.e., the set of complex dimensions of an RFD $(\ao)$, $\mcd_{\ao} = \mcd(\zao) = \mcd (\tzao)$, is independent of the choice of $\delta > 0.$
It is also true, in particular, for a bounded set $A$  instead of an RFD $(\ao)$ (by considering the RFD $(A, A_{\delta_2}),$ for any fixed $\delta_2 > 0$). This comment, being valid for any RFD $(\ao)$ (and, in particular, bounded set) in $\mbr^N$, will no longer be repeated in this  section (i.e., \S \ref{Sec:3.4}).}
$\zeta_A$ has a meromorphic extension to all of $\mbc$ given by
\begin{equation}\label{3.39.2}
\zeta_A (s) = \frac{6 (\sqrt{3})^{1-s} 2^{-s}}{s (s-1) (2^s - 3)} + 2 \pi \frac{\delta^s}{s} + 3 \frac{\delta^{s-1}}{s-1},
\end{equation}
for every $s \in \mbc.$ Consequently, the set of principal complex dimensions of the Sierpinski gasket $A$ is given by 
\begin{align}\label{3.40.2}
\dim_{PC} A = \mcd_{PC} (\zeta_A) &= \bigg\{ \log_2 3 + i \frac{2 \pi}{\log 2}k : k \in \mbz \bigg\} \notag \\
&= \log_2 3 + i \frac{2 \pi}{\log 2} \mbz
\end{align}
and the set of all complex dimensions of $A$ is given by
\begin{align}\label{3.41}
\mcd_A := \mcd (\zeta_A) &= \{ 0 \} \cup \bigg(\log_2 3 + i \frac{2 \pi}{\log 2} \mbz \bigg) = \{ 0\} \cup \dim_{PC} A \notag \\
&= \{ 0 \} \cup \bigg\{ s_k := \log_2 3 + i \frac{2 \pi}{\log 2} k: k \in \mbz \bigg\}.
\end{align}
Each complex dimension 0 or $s_k := \log_2 3 +i (2 \pi/\log 2) k \ (k \in \mbz)$ is simple (i.e., is a simple pole of $\zeta_A$) and the corresponding residue is given respectively by\footnote{A priori, in light of \eqref{3.39.2}, $s=1$ should be a pole of $\zeta_A$. However, a direct computation shows that $\res (\zeta_A, 1) = -3 +3 =0,$ as indicated in \eqref{3.42}. Hence, $1$ is {\em not} a complex dimension of the Sierpinski gasket. In fact, the corresponding fractal tube formula will {\em not} contain a term corresponding to $s=1$; see the relevant parts of \S \ref{Sec:3.5.3}, along with \cite[\S 5.5.3]{LapRaZu1}, especially, the last equation before Example 5.5.13 in {\em loc. cit.}}
\begin{equation}\label{3.42}
\res (\zeta_A, 0) = 3 \sqrt{3} + 2\pi, \ \res (\zeta_A, 1) = 0
\end{equation} 
and  for each $k \in \mbz$,
\begin{equation}\label{3.43}
\res (\zeta_A, s_k) = \frac{6 (\sqrt{3})^{1-s_k}}{(\log 2) 4^{s_k} s_k (s_k -1)}.
\end{equation}
Finally, note that in \eqref{3.40.2} and \eqref{3.41}, $D := \log_2 3$ is the Minkowski dimension of $A$ and ${\bf p} := 2 \pi/\log 2$ is the oscillatory period of $A$. Also, the expression obtained for $\dim_{PC} A$ in \eqref{3.40.2} is compatible with the results of part ($c$) of \S \ref{Sec:3.3.2}; see \eqref{3.31}.\\
 
\subsubsection{The Sierpinski carpet}\label{Sec:3.4.2}

Let $A \subseteq \mbr^2$ be the classic Sierpinski carpet (with generator the unit square and 8 equal scaling ratios $r_1 = \cdots = r_8 = 1/3$). As is well known, $A$ is the unique (nonempty) compact subset of $\mbr^2$ such that $A = \cup_{j=1}^8 S_j (A)$, where $S_1, \cdots, S_8$ are suitable contractive similarities of $\mbr^2$. Clearly, $D:= \dim_B A$ exists and $D = \log_3 8$, the similarity dimension of the self-similar set $A$. Then, much as in the case of the Sierpinski gasket $A$ from \S \ref{Sec:3.4.1} just above, it can be shown (see \cite[Prop. 3.21]{LapRaZu1}) that $\zeta_A$ has a meromorphic extension to all of $\mbc$, and that for every $\delta > 1/6$ (so that the $\delta$-neighborhood $A_\delta$ of $A$ be connected) and every $s \in \mbc$,
\begin{equation}\label{3.44}
\zeta_A (s) = \frac{8}{2^s s(s-1) (3^s -8)} + 2\pi \frac{\delta^s}{s} + 4 \frac{\delta^{s-1}}{s-1}.
\end{equation}
It follows that
\begin{equation}\label{3.45}
\dim_{PC} A = \mcd_{PC} (\zeta_A) = \bigg\{ \log_3 8 + i \frac{2\pi}{\log 3} k: k \in \mbz \bigg\} = \log_3 8 + i \frac{2\pi}{\log 3} \mbz
\end{equation}
and the set of all complex dimensions of the Sierpinski carpet is given by 
\begin{align}\label{3.43.2}
\mcd_A = \mcd (\zeta_A) &= \{0,1 \} \cup \Big(\log_3 8 + i \frac{2\pi}{\log 3} \mbz \Big) \notag \\
&=\{ 0,1\} \cup \dim_{PC} A = \{0,1 \} \cup \Big\{ s_k := \log_3 8 + i \frac{2 \pi}{\log 3} k: k \in \mbz \Big\}. 
\end{align}
Furthermore, the complex dimensions of $A$ are all simple and the residues at $0,1$ and $s_k \, (k \in \mbz)$ are given, respectively, by
\begin{equation}\label{3.44.2}
\res (\zeta_A, 0) = 2 \pi + \frac{8}{7}, \ \res (\zeta_A, 1) = \frac{16}{5}
\end{equation}
and
\begin{equation}\label{3.45.2}
\res (\zeta_A, s_k) = \frac{2^{-s_k}}{(\log 3) s_k (s_k -1)}, \text{ for all } k \in \mbz.
\end{equation}
Again, in \eqref{3.45} and \eqref{3.43.2}, $D := \log_3 8$ and ${\bf p} := 2 \pi/\log 3$ are, respectively, the Minkowski dimension and the oscillatory period of $A$, in agreement with the results stated in part ($c$) of \S \ref{Sec:3.3.2}.\\ 

\subsubsection{The $3$-$d$ Sierpinski carpet}\label{Sec:3.4.3}

We refer to \cite[Exple. 5.5.13]{LapRaZu1} for the precise definition of this version of the three-dimensional Sierpinski carpet $A$ and for the corresponding results. It is shown there that for any $\delta > 1/6$ (so that $A_\delta$ be connected), $\zeta_A$ has a meromorphic extension to all of $\mbc$ given by 
\begin{equation}\label{3.46}
\zeta_A (s) = \frac{48 \, . \, 2^{-s}}{s(s-1)(s-2)(3^s - 26)} + 4\pi \frac{\delta^s}{s} + 6\pi \frac{\delta^{s-1}}{s-1} + 6 \frac{\delta^{s-2}}{s-2},
\end{equation}
for every $s \in \mbc$. Therefore, 
\begin{equation}\label{3.47}
\dim_{PC} (\zeta_A) = \mcd_{PC} (\zeta_A) = \log_3 26 + i \frac{2 \pi}{\log 3} \mbz
\end{equation}
and 
\begin{equation}\label{3.48}
\mcd (\zeta_A) = \{ 0,1,2 \} \cup \dim_{PC} A = \{ 0, 1, 2 \} \cup \{ s_k := D + i k {\bf p} : k \in \mbz \},
\end{equation}
where $D: = D (\zeta_A) = \log_3 26 $ and ${\bf p} := 2\pi/\log 3$ are, respectively, the Minkowski dimension (as well as the similarity dimension) and the oscillatory period of $A$. Each complex dimension in \eqref{3.47} and \eqref{3.48} is simple and 
\begin{equation}\label{3.49}
\res (\zeta_A, j) = 4\pi - \frac{24}{25}, \ 6\pi + \frac{24}{23}, \ \frac{96}{17} \text{ for } j = 0,1,2,
\end{equation}
respectively; also, for every $k \in \mbz$,
\begin{equation}\label{3.50}
\res (\zeta_A, s_k) = \frac{24}{13.2^{s_k} s_k (s_k -1) (s_k -2) \log 3}.
\end{equation}\\

\subsubsection{The $N$-dimensional relative Sierpinski gasket}\label{Sec:3.4.4}

Let $(A_N, \Omega_N)$ denote the $N$-{\em dimensional relative Sierpinski gasket,} also called the ({\em inhomogeneous}) $N$-{\em gasket} RFD, in short, and as introduced in \cite[Exple. 4.2.26]{LapRaZu1} (as well as in \cite{LapRaZu4}). We refer the interested reader to {\em loc. cit.} for a detailed description of its geometric construction and for the corresponding figures. (See also Remark \ref{Rem:3.7} below for a synopsis of the construction.) We simply mention here that for each fixed integer $N \geq 2, (A_N, \Omega_N)$ is an RFD in $\mbr^N$ which can also be viewed as a self-similar spray (or RFD) in $\mbr^N$ (in the refined sense of \cite[\S 4.2.1 and \S 5.5.6]{LapRaZu1} rather than in the original sense of \cite{Lap3}, \cite{LapPo3} and \cite{LapPe2} or \cite{LapPeWi1}), with $N+1$ equal scaling ratios $r_1 = \cdots = r_{N+1} = 1/2$ and with a single generator RFD $(\partial \Omega_{N,0}, \Omega_{N, 0})$, where the bounded open set $\Omega_{N,0}$ in $\mbr^N$ (called the $N$-plex) is described in Remark \ref{Rem:3.7}. Furthermore, unlike in our previous examples in \S \ref{Sec:3.4.1}--\S \ref{Sec:3.4.3}, $A_N$ is {\em not} a self-similar set (in the usual sense of the term) but is instead an {\em inhomogeneous} self-similar set, in the sense of \cite{BarnDemk} and (with a different terminology) of \cite{Hat}. More specifically, it is the unique (nonempty) compact subset $A$ of $\mbr^N$ satisfying the {\em inhomogeneous} fixed point equation
\begin{equation}\label{3.51}
A = \bigcup_{j=1}^{N+1} S_j (A) \cup B,
\end{equation}
where the maps $S_j$ (for $j=1, \cdots, N+1$) are $N+1$ contractive similarity transformations\footnote{These similitudes $(S_j)_{j=1}^{N+1}$ have for respective fixed points $(P_j)_{j=1}^{N+1}$, the points chosen at the beginning of Remark \ref{Rem:3.7}.}
of $\mbr^N$ (the same ones as those defining the $N$-dimensional analog of the usual self-similar gasket, which is an homogeneous or a classic self-similar set, satisfying the counterpart of \eqref{3.51} with $B:= \emptyset$, the empty set) and $B$ is a certain {\em nonempty} compact subset of $\mbr^N$; in fact, $B$ can be chosen to be equal to $\partial \Omega_{N,0}$, the boundary of the $N$-plex $\Omega_{N,0}$ described in Remark \ref{Rem:3.7}.

\begin{remark}({\em Construction of the generator $\Omega_{N,0}$ and of the inhomogeneous $N$-gasket RFD $(A_N, \Omega_N)$.})\label{Rem:3.7}
Here, the generator $\Omega_{N,0}$ and the compact set $A_N$ can be constructed as follows. Let $V_N = \{P_1, \cdots, P_{N+1} \}$ be a set of $N+1$ points in $\mbr^N$ such that $|P_j - P_k| =1$, for any $j \neq k$. (Such a set can be constructed inductively.) Let $\Omega_N$ be the (necessarily closed) convex hull of $V_N$. Clearly, $\Omega_N$ is an $N$-simplex. Then, $\Omega_{N,0}$, called the $N$-plex, is the bounded open subset of $\mbr^N$ obtained by taking the interior of the convex hull of the set of midpoints of all of the $\frac{(N+1)N}{2} = \binom{N}{2}$ edges of the $N$-simplex $\Omega_N$. (Note that for $N=2, \Omega_{N,0}$ is the first deleted open triangle in the construction of the Sierpinski gasket, while for $N=3$, it is an octahedron; see \cite[Fig. 4.7]{LapRaZu1} for an illustration.)

Now, the set $\overline{\Omega_N} \backslash \Omega_{N,0}$ is the union of $N+1$ congruent and compact $N$-simplices with disjoint interiors and having all their sides (edges) of length $1/2$. This is the first step in the construction of $A_N$. We proceed analogously with each of the aforementioned $N$-simplices. We then repeat the construction, ad infinitum. The compact subset $A_N$ of $\mbr^N$ obtained in this manner is called the {\em inhomogeneous} \linebreak {\em $N$-gasket.} For $N=2$, it coincides with the classic Sierpinski gasket (studied in \S \ref{Sec:3.4.1}), but when $N \geq 3$, it does not coincide with the usual $N$-dimensional Sierpinski gasket (studied, e.g., in \cite{KiLap1}). In fact, still for $N \geq 3$, it is no longer self-similar (in the classic sense) but is instead an inhomogeneous self-similar set satisfying \eqref{3.51}, with $B := \partial \Omega_{N,0}$, the boundary of the $N$-plex. (See \cite[Fig. 4.8]{LapRaZu1} for an illustration of the case when $N=3$.)

Finally, the {\em relative} (or {\em inhomogeneous}) $N$-{\em gasket RFD} is given by $(A_N, \Omega_N)$, where $A_N$ is the above inhomogeneous $N$-gasket and $\Omega_N$ is the above $N$-simplex. 
\end{remark}
Then (see \cite[Exple. 4.2.26]{LapRaZu1}), for the inhomogeneous $N$-gasket RFD $(A_N, \Omega_N)$, the distance zeta function $\zeta_{A_N, \Omega_N}$ has a meromorphic extension to all of $\mbc$, given for every $s \in \mbc$ by
\begin{equation}\label{3.52}
\zeta_{A_N, \Omega_N} (s) = \frac{g_N (s)}{s(s-1) \cdots (s-(N-1)) (1- (N+1) 2^{-s})},
\end{equation}
for some nowhere vanishing entire function $g_N$. (For example, when $N=3$, we have $g_3 (s) := 8(\sqrt{3})^{3-s} (2 \sqrt{2})^{-s}$ and if $N=2, \ g_2 (s) := 6(\sqrt{3})^{1-s} 2^{-s}$, still for all $s \in \mbc$; see, respectively, \cite[Eq. (4.2.89) and Prop. 4.2.25]{LapRaZu1}.) 

In order to explain the form of $\zeta_{A_N, \Omega_N}$ given in \eqref{3.52}, we recall that $(A_N, \Omega_N)$ is a self-similar RFD with generator the RFD $(\partial \Omega_{N, 0}, \Omega_{N,0})$ and with equal scaling ratios $r_j \equiv 1/2$, for $j=1, \cdots, N+1$. Thus, according to the results of \cite[\S 4.2.1 and \S 5.5.6]{LapRaZu1} about self-similar sprays (and recalled in \S \ref{Sec:3.4.10} below),
\begin{equation}\label{3.53}
\zeta_{A_N, \Omega_N} (s) = \zeta_\mfs (s) \cdot \zeta_{\partial \Omega_{N,0}, \Omega_{N, 0}} (s),
\end{equation}
\\
where the {\em scaling zeta function} $\zeta_{\mfs}$ (here, the geometric zeta function of the underlying unbounded self-similar string with equal scaling ratios $r_j \equiv 1/2$ for  $j=1, \cdots, N+1$) is given by
\begin{equation}\label{3.54}
\zeta_{\mfs} (s) = \frac{1}{1- (N+1) 2^{-s}}
\end{equation}
for all $s \in \mbc$, and where via a direct computation,\footnote{See {\em loc. cit.} for the case when $N=2$ or when $N=3$.} one can show that $\zeta_{\partial \Omega_{N,0}, \Omega_{N, 0}}$ is given for all $s \in \mbc$ by
\begin{equation}\label{3.55}
\zeta_{\partial \Omega_{N,0}, \Omega_{N, 0}} (s) = \frac{g_N (s)}{s(s-1) \cdots (s - (N-1))},
\end{equation}
with $g_N$ as above. Now combining \eqref{3.53}--\eqref{3.55}, we obtain \eqref{3.52}, as desired.

Next, since $g_N$ is nowhere vanishing and is entire, we deduce from \eqref{3.52} that 
\begin{equation}\label{3.56}
\mcd_{A, \Omega} := \mcd (\zao) = \{0,1, \cdots, N-1 \} \cup \mcd (\zeta_\mfs),
\end{equation}
where (in light of \eqref{3.54})
\begin{equation}\label{3.57}
\mcd (\zeta_\mfs) = \log_2 (N+1) + i \frac{2\pi}{\log 2} \mbz.
\end{equation}
Note that, by \eqref{3.55},
\begin{equation}\label{3.58}
\mcd (\zeta_{\partial \Omega_{N,0}, \Omega_{N,0}}) = \{0, 1, \cdots, N-1 \}.
\end{equation}
Therefore, for every $N \geq 2$, the set of complex dimensions of the inhomogeneous $N$-gasket is given by
\begin{equation}\label{3.59}
\mcd (A_N, \Omega_N) := \mcd (\zeta_{A_N, \Omega_N}) = \{ 0, 1, \cdots, N-1 \} \cup \Big(\log_2 (N+1) + i \frac{2\pi}{\log 2} \mbz \Big).
\end{equation}
Except when $\log_2 (N+1) = j$, for some $j \in \{ 0, 1, \cdots, N-1 \}$ (i.e., $N =2^j -1$, for some $j \in \{ 2, \cdots, N-1 \}$, since $N \geq 2$ here), all of the complex dimensions of the RFD $(A_N, \Omega_N)$ in \eqref{3.56} are simple.\footnote{In view of the aforementioned results of {\em loc. cit.}, the residues of $\zeta_{A_N, \Omega_N}$ at each complex dimension $\omega_j = j$ (for $j \in \{ 0, 1, \cdots, N-1 \}$ and $s_k := \log_2 (N+1) + i (2\pi/\log 2) k \ (k \in \mbz)$ can be explicitly computed when $N =2$ and when $N=3$; see \cite[Exple. 4.2.24 and Eq. (4.2.88)]{LapRaZu1}.}

It is instructive (although easy) to determine $\dim_B (A_N, \Omega_N)$ and $\dim_{PC} (A_N, \Omega_N)$. In light of \eqref{3.53}, we have that
\begin{equation}\label{3.60}
\overline{\dim_B} (A_N, \Omega_N) = \max (D (\zeta_{\partial \Omega_{N,0}, \Omega_{N,0}}), \ D (\zeta_\mfs)) = \max (N-1, \log_2 (N+1)).
\end{equation}
Observe that in \eqref{3.60}, $\dim_B (\partial \Omega_{N,0}, \Omega_{N,0}) = N-1$ is the Minkowski dimension (which exists) of the generating RFD $(\partial \Omega_{N,0}, \Omega_{N,0})$ and $\sigma_N := N-1$ is the similarity dimension of the self-similar spray or RFD $(A_N, \Omega_N)$. Also, since one can show that $\dim_B (A_N, \Omega_N)$ exists, we conclude that 
\begin{align}\label{3.61}
D := \dim_B (A_N, \Omega_N) (= D(\zeta_{A_N, \Omega_N})) &= \max (N-1, \log_2 (N+1))\\
&=\begin{cases}\notag
\log_2 3, \quad \text{for } N=2,\\
N-1, \quad \text{for } N \geq 3.
\end{cases}
\end{align}
Furthermore, we deduce from \eqref{3.59} and \eqref{3.61} that the set of principal complex dimensions of $(A_N, \Omega_N)$ is given by
\begin{equation}\label{3.62}
\dim_{PC} (A_N, \Omega_N) =\begin{cases}
\log_2 3 + i \frac{2 \pi}{\log 2} \mbz, \quad &\text{for } N=2,\\
2 + i \frac{2 \pi}{\log 2} \mbz, \quad &\text{for } N =3,\\
\{ N-1 \}, \quad &\text{for } N \geq 4. 
\end{cases}
\end{equation}
Observe that for $N=2$, we have that $\sigma_2 = \log_2 3 > \dim_B (\partial \Omega_{2,0}, \Omega_{2,0}) = 1$, and hence, $\dim (A_2, \Omega_2) = \sigma_2$, the similarity dimension of the Sierpinski gasket, in agreement with a well-known result about classic or homogeneous self-similar sets (satisfying the open set condition); see, e.g., \cite{Hut} or \cite[Ch. 9]{Fa1}. By contrast, when $N \geq 4$, we have the reverse inequality;\footnote{In fact, this inequality \eqref{3.92.1/2} is always strict; indeed, it is easy to check by induction on $N$ that we never have $N=2^{N-1} -1$, for some integer $N \geq 4$.}
namely,
\begin{equation}\label{3.92.1/2}
\sigma_N = \log_2 (N+1) \leq \dim_B (\partial \Omega_{N,0}, \Omega_{N, 0}) = N-1.
\end{equation} 
Therefore, $\dim_B (A_N, \Omega_N) = \dim_B (\partial \Omega_{N,0}, \Omega_{N,0}) = N-1,$ in this case.\footnote{There is no contradiction  because, as we recall from our earlier discussion, $(A_N, \Omega_N)$ is an inhomogeneous (but unless $N=2$) is {\em not} a standard (or homogeneous) self-similar set.}

Finally, if $N=3,$ we have $\sigma_3 = \dim _B (\partial \Omega_{3,0}, \Omega_{3,0}) =2$. This coincidence between the geometry of the generator $(\partial \Omega_{3,0}, \Omega_{3,0})$ and the scaling of the self-similar spray $(A_3, \Omega_3)$ explains why $D=2$ is a complex dimension of multiplicity two if $N=3$. In some sense, one can say that {\em there is a resonance between the underlying geometry and the underlying scaling of the relative 3-gasket} RFD $(A_3, \Omega_3).$

The above facts have interesting geometric consequences, as is explained in detail in \cite[\S 5.5.6]{LapRaZu1}, by using either the fractal tube formulas of \cite[\S 5.1--\S 5.3]{LapRaZu1} or the Minkowski measurability criteria of \cite[\S 5.4]{LapRaZu1} (both to be briefly discussed in \S \ref{Sec:3.5}). Firstly, if $N=2,$ the RFD $(A_2, \Omega_2)$ is not Minkowski measurable because in light of \eqref{3.62}, it has nonreal principal complex dimensions; however, $(A_2, \Omega_2)$ is Minkowski nondegenerate. Secondly, if $N \geq 4$ (and since then, $2^{N-1} \neq N-1$, so that the dimension $D = N-1$ of $(A_N, \Omega_N)$  is simple), the RFD is not Minkowski measurable but is still Minkowski nondegenerate. 

Lastly, if $N=3$, $(A_3, \Omega_3)$ is not Minkowski measurable (since its Minkowski dimension $D=2$ has multiplicity two); further, it is also Minkowski {\em degenerate}, which suggests that the usual power law is no longer appropriate to measure the ``fractality'' of $(A_3, \Omega_3)$.\footnote{All of these facts are established in \cite[\S 5.5.6]{LapRaZu1}; see, especially, part ($c$) of Remark 5.5.26 of {\em loc. cit.}.}
However, one can use a suitably generalized Minkowski content (as in \cite{HeLap} and \cite[\S 6.1.1.2]{LapRaZu1}), involving the choice of the gauge function $h(t) := \log (t^{-1})$ for all $t \in (0,1)$, so that the RFD $(A_3, \Omega_3)$ be not only Minkowski nondegenerate but also {\em Minkowski measurable, relative to} $h$ (by contrast  with the cases when $N=2$ and $N=3$); see \cite[Thm. 5.4.27]{LapRaZu1}. 

\begin{exercise}\label{Exc:3.8.1/2}
Verify that when $N=3$, we have
\begin{equation}\label{3.23.1/2}
g_3 (s) = 8 (\sqrt{3})^{3-s} (2 \sqrt{2})^{-s},
\end{equation}
for all $s \in \mbc$ in \eqref{3.52} and \eqref{3.55}.
\end{exercise}

\begin{exercise}\label{Exc:3.7}
Calculate the fractal zeta functions and the complex dimensions of the $N$-carpet RFD $(A, \Omega)$ (the $N$-dimensional relative Sierpinski carpet), which extends to $\mbr^N$ both the Sierpinski carpet ($N=2$; see \S \ref{Sec:3.4.2}) and $3$-carpet ($N=3$; see \S \ref{Sec:3.4.3}). 

Note that this example is significantly simpler than that of the $N$-gasket RFD studied in the present subsection (i.e., \S \ref{Sec:3.4.4}); indeed, unlike for the relative $N$-gasket, which is an inhomogeneous self-similar set, the compact set $A$ is an homogeneous (i.e., classical) self-similar set in $\mbr^N$. In fact, $A$ coincides with the standard $N$-Sierpinski carpet, while $\Omega = (0,1)^N$.

We refer the interested reader to \cite[Exple. 4.2.31]{LapRaZu1} for the complete answers and the corresponding computation.  We simply mention here that
\begin{equation}\label{3.63}
\dim_{PC} (A, \Omega) = \log_3 (3^N -1) + i \frac{2 \pi}{\log 3} \mbz,
\end{equation}
where, as before, $D := \log_3 (3^N -1)$ is the Minkowski dimension of $(\ao)$ and ${\bf p} := 2\pi/\log 3$ is the oscillatory period of $(A, \Omega)$, while
\begin{equation}\label{3.64.1}
\mcd_{A, \Omega} = \mcd (\zao) = \{ 0, 1, \cdots, N-1 \} \cup \dim_{PC} (\ao).
\end{equation}
Furthermore, in either \eqref{3.63} or \eqref{3.64.1}, each complex dimension is simple.
\end{exercise}
 
\subsubsection{The $\frac{1}{2}$-square and $\frac{1}{3}$-square fractals}\label{Sec:3.4.5} We discuss here in parallel two related relative fractal drums, namely, the $\frac{1}{2}$-square and $\frac{1}{3}$-square fractals, which exhibit somewhat different properties. \\

{\bf ({\em a})} ({\em The $\frac{1}{2}$-square fractal}). Starting with the closed unit square $[0,1]^2 \subseteq \mbr^2$, we remove the two open squares of side length $\frac{1}{2}$, denoted by $G_1$ and $G_2$, along the main diagonal. Next, we repeat this step with the two remaining closed squares of side length $1/2$; and so on, ad infinitum. The $\frac{1}{2}$-{\em square fractal} $A$ is the compact set that is left at the end of the process. (See also \cite[Fig. 4.10]{LapRaZu1} for an illustration.)

The $\frac{1}{2}$-square fractal is a {\em nonhomogeneous} self-similar fractal (as was the case of the set $A$ in the construction of the relative Sierpinski $N$-gasket in \S\ref{Sec:3.4.4});
more specifically, it is the unique nonempty compact subset $A$ of $\mbr^2$ satisfying the following inhomogeneous fixed point equation:
\begin{equation}\label{3.64}
A = \bigcup_{j=1}^2 S_j (A) \cup B,
\end{equation}
where the nonempty compact set $B \subseteq \mbr^2$ is the union of the left and upper sides of the closed square $\overline{G}_1$ and of the right and lower sides of the closed square $\overline{G}_2$. Here, the contractive similitudes of $\mbr^2$ involved, namely, the maps $S_1$ and $S_2$, have respective fixed points  at the lower left vertex and the upper right vertex of the unit square, and have scaling ratios $r_1 = r_2 = 1/2$.\footnote{It is noteworthy that the homogeneous self-similar set $E$ which is the unique nonempty compact subset of $\mbr^2$ satisfying the homogeneous fixed point equation associated with \eqref{3.64} (namely, $E = \cup_{j=1}^2 S_j (E)$), is the main diagonal of the unit square $[0,1]^2$.}
(See \cite[Fig. 4.11]{LapRaZu1} for an illustration.)

Let $\Omega := (0,1)^2$ and consider the RFD $(\ao)$; by construction, it is a self-similar spray (or RFD) with generator $G = G_1 \cup G_2$ and scaling ratios $r_1 = r_2 = 1/2$.

It is shown in \cite[Exple. 4.2.33]{LapRaZu1} that for every $\delta > 
1/4$, $\zao$ and hence also, $\zeta_A$, in light of \eqref{3.66} below, have a meromorphic continuation to all of $\mbc$ given for every $s \in \mbc$ respectively by 
\begin{equation}\label{3.65}
\zao (s) = \frac{\zeta_{\partial G, G} (s)}{1-2 \cdot 2^{-s}} = \frac{2^{-(s+1)}}{s(s-1) (2^{s-1} -1)}
\end{equation} 
and 
\begin{align}\label{3.66}
\zeta_A (s) &= \zao (s) +\zeta_{[0,1]^2} (s) \notag \\
&= \frac{2^{-(s+1)}}{s(s-1)(2^{s-1} -1)} + 4 \frac{\delta^{s-1}}{s-1} + 2 \pi \frac{\delta^s}{s}.
\end{align}
It follows that
\begin{equation}\label{3.67}
D := D(\zeta_A) = \dim_B A = D(\zao) = \dim_B (\ao) = 1
\end{equation}
and
\begin{equation}\label{3.68}
\dim_{PC} A = \dim_{PC} (\ao) = 1 + \frac{2 \pi}{\log 2} \mbz,
\end{equation}
as well as
\begin{equation}\label{3.69}
\mcd_A = \mcd (\zeta_A) = \mcd_{\ao} = \mcd (\zao) = \{ 0,1 \} \cup \Big(1 + i \frac{2 \pi}{\log 2} \mbz \Big),
\end{equation}
where these are equalities between multisets.

All of the complex dimensions in \eqref{3.68} and \eqref{3.69} (namely, $0$ and $s_k := 1 + i \, (2 \pi/\log 2) \, k$, for $k \in \mbz$) are simple, except for the dimension $D=1$ which is double.

Furthermore, for all $k \in \mbz \backslash \{ 0\}$, we have 
\begin{equation}\label{3.102.1/2}
\res (\zeta_A, 0) = 1 + 2 \pi \ \text{ and } \ \res (\zeta_A, s_k) = \frac{4^{-i {\bf p}k}}{4 s_k (s_k -1)},
\end{equation}
where ${\bf p} := 2 \pi/\log 2$ is the oscillatory period of the self-similar spray $(\ao).$\\

{\bf ({\em b})} ({\em The $\frac{1}{3}$-square fractal}). As in part ($a$), we begin with the unit square $[0,1]^2$; we then divide it into nine congruent smaller squares. We further delete seven of those smaller squares; that is, we only keep the lower and upper right squares. We then repeat the process, ad infinitum. What is left at the end of the process is denoted by $A$ and called the $\frac{1}{3}$-{\em square fractal}.

If $\Omega := (0,1)^2$, then we consider the RFD $(\ao)$ in $\mbr^2$. Note that (as in part ($a$)), $A$ is an {\em inhomogeneous} self-similar fractal; more specifically, it is the unique  (nonempty) compact subset of $\mbr^2$ satisfying the following inhomogeneous fixed point equation:
\begin{equation}\label{3.71}
A = \bigcup_{j=1}^2 S_j (A) \cup B,
\end{equation}  
where $B \subseteq \mbr^2$ is the nonempty compact set defined by $B := \partial G$ and $G$ (called the generator of the self-similar spray $(\ao)$) is a suitable open convex polygon. Furthermore, $S_1$ and $S_2$ are contractive similitudes of $\mbr^2$, with respective fixed points located at the lower left vertex and the upper right vertex of the unit square.

The RFD $(\ao)$ is a self-similar spray (or RFD) with generator $G$ and scaling ratios $r_1 = r_2 = r_3 = 1/3$.

Much as in part ($a$), it is shown in \cite[Exple. 4.2.34]{LapRaZu1} that $\zao$ (and thus also $\zeta_A$) admits a meromorphic continuation to all of $\mbc$ given for every $s \in \mbc$ (and for all sufficiently large positive $\delta$) by
\begin{equation}\label{3.72}
\zao (s) = \frac{\zeta_{\partial G, G} (s)}{1-2 \cdot 3^{-s}} = \frac{2}{s(3^s -2)} \Big(\frac{6}{s-1} + \Psi (s)\Big),
\end{equation}
where $\Psi$ is a suitable entire function (which is explicitly known), and 
\begin{align}\label{3.73}
\zeta_A (s) &= \zao (s) + \zeta_{[0,1]^2} (s)\\
&= \frac{2}{s(3^s -2)} \Big( \frac{6}{s-1} + \Psi (s) \Big) + 4 \frac{\delta^{s-1}}{s-1} + 2 \pi \frac{\delta^s}{s}.\notag
\end{align}
As a result,
\begin{equation}\label{3.74}
D = D_A = D(\zeta_A) = D_{\ao} = D(\zao) = 1
\end{equation}
and (see \eqref{3.110.1/2} below for a more precise statement)
\begin{equation}\label{3.75}
\dim_{PC} A = \dim_{PC} (\ao) \subseteq \{1 \} \cup \Big( \log_3 2 + i \frac{2 \pi}{\log 3} \mbz \Big),
\end{equation}
as well as 
\begin{equation}\label{3.76}
\{ 0,1 \} \cup F \subseteq \mcd_A = \mcd (\zeta_A) = \mcd_{\ao} = \mcd (\zao) \subseteq  \{0,1 \} \cup \Big( \log_3 2 + i \frac{2 \pi}{\log 3} \mbz \Big),
\end{equation}
each complex dimension in \eqref{3.75} and \eqref{3.76} being simple. Here, $F$ is a subset of $\log_3 2 + i (2 \pi/\log 3) \mbz$ containing $\log_3 2$ and at least finitely many (but more than two) nonreal principal complex dimensions.\footnote{At this stage, the inclusion appearing on the left of \eqref{3.76} is only verified numerically. The difficulty  here is due to the presence of the entire function $\Psi$.}
We conjecture that the set $F$ is in fact (countably) infinite and furthermore, that the inclusions in \eqref{3.76} should actually be equalities and $\dim_{PC} A = \dim_{PC} (\ao) = \{ 1\}$. 

It is noteworthy that  $\log_3 2$ is the dimension of the homogeneous self-similar set $E$ associated with \eqref{3.71} (i.e., $E = \cup_{j=1}^2 S_j (E)$, with $E \subseteq \mbr^2$ nonempty and compact); indeed, $E$ is just the ternary Cantor set located along the main diagonal of the unit square $[0,1]^2$.

Finally, a simple computation yields that 
\begin{equation}\label{3.77}
\res (\zeta_A, 0) = 12 + \pi, \ \res (\zeta_{A, 1}) = 16
\end{equation}
and (with ${\bf p} := 2 \pi/\log 3$ and $s_k := \log_3 2 + i (2 \pi/\log 3) k$, for each $k \in \mbz$)
\begin{equation}\label{3.78}
\res (\zeta_A, s_k) = \frac{3^{-i {\bf p} k}}{(\log 3) s_k} \Big( \frac{6}{s_k -1} + \Psi (s_k) \Big).
\end{equation}
Therefore, in light of \eqref{3.77}, we can now specify the statement made in \eqref{3.75} by affirming that
\begin{equation}\label{3.110.1/2}
\dim_{PC} A = \dim_{PC} (\ao) = \{ 1\}.
\end{equation}

\subsubsection{The $(N-1)$-sphere and its associated RFD}\label{Sec:3.4.6}
In this subsection, we study the complex dimensions of the $(N-1)$-sphere
\begin{equation}\label{3.79}
A := S^{N-1} = \{ x \in \mbr^N: |x| = 1 \},
\end{equation}
where $|\cdot|$ denotes the Euclidean norm in $\mbr^N$, and of the associated RFD $(\ao)$, relative to the open unit ball $\Omega$ of $\mbr^N$, called the $(N-1)$-{\em sphere RFD}. We shall see that the answer obtained in the latter case is very natural. The difference between the answers in the former case (the $(N-1)$-sphere) and the latter case (the $(N-1)$-sphere RFD) is simply due to the fact that in the former case, we consider two-sided $\varepsilon$-neighborhoods of $A = S^{N-1}$ whereas in the latter case, we deal with one-sided (or ``inner'') $\varepsilon$-neighborhoods of $A = S^{N-1}$.\\

{\bf ({\em a})} ({\em The $(N-1)$-sphere}). Let $A$ be the $(N-1)$-dimensional sphere, as given by \eqref{3.79}. Then, in \cite[Exple. 2.2.21]{LapRaZu1}, the tube zeta function $\zeta_A$ of $A$ is shown to have a meromorphic extension to all of $\mbc$ given for every $s \in \mbc$ and for any fixed $\delta \in (0, 1)$ by
\begin{equation}\label{3.80}
\widetilde{\zeta}_A (s) = \Theta_N \sum_{k=0}^N (1-(-1)^k) \binom{N}{k} \frac{\delta^{s-N +k}}{s - (N-k)},
\end{equation}
where $\Theta_N$ denotes the volume of the unit ball in $\mbr^N$ and the numbers $\binom{N}{k}$ are the usual binomial coefficients.\footnote{In light of \eqref{3.11}, we deduce at once the value of the distance zeta function $\zeta_A (s)$ for $s \in \mbc$.}
Therefore, independently of the value of $\delta > 0$,
\begin{equation}\label{3.81}
D := D(\widetilde{\zeta}_A) = D(\zeta_A) = \dim_B A = N-1,
\end{equation}

\begin{equation}\label{3.82}
\dim_{PC} A = \{ N-1 \}
\end{equation}
and
\begin{equation}\label{3.83}
\mcd_A = \mcd (\widetilde{\zeta}_A) = \mcd (\zeta_A) = \Bigg\{ N-1, N-3, \cdots, N- \Big( 2 \Big[ \frac{N-1}{2} \Big] + 1 \Big) \Bigg\},
\end{equation}
each complex dimension in \eqref{3.82} and \eqref{3.83} being simple. Note that for $N \geq 1$ odd (respectively, even), the last number in this set is equal to $0$ (respectively, $1$). 

Finally, for every $d \in \mcd_A$,
\begin{equation}\label{3.84}
\res (\widetilde{\zeta}_A, d) = 2 \Theta_N \binom{N}{d}.
\end{equation}
In particular, for $d := D = N-1,$ a direct computation (based on the definition of the Minkowski content given in \S \ref{Sec:3.2} above) yields
\begin{equation}\label{3.85.2}
\mcm = \mcm (A) = 2 N \Theta_N = \res (\widetilde{\zeta}_A, D), 
\end{equation}
in agreement with a result stated in part ($a$) of \S \ref{Sec:3.3.2}.\footnote{If the (absolute) $(N-1)$-sphere $A$ had radius $R$ instead of radius $1$, then for any $\delta \in (0, R],$ one should simply substitute $\Theta_N R^d$ and $\Theta_N R^{N-1}$ for $\Theta_N$ in \eqref{3.84} and \eqref{3.85.2}, respectively. (See also part ($iii$) of Exercise \ref{Exc:3.10}.)}
Note that, clearly, $A$ is Minkowski measurable with Minkowski content $\mcm$.
\begin{exercise}\label{Exc:3.9}
Show directly that $A = S^{N-1}$ is Minkowski measurable, with Minkowski content $\mcm$ given by the second equality of \eqref{3.85.2}.
\end{exercise}

{\bf ({\em b})} ({\em The $(N-1)$-sphere RFD}). Consider the $(N-1)$-sphere RFD $(\ao)$, where $A:= S^{N-1}$ is the unit sphere of $\mbr^N$ (as in part ($a$) just above) and $\Omega$ is the open unit ball in $\mbr^N$; so that $A = \partial \Omega$ and hence, $(\ao) = (\partial \Omega, \Omega).$ Clearly, for any $N \geq 1$, the $(N-1)$-sphere RFD (or {\em relative $(N-1)$-sphere}) is an RFD in $\mbr^N$.

It is shown in \cite[Exple. 4.1.19]{LapRaZu1} that the distance zeta function $\zao$ of $(\ao)$ admits a (necessarily unique) meromorphic extension to all of $\mbc$ given for every $s \in \mbc$ and for any fixed $\delta \in (0,1)$ by the following expression:\footnote{Unlike in part ($a$) of the present subsection, it is easier to compute directly $\zao$ rather than $\tzao$. Of course, in light of \eqref{3.8.1/4}, one can then deduce $\tzao$ from \eqref{3.85}.}
\begin{equation}\label{3.85}
\zao (s) = N \Theta_N \sum_{j=0}^{N-1} \binom{N-1}{j} \frac{(-1)^{N-j-1}}{s-j}.
\end{equation}
Therefore, one deduces at once that
\begin{equation}\label{3.86}
D= D(\zao) = D(\tzao) = \dim_B (\ao) = N-1,
\end{equation}
\begin{equation}\label{3.87}
\dim_{PC} = \{ N-1 \}
\end{equation}
and
\begin{equation}\label{3.88}
\mcd_{\ao} = \mcd(\zao) = \mcd(\tzao) = \{ 0,1, \cdots N-1\}.
\end{equation}
\vspace{1mm}
Observe the contrast between the result obtained for $\mcd_A$ and $\mcd_{\ao}$ in \eqref{3.83} and \eqref{3.88}, respectively, as was alluded to in the introduction to this subsection (i.e., \S \ref{Sec:3.4.6}). In particular, the set of complex dimensions $\mcd_{\ao} = \{ 0,1, \cdots, N-1 \}$ of the relative $(N-1)$-sphere obtained in \eqref{3.88} is exactly the one we would have expected, a priori.

Finally, for every $j \in \{ 0,1, \cdots, N-1 \}$,
\begin{equation}\label{3.89}
\res (\zao, j) = (-1)^{N-j-1} N \Theta_N \binom{N-1}{j}.
\end{equation}
In particular, for $j := D = N -1$, we see that the RFD $(\ao)$ is Minkowski measurable with (relative) Minkowski content given by
\begin{equation}\label{3.90}
\mcm = \mcm (\ao) = (N-D) \res (\zao, D) = N \Theta_N;
\end{equation}
note that here, $N-D=1$. (Compare with \eqref{3.22}.)

\begin{exercise}\label{Exc:3.10}
($i$) Show via a direct computation that $\zao$ is given by \eqref{3.85}.

($ii$) Address the same question as in ($i$) for $\widetilde{\zeta}_A$ (in part ($a$) of this subsection) in order to recover the expression stated in \eqref{3.80}.

($iii$) How are the expressions of $\widetilde{\zeta}_A$ in \eqref{3.80} and of $\zao$ in \eqref{3.85} modified if $A$ and $\Omega$ are, respectively, the $(N-1)$-sphere and the (open) $N$-ball of radius $R$ (instead of radius $1$)?
\end{exercise}

\subsubsection{The Cantor grill}\label{Sec:3.4.7}

Let $A:= C \times [0,1]$ be the Cartesian product of the ternary Cantor set by the unit interval. Henceforth, $A \subseteq \mbr^2$ is referred to as the {\em Cantor grill.} (See \cite[Fig. 2.15]{LapRaZu1} for an illustration.) Then, according to \cite[Exple. 2.2.34 in \S 2.2.3]{LapRaZu1}, 
\begin{equation}\label{3.90.1/2}
D = D_C = \dim_B A = 1 + \log_3 2,
\end{equation}
the set of principal complex dimensions of $A$ is given by
\begin{equation}\label{3.91}
\dim_{PC} A = (1 + \log_3 2) + i \frac{2 \pi}{\log 3} \mbz, 
\end{equation}
while the set of all complex dimensions (in $\mbc$) of $A$ is given by
\begin{align}\label{3.92}
\mcd_A &= \mcd (\zeta_A) = \{ 0,1 \} \cup \bigcup_{m=0}^1 \Big((m + \log_3 2) + i \frac{2 \pi}{\log 3} \mbz \Big) \\
&= \{ 0,1 \} \cup (D_C + i {\bf p} \mbz) \cup ((1+D_C) + i {\bf p} \mbz), \notag
\end{align}
where $D_C = \log_3 2 = \dim_B C$ is the Minkowski dimension of the Cantor set $C$ and ${\bf p} := 2 \pi/ \log 3$ is the oscillatory period of $C$. Each complex dimension in \eqref{3.91} and \eqref{3.92} is simple.

\begin{exercise}\label{Exc:3.11}
($i$) ({\em Higher-dimensional Cantor grills}). Generalize the results of the present subsection to the higher-dimensional Cantor grill $A:= C \times [0,1]^d$, where $d$ is an arbitrary positive integer. In particular, show that  (with ${\bf p}:= 2 \pi/\log 3$ and $D_C = \log_3 2$, as above)
\begin{equation}\label{3.93}
D_A = \dim_B A = d + \dim_B C = d + D_C
\end{equation}
\begin{equation}\label{3.94}
\dim_{PC} A = D_A + i {\bf p} \mbz,
\end{equation}
and
\begin{equation}\label{3.95}
\mcd_A = \mcd (\zeta_A) = \{ 0,1, \cdots, d \} \cup \bigcup_{m=0}^d ((m + D_C) + i {\bf p} \mbz).
\end{equation}
[{\em Hint}: In order to establish \eqref{3.93} and \eqref{3.94}, compare $\zeta_A (s)$ and $\zeta_A (s-d)$ and show that the difference of these two functions is holomorphic in a suitable half-plane, namely,
\[ \{ Re (s) > D_C + (d-1) \} \supseteq \{ Re (s) \geq D_A \}.\] 
The proof of \eqref{3.95} is significantly more complicated; if needed, see \cite[Exple. 2.2.34 and \S 4.7.1]{LapRaZu1}.]\\

($ii$) ({\em Fractal combs}). Let $K$ be any compact subset of $\mbr$ and let $A:= K \times [0,1]^d$, with $d \in \mbn$. Extend the results of part ($i$) to this more general situation.
\end{exercise}

\begin{exercise}({\em Two different Cantor grill RFDs}.)\label{Exc:3.12}
Let $A:= C$, the ternary Cantor set, and $\Omega_1 := (0,1)^2$ while $\Omega_2 := (-1,0) \times (0,1)$. Then, show that the complex dimensions of the RFDs $(A, \Omega_1)$ and $(A, \Omega_2)$ are very different. More specifically, as is observed in \cite[\S 1.1]{LapRaZu1}, it turns out that 
\begin{equation}\label{3.96}
\mcd_{A, \Omega_1} = \mcd_{C \times [0,1]} = \{ 0,1 \} \cup (D_C + i {\bf p} \mbz) \cup ((1+D_C) + i {\bf p} \mbz), 
\end{equation}
as in \eqref{3.92}, where $D_C = \log_3 2$ and ${\bf p} = 2 \pi/\log 3$ are, respectively, the Minkowski dimension and the oscillatory period of $C$. By contrast, 
\begin{equation}\label{3.97}
\mcd_{A, \Omega_2} = \{ 0,1 \} \cup \mcd_C = \{ 0,1 \} \cup (D_C + i {\bf p} \mbz).
\end{equation}
(Further, all the complex dimensions in either \eqref{3.96} or \eqref{3.97} are simple.) Thus, the RFD $(A, \Omega_2)$ no longer ``sees'' the principal complex dimensions of the Cantor grill $C \times [0,1]$ in \eqref{3.91} (or of the RFD $(A, \Omega_1)$, according to \eqref{3.96}).

In addition to establishing \eqref{3.96} and \eqref{3.97}, provide an intuitive explanation for the striking difference between these two results.
\end{exercise}

\subsubsection{The Cantor dust}\label{Sec:3.4.8}

Let $A := C \times C$, where as before, $C$ is the ternary Cantor set. Henceforth, $A$, the Cartesian product of $C$ by itself, is referred to as the {\em Cantor dust}. (See \cite[Fig. 1.2]{LapRaZu1} for a depiction of $A$.) The associated {\em Cantor dust RFD} $(\ao)$ is defined by $\Omega := (0,1)^2$ and $A:= C \times C$, as above.

Then, it is shown in \cite[Exple. 4.7.15]{LapRaZu1} that $\zao$ has (for all $\delta > 0$ large enough) a meromorphic continuation to all of $\mbc$ given for every $s \in \mbc$ by
\begin{equation}\label{3.98}
\zao (s) = \frac{8}{s(3^s -4)} \Bigg( \frac{J(s)}{6^s} + \frac{\Gamma (\frac{1-s}{2})}{\Gamma (\frac{2-s}{2})} \frac{\sqrt{\pi}}{6^s s(3^s -2)} + K(s) \Bigg),
\end{equation}
where $\Gamma = \Gamma (s)$ denotes the classic gamma function (which, as we recall, does not have any zeros anywhere in $\mbc$ but has simple poles at $0, -1, -2, -3, \cdots$. Here, $J(s) := \int_0^{\pi/4} (\cos \theta)^{-s} d \theta$ is an entire function and $K = K(s)$ is a meromorphic function in all of $\mbc$ with (simple) poles at $1,3, 5, \cdots$.

It follows from \eqref{3.98} and the aforementioned properties of $\Gamma, J$ and $K$ that 
\begin{equation}\label{3.99}
D_{\ao} = D(\zao) = \dim_B (\ao) = \log_3 4 =D_A = \dim_B A, 
\end{equation}
as expected since $\log_3 4 = \log_3 2 + \log_3 2 = 2 \dim_B C$, and that the set of complex dimensions (in $\mbc$) of the Cantor dust RFD consists of simple poles of $\zao$ and is given by
\begin{equation}\label{3.100}
\mcd_{\ao} = \mcd(\zao) = \{ 0 \} \cup \Big(\log_3 2 + i \frac{2\pi}{\log 3} \mbz \Big) \cup \Big(\log_3 4 + i \frac{2 \pi}{\log 3} \mbz \Big).
\end{equation}
More specifically, due to possible zero-pole cancellations, $\mcd_{\ao}$ is a subset of the set given on the right-hand side of \eqref{3.100} and contains $D_C = \log_3 2, D_A = D_{C \times C} = \log_3 4$, as well as at least two nonreal complex conjugate principal complex dimensions. It is conjectured in {\em loc. cit.} that, in fact, we have a true equality in \eqref{3.100} and hence, in particular, that the set of principal complex dimensions of the RFD $(\ao)$ (as well as of the compact set $A \subseteq \mbr^2$) is  given by 
\begin{equation}\label{3.101}
\dim_{PC} (\ao) =\dim_{PC} A = \log_3 4 + i \frac{2 \pi}{\log 3} \mbz
\end{equation}
or is, at least, an infinite subset of the `periodic set' $\log_3 4 + i (2 \pi/ \log 3) \mbz$. Furthermore, the author conjectures entirely similar statements about the set of subcritical complex dimensions; namely, the set of complex dimensions of $(\ao)$ (and of $A$) should be equal to (or, at least, coincide with an infinite subset of) the periodic set $\log_3 2 + i (2 \pi/\log 3) \mbz$.

We refer to Conjecture \ref{Con:3.14} below for a much more general and precise statement about the complex dimensions (interpreted there in an extended sense) of Cartesian products.

\begin{exercise}\label{Exc:3.13}
($i$) Deduce from the above results about the Cantor dust RFD $(\ao)$ analogous results concerning the Cantor dust itself, $A:= C \times C$.

($ii$) Generalize the above results about the Cantor dust RFD and the Cantor dust to $A= C^d := C \times \cdots \times C$, the $d$-fold Cartesian product of $C$ with itself (where $d \geq 1$ is arbitrary) and to the associated RFD $(\ao)$, with $A$  as just above and with $\Omega:= (0,1)^d$. We expect that $\dim_{PC} A = \dim_{PC} (\ao)$ and, similarly, $\mcd_A =\mcd_{\ao}$.

($iii$) Finally, replace the ternary Cantor set $C$ by other Cantor-type sets and by more general compact subsets of $\mbr$ (including lattice and nonlattice self-similar sets).
\end{exercise}

The following conjecture of the author was motivated, in part, by the results from \cite{LapRaZu1} (and [\hyperlinkcite{LapRaZu2}{LapRaZu2--6}]) stated about the complex dimensions of the Cantor grill (in \S \ref{Sec:3.4.7}) and of the Cantor dust (in \S \ref{Sec:3.4.8}), along with the results of {\em loc. cit.} briefly discussed in \S \ref{Sec:3.3.4} about the invariance of the complex dimensions under embeddings into higher-dimensional Euclidean spaces. More specifically, the author was first led to stating this conjecture (in December 2016), on the basis of his joint work on quantized number theory and fractal cohomology ([\hyperlinkcite{CobLap1}{CobLap1--2}], \cite{Lap10}), and especially, due to the requirement that `fractal cohomology' should satisfy an appropriate analog of the K\" unneth formula for the cohomology of Cartesian products. (See, especially, \cite[esp., Chs. 4--6]{Lap10} for more details and motivations; see also \S \ref{Sec:4.1} and \S \ref{Sec:4.4} in the epilogue for a brief discussion of the general context.)

For simplicity, we state the conjecture for compact sets rather than for general RFDs, but clearly, an analogous conjecture can be made about RFDs. We also implicitly assume that the corresponding fractal zeta functions are meromorphic in all of $\mbc$,\footnote{See, in particular, Remark \ref{Rem:3.14.1/2} for a more general situation.} 
but we can also state the conjecture relative to a common window (in the sense of \cite{Lap-vF4, LapRaZu1}) or more generally, a domain of $\mbc$ to which the fractal zeta functions can be meromorphically continued.\footnote{If $f$ and $g$ are two meromorphic functions defined on the same domain $U$ of $\mbc$, we simply let $\mfD (f) = \mfD (f; U)$ and $\mfD(g) = \mfD (g; U)$ in the statement of Conjecture \ref{Con:3.14}, and with the notation of Definition \ref{Def:3.15}.}
Finally, the mathematical notion of {\em divisor} (denoted by $\mfD (f)$) of a meromorphic function $f$ is well known and will be recalled in Definition \ref{Def:3.15} below (following \cite[\S 3.4]{Lap-vF4}). For now, we simply mention that $\mfD (f)$ is the multiset of zeros and poles of $f$; that is, the graded set of zeros and poles of $f$, counted according to their multiplicities (with the zeros counted positively and the poles counted negatively).

\begin{conjecture}[Complex dimensions of Cartesian products, \cite{Lap10}]\label{Con:3.14}
Let $A_1$ and $A_2$ be two bounded $($or, equivalently, compact$)$ subsets of $\mbr^{N_1}$ and $\mbr^{N_2}$, respectively. For $j=1,2$, let $\mfD (A_j)$ denote the divisor of $\zeta_{A_j}, \mfD (A_j) := \mfD (\zeta_{A_j})$, and $\mcd (A_j) = \mcd (\zeta_{A_j})$ denote the $($multi$)$set of complex dimensions of $A_j$.\footnote{For simplicity, we assume implicitly that $\overline{\dim}_B A_j < N_j,$ for $j=1,2$. We also use $\zeta_{A_j}$ in order to define both $\mfD_{A_j}$ and $\mcd_{A_j}$ for $j=1,2$; namely, for $j=1,2$, $\mfD_{A_j}:= \mfD(\zeta_{A_j})$ and $\mcd_{A_j}:= \mcd (\zeta_{A_j}) (= \mcd (\widetilde{\zeta}_{A_j})$, in this case). [Note that in spite of the functional equation \eqref{3.11} (or, more generally, \eqref{3.8.1/4}, in the case of RFDs), $\zeta_{A_j}$ and $\widetilde{\zeta}_{A_j}$ have the same poles but not necessarily the same zeros.] One could make a similar conjecture without assuming that $\overline{\dim}_B A_j < N_j$ for $j=1,2$ and by using $\widetilde{\zeta}_{A_j}$ instead of $\zeta_{A_j}$ in order to define both $\mfD_{A_j}$ and $\mcd_{A_j}$, for $j=1,2$.}

Then, we have the identity
\begin{equation}\label{3.102}
\mfD (A_1 \times A_2) = \mfD (A_1) + \mfD (A_2),\footnotemark
\end{equation}\footnotetext{Given two subsets (or submultisets) $E_1$ and $E_2$ of the same additive group, their Minkowski sum $E_1 + E_2$ is defined by 
\begin{equation}
E_1 + E_2 = \{ e_1 +e_2 :e_1 \in E_1, e_2 \in E_2 \},\label{3.103.2}
\end{equation}
viewed as a subset (or submultiset) of this same group.}
Also, we have the inclusion 
\begin{equation}\label{3.103}
\mcd (A_1 \times A_2) \subseteq \mcd (A_1) + \mcd (A_2),
\end{equation}
the Minkowski sum of $\mcd (A_1)$ and $\mcd (A_2)$. Furthermore, typically $($or ``generically'', in a vague sense$)$, we have an equality in \eqref{3.103}, because we do not have zero-pole cancellations in such cases{\em :}
\begin{equation}\label{3.103.1/2}
\mcd (A_1 \times A_2) = \mcd (A_1) + \mcd (A_2).
\end{equation}
\end{conjecture}

\begin{remark}\label{Rem:3.14.1/2}
If $\zeta_{A_1 \times A_2}$ is not necessarily meromorphic in all of $\mbc$ (or in the given domain $U \subseteq \mbc$ under consideration), but $\zeta_{A_1}$ and $\zeta_{A_2}$ still are, then, under appropriate hypotheses, we expect (based in part on cohomological and spectral considerations; see \S \ref{Sec:4.3} and \S \ref{Sec:4.4}, along with \cite{Lap10}) that Conjecture \ref{Con:3.14} can be suitably modified and extended by substituting for the ordinary divisor $\zeta_{A_1 \times A_2}$ a ``generalized divisor'' (still denoted by $\bm{\mcd} (\zeta_{A_1 \times A_2})$ or, in short, $\bm{\mcd} (A_1 \times A_2)$) which takes into account the (nonremovable) singularities (and not just the poles) of $\zeta_{A_1 \times A_2}$; in that case, we must also replace the Minkowski sum by its closure in the right-hand side of the counterpart of \eqref{3.102}:
\begin{equation}\label{3.104.1/4}
\bm{\mcd} (A_1 \times A_2) = c \ell (\bm{\mcd} (A_1) + \bm{\mcd} (A_2)).
\end{equation}
In particular, the counterpart of \eqref{3.103} becomes
\begin{equation}\label{3.104.1/2}
\mcd (A_1 \times A_2) \subseteq c \ell (\mcd (A_1) + \mcd (A_2)).
\end{equation}
\end{remark}
\vspace{1mm}
Next, as promised, we recall the definition of the divisor of a meromorphic function (see, e.g., \cite[Def. 3.11]{Lap-vF4}). It goes back at least to Riemann in the related context of Riemann surfaces and has counterparts and various generalizations in many fields, including arithmetic and algebraic geometry, as well as in algebraic combinatorics. The notion of divisor is ideally suited to making precise sense of the possible cancellations between the zeros and poles of a meromorphic function. Hence, its key use in the statement of Conjecture \ref{Con:3.14} above.

\begin{definition}({\em Divisor of a meromorphic function}).\label{Def:3.15}
Let $f$ be a meromorphic function on a given domain $U \subseteq \mbc$. Then, the {\em divisor} of $f$, denoted by $\mfD = \mfD (f)$, is defined as the formal sum\footnote{This is an at most countable sum since clearly, $ord(f, \omega) = 0$ whenever $\omega \in U$ is neither a pole nor a zero of $f$. Also, $\mfD (f)$ can be viewing as lying in the free ablian group generated by the distinct zeros and poles of $f$.}
\begin{equation}\label{3.104}
\mfD (f) = \mfD(f;U) := \sum_{\omega \in U} ord (f; \omega) \cdot \omega,
\end{equation}
where the order of $f$ at $\omega \in U$ is defined as the integer $m \in \mbz$ such that the function $|f (s)(s -\omega)^{-m}|$ is bounded away from $0$ and $\infty$ in a neighborhood of $\omega$. Therefore, if $\omega$ is a zero of (positive) multiplicity $m$, then $order (f, \omega) = m$, a positive integer, whereas if $\omega$ is a pole of $f$ of (positive) multiplicity $n$, then $ord (f, \omega) = -n =: m$, a negative integer. Furthermore, $ord (f,\omega) = 0$ if $\omega$ is neither a zero nor a pole of $f$ (in $U$). 

Formally, the {\em Minkowski sum} of the divisors $\mfD (f)$ and $\mfD (g)$ of two meromorphic functions on the same domain $U$ of $\mbc$ is denoted  by $\mfD (f) + \mfD (g)$ and is given by 
\begin{equation}\label{3.105}
\mfD (f) + \mfD (g) = \sum_{\omega \in U} (ord (f,\omega) + ord (g, w)) \cdot \omega.
\end{equation}
\end{definition}

Note that on the right-hand side of \eqref{3.105}, some of the coefficients of $ord (f,\omega) + ord (g, \omega)$ may vanish, corresponding precisely to the aforementioned (exact) zero-pole cancellations.

\begin{remark}\label{Rem:3.16}
($a$) Conjecture \ref{Con:3.14} above is consistent with all the known examples  and results, including those described in \S \ref{Sec:3.3.4}, \S \ref{Sec:3.4.7} and \S \ref{Sec:3.4.8}. (See, in particular, parts ($b$) and ($c$) of the present remark for the examples of the Cantor grill and the Cantor dust, respectively.) We warn the interested reader, however, that in the general case, this conjecture is likely to be quite difficult to prove although the coherence and consistency of the fractal cohomology theory partly developed in \cite{Lap10} (and briefly discussed in \S \ref{Sec:4.1} and \S \ref{Sec:4.4}) clearly requires that Conjecture \ref{Con:3.14} must be true.

($b$) ({\em Cantor grill, revisited}) As a simple verification, for the example of the Cantor grill $C \times [0,1]$ discussed in the main text of \S \ref{Sec:3.4.7} (and corresponding to the case when $d=1$) or, more generally, of its higher-dimensional counterpart $A := C \times [0,1]^d$ considered in Exercise \ref{Exc:3.11}($i$), we have
\begin{equation}\label{3.106}
\mcd_{[0,1]^d} := \mcd (\widetilde{\zeta}_{[0,1]^d}) = \{ 0,1, \cdots, d \},
\end{equation}
as can be easily verified,\footnote{Note that we only use $\widetilde{\zeta}_{[0,1]^d}$ here because $D_{[0,1]^d} = \dim_B [0,1]^d =d.$} and
\begin{equation}\label{3.107}
\mcd_C = \{ 0\} \cup \mcd_{CS} = \{ 0 \} \cup (D_C + i{\bf p} \mbz),
\end{equation} 
where $D_C = \log_3 2$ and ${\bf p} = 2 \pi/\log 3$. (See \eqref{3.13} and \eqref{3.14.1/2} for a closely related fact.) Therefore, the Minkowski sum $\mcd_C + \mcd_{[0,1]^d}$ is given by
\begin{equation}
\{0,1, \cdots, d \} \cup \Bigg( \bigcup_{m=0}^d ((m + D_C) + i{\bf p} \mbz) \Bigg),
\end{equation}
which is precisely $\mcd_A = \mcd_{C \times [0,1]^d},$ as given by \eqref{3.95} (and, in particular, for the usual Cantor grill corresponding to the choice $d=1$, as given by \eqref{3.92}). Of course, this is also in agreement with the result predicted by Conjecture \ref{Con:3.14} in \eqref{3.103}, and even in \eqref{3.103.1/2} since we have
\begin{equation}\label{3.109}
\mcd_C + \mcd_{[0,1]^d} = \mcd_{C \times [0,1]^d}
\end{equation}
in the present situation, as expected in the ``generic'' case. 

Observe that it is absolutely crucial here that $\mcd_C$ be given by \eqref{3.107}, as obtained above via the higher-dimensional theory of fractal zeta functions, and not just by $\mcd (\zeta_{CS}) = D_C + i{\bf p} \mbz$, as given by the theory of fractal strings and their associated geometric zeta functions. Recall that $\mcd_C = \mcd (\zeta_C) = \mcd (\widetilde{\zeta}_C) = \{ 0 \} \cup \mcd (\zeta_{CS}),$ where $\mcd (\zeta_{CS})$ is given just above. (An entirely analogous comment can be made about the example of the Cantor dust dealt with in part ($c$).)

($c$) ({\em The Cantor dust, revisited}). Let $A = C \times C$ be the Cantor dust, as in \S \ref{Sec:3.4.8}. Then, in light of \eqref{3.107} in part ($b$) just above, it is easy to check that $\mcd_C + \mcd_C$ is given by (with the same notation as in part ($b$) of this remark)
\begin{equation}\label{3.111}
\mcd_C + \mcd_C = \{0 \} \cup (D_C +i{\bf p} \mbz) \cup (2 D_C + i{\bf p} \mbz),
\end{equation}
with $2 D_C = \log_3 4 = D_A$, in accord with the result stated in \eqref{3.100} and more precisely, immediately after \eqref{3.100}, and also in agreement with the containment \eqref{3.103} of Conjecture \ref{Con:3.14}, even though it has not yet been fully proved for this example. In fact, we expect once again that we have an equality in the present situation, as predicted for the ``generic'' case in \eqref{3.103.1/2} of Conjecture \ref{Con:3.14}; namely, we expect that
\begin{equation}\label{3.112}
\mcd_{C \times C} = \mcd_C + \mcd_C.
\end{equation} 

We leave it to the interested reader to calculate the $d$-fold Minkowski sum of $\mcd_C$ and to deduce from it (assuming that we are in the generic case of Conjecture \ref{Con:3.14}) the expression for $\mcd_{C^d}$, where $C^d$ is the $d$-fold Cartesian product of $C$ by itself, as sought for in part ($ii$) of Exercise \ref{Exc:3.13}.
\end{remark}

The following example (the Cantor graph RFD) will play an important role in \S \ref{Sec:3.6} in order to illustrate the definition of fractality in terms of the existence of nonreal complex dimensions.\\

\subsubsection{The devil's staircase $($or Cantor graph$)$ RFD}\label{Sec:3.4.9}

Let $A$ denote the graph of the classic Cantor function, also called from now on the devil's staircase or the Cantor graph. It is well known that $A$ is a self-affine (rather than a self-similar) set in $\mbr^2$; more specifically, it scales distances by the factors $1/2$ and $1/3$ along the horizontal and vertical directions, respectively. (See, for example, \cite[Remark 1.2.1]{LapRaZu1} for a more detailed description, along with \cite[Plate 83, p. 83]{Man} and \cite[Figs. 1.5--1.7]{LapRaZu1} for an illustration.)

Next, let $\Omega$ be the union of the triangles $\Delta_k$ and $\widetilde{\Delta}_k$ above and below the horizontal parts of the Cantor graph $A$, for $k=1,2, \cdots$. (For each $k \geq 1$, at the $k$-th stage of the construction of $A$ and of the RFD $(A, \Omega)$ below, there are $2^{k-1}$ pairs of congruent triangles $\Delta_k$ and $\widetilde{\Delta}_k$.) Then, $(A,\Omega)$ is called the {\em Cantor graph RFD}; clearly, it is an RFD in $\mbr^2$. Intuitively, it can be thought of as the Cantor graph viewed from the perspective of the (non-Euclidean) $\ell^\infty$ metric of $\mbr^2$, given by $||(x,y)||_\infty = \max (|x| , |y|),$ for ${x,y} \in \mbr^2$. In spite of that, $(A, \Omega)$ captures the essence of the Cantor graph.

It is shown in \cite[Exple. 5.5.14]{LapRaZu1} that $\zeta_{A, \Omega}$ admits a (necessarily unique) meromorphic continuation to all of $\mbc$, given for every $s \in \mbc$ by 
\begin{equation}\label{3.113}
\zeta_{A, \Omega} (s) = \frac{2}{s(3^s -2)(s-1)}.
\end{equation}
It follows that the set of principal complex dimensions of the Cantor graph RFD $(A, \Omega)$ is given by
\begin{equation}\label{3.114}
\dim_{PC} (A,\Omega) = \{1 \},
\end{equation}
where
\begin{equation}\label{3.115}
D_{A,\Omega} = \dim_B (A, \Omega) =1,
\end{equation}

\noindent and that the set of all complex dimensions of $(A, \Omega)$ is
\begin{equation}\label{3.116}
\mcd_{A,\Omega} = \mcd (\zeta_{\ao}) = \{0,1 \} \cup \Big(\log_3 2 + i \frac{2\pi}{\log 3} \mbz \Big),
\end{equation}
with each complex dimension being simple. Apart from the complex dimension at $s = 0$, this is in complete agreement with the values of the complex dimensions of the Cantor graph $A$ predicted in [\hyperlinkcite{Lap-vF2}{Lap-vF2--4}] (see, especially, \cite[\S 12.1.2]{Lap-vF4}), on the basis of an approximate tube formula for $A$ (but, at the time, without an appropriate definition of higher-dimensional fractal zeta functions).

Furthermore, with $s_k := \log_3 2 + i (2 \pi/ \log 3)$ for each $k \in \mbz$, we have that 
\begin{equation}\label{3.117}
\res (\zeta_{A, \Omega}, s_k) = \frac{1}{(\log 3) (s_k-1) s_k}.
\end{equation}
Moreover,
\begin{equation}\label{3.118}
\res (\zeta_{A, \Omega}, 0) = \res (\zao, 1) = 2.
\end{equation}

Let us  next briefly consider the {\em Euclidean Cantor graph RFD} $(A, A_{1/3})$, where $A_{1/3}$ is the $(1/3)$-neighborhood (with respect to the Euclidean metric of $\mbr^2$) of the Cantor graph $\mbr^2$. Then, it is also shown in \cite[\S 5.5.4]{LapRaZu1} that
\begin{equation}\label{3.119}
\mcd_{A, A_{1/3}} = \mcd (\zeta_{A, A_{1/3}}) \subseteq \mcd_{(\ao)} = \{ 0,1 \} \cup \Big( \log_3 2 + i \frac{2 \pi}{\log 3} \mbz \Big)
\end{equation}
and that\footnote{Observe that we can rewrite \eqref{3.119} as follows:
\begin{equation}
\mcd_{\ao} = \mcd_{CS} \cup \{ 1\},
\end{equation}
where the set of complex dimensions of the Cantor string $CS$ (and of the set $C$) is given by
\begin{equation}\label{3.119.1/2}
\mcd_{CS} = \mcd_C = \{ 0 \} \cup (D_{CS} + i{\bf p}\mbz).
\end{equation}
Here, $D_{CS} = D_C = \log_3 2$ is the Minkowski dimension of the Cantor string (and set) and ${\bf p} = (2\pi / \log 3)$ is its oscillatory period. \label{Fn:100}}

\begin{equation}\label{3.120}
\dim_{PC} (A, A_{1/3}) = \dim_{PC} (\ao) = \{ 1 \},
\end{equation}
where
\begin{equation}\label{3.121}
\dim_B (A, A_{1/3}) = \dim_B (\ao) = D(\zao) = D(\zeta_{A, A_{1/3}}) = 1.
\end{equation}

It can be checked numerically (see {\em loc. cit.}) that $\mcd_{A, A_{1/3}}$ contains several pairs of {\em nonreal} complex conjugate complex dimensions with real part $D_{CS} = D_C = \log_3 2$. In fact, we expect that a lot more is true, as is stated in the following conjecture. (Compare with the inclusion obtained in \eqref{3.119}.)

\begin{conjecture}{\em ([\hyperlinkcite{LapRaZu4}{LapRaZu4, 7}], \cite[\S 5.5.4]{LapRaZu1}).}\label{Con:3.17}
We not only have that $\mcd_{A, A_{1/3}} \subseteq \mcd_{\ao}$ $($as stated in \eqref{3.119}$)$ but also that
\begin{equation}\label{3.122}
\mcd_A = \mcd_{A, A_{1/3}} = \mcd_{\ao},
\end{equation}
where $($as in \eqref{3.116} and with the notation of footnote \ref{Fn:100}$)$
\begin{equation}\label{3.123}
\mcd_{\ao} = \{ 0,1 \} \cup (D_{CS} + i{\bf p} \mbz),
\end{equation}
with $D_{CS} = \log_3 2$ and ${\bf p} = (2\pi / \log 3)$.
\end{conjecture}
\vspace{2mm}
\subsubsection{Self-similar sprays}\label{Sec:3.4.10}

The notion of a self-similar (or, more generally, fractal) spray was introduced in \cite{LapPo3}, formalizing a number of examples discussed in \cite{Lap1, Lap3}. It was then used extensively, in particular, in [\hyperlinkcite{LapPe1}{LapPe1--2}, \hyperlinkcite{LapPeWi1}{LapPeWi1--2}] where were established fractal tube formulas extending (and using) those obtained for fractal strings in [\hyperlinkcite{Lap-vF1}{Lap-vF1--4}]. (See also, e.g., \cite{DemDenKoU}, \cite{DemKoOU}, \cite{DenKoOU} and \cite{KoRati}.) The notion of a self-similar spray or RFD introduced in \cite{LapRaZu1} (and in [\hyperlinkcite{LapRaZu4}{LapRaZu4,6}]) is slightly different from the one used in those references; in particular, its greater flexibility is well suited to a variety of examples discussed in {\em loc. cit.} as well as in the present section (e.g., the relative $N$-gasket in \S \ref{Sec:3.4.4}). Since the precise definition is a bit technical, we simply mention that, roughly speaking, a {\em self-similar spray} (or RFD) is an RFD $(\ao)$ in $\mbr^N$ which, up to displacements, is obtained from a single {\em generator} (also called the {\em base RFD}) $(\partial G, G)$),\footnote{For the simplicity of exposition, we consider here the case of a single generator. The case of multiple generators (i.e., finitely many generators) is an immediate consequence of the present case of a single generator; see part ($a$) of Remark \ref{Rem:3.18}.} itself an RFD in $\mbr^N$, via a {\em scaling sequence} $\mcl = (\ell_j)$, which is a possibly unbounded self-similar fractal string (in the sense of \cite[Ch. 3]{Lap-vF4}) with (not necessarily distinct) scaling ratios $r_1, \cdots, r_J$. Here, $J \geq 2$ and $0 < r_1, \cdots, r_J < 1$; furthermore, $\mcl$ is comprised of all the finite products of elements of the {\em ratio list} $\{ r_1 \cdots, r_J \}$.

It follows that $(\ao)$ is the disjoint union of scaled copies of the generating RFD $(\partial \Omega, \Omega)$, scaled by the self-similar string $\mcl$. In order for $\Omega$ to have finite total volume, we assume that $G$ is open, $|G| < \infty, \dim_B (\partial G, G) < N$ and $\sum_{j=1}^J r_j^N < 1$. 

Let $(\ao)$ be an arbitrary self-similar spray (or RFD) in $\mbr^N$. It is shown in \cite[Thm. 4.2.17]{LapRaZu1} that $\zeta_{\ao}$ admits a meromorphic continuation to all of $\mbc$ given for every $s \in \mbc$ by 
\begin{equation}\label{3.124}
\zeta_{\ao} (s) = \frac{\zeta_{\partial G, G} (s)}{1 - \sum_{j=1}^J r_j^s}
\end{equation} 
or equivalently, by the following {\em factorization formula}:\footnote{A priori, the equivalent identities \eqref{3.124} and \eqref{3.125} are valid for $Re(s) > \overline{D}$, with $\overline{D} = D_{\ao} = D(\zao)$, as given by \eqref{3.129}. However, upon meromorphic continuation, it remains valid in any domain $U$ to which $\zeta_{\partial G, G}$ can be meromorphically extended. For example, we can take $U = \mbc$ if $G$ is sufficiently ``nice'' (e.g., monophase or even pluriphase, in the sense of \cite{LapPe2, LapPeWi1}, and in particular, a nontrivial polytope \cite{KoRati}.}
\begin{equation}\label{3.125}
\zeta_{A, \Omega} (s) = \zeta_{\mfs} (s) \cdot \zeta_{(\partial G, G)} (s),
\end{equation} 
where $\zeta_{{\bf s}} (s) := \zeta_\mcl (s)$, the geometric zeta function of the possibly unbounded self-similar string $\mcl$, is called the {\em scaling zeta function} of the fractal spray $(\ao)$ and is given (also for  every $s \in \mbc)$ by
\begin{equation}\label{3.126}
\zeta_{\mfs} (s) = \frac{1}{1-\sum_{j=1}^J r_j^s}.
\end{equation}
It then follows from \eqref{3.124}--\eqref{3.126} that (with $\mcd_{\ao} = \mcd(\zao), \mcd_{(\partial G, G)} = \mcd (\zeta_{(\partial G, G)})$ and $\mcd_{\mfs} = \mcd(\zeta_{\mfs}) = \mcd (\zeta_\mcl)$)
\begin{equation}\label{3.127}
\mcd(\zeta_{\ao}) \subseteq \mcd(\zeta_\mfs) \cup \mcd(\zeta_{\partial G, G}),
\end{equation}
where the containment comes from the fact that there could, in general, be zero-pole cancellations in the expression on the right-hand side of \eqref{3.125} (or, equivalently, of \eqref{3.124}). Note that we also have the following equality:
\begin{equation}\label{3.128}
\mfD (\zeta_{\ao}) = \mfD(\zeta_\mfs) \cup \mfD (\zeta_{\partial G, G}),
\end{equation}
where (as in \eqref{3.104}) $\mfD (f)$ denotes the divisor of the meromorphic function $f$.

We also deduce from \eqref{3.125} that (with $D_{\ao} = D(\zao)$ and $D_{\partial G, G} = D(\zeta_{\partial G, G})$)
\begin{equation}\label{3.129}
D_{\ao} = \max (\sigma_0, D_{\partial G,G})
\end{equation}
or, equivalently,
\begin{equation}\label{3.130}
\overline{\dim}_B (\ao) = \max (\sigma_0, \overline{\dim}_B (\partial G,G)),
\end{equation}
where $\sigma_0 = D(\zeta_\mfs)$ is the {\em similarity dimension} of the self-similar spray $(\ao)$ defined as the unique {\em real} solution of the Moran equation \cite{Mora}:\footnote{Clearly, in light of the definition \eqref{3.131} of $\sigma_0$, we have $\sigma_0 \in (0, N)$; indeed, by assumption, $J \geq 2 > 1$ (hence, $\sigma_0 > 0$) and $\overline{\dim}_B (\partial G, G) < N$ (hence, $\sigma_0 < N$). Observe that as a result, and in light of \eqref{3.129} and \eqref{3.130}, we have that $\overline{\dim}_B (\ao) < N$. Thus the complex dimensions of $(\ao)$ can be defined indifferently via $\zeta_{\ao}$ or $\tzao$.  }
\begin{equation}\label{3.131}
\sum_{j=1}^J r_j^{\sigma_0} =1.
\end{equation}

In addition, in light of \eqref{3.126}, $\mcd_\mfs = \mcd(\zeta_\mfs)$ is given as the {\em set of complex  solutions of the complexified Moran equation}, also called the {\em set of scaling complex dimensions of} $(\ao)$:
\begin{equation}\label{3.132}
\mcd_{\mfs} = \Bigg\{ s \in \mbc: \sum_{j=1}^J r_j^s =1 \Bigg\},
\end{equation}
so that this latter expression for $\mcd_\mfs$ can be substituted in \eqref{3.127}:
\begin{equation}\label{3.133}
\mcd_{\ao} \subseteq \mcd_{\partial G, G} \cup \Bigg\{ s \in \mbc: \sum_{j=1}^J r_j^s = 1 \Bigg\}.
\end{equation}
If there are no zero-pole cancellations in \eqref{3.124} (or, equivalently, in \eqref{3.125}), which (for ``nice'' generators) is the case ``generically'' (in a vague sense), then we have an actual equality in \eqref{3.133}. (See \cite[Thm. 4.2.19]{LapRaZu1}.)

If the generator $G$ (a bounded open subset of $\mbr^N$) is sufficiently nice (e.g., according to the main result of \cite{KoRati} proving a conjecture in [\hyperlinkcite{LapPe1}{LapPe1--2}], if $G$ is a nontrivial polytope, a very frequent situation for the classic self-similar fractals), then it is shown in \cite[\S 5.5.6]{LapRaZu1} that
\begin{equation}\label{3.152}
\mcd_{\partial G, G} \subseteq \{ 0, 1, \cdots, N-1 \}.
\end{equation}
It therefore follows from \eqref{3.128} and \eqref{3.133} that
\begin{equation}\label{3.152.1/2}
\mcd_{\ao} \subseteq \mcd_{\partial G, G} \cup \mcd_{\mfs} \subseteq \{ 0, 1, \cdots N-1 \} \cup \Bigg\{ s \in \mbc: \sum_{j=1}^J r_j^s =1 \Bigg\},
\end{equation}
with equalities in the ``generic'' case. In such a situation (much as in [\hyperlinkcite{LapPe1}{LapPe1--2}), \hyperlinkcite{LapPeWi1}{LapPeWi1--2}]), $\mcd_{\partial G,G} \subseteq \{ 0,1, \cdots, N-1 \}$ and $\mcd_\mfs \subseteq \{ s \in \mbc : \sum_{j=1}^J r_j^s =1 \}$ are respectively called the {\em integer dimensions} and the {\em scaling dimensions} of the self-similar spray $(\ao)$.\\

{\em Sketch of the proof of the factorization formulas} \eqref{3.124}--\eqref{3.125}. It is instructive to provide the proof of formula \eqref{3.124} (and hence, equivalently, of the factorization formula \eqref{3.125}), as given in the proof of \cite[Thms. 4.2.17 and 4.2.19, p. 289]{LapRaZu1}.

First, we note that since the (possibly unbounded) self-similar string $\mcl = (\ell_j)_{j \geq 1}$ has for `{\em scales}' $\{ \ell_j: j \geq 1 \}$, the free monoid with generators $r_1, \cdots, r_J$, $\mcl$ satisfies the following {\em self-similar identity} (much as in \cite[\S 4.4.1]{Lap-vF4}):\footnote{Here, in \eqref{3.153}, the symbol $\sqcup$ denotes the disjoint union of fractal strings.}
\begin{equation}\label{3.153}
\mcl = \mcl_0 \sqcup \bigsqcup_{j=1}^J (r_j \mcl), \text{ where } \mcl_0 := \{ 1 \}.
\end{equation}

Further, a moment's reflection shows that we can deduce from \eqref{3.153} the following {\em self-similar identity} satisfied by the self-similar spray $(\ao)$:
\begin{equation}\label{3.154}
(\ao) = (\partial G, G) \sqcup \bigsqcup_{j=1}^J (r_j (\ao)),
\end{equation}
where much as in \eqref{3.153}, the symbol $\sqcup$ denotes the disjoint union of RFDs in $\mbr^N$. (See \cite[Def. 4.1.43]{LapRaZu1} for the precise definition of such a disjoint union of RFDs in $\mbr^N$, which is itself an RFD in $\mbr^N$.)

Now, by combining the scaling (and the invariance) property of the distance zeta function (see \S \ref{Sec:3.3.3} and \S \ref{Sec:3.3.4}) according to which (for all $s \in \mbc$) 
\begin{equation}\label{3.155}
\zeta_{r_j (\ao)} (s) = r_j^s \, \zeta_{\ao} (s), \text{ for each } j=1, \cdots, J,
\end{equation}
where $r_j (\ao) := (r_j A,r_j \Omega)$, along with the (finite) additivity of the distance zeta function under disjoint unions (see a special case of \cite[Prop. 4.1.17]{LapRaZu1}),\footnote{This finite additivity property can be easily established. The significantly more delicate countable additivity property  (precisely stated and established in \cite[Prop. 4.1.17]{LapRaZu1}) is not needed here but is frequently used throughout \cite{LapRaZu1}.} we obtain the following functional equation,
\begin{equation}\label{3.156}
\zeta_{\ao} (s) = \zeta_{\partial G, G}(s) + \sum_{j=1}^J r_j^s \zeta_{\ao} (s),
\end{equation}
which is equivalent to \eqref{3.124} (and to \eqref{3.125}), after an elementary factorization.

\begin{remark}\label{Rem:3.18}
($a$) ({\em Multiple generators}). In the case of a self-similar spray $(\ao)$ with multiple generators, $G_1, \cdots, G_Q$, we simply add up the results obtained for each of the generators. More specifically,  in light of \eqref{3.124}--\eqref{3.126}  applied to each of the generators, we then have that
\begin{equation}\label{3.157}
\zeta_{\ao} (s) = \frac{\sum_{q=1}^Q \zeta_{\partial G_q, G_q} (s)}{1 - \sum_{j=1}^J r_j^s}.
\end{equation}

($b$) ({\em Fractal sprays}). For (not necessarily self-similar) fractal sprays, several (but not all) of the above results are still valid. More specifically, if $(\ao)$ is a fractal spray RFD in $\mbr^N$, with a single generator $(\partial G, G)$ scaled by the (not necessarily bounded or self-similar) fractal string $\mcl = (\ell_j)_{j \geq 1}$ and such that $|\Omega| < \infty$ and $\overline{\dim}_B (\partial G, G) < N$, then
\begin{equation}\label{3.158}
\overline{D} = \overline{\dim}_B (\ao) = \max (\overline{D}_\mfs, \overline{\dim}_B (\partial G,G)),
\end{equation}
where $D_{\mfs} := D (\zeta_\mcl) = D(\zeta_{\mfs}), \overline{D} := D (\zeta_{\ao})$ and $\overline{D}_{\partial G, G} := D (\zeta_{\partial G, G}) = \overline{\dim}_B (\partial G, G)$. 

Then, for all $s \in \mbc$ such that $Re(s) > \overline{D}$ (and, hence, upon meromorphic continuation and for $G$ nice enough,\footnote{This is the case, for example, if $G$ is a nontrivial polytope \cite{KoRati} or more generally, if $G$ is either monophase or pluriphase (in the sense of \cite{LapPe2, LapPeWi1}).} 
for all  $s \in \mbc$),
\begin{equation}\label{3.159}
\zeta_{\ao} (s) = \zeta_{\mfs} (s) \cdot \zeta_{\partial G, G} (s),
\end{equation}
where $\zeta_{\mfs} := \zeta_\mcl$ denotes the meromorphic continuation of $\zeta_\mcl$:
\begin{equation}\label{3.160}
\zeta_{\mfs} (s) := \zeta_{\mcl} (s) \Bigg(= \sum_{j \geq 1} \ell_j^s, \text{ for } Re(s) > \overline{D}_\mfs \Bigg).
\end{equation}
It then follows from \eqref{3.159} that (with $\mcd_{\mfs} := \mcd (\zeta_\mcl) = \mcd(\zeta_{\mfs}))$
\begin{equation}\label{3.161}
\mcd_{\ao} \subseteq \mcd_{\partial G, G} \cup \mcd_{\mfs},
\end{equation}
with equality instead of an inclusion, unless there are zero-poles cancellations. Furthermore, we always have the following identity between divisors:
\begin{equation}\label{3.162}
\mfD_{\ao} = \mfD_{\partial G, G} \cup \mfD_\mfs.
\end{equation}

Naturally, a comment entirely similar to the one made in part ($a$) of this remark applies here if the fractal spray has multiple (i.e., finitely many) generators instead of a single generator.
\end{remark}

In the next exercise, the interested reader is asked implicitly to use, in particular, the results of the present subsection in order to answer the various questions.

\begin{exercise}\label{Exc:3.19}
($a$) ({\em Sierpinski gasket}). Let $A$ be the classic Sierpinski gasket. Then, calculate $\zeta_A, \zeta_{A, T}$ (where $T$ is the equilateral triangle with side lengths $1$), and deduce from this computation the two sets of complex dimensions $\mcd_A$ and $\mcd_{A,T}$, respectively. Compare your results with those stated in \S \ref{Sec:3.4.1} (or with the $N=2$ case of \S \ref{Sec:3.4.4}).

($b$) ({\em Sierpinski carpet}). Let $A$ be the classic Sierpinski carpet. Then, answer the same questions as in part ($a$) just above, except with the unit triangle $T$ replaced by the open unit square $S := (0,1)^2$, and with \S \ref{Sec:3.4.1} and \S \ref{Sec:3.4.4} replaced by \S \ref{Sec:3.4.2} and \S \ref{Sec:3.4.3}, respectively.

($c$)  Answer analogous questions for the relative (or inhomogeneous) $N$-gasket, as discussed in \S \ref{Sec:3.4.4}, as well as for the $\frac{1}{2}$-square and $\frac{1}{3}$-square fractals, as discussed in   \S \ref{Sec:3.4.5}.
\end{exercise}

We refer the interested reader to [\hyperlinkcite{LapRaZu1}{LapRaZu1--10}] for many other examples of computations of fractal zeta functions and complex dimensions (as well as of fractal tube formulas, in regard to \S \ref{Sec:3.5} just below). These examples include non self-similar fractals or RFDs such as fractal nests (see \cite[\S 3.5 and Exples. 5.5.16 and 5.5.24 in \S 5.5.5]{LapRaZu1}) as well as bounded and unbounded geometric chirps (see  \cite[\S 3.6, \S 4.4.1 and Exple. 5.5.19 in \S 5.5.5]{LapRaZu1}). They also include examples of fractal strings and RFDs, as well as compact sets, with principal complex dimensions having arbitrarily prescribed (finite or infinite) multiplicities (see \cite[Thms. 3.3.6 and 4.2.19]{LapRaZu1}).

\subsection{Fractal tube formulas and Minkowski measurability criteria: Theory and examples}\label{Sec:3.5}

The goal of this subsection is to briefly state (in \S \ref{Sec:3.5.1}) the higher-dimensional analog (obtained in \cite{LapRaZu6} and \cite[\S\S 5.1--5.3]{LapRaZu1}) of the fractal tube formulas obtained for fractal strings in [\hyperlinkcite{Lap-vF2}{Lap-vF2--4}] (see, especially, \cite[Ch. 8]{Lap-vF4}). In fact, the latter tube formulas (briefly discussed in \S \ref{Sec:2.1}) are now but a very special case of their higher-dimensional counterparts.\footnote{However, as is often the case in such situations, the results and techniques developed in \cite[Chs. 5 and 8]{Lap-vF4} in order to prove fractal tube formulas and other explicit formulas for (generalized) fractal strings are key to establishing the higher-dimensional (pointwise and distributional, exact or with error term) fractal tube formulas. Nevertheless, a significant amount of additonal work is required  in order to prove those tube formulas for RFDs (or, in particular, bounded sets) in $\mbr^N$; see, especially, \cite[\S\S 5.1--5.3]{LapRaZu1}.}

We will also state (in \S \ref{Sec:3.5.2}) some of the Minkowski measurability criteria obtained in \cite[\S 5.4.3 and \S 5.4.4]{LapRaZu1}.

In fact, the main goal of this section is to illustrate (in \S \ref{Sec:3.5.3}) the aforementioned results concerning fractal tube formulas and Minkowski measurability criteria (and established, in particular, in \cite{LapRaZu7} and \cite[\S 5.4]{LapRaZu1}), in \S \ref{Sec:3.5.1} and \S \ref{Sec:3.5.2}, respectively) via a variety of examples, many of which have been discussed from other points of view earlier in the paper, especially in \S \ref{Sec:3.4}.\\  

\subsubsection{Fractal tube formulas for RFDs in $\mbr^N$, via distance and tube zeta functions}\label{Sec:3.5.1}

Recall from \S \ref{Sec:2.1}, adapted to the present much more general situation of RFDs in $\mbr^N$, that a fractal tube formula enables us to express the tube function 
\begin{equation}\label{3.163}
\varepsilon \mapsto V(\varepsilon) = V_{\ao} (\varepsilon) := |A_\varepsilon \cap \Omega|_N
\end{equation} 
of a RFD $(\ao)$ in $\mbr^N$ in terms of the complex dimensions of $(\ao)$ and the associated residues of the corresponding fractal zeta function (here, either the distance or the tube zeta function of $(\ao)$).

An important feature of the higher-dimensional theory of complex dimensions is that such fractal tube formulas can be established under great generality for RFDs in $\mbr^N$, via either distance zeta functions or tube zeta functions, as well as via other fractal zeta functions, such as the so-called shell zeta functions and the Mellin zeta functions, which are of interest in their own right but are also used in \cite[Ch. 5]{LapRaZu1} in key steps towards the proof of the fractal tube formulas and of related Minkowski measurability criteria. We will limit ourselves here to the case of distance and tube zeta functions.

Depending on the growth assumptions made about the fractal zeta functions under consideration,\footnote{These polynomial-type growth conditions are referred to (much as in \cite[Chs. 5 and 8]{Lap-vF4}) as {\em languidity conditions} in the case of tube formulas with error term and as {\em strong languidity conditions} in the case of exact tube formulas. These conditions are slightly different for distance and tube zeta functions; see \cite[Ch. 5]{LapRaZu1} for more details.}
we obtain fractal tube formulas with error term or without error term (i.e., {\em exact}), as well as interpreted  either pointwise or distributionally. Therefore, just as in the one-dimensional case of fractal strings, but now in the much more general case of RFDs (and, in particular, of bounded sets) in $\mbr^N$, there is a lot of flexibility for obtaining and applying such tube formulas.

Let us now be a little bit more specific, while avoiding technicalities and cumbersome (although useful) definitions. We state all of the results for RFDs $(\ao)$ in $\mbr^N$ but, as usual, the special case of bounded subsets $A$ of $\mbr^N$ is obtained by simply considering the associated RFD $(A, A_{\delta_1})$ for some fixed $\delta_1 >0$; as before, which $\delta_1 > 0$ is chosen turns out to be unimportant.

We begin by (loosely) stating the fractal tube formula via distance zeta functions. Let $(\ao)$ be a RFD in $\mbr^N$ such that $\overline{D} = \overline{\dim}_B (\ao) < N$. Then, if $(\ao)$ is $d$-languid (which roughly means that the distance zeta function satisfies some mild polynomial growth conditions), we have the following fractal formula, expressed via $\zao$:\footnote{For the simplicity of exposition, we assume at first that all of the (visible) complex dimensions (i.e., the visible poles of $\zao$) are simple; see \eqref{3.167} for the general case when some (or possibly all) of the complex dimensions are multiple.}
\begin{equation}\label{3.164}
V_{\ao} (\varepsilon) = \sum_{\omega \in \mcd_{\ao}} c_\omega \frac{\varepsilon^{N- \omega}}{N- \omega} + R(\varepsilon),
\end{equation}
where $\mcd_{\ao} = \mcd_{\ao} (U)$ is the set of visible complex dimensions of $(\ao)$ (defined as the poles of $\zao$ belonging to $U$) and $U$ is a suitable domain of $\mbc$ (to which $\zao$ can be meromorphically extended to a $d$-languid function). Furthermore, for each $\omega \in \mcd_{\ao}$,
\begin{equation}\label{3.165}
c_\omega := \res (\zeta_A, \omega);
\end{equation}
so that \eqref{3.164} can be rewritten equivalently in the following form (still in the case of simple complex dimensions):
\begin{equation}\label{3.165.2}
V_{\ao} (\varepsilon) = \sum_{\omega \in \mcd_{\ao}} \res (\zao, \omega) \frac{\varepsilon^{N- \omega}}{N- \omega} + R(\varepsilon).
\end{equation}

Moreover, in \eqref{3.164} and \eqref{3.165.2}, $R(\varepsilon)$ is an error tem (of lower order than the sum over the complex dimensions) which can be estimated explicitly.\footnote{Naturally, in the case of a pointwise (respectively, distributional) tube formula, $R(\varepsilon)$ is a pointwise (respectively, distributional) error term.}

If, in addition, $(\ao)$ (i.e., $\zao$) is strongly $d$-languid (a growth condition which requires that $U := \mbc$ and is stronger than $d$-languidity), then $R(\varepsilon) \equiv 0$ in \eqref{3.164} (and equivalently, in \eqref{3.165.2} or in \eqref{3.167}). Hence, we obtain an {\em exact} fractal tube formula in this case (i.e., a tube formula without error term):
\begin{equation}\label{3.166}
V_{\ao} (\varepsilon) = \sum_{\omega \in \mcd_{\ao}} \res (\zeta_{\ao}, \omega) \frac{\varepsilon^{N-\omega}}{N- \omega}.
\end{equation}

Finally, we note that if the visible complex dimensions are not necessarily all simple, then (for example) the tube formula \eqref{3.165.2} takes the following, slightly more complicated, form (except for this modification, all the other statements and hypotheses are identical, otherwise):\footnote{Here, for clarity, we write $\res (f(s), \omega)$ (instead of $\res (f, \omega)$) to denote the residue of a meromorphic function $f=f(s)$ at $s =\omega$.}
\begin{equation}\label{3.167}
V_{\ao} (\varepsilon) = \sum_{\omega \in \mcd_{\ao}} \res \Big( \frac{\varepsilon^{N-s}}{N-s} \zao (s), \omega \Big) + R(\varepsilon), 
\end{equation}
with $R(\varepsilon) \equiv 0$ in the case of an exact formula.

Entirely analogous fractal tube formulas are also obtained via the tube (instead of the distance) zeta function; that is, for $\tzao$ instead of $\zao$. The main difference is that in the counterparts of the fractal tube formulas \eqref{3.164} and \eqref{3.165.2}--\eqref{3.167}, we must replace $\varepsilon^{N-\omega}/(N-\omega)$ by $\varepsilon^{N-\omega}$, while in the counterpart of \eqref{3.167}, we must replace $\frac{\varepsilon^{\omega -s}}{N-s} \zao (s)$ by $\varepsilon^{N-s} \zao (s)$. The only other difference, truly minor this time, is that the $d$-languidity (respectively, strong $d$-languidity) condition must be replaced by the slightly different languidity (respectively, strong languidity) condition, in the case of a (pointwise or distributional) fractal tube formula with (respectively, without) error term.

For example, the (pointwise or distributional) fractal tube formula, expressed via the tube zeta function $\tzao$, takes the following form (in the case of simple complex dimensions):\footnote{Compare with its counterpart, \eqref{3.165.2}, expressed via the distance zeta function $\zao$.}
\begin{equation}\label{3.168}
V_{\ao} (\varepsilon) = \sum_{\omega \in \mcd_{\ao}} \res (\tzao, \omega) \varepsilon^{N-\omega} + R(\varepsilon).
\end{equation}
In addition, if $(\ao)$ (i.e., $\tzao$) is strongly languid (which requires that $U:=\mbc$), then we can let $R(\varepsilon) \equiv 0$ in \eqref{3.168}; that is, we obtain an exact fractal tube formula in this case.

\begin{remark}\label{Rem:3.20}
($a$) For the precise statements and the full proofs of the above $($as well as of other pointwise and distributional$)$ fractal tube formulas, we refer the interested reader to \cite{LapRaZu6} or \cite[\S\S 5.1--5.3]{LapRaZu1}. 

($b$) Much as was discussed in \S \ref{Sec:2.5} (see, especially, Remark \ref{Rem:2.7}), and as is apparent in \eqref{3.164}, \eqref{3.165.2}, \eqref{3.166} and \eqref{3.168}, each (visible) simple complex dimension $\omega \in \mcd_{\ao}$ gives rise to an {\em oscillatory term}, proportional to $\varepsilon^{N-\omega}$, with exponent the {\em fractal complex co-dimension} $N-\omega$ of $\omega$.\footnote{In the case when $\omega$ is a multiple complex dimension, this oscillatory term is modulated by the multiplication by a suitable polynomial in the variable $x:= \log(\varepsilon^{-1})$ and of degree equal to the multiplicity minus one.}
(Naturally, for each individual complex dimension, these oscillations are multiplicatively periodic. In order to obtain the associated additive or ordinary oscillations, it suffices to let $x := \log (\varepsilon^{-1})$.) Furthermore, as was mentioned before in \S\ref{Sec:2} in the case of fractal strings, as $\omega$ varies in $\mcd_{\ao}$, the amplitudes (respectively, frequencies) of these oscillations are governed by the real (respectively, imaginary) parts of the complex dimensions. Observe that in the case of simple poles (as in \eqref{3.164}, \eqref{3.165.2}--\eqref{3.166} and \eqref{3.168} in \S \ref{Sec:3.5.1}), the (typically) infinite sum over the complex dimensions of the RFD $(\ao)$ can be thought of as (pointwise or distributional) natural generalizations of Fourier series or of almost periodic functions (or distributions). (Compare with \cite[\S VII, I]{Schw} and \cite{Boh}, respectively.)
\end{remark}

We close this subsection by placing the above results in a broader context and providing related references in geometric measure theory and convex geometry. Prior to that, let us recall that the first fractal tube formulas (expressed in terms of complex dimensions and geometric zeta functions) were obtained in [\hyperlinkcite{Lap-vF1}{Lap-vF1--4}] (see, especially, \cite[Ch. 8]{Lap-vF4}), in the case of fractal strings. In the case of fractal sprays (roughly, higher-dimensional analogs of fractal strings), fractal tube formulas were obtained in [\hyperlinkcite{LapPe2}{LapPe2--3}] and, in greater generality, in \cite{LapPeWi1}, but without a natural notion of associated fractal zeta function. Finally, the fractal tube formulas described in the present subsection (i.e., \S \ref{Sec:3.5.1}) were obtained in \cite{LapRaZu6} and \cite[Ch. 5]{LapRaZu1}. They are expressed in terms of complex dimensions and natural fractal zeta functions and include the earlier fractal tube formulas for fractal strings from [\hyperlinkcite{Lap-vF1}{Lap-vF1--4}] and for fractal sprays [\hyperlinkcite{LapPe2}{LapPe2--3}, \hyperlinkcite{LapPeWi1}{LapPeWi1}]; see, in particular, Example \ref{Exe:3.21} and Example \ref{Exe:3.31}.

\begin{remark}({\em Tube formulas and curvatures}: {\em Brief history, references, and beyond.})\label{Rem:3.20.1/2}
It is well known that for a (nonempty) compact convex set $A$ of $\mbr^N$, the tube function $V_A (\varepsilon):= |A_\varepsilon|_N$ is a polynomial of degree (at most) $N$ in $\varepsilon$, the coefficients of which can be interpreted geometrically. This is known as Steiner's formula (\cite{Stein}, 1840). More specifically, 
\begin{equation}\label{3.167.1/2}
V_A (\varepsilon) = \sum_{\alpha =0}^N \gamma_\alpha \mu_\alpha (A) t^{N-\alpha},
\end{equation}
where for each $\alpha \in \{ 0,1, \cdots, N\}$, $\gamma_\alpha$ is the volume of the unit ball in $\mbr^{N-\alpha}$ and the (normalized) coefficients $\mu_\alpha (A)$ have either a geometric, combinatorial or algebraic interpretation. (See, e.g., \cite{Schm} and \cite{KlRot}.) For example, $\mu_N (A) = |A|_N$ is the volume of $A$, $\mu_{N-1} (A)$ is the surface area of $A, \cdots, \mu_1 (A)$ is its mean width, while $\mu_0 (A)$ is its Euler characteristic (equal to $1$ in the present case, but equal to any integer in more general situations).

Furthermore, H. Weyl (\cite{Wey3}, 1939) has obtained an analog of Steiner's formula when $A$ is a smooth compact submanifold of Euclidean space $\mbr^N$, thereby interpreting the coefficients $\mu_\alpha (A)$ as suitable curvatures, now called the  {\em Weyl curvatures} of $A$ (see, e.g., \cite{BergGos} and \cite{Gra}).

Moreover, H. Federer (\cite{Fed1}, 1969) has unified and extended both Steiner's formula and Weyl's tube formula by establishing their counterpart for sets of positive reach.\footnote{A (compact) subset $A$ of $\mbr^N$ is said to be of {\em positive reach} if there exists $\eta \in (0, +\infty]$ such that each point of $A_\eta$, the $\eta$-neighborhood of $A$, has a unique metric projection onto $A$. Clearly, a convex set has infinite reach.}
In the process, he introduced (real or signed) measures $\mu_\alpha$ (for $\alpha \in \{0,1, \cdots, N\}$), now called {\em Federer's curvature measures}, whose total mass $\mu_\alpha (A)$ coincide with the normalized coefficients of \eqref{3.167.1/2}. He also obtained a localized version of the tube formula \eqref{3.167.1/2}, expressed in terms of the values of the measures $\mu_\alpha$, for $\alpha \in \{0,1, \cdots, N\}$, at suitable Borel subsets of $\mbr^N$ (or of $\mbr^N \times S^{N-1}$).

Much later, `fractal curvatures' were introduced by S. Winter in \cite{Wi}, M. Z\" ahle [\hyperlinkcite{Za3}{Z\" a3--4}], and S. Winter and M. Z\" ahle \cite{WiZa}, for certain classes of self-similar deterministic or random fractals,  but still for integer values of the index. The author has long conjectured that there should exist suitable notions of complex measures (\cite{Coh}, \cite{Fo}, \cite{Ru1}) or distributions (\cite{Schw}, \cite{GelSh}), called `{\em fractal curvature measures}', associated with each (visible) complex dimension $\omega \in \mcd_A$ and enabling us to interpret geometrically the (normalized) coefficients of the fractal tube formulas (see, for example, \eqref{3.164} and \eqref{3.168}). Also, a {\em local} fractal tube formula, yet to be precisely formulated and to be rigorously established, should extend Federer's local tube formula, and be expressed (in the case of simple poles) in terms of the residues of a suitably defined {\em local fractal zeta function} evaluated at each of the visible complex dimensions. (See \cite[Pb. 6.2.3.8 and App. B]{LapRaZu1}.)\footnote{See also the earlier work in [\hyperlinkcite{LapPe2}{LapPe2--3}] and \cite{LapPeWi1} for the very special case of fractal sprays, as well as \cite{LapPe1} which dealt with the example of the Koch snowflake curve (but without the use of any zeta functions).}

We close this discussion of tube formulas by mentioning several relevant references (beside \cite{Stein}, \cite{Mink}, \cite{Wey3} and \cite{Fed1}), including the books \cite{Fed2}, \cite{KlRot}, \cite{BergGos}, \cite{Gra}, and \cite{Schn}, along with the papers \cite{HugLasWeil}, \cite{KeKom}, \cite{Kom}, [\hyperlinkcite{LapPe1}{LapPe1--3}], \cite{LapPeWi1}, [\hyperlinkcite{LapLu-vF1}{LapLu-vF1--2}], \cite{LapRaZu6}, [\hyperlinkcite{Ol1}{Ol1--2}], \cite{Sta}, \cite{Wi}, \cite{WiZa}, [\hyperlinkcite{Za1}{Za1--4}], as well as \cite[\S 13.1]{Lap-vF4} and \cite[Ch. 5]{LapRaZu1}, and the many relevant references therein.
\end{remark}
 
\subsubsection{Minkowski measurability criteria for RFDs in $\mbr^N$}\label{Sec:3.5.2}

In this subsection, we briefly discuss, in particular, a few of the Minkowski measurability criteria obtained in [\hyperlinkcite{LapRaZu5}{LapRaZu5,7}] or in \cite[\S 5.4]{LapRaZu1}, to which we refer for further information and for closely related necessary or sufficient conditions for Minkowski measurablity. The criteria will be stated for RFDs in $\mbr^N$, but as usual, can also be applied to bounded subsets (which, as we know by now, are special cases of RFDs).

The main criterion  (see \cite[Thm. 5.4.20]{LapRaZu1}) is expressed in terms of the distance zeta function:\\

{\em Let $(\ao)$ be a RFD in $\mbr^N$. Assume that $D:= \dim_B A$ exists and $D < N$. Then, under suitable hypotheses,}\footnote{Namely, $(\ao)$ (i.e., $\zao$) is assumed to be $d$-languid (as in \S \ref{Sec:3.5.1}) for a screen $S$  passing between the critical line $\{ Re(s) =D\}$ and all of the complex dimensions of $(\ao)$ with real part $<D$. (Roughly speaking, a screen $S$ is a suitable curve bounding the region $U$, with $\overline{U} \supseteq \{ Re(s) =D \}$, to which $\zao$ is meromorphically continued, and extending to infinity in the vertical direction. Also, $S$ is required not to contain any pole of $\zao$; see \cite[\S 5.3]{Lap-vF4} or \cite[\S 5.1.1]{LapRaZu1}.)}
{\em the following statements are equivalent}:\\

($i$) {\em The RFD $(\ao)$ is Minkowski measurable.}\\

($ii$) {\em The Minkowski dimension $D$ is the only principal complex dimension of $(\ao)$ $($i.e., the only pole of $\zao$ with real part equal to $D)$ and it is simple.}\footnote{Equivalently, the RFD $(\ao)$ does {\em not} have any {\em nonreal} complex dimension, and $D$ is simple.}\\

The above Minkowski measurability criterion is the higher-dimensional counterpart of Theorem \ref{Thm:2.2} in \S \ref{Sec:2.1}, the Minkowski measurability criterion for fractal strings obtained in \cite[\S 8.3]{Lap-vF4}. Its proof involves several key ingredients including, especially, the Wiener--Pitt Tauberian theorem \cite{PitWie} (stated, e.g., in \cite{Pit}, \cite{Kor} and \cite[Thm. 5.4.1]{LapRaZu1}), a version of the fractal tube formula (with error term) via $\zao$ discussed in \S \ref{Sec:3.5.1} just above, as well as a uniqueness theorem for almost periodic functions (or rather, distributions \cite[\S VI.9.6]{Schw}).

An analog for the tube (rather than the distance) zeta function $\tzao$ of the above Minkowski measurability criterion for RFDs is also obtained in \cite[Thm. 5.4.25]{LapRaZu1}.

Moreover, a necessary (respectively, sufficient) condition for the Minkowski measurability of RFDs is obtained in \cite[Thm. 5.4.15]{LapRaZu1} (respectively, \cite[Thm. 5.4.2]{LapRaZu1}). 

In addition, the case of RFDs for which the underlying scaling law is no longer a power law but is governed instead by a nontrivial gauge function $h = h(t)$ (with $t \in (0,1)$, say)\footnote{The standard case when the underlying scaling law is a power law corresponds to the trivial gauge function $h(t) \equiv 1$.} 
is examined in [\hyperlinkcite{LapRaZu6}{LapRaZu6--7}]; see also \cite{LapRaZu4}) and \cite[\S 5.4.4]{LapRaZu1} where both an $h$-Minkowski measurability criterion and an optimal $h$-fractal tube expansion are obtained, especially for gauge functions of the form $h(t) := (\log t^{-1})^{m-1}$ for some integer $m \geq 2$ (corresponding, e.g., under appropriate hypotheses, to $D = \dim_B (\ao)$ being a multiple pole of order $m$.

In the definition of the (upper, lower)  Minkowski content, relative to a given gauge function $h$, where for some $\varepsilon_0 > 0$, $h: (0, \varepsilon_0) \rightarrow (0, + \infty)$ is a function of slow growth satisfying suitable conditions near $0$ (including the fact that $h(t) \rightarrow + \infty$ as $t \rightarrow 0^+$), one simply replaces $\varepsilon^{N-D}$ by $\varepsilon^{N-D} h (\varepsilon)$. For example, assuming that $D = \dim_B (\ao)$ exists, the {\em upper $h$-Minkowski content} of $(\ao)$ is given by
\begin{equation}\label{3.191.1/2}
\mcm^* ((A, \Omega), h) := \varlimsup_{\varepsilon \rightarrow 0^+} \frac{V_{A, \Omega} (\varepsilon)}{\varepsilon^{N-D} h(\varepsilon)},
\end{equation}
and similarly for the {\em lower $h$-Minkowski content}, $\mcm_* ((A, \Omega), h)$, and the $h$-{\em Minkowski content}, $\mcm (A, \Omega, h)$. (See \cite{HeLap} and \cite[Eq. (4.5.10), p. 352, and \S 6.1.1.2, pp. 544--545]{LapRaZu1}.) Examples of allowable gauge functions considered  in \cite{HeLap}, \cite{LapRaZu1}, \cite{LapRaZu7} and \cite{LapRaZu10} include $\log^k (t^{-1})$, for all $t \in (0,1)$, where $k \in \mbn$ is arbitrary and $\log^k$ denotes the $k$-th iterated logarithm. The reciprocals $1/h (t)$ of allowable gauge functions are also considered in the above definitions and references.

Although we do not wish to go into the technical details here, we mention that in \cite{LapRaZu10}, connections between generalized Minkowski contents with logarithm-type gauge functions, Minkowski measurability criteria, and appropriate Riemann surfaces (see, e.g., \cite{Ebe, Schl}), are explored. These types of generalized Minkowski contents arise naturally in certain geometric situations and in the study of certain dynamical systems.

\begin{remark}({\em Extension to Ahlfors metric spaces.})\label{Rem:3.22.1/2}
It is noteworthy that essentially the entire (higher-dimensional) theory of complex dimensions, fractal zeta functions and fractal tube formulas (including the Minkowski measurability criteria) from \cite{LapRaZu1} and the accompanying series of papers can be extended without change to a large class of metric measure spaces (see, e.g.,  \cite{DaMcCS}) called Ahlfors spaces, doubling spaces or else, spaces of homogeneous type, and of frequent use in harmonic analysis, nonsmooth analysis, fractal geometry and the theory of dynamical systems.\footnote{A metric space $(X, {\bf d})$ of finite diameter and equipped with a positive Borel measure $\mu$ (i.e., a {\em metric measure space}) is said to be an {\em Ahlfors space of Ahlfors dimension} $\alpha$ if $\mu (B_r (x))$ is comparable to $r^\alpha$, where $B_r (x)$ is the closed ball (with respect to the metric ${\bf d}$) of center $x$ and radius $r$, and with $x \in X$ arbitrary. (The implicit constants must, of course, be independent of $x \in X$ and of $r > 0$ sufficiently small.) In this case, $\alpha$ coincides with the Hausdorff and Minkowski dimensions of $X$. \label{Fn:119}}
This extension is carried out in \cite{Wat} and \cite{LapWat} and the corresponding theory is illustrated by several examples of computation of complex dimensions and concrete fractal tube formulas, both in the setting of Ahlfors spaces and a little beyond.
\end{remark}

\subsubsection{Examples}\label{Sec:3.5.3}

In the present subsection, we illustrate by means of a variety of examples some of the results about fractal tube formulas and Minkowski measurability criteria discussed in \S \ref{Sec:3.5.1} and \S \ref{Sec:3.5.2}, respectively. In the process, in order to avoid unnecessary repetitions, we often refer to the corresponding examples in \S \ref{Sec:3.4}.

\begin{example}({\em Fractal strings}).\label{Exe:3.21}
We briefly explain here how to recover the fractal tube formulas for fractal strings from \cite{Lap-vF4}) discussed in \S \ref{Sec:2.1} above.

Let $\mcl = (\ell_j)_{j \geq 1}$ be a (bounded) fractal string and let $\Omega \subseteq \mbr$ be any geometric realization of $\mcl$ as an open set with finite length (i.e., $|\Omega|_1 < \infty$). Furthermore, let $(\partial \Omega, \Omega)$ be the associated RFD in $\mbr$. Then, as we have seen in \S \ref{Sec:3.2.2}, the {\em distance zeta function} $\zeta_{\partial \Omega, \Omega}$ of the RFD $(\partial \Omega, \Omega)$ and the {\em geometric zeta function} $\zeta_\mcl$ of the fractal string $\mcl$ are connected via the following functional equation:
\begin{equation}\label{3.169}
\zeta_{\partial \Omega, \Omega} (s) = \frac{2^{1-s}}{s} \zeta_\mcl (s),
\end{equation}
for all $s \in U$, where $U$ is any domain of $\mbc$ to which $\zeta_\mcl$ (or, equivalently, $\zeta_{\partial \Omega, \Omega}$) can be meromorphically continued. As a result, provided $0 \in U$ (and assuming for simplicity that $\zeta_\mcl (0) \neq 0$), then 
\begin{equation}\label{3.170}
\mcd_{\partial \Omega, \Omega} = \mcd_\mcl \cup \{ 0 \},
\end{equation}
where the union is taken between multisets, as usual. More specifically, if $0$ is a pole of $\zeta_\mcl$ of multiplicity $m \geq 0$ (the case when $m=0$ corresponding to $0$ not being a pole of $\zeta_\mcl$), then it is a pole of $\zeta_{\partial \Omega, \Omega}$ of multiplicity $m+1$.

The identity \eqref{3.170} explains why the expressions for the fractal tube formulas in the case of simple poles to be discussed a little further on looked somewhat awkward in \cite[\S 8.1]{Lap-vF4} but are now significantly simplified, both conceptually and concretely. In particular, as is now clear, even in the present case when $N=1$, the distance zeta function $\zeta_{\ptoo}$ is the proper theoretical tool to define and understand the complex dimensions of fractal strings as well as to formulate the associated fractal tube formulas.\footnote{This is so even though the geometric zeta function has been (and will continue to be) a very useful tool as well. We note that an entirely analogous comment can be made about (not necessarily self-similar) fractal sprays; recall from \S \ref{Sec:3.4.10} that in the latter case, $\mcl$ may be unbounded and $\zeta_\mcl$ is then called the {\em scaling zeta function} of the fractal spray (see \cite{LapPe2, LapPeWi1} and \cite[\S 13.1]{Lap-vF4}).}

In light of \eqref{3.169}--\eqref{3.170} and since $N=1$, the fractal tube formula \eqref{3.167} yields its counterpart for fractal strings (for every $\delta_1 \geq \ell_1/2$):\footnote{Note that in light of \eqref{2.8} and \eqref{3.163}, we have that $V_\mcl (\varepsilon) = V_{\ptoo} (\varepsilon)$. Also, we let $\mcd_\mcl = \mcd_{\mcl} (U)$ and $\mcd_{\ptoo} = \mcd_{\ptoo} (U)$. }
\begin{align}\label{3.171}
V_\mcl (\varepsilon) &= V_{\ptoo} (\varepsilon)\\
&= \sum_{\omega \in \mcd_{\ptoo}} \res \Big( \frac{\varepsilon^{1-s}}{1-s} \zeta_{\ptoo} (s), \omega \Big) + R(\varepsilon)\notag \\
&= \sum_{\omega \in \mcd_\mcl \cup \{ 0 \}} \res \Big( \frac{(2 \varepsilon)^{1-s}}{s(1-s)} \zeta_\mcl (s), \omega \Big) + R (\varepsilon) \notag \\
&= \sum_{\omega \in \mcd_\mcl} \res \Big( \frac{(2 \varepsilon)^{1-s}}{s (1-s)} \zeta_\mcl (s), \omega \Big) + \{ 2 \varepsilon \zeta_\mcl (0) \}_{0 \in U \backslash \mcd_\mcl} + R(\varepsilon), \notag
\end{align}
when $\mcl$ (i.e., $\zeta_\mcl$) is languid, and with $R (\varepsilon) \equiv 0$ when $\mcl$ (i.e., $\zeta_\mcl$) is strongly languid. Here, by definition, the term $\{ 2 \varepsilon \zeta_\mcl (0) \}_{0 \in U \backslash \mcd_\mcl}$ between braces in the last equality of \eqref{3.171} is included only if $0 \in U \backslash \mcd_\mcl$.

We note that the expression obtained in \eqref{3.171} is in complete agreement with the (pointwise or distributional) fractal tube formulas obtained for fractal strings in \cite[Thm. 8.1 or Thm. 8.7]{Lap-vF4}).\footnote{In \S \ref{Sec:2.1}, for the simplicity of exposition, we gave a less precise statement of the formula; compare with formula \eqref{2.10.1/2}.}

Since, in view of \eqref{3.169}, 
\begin{equation}\label{3.172}
\res (\zeta_{\ptoo}, \omega) = \frac{2^{1-\omega}}{\omega} \res (\zeta_\mcl, \omega)
\end{equation}
for every simple complex dimension $\omega \in U \backslash \{ 0 \}$, we deduce from \eqref{3.171} the following (pointwise or distributional) tube formula in the special case  when all of the visible complex dimensions are simple (and $\zeta_\mcl$ is languid):
\begin{align}\label{3.173}
V_\mcl (\varepsilon) &= V_{\ptoo} (\varepsilon) \notag \\
&= \sum_{\omega \in \mcd_\mcl (W) \backslash \{ 0\}} \res (\zeta_\mcl, \omega) \frac{(2\varepsilon)^{1-\omega}}{\omega(1-\omega)} \notag \\
&+ \{ 2 \varepsilon (1 - \log (2 \varepsilon)) \res (\zeta_\mcl, 0) + 2 \varepsilon \zeta_\mcl (0) \}_{0 \in U} + R(\varepsilon),
\end{align}
where, by definition, the term between braces in the last equality of \eqref{3.173} is included only if $0 \in U$. Also, as before, if $\mcl$ (i.e., $\zeta_\mcl$) is strongly languid (which implies that $U = \mbc$), then we obtain an exact tube formula; that is, we can let $R(\varepsilon) \equiv 0$ in \eqref{3.173}.

We point out that \eqref{3.173} is in agreement with the results stated in \cite[Cor. 8.3 or Cor. 8.10]{Lap-vF4} in the distributional or pointwise case, respectively.\footnote{Again, we note that in \S \ref{Sec:2.1}, for the simplicity of exposition, the corresponding formula stated in \eqref{2.10} was not as precise as in \eqref{3.173}.}
\end{example}

\begin{example}({\em The Sierpinski gasket}).\label{Exe:3.22}
We briefly revisit the first example from \S \ref{Sec:3.4}, studied in \S \ref{Sec:3.4.1}, in which $A \subseteq \mbr^2$ is the classic Sierpinski gasket and, by \eqref{3.41}, all of the complex dimensions $s =0$ and $s_k := \log_2 3 + i (2 \pi/ \log 2) k$ (for any $k \in \mbz$) are simple with associated residues given by \eqref{3.42} and \eqref{3.43}, respectively. Then, in light of \eqref{3.166} and since $N=2$ here, we obtain the following pointwise, exact fractal tube formula (the hypotheses of which are shown to be satisfied in \cite[Expl. 5.5.12]{LapRaZu1}), valid for all $\varepsilon \in (0, 1/2 \sqrt{3})$:
\begin{align}\label{3.174}
V_A (\varepsilon) &= \sum_{k \in \mbz} \res (\zeta_A, s_k) \frac{\varepsilon^{2- s_k}}{2-s_k} + \res (\zeta_A, 0) \frac{\varepsilon^{2-0}}{2} \\
&= \sum_{k \in \mbz} \frac{6(\sqrt{3})^{1-s_k}}{4^{s_k} (\log 2) s_k (s_k -1)} \frac{\varepsilon^{2-s_k}}{2-s_k} + (3 \sqrt{3} + 2 \pi) \frac{\varepsilon^2}{2}. \notag
\end{align}

Since $s_k = D + ik{\bf p}$, for all $k \in \mbz$, where $D = \log_2 3$ and ${\bf p} = (2 \pi / \log 2)$ are, respectively, the Minkowski dimension and the oscillatory period of $A$, we can clearly rewrite \eqref{3.174} in the following form:
\begin{equation}\label{3.175}
V_A (\varepsilon) = \varepsilon^{2-D} G (\log_2 \varepsilon^{-1}) + \Big( \frac{3 \sqrt{3} + 2 \pi}{2} \Big) \varepsilon^2,
\end{equation}
where $G$ is a continuous, nonconstant $1$-periodic function, which is bounded away from zero and infinity, and is given by the following absolutely convergent (and hence, pointwise convergent) Fourier series expansion, for all $x \in \mbr$:
\begin{equation}\label{3.176}
G(x) := \frac{6 \sqrt{3}}{\log 2} \sum_{k \in \mbz} \frac{(4 \sqrt{3})^{-s_k}}{(2 -s_k) (s_k -1) s_k} e^{i2\pi kx}.
\end{equation}
\end{example}

It is apparent from \eqref{3.176} that  since $G$ is nonconstant, the function $\varepsilon^{-(2-D)} V_A (\varepsilon)$ is oscillatory, and hence does not have a limit as $\varepsilon \rightarrow 0^+$. Therefore, the Sierpinski gasket is {\em not} Minkowski measurable, in agreement with the Minkowski measurability criterion  stated in \S \ref{Sec:3.5.2}, the hypotheses of which are easily verified since $D = \log_2 3$ is simple and $A$ has infinitely many nonreal principal complex dimensions (here, $s_k = D + ik{\bf p}$, with $k \in \mbz \backslash \{ 0\}$). 

\begin{example}({\em The Sierpinski carpet}).\label{Exe:3.23}
Let $A \subseteq \mbr^2$ be the classic Sierpinski carpet studied in \S \ref{Sec:3.4.2}. Then, in light of \eqref{3.44}--\eqref{3.43.2}, \eqref{3.166} yields the following pointwise, exact fractal tube formula:
\begin{equation}\label{3.177}
V_A (\varepsilon) = \sum_{k \in \mbz} \frac{2^{-s_k}}{(\log 3) (2 -s_k) (s_k -1) s_k} \varepsilon^{2-s_k} + \frac{16}{5} \varepsilon^2 + \frac{1}{2} \Big(2 \pi + \frac{8}{7} \Big) \varepsilon^2,
\end{equation}
where $s_k := D + i k {\bf p}$, for each $k \in \mbz$; here, $D := \log_3 8$ and ${\bf p} := (2 \pi/ \log 3)$ are, respectively, the Minkowski dimension and the oscillatory period of $A$.
\end{example}

\begin{exercise}\label{Exc:3.24}
($i$) Much as in \eqref{3.175} and \eqref{3.176} above, rewrite the leading term (i.e., the sum over all $k \in \mbz$) in \eqref{3.177} in the form $\varepsilon^{2-D} G (\log_3 \varepsilon^{-1}),$ where $G$ is a continuous, nonconstant $1$-periodic function, which is bounded away from zero and infinity.

($ii$) Deduce from part ($i$) via a direct computation that the Sierpinski carpet is {\em not} Minkowski measurable, also in agreement with the Minkowski measurability criterion stated in \S \ref{Sec:3.5.2} (and of which you should verify the hypotheses).

($iii$) Finally, by means of a direct computation (based, e.g., on part ($i$)), show that the Sierpinski carpet is Minkowski nondegenerate and calculate its average Minkowski content $\widetilde{\mcm}$. Furthermore, verify the latter results by using \eqref{3.34.1/2}, connecting $\widetilde{\mcm}$ and the residue of $\zeta_A (s)$ at $s = D$ in the non-Minkowski measurable case (and of which the hypotheses are satisfied, in light of part ($i$)).
\end{exercise}

\begin{example}({\em The 3-$d$ carpet}).\label{Exe:3.25}
Let $A \subseteq \mbr^2$ be the 3-$d$ carpet studied in \S \ref{Sec:3.4.3}. Then, in light of \eqref{3.46}--\eqref{3.50}, \eqref{3.166} yields the following pointwise, exact fractal tube formula (valid for all $\varepsilon \in (0, 1/2)$):
\begin{equation}\label{3.178}
V_A (\varepsilon) = \frac{24}{13 \log 3} \varepsilon^{3-D} G (\log_3 \varepsilon^{-1}) + \Big(6 - \frac{6}{17} \Big) \varepsilon + \Big(3 \pi + \frac{12}{23}\Big) \varepsilon^2 + \Big( \frac{4 \pi}{3} - \frac{8}{25} \Big) \varepsilon^3,
\end{equation}
where $D = \log_3 26$ is the Minkowski dimension of $A$ and $G$ is a continuous, nonconstant $1$-periodic function which is bounded away from zero and infinity, and is given by the following pointwise (absolutely convergent and hence) convergent Fourier series expansion:
\begin{equation}\label{3.179}
G(x) = \frac{24}{3 \log 3} \sum_{k \in \mbz} \frac{2^{-s_k}}{(3- s_k)(s_k -1)(s_k-2)s_k} e^{i2\pi kx}, \text{ for all } x \in \mbr,
\end{equation}
with $(s_k := D+ik{\bf p})_{k \in \mbz}$ denoting the sequence of principal complex dimensions of $A$ and ${\bf p} := 2 \pi/\log 3$ denoting the oscillatory period of $A$.
\end{example}

\begin{exercise}\label{Exc:3.26}
($a$) (3-$d$ {\em carpet, revisited}). For the Sierpinski 3-$d$ carpet $A$ in Example \ref{Exe:3.25}, answer the analog of questions ($ii$) and ($iii$) of Exercise \ref{Exc:3.24}.

($b$) ({\em 3-gasket RFD}). For the relative (or inhomogeneous) Sierpinski 3-gasket $(A_3, \Omega_3)$ studied in \S  \ref{Sec:3.4.4} (specialized to $N=3)$, use \eqref{3.166}, along with the $N=3$ case of \eqref{3.52} and \eqref{3.59} (with $g_3 = g_3 (s)$ given by formula \eqref{3.23.1/2} of Exercise \ref{Exc:3.8.1/2}), in order to obtain a pointwise, exact fractal tube formula for $(A_3, \Omega_3)$. (Recall from the discussion following \eqref{3.62} that $D = \log_2 4 =2$ is a complex dimension of $(A_3, \Omega_3)$ of multiplicity {\em two}, whereas the other complex dimensions are simple.)

Show via a direct computation that the RFD $(A_3, \Omega_3)$ is Minkowski degenerate and therefore not Minkowski measurable but that with respect to the gauge function $h(t) := \log (t^{-1})$, for all $t \in (0,1), (A_3, \Omega_3)$ is $h$-Minkowski measurable (and hence also $h$-Minkowski nondegenerate), as was stated towards the end of \S \ref{Sec:3.4.4}. In order to establish the latter facts, you may use the results from \cite{LapRaZu6} and \cite[\S 5.4.4]{LapRaZu1} briefly discussed towards the end of \S \ref{Sec:3.5.2}.
\end{exercise}

\begin{example}({\em The $\frac{1}{2}$-square and $\frac{1}{3}$-square fractals}).\label{Exe:3.27}
We revisit and complete here part ($a$) (the $\frac{1}{2}$-square fractal) and part ($b$) (the $\frac{1}{3}$-square fractal) of \S \ref{Sec:3.4.5}.\\

{\bf ({\em a})} ({\em The $\frac{1}{2}$-square fractal}). Recall from part ($a$) of \S \ref{Sec:3.4.5} that all of the fractal complex dimensions of the $\frac{1}{2}$-square fractal $A$ are simple, except for $s=1$, which is equal to $\dim_B A$. As a result, it follows from \eqref{3.166} in light of \eqref{3.65}, \eqref{3.69} and \eqref{3.102.1/2} that the following pointwise, exact fractal tube formula holds (for all $\varepsilon \in (0, 1/2)$):
\begin{equation}\label{3.180}
V_A (\varepsilon) = \frac{1}{4 \log 2} \varepsilon \log \varepsilon^{-1} + \varepsilon G(\log_2 (4 \varepsilon)^{-1}) + \frac{1 + 2 \pi}{2} \varepsilon^2,
\end{equation}
where $G$ is a nonconstant, continuous $1$-periodic function which is bounded away from zero and infinity and is given by the following convergent (because absolutely convergent) Fourier expansion:
\begin{equation}\label{3.181}
G(x) := \frac{29 \log 2-4}{8 \log 2} + \frac{1}{4} \sum_{k \in \mbz \backslash \{ 0\}} \frac{e^{2 \pi ikx}}{(2 -s_k)(s_k -1) s_k}, \text{ for all } x \in \mbr, 
\end{equation} 
with $s_k := ik{\bf p}$, for all $k \in \mbz \backslash \{ 0\}$, and ${\bf p} := (2 \pi/\log 2)$, the oscillatory period of $A$.

Let us briefly explain how to obtain \eqref{3.180} and \eqref{3.181}. In light of \eqref{3.167} (applied with $R(\varepsilon) \equiv 0$ because we are in the strongly $d$-languid case),
\begin{align}\label{3.182}
V_A (\varepsilon) &= \sum_{\omega \in \mcd (\zeta_A)} \res \Big( \frac{\varepsilon^{2-s}}{2-s} \zeta_A (s), \omega \Big) \\
&= \res \Big( \frac{\varepsilon^{2-s}}{2-s} \zeta_A (s), 1 \Big) + \sum_{\omega \in \mcd (\zeta_A) \backslash \{ 1\}} \res (\zeta_A (s), \omega) \frac{\varepsilon^{2 - \omega}}{2 -\omega}. \notag
\end{align}
In order to calculate the above residue at $s=1$ in the last equality of \eqref{3.182}, one computes the Laurent series expansion of $\zeta_A$ around $s=1$ (which is a double pole of $\zeta_A$), as follows:
\begin{equation}\label{3.183}
\zeta_A (s) = \frac{d_{-2}}{(s-1)^2} + \frac{d_{-1}}{s-1} + O(1) \quad \text{ as  } s \rightarrow 1,
\end{equation}
with
\begin{equation}\label{3.184}
d_{-2} := \frac{1}{4 \log 2} \quad \text{and} \quad d_{-1} := \frac{29 \log 2-2}{8 \log 2}. 
\end{equation}
One then deduces from combining \eqref{3.183} and \eqref{3.184} that
\begin{align}\label{3.185}
\res \Big( \frac{\varepsilon^{2-s}}{2-s} \zeta_A (s), 1 \Big) &= \varepsilon (d_{-1} -d_{-2} + d_{-2} \log \varepsilon^{-1})\\
&= \frac{1}{4 \log 2} \varepsilon \log \varepsilon^{-1} + \frac{29 \log 2 - 4}{8 \log 2}. \notag
\end{align}
Finally, in light of \eqref{3.185} and the expression of $\res (\zeta_A (s), s_k)$ for any $k \in \mbz \backslash \{ 0 \} $ given in \eqref{3.102.1/2}, as well as of the value of $\res (\zeta_A (s), 0)$ obtained in the just mentioned equation, we deduce the exact tube formula \eqref{3.180} from \eqref{3.182}.

Next, recalling from \eqref{3.67} that $\dim_B A =1$, we can easily deduce from the fractal tube formula \eqref{3.180} that $A$ is Minkowski degenerate with Minkowski content $\mcm = + \infty$. In particular, $A$ is not Minkowski measurable, in the usual sense. Moreover, it follows from the $h$-Minkowski measurability criterion discussed at the end of \S \ref{Sec:3.5.2}, that for the choice of the gauge function $h(t) := \log t^{-1}$ (for all $t \in (0,1)$), $A$ is $h$-Minkowski measurable with $h$-Minkowski content $\mcm (A, h)$ given by
\begin{equation}\label{3.186}
\mcm (A, h) = \frac{1}{4 \log 2}.
\end{equation}
This concludes the discussion of the $\frac{1}{2}$-square fractal, for now.\\

{\bf ({\em b})} ({\em The $\frac{1}{3}$-square fractal}). Let us next briefly consider the $\frac{1}{3}$-square fractal $A \subseteq \mbr^2$ studied in part ($b$) of \S \ref{Sec:3.4.5}. Recall from that discussion (see, especially, \eqref{3.74}, \eqref{3.76} and \eqref{3.110.1/2}) that $\dim_B A =1, \dim_{PC} A = \{1\}$, and all of the complex dimensions of $A$ are simple with 
\begin{equation}\label{3.187}
F \cup \{0,1\} \subseteq \mcd_A \subseteq \{ 0,1 \} \cup \{ s_k := \log_3 2 + i{\bf p}k :k \in \mbz \},
\end{equation}
where $F$ is a nonempty finite subset of $\log_3 2 + i{\bf p} \mbz$ containing $\log_3 2$ as well as several nonreal complex dimensions (and conjectured to be infinite). Here and henceforth, ${\bf p} := 2 \pi/ \log 3$.

Then, in light of the exact fractal tube formula (in the case of simple poles) stated in \eqref{3.166}, combined with \eqref{3.187} and \eqref{3.76}--\eqref{3.110.1/2}, we obtain the following pointwise, exact fractal tube formula for $A$ (valid for all $\varepsilon \in (0, 1/\sqrt{2})$):
\begin{equation}\label{3.188}
V_A (\varepsilon) = 16\varepsilon + \varepsilon^{2-\log_3 2} G (\log_3 (3 \varepsilon)^{-1}) + \frac{12 + \pi}{2} \varepsilon^2,
\end{equation}
where $G$ is a nonconstant, continuous $1$-periodic function which is bounded away from zero infinity and is given by the following absolutely convergent (and hence, pointwise convergent) Fourier series, for all $x \in \mbr$:
\begin{equation}\label{3.189}
G(x) := \frac{1}{\log 3} \sum_{k \in \mbz} \frac{e^{2 \pi ikx}}{(2-s_k) s_k} \Big( \frac{6}{s_k -1} + \Psi (s_k) \Big),
\end{equation}
where $\Psi = \Psi (s)$ is the entire function occurring in \eqref{3.72} and \eqref{3.78}.

Finally, it follows from \eqref{3.188} that $A$ is Minkowski measurable with Minkowski content (in the usual sense) given by
$\mcm =16$.
This concludes for now our discussion of the $\frac{1}{3}$-square fractal.
\end{example}

\begin{exercise}\label{Exc:3.28}
($a$) ({\em Cantor grill}). Use the results of \S \ref{Sec:3.4.7}  to calculate the residues of $\zeta_A$ at the complex dimensions, and then to obtain a (pointwise, exact) fractal tube formula for the Cantor grill $A = C \times [0,1]$, where $C$ is the ternary Cantor set.

($b$) ({\em Cantor dust}). Answer an analogous question for the Cantor dust $A = C \times C$ studied in \S \ref{Sec:3.4.8}. Then, extend your result to $A = C^d$, where $C^d$ is the Cartesian product of $d$ copies of the Cantor set $C$, with $d \geq 2$.

[{\em Caution}: Question ($b$) is more difficult than question ($a$).]
\end{exercise}

\begin{example}({\em The Cantor graph RFD}).\label{Exe:3.29}
Let $(\ao)$ denote the Cantor graph RFD (in $\mbr^2)$ described and studied in \S \ref{Sec:3.4.9}. Recall that the compact set $A \subseteq \mbr^2$ is the graph of the Cantor function (also called the devil's staircase). Then, in light of \eqref{3.115}--\eqref{3.116}, $D_{\ao} = D_A =1$ and 
\begin{equation}\label{3.191}
\mcd_{\ao} = \{ 0,1 \} \cup \{ s_k := \log_3 2 + ik{\bf p} : k \in \mbz \},
\end{equation}
with ${\bf p} := 2 \pi/ \log 3$ and each complex dimension being simple. We therefore deduce from \eqref{3.166} the following pointwise, exact fractal tube formula, valid for every $\varepsilon \in (0,1)$ (where we use the values of the residues of $\zeta_A$ given in \eqref{3.117}--\eqref{3.118}):
\begin{align}\label{3.192}
V_{\ao} (\varepsilon) &= 2 \varepsilon + \varepsilon^{2 - \log_3 2} G (\log_3 \varepsilon^{-1}) + \varepsilon^2 \notag \\
&= 2 \varepsilon^{2 - D_{\ao}} + \varepsilon^{2-D_{CS}} G (\log_3 \varepsilon^{-1}),
\end{align}
where (as above) $D_{\ao} = D_A = 1$, the dimension of the Cantor graph, and $D_{CS} = D_C = \log_3 2$, the dimension of the Cantor set (or of the Cantor string). Furthermore, in \eqref{3.192}, $G$ is a nonconstant, continuous $1$-periodic function which is bounded away from zero and infinity, and is given by the following absolutely convergent (and hence, pointwise convergent) Fourier series expansion, for all $x \in \mbr$:
\begin{equation}\label{3.193}
G (x) := \frac{1}{\log 3} \sum_{k \in \mbz} \frac{e^{2 \pi i kx}}{(2 -s_k) (s_k -1) s_k}.
\end{equation}

Finally, it easily follows from \eqref{3.192}--\eqref{3.193} that $(\ao)$ is Minkowski measurable, with Minkowski content given by
\begin{equation}\label{3.194}
\mcm_{\ao} = \frac{\res (\zao,1)}{2-1} = 2,
\end{equation}
where we have used \eqref{3.118} in the second equality. Exactly the same property and identity as in \eqref{3.194} holds for the Cantor graph $A$ instead of for $(\ao)$, as the interested reader can verify; in particular, $\mcm_A = 2$. This concludes our discussion of the Cantor graph RFD for now. We will return  to this example (and to the associated Cantor graph) in \S \ref{Sec:3.6}, when discussing the notion of fractality; see the text surrounding \eqref{3.212}--\eqref{3.218}.
\end{example}

\begin{exercise}\label{Exc:3.30}
Directly calculate the length of the Cantor graph (i.e., of the devil's staircase) $A$ and compare your result with the value of the Minkowski content of $A$ given above. Furthermore, give a heuristic, geometric argument that enables you to guess the length of $A$ without any computation.
\end{exercise}

\begin{example}({\em Self-similar sprays}).\label{Exe:3.31}
We briefly discuss the important class of self-similar sprays, studied in \S \ref{Sec:3.4.10} above. Since this discussion could be quite lengthy, otherwise, we refer to \cite[\S 5.5.6]{LapRaZu1} for the details.

Let $(\ao)$ be a self-similar spray with scaling ratios $r_1, \cdots, r_J$ (with $J \geq 2$) and (for simplicity, but without loss of generality) with a simple generator $G$ (or rather, generating or base RFD $(\partial G, G)$), as in \S \ref{Sec:3.4.10}. Also as in \S \ref{Sec:3.4.10}, we assume that $G$ is a (nonempty) bounded open subset of $\mbr^N$, with $D_{\partial G,G} = \dim_B (\partial G,G) < N$, and that $\sum_{j=1}^J r_j^N < \infty$, so that the fractal spray has finite total volume.

Recall that $\zao$ is then given by the key factorization formula \eqref{3.124} or \eqref{3.125}, expressing $\zao$ in terms of the distance zeta function $\zeta_{\partial G,G}$ of the generating RFD and of the scaling zeta function $\zeta_\mfs$ of the spray. Namely,
\begin{equation}\label{3.195}
\zao (s) = \zeta_\mfs (s) \cdot \zeta_{\partial G,G} (s) = \frac{\zeta_{\partial G,G} (s)}{1 - \sum_{j=1}^J r_j^s}. 
\end{equation}

Furthermore, in light of \eqref{3.129}, \eqref{3.132}--\eqref{3.133} and \eqref{3.152.1/2},
\begin{equation}\label{3.196}
D_{\ao} = \max (\sigma_0, D_{\partial G,G}),
\end{equation} 
where $\sigma_0 = D(\zeta_\mfs) \in (0,N)$ (the similarity dimension of the spray) is the unique real solution of the Moran equation \cite{Mora}; i.e., $\sigma_0 \in \mbr$ and $\sum_{j=1}^J r_j^{\sigma_0} =1$. Moreover, in light of \eqref{3.132}--\eqref{3.133} and \eqref{3.152.1/2}, if we assume that the generator $G$ is sufficiently ``nice'' [e.g., $G$ is monophase, in the sense of [\hyperlinkcite{LapPe1}{LapPe1--2}, \hyperlinkcite{LapPeWi1}{LapPeWi1--2}]\footnote{Roughly, this means that $V_{\partial G,G} (\varepsilon)$ is polynomial for all $\varepsilon$ sufficiently small.} 
and, in particular, if $\partial G$ is a nontrivial polytope (by a result in \cite{KoRati})], we have that
\begin{equation}\label{3.197}
\mcd_{\ao} = \mcd_{\partial G,G} \cup \mcd_\mfs \subseteq \{0,1, \cdots, N-1\} \cup \Bigg\{s \in \mbc : \sum_{j=1}^J r_j^s =1 \Bigg\},
\end{equation}
with frequent or ``typical'' equality in \eqref{3.197} and with $\mcd_{\ao}$ always containing $D_{\ao}$ and (by the results in \cite[Thm. 3]{Lap-vF4}) also containing infinitely many (scaling) complex dimensions with real part $\sigma_0$.

We then deduce from \eqref{3.167} the following exact, pointwise fractal tube formula:
\begin{equation}\label{3.198}
V_{\ao} (\varepsilon) = \sum_{\omega \in \mcd_\mfs \cup \{ 0,1, \cdots, N-1 \} } \res \Bigg( \frac{\varepsilon^{N-s} (\sum_{\alpha =0}^N \kappa_\alpha \frac{g^{s-\alpha}}{s-\alpha})}{(N-s) (1- \sum_{j=1}^J r_j^s)}, \omega \Bigg).
\end{equation}
Here, $g$ is the inner radius of the generator $G$ and the coefficients $\kappa_\alpha$ (some of which could vanish) are the coefficients of the polynomial expansion of $V_{\partial G,G} (\varepsilon)$.\footnote{More specifically,
\begin{equation}\label{3.199}
V_{\partial G,G} (\varepsilon) = \sum_{\alpha =0}^{N-1} \kappa_\alpha t^{N-\alpha}, \quad \text{for } 0 < \varepsilon <g.
\end{equation}}  

In the important special case when all of the scaling complex dimensions are simple and when $\sigma_0$ is not an integer, the fractal tube formula \eqref{3.198} takes the following simpler form:
\begin{equation}\label{3.199.5}
V_{\ao} (\varepsilon) = \sum_{\omega \in \mcd_\mfs \cup \{0,1, \cdots, N-1 \}} d_\omega \varepsilon^{N-\omega},
\end{equation} 
where
\begin{equation}\label{3.200}
d_\omega = \res (\zeta_\mfs, \omega) \Bigg( \sum_{\alpha =0}^N \frac{\kappa_\alpha g^{\omega - \alpha}}{\omega - \alpha} \Bigg), \quad \text{if } \omega \in \mcd_\mfs,
\end{equation}
and
\begin{equation}\label{3.201}
d_\omega = \zeta_\mfs (\omega) \kappa_\omega, \quad \text{if } \omega \in \{ 0,1, \cdots, N-1\}.
\end{equation}
Therefore,
\begin{equation}\label{3.202}
V_{\ao} (\varepsilon) = \sum_{\omega \in \mcd_\mfs} \res (\zeta_{\mfs}, \omega) \Bigg( \sum_{\alpha =0}^N \frac{\kappa_\alpha g^{\omega - \alpha}}{\omega - \alpha} \Bigg) + \sum_{\alpha = 0}^{N-1} \zeta_\mfs (\alpha) \kappa_\alpha \varepsilon^{N - \alpha.}
\end{equation}

We note that if the generator $G$ is pluriphase\footnote{That is, roughly speaking, if $V_{\partial G,G} (\varepsilon)$ is a piecewise polynomial function; see \cite{LapPe2, LapPeWi1}.}
instead of monophase, then we can easily extend the above results and at the same time recover (as well as significantly extend) the results of \cite{LapPe2} and, especially, of \cite{LapPeWi1}. Also, even if $G$ is not pluriphase (and thus certainly not monophase), it is clear from \eqref{3.195} that under suitable polynomial-type growth assumptions on $\zeta_{\partial G,G}$, one can use the tube formulas (from \cite{LapRaZu6} and \cite[\S\S 5.1--5.3]{LapRaZu1}) recalled in \S \ref{Sec:3.5.1} in order to obtain pointwise or distributional fractal tube formulas (with or without error term) for $(\partial G,G)$, expressed via $\zeta_{\partial G,G}$. This is so even if the fractal spray RFD $(\ao)$ is not necessarily self-similar. (Indeed, in that general case, the first equality in \eqref{3.195} still holds.) We let the interested reader elaborate on the latter comments.

As was already pointed out in \S \ref{Sec:3.4.10}, in the present case of self-similar sprays (and unlike for ordinary self-similar sets $A$ for which one always has $D_A = \sigma_0$), we must distinguish three cases here, all of which are realized (see, e.g., the example of the relative Sierpinski $N$-gasket discussed in \S \ref{Sec:3.4.4} and at the end of \S \ref{Sec:3.6.1}):\footnote{Henceforth, we use the above notation and write $D_G:= D_{\partial G,G} = \dim_B (\partial G,G)$.}\\

{\bf {\em Case} ({\em i}):} $D_G < \sigma_0$. Then, by \eqref{3.196}, $D_{\ao} = \sigma_0$ and is simple. Hence,
\begin{equation}\label{3.203}
\dim_{PC} (\ao) = \mcd_\mfs = \Bigg\{ s \in \mbc: \sum_{j=1}^J r_j^s = 1 \Bigg\}.
\end{equation}
Consequently, in light of the structure of the scaling complex dimensions of self-similar sprays given in \S \ref{Sec:3.4.10} (and based on the results of \cite[Ch. 3]{Lap-vF4}), we obtain the precise counterpart of the Minkowski measurablity criterion of self-similar strings from \cite[\S 8.4]{Lap-vF4} (see part ($d$) of Remark \ref{Rem:2.4}). Namely, {\em the self-similar spray RFD $(\ao)$ is Minkowski measurable if and only if it is nonlattice}; {\em that is, if and only if it does not have any nonreal principal complex dimensions. Hence, if $(\ao)$ is lattice, it is not Minkowski measurable, whereas if it is nonlattice, it is Minkowski measurable.} This result follows from the refined version of the Minkowski measurability criterion stated in \S \ref{Sec:3.5.2}, and more precisely, from the necessary (respectively, sufficient) condition for Minkowski measurability obtained in \cite[\S 5.4]{LapRaZu1} and in \cite{LapRaZu6} (and briefly alluded to in \S \ref{Sec:3.5.2}).\footnote{{\em Caution}: The hypotheses of the Minkowski measurability criterion stated in \S \ref{Sec:3.5.2} are not always satisfied in the nonlattice case (see \cite[Exple. 5.32]{Lap-vF4} for a counterexample). This is why we have to use a refined form (the aforementioned sufficient condition) in order to establish the Minkowski measurability of $(\ao)$ in the nonlattice case. On the other hand, in the lattice case, the hypotheses of the above  criterion are clearly satisfied and therefore we can conclude that $(\ao)$ is not Minkowski measurable; alternatively, one can use the aforementioned sufficient condition (see \cite[Thm. 5.4.15]{LapRaZu1}) in order to reach the same conclusion.}\\

{\bf {\em Case} ({\em ii}):} $D_G = \sigma_0$ (and hence, in light of \eqref{3.196}, $D_{\ao} =D_G = \sigma_0$ is an integer). Then, as was noted in \S \ref{Sec:3.4.10}, it follows from the factorization formula \eqref{3.195} that $D_{\ao}$ is a complex dimension of $(\ao)$ of multiplicity two. As a result, due to the Minkowski measurability criterion discussed in \S \ref{Sec:3.5.2}, $(\ao)$ cannot be Minkowski measurable (in the usual sense), irrespective of whether $(\ao)$ is lattice or nonlattice. However, if we use the gauge function $h(t) := \log t^{-1}$, for all $t \in (0,1)$, then provided its hypotheses are satisfied, the $h$-Minkowski measurability criterion briefly discussed towards the end of \S \ref{Sec:3.5.2} (\cite[Thm. 5.4.32]{LapRaZu1}), the RFD $(\ao)$ is $h$-Minkowski measurable. \\

{\bf {\em Case} ({\em iii}):} $D_G > \sigma_0$. Then, by \eqref{3.196}, $D_{A,G} = D_G$. Since all of the scaling complex dimensions of $(\ao)$ have real parts not exceeding $\sigma_0$ and hence, strictly less than $D_{\ao}$, we deduce from \eqref{3.197} that $D_{\ao} = D_G$ is the only principal complex dimension of $(\ao)$ and (since $D_G$ is a simple pole of $\zeta_{\partial G,G}$) that it is simple. Consequently, {\em in this case} ({\em i.e., in case} ($iii$)), $(\ao)$ {\em is always Minkowski measurable}, whether or not the self-similar spray $(\ao)$ is lattice or nonlattice.

However, this does not preclude $(\ao)$ from having lower-order oscillations in its geometry (e.g., in its fractal tube formula \eqref{3.202}). This is indeed what happens (generically) for the relative Sierpinski $N$-gasket when $N \geq 4$.
\end{example}

\begin{exercise}({\em Relative $N$-gasket}).\label{Exc:3.32}
By using the trichotomy outlined in cases ($i$)--($iii$) just above, obtain (as explicitly as possible) fractal tube formulas for the relative (or inhomogeneous) Sierpinski $N$-gasket studied  in \S \ref{Sec:3.4.4}. Also, determine the Minkowski measurability (or, more generally and when necessary, the $h$-Minkowski measurability) of the relative $N$-gasket, depending on the value of $N$. When appropriate, calculate the corresponding Minkowski content $\mcm$ or the average Minkowski content $\widetilde{\mcm}$.

[{\em Hint}: Distinguish the three cases when $N=2, N=3$ and $N \geq 4$, respectively.]
\end{exercise}

\begin{remark}({\em Ordinary self-similar sets.})\label{Rem:3.32.1/2}
($a$) It has long been conjectured by the author (see \cite[Conj. 3, p. 163]{Lap3}) that (classic or homogeneous) self-similar sets in $\mbr^N$ satisfying the open set condition (in the sense of \cite{Hut}; see also \cite{Fa1}) are Minkowski measurable  if and only if they are nonlattice (and, equivalently, are not Minkowski measurable if and only if they are lattice). When $N=1$ (i.e., for self-similar strings), the fact that nonlattice self-similar sets (i.e., strings) are Minkowski measurable was first proved by the author in \cite{Lap3} and then, independently, by K. Falconer in \cite{Fa2}, in both cases by using the renewal theorem (first used in a related context by S. Lalley in [\hyperlinkcite{Lall1}{Lall1--3}]). Then, this result was extended to higher dimensions (and to certain random fractals, as was also conjectured in \cite{Lap3}) by D. Gatzouras in \cite{Gat}. There remained to prove that the nonlattice condition was also necessary for obtaining the Minkowski measurability of a given self-similar set. This was first established when $N=1$ (i.e., for self-similar strings) by the author and M. van Frankenhuijsen in \cite{Lap-vF2} (see \cite[\S 8.4]{Lap-vF4}), where both the necessary and sufficient conditions were proved by using the theory of complex dimensions of fractal strings (combined with a suitable Tauberian theorem) and the associated fractal tube formulas; see \cite[\S 8.4, Thms. 8.23 and 8.36]{Lap-vF4}. Finally, in higher dimensions (i.e., when $N \geq 2$), the sufficient condition was recently established  (independently of the above results from \cite{LapRaZu1} and the accompanying papers about self-similar sprays or RFDs) by S. Kombrink, E. Pearse and S. Winter in \cite{KomPeWi}, also by using the renewal theorem but now combined with several nice new observations.

($b$)  We conjecture that under suitable hypotheses and still for (classic or homogeneous) self-similar sets satisfying the open set condition, the above characterization expressed in terms of complex dimensions is still valid; that is, the presence of a nonreal principal complex dimension is equivalent to the self-similar set being not Minkowski measurable (and hence, Minkowski measurability is equivalent to the Minkowski dimension $D$ being the only principal complex dimension; see \cite[Pb. 6.2.36]{LapRaZu1} for more details. This problem still remains open, for now, and its resolution will require, in particular, suitably extending the factorization formula \eqref{3.195} to this setting or finding an appropriate substitute for it.
\end{remark}

\subsection{Fractality, hyperfractality and unreality, revisited}\label{Sec:3.6}

We pursue and complete here the discussion of fractality and unreality (as well as of the closely related topic of the meaning of complex dimensions) started in \S \ref{Sec:2.4} and \S \ref{Sec:2.5}. In light of the higher-dimensional theory of complex dimensions developed in \cite{LapRaZu1} (and in the accompanying series of papers, [\hyperlinkcite{LapRaZu2}{LapRaZu2--10}]) and, in particular, of the fractal tube formulas obtained in \cite[Ch. 5]{LapRaZu1} and \cite{LapRaZu6} (as well as discussed in \S \ref{Sec:3.5.1} and \S \ref{Sec:3.5.3}), the  interpretation of (necessarily complex conjugate pairs of) {\em nonreal} complex dimensions as giving rise to (or detecting) the {\em intrinsic oscillations} of a given geometric object is exactly the same as in the one-dimensional case of fractal strings discussed in \S \ref{Sec:2.4}. Therefore, we refer the interested reader to \S \ref{Sec:2.4}  for the corresponding discussion, which can easily be adapted to higher dimensions and is illustrated by the many examples of complex dimensions and fractal tube formulas provided in \S \ref{Sec:3.4} and \S \ref{Sec:3.5.3}, respectively. 

We further mention that for clarity, we will focus  here on geometric objects which are relative fractal drums or RFDs (and, in particular, bounded sets) in $\mbr^N$. However, since the entire theory of fractal zeta functions and complex dimensions (in particular, of the associated fractal tube formulas) extends (essentially without change) to suitable metric measure spaces (namely, Ahlfors-type spaces and beyond), as is shown in \cite{LapWat, Wat}, we could instead work within the much greater generality of RFDs (and, in particular, of bounded sets) in such metric measure spaces. (See Remark \ref{Rem:3.22.1/2}.) The main difference is that the embedding dimension $N$ would have to be replaced by the Ahlfors dimension (or its appropriate analog) of the embedding metric measure space; see footnote \ref{Fn:119}.

Perhaps more importantly, we could also apply to various counting functions the very general (pointwise or distributional) explicit formulas from \cite[Ch. 5]{Lap-vF4} in order to detect the potential geometric, spectral, dynamical, or arithmetic oscillations that are intrinsic to  fractal-like (physical or mathematical) objects. These counting functions could be geometric (such as the box-counting function), spectral (such as the eigenvalue or frequency counting function or else, the partition function or trace of the heat semigroup) or dynamical (such as the prime orbit counting function, counting the number of homology classes of primitive periodic orbits of the corresponding dynamical system\footnote{See \cite[Ch. 7]{Lap-vF4} for a simple but illuminating example. The author has long thought that the just referred work could be greatly extended to a variety of hyperbolic and other dynamical systems, for example, within the setting of the theory developed by Parry and Pollicott in \cite{ParrPol1, ParrPol2}, where Ruelle or dynamical zeta functions ([\hyperlinkcite{Rue1}{Rue1--4}], \cite{Lag}) were used. This potentially significant extension still remains to be achieved. \label{Fn:129}}).

Moreover, in a similar spirit, we could use other types of fractal-like zeta functions (such as, for example, suitably weighted Ruelle's dynamical zeta functions; see, e.g., [\hyperlinkcite{Rue1}{Rue1--4}, \hyperlinkcite{ParrPol1}{ParrPol1--2}], \cite{Lag}, along with footnote \ref{Fn:129}) could be used to define ``fractality'' via the existence of nonreal complex dimensions. Finally, as will be further explained below, even the notion of ``complex dimensions'' itself can be relaxed, by considering (nonremovable) singularities  that are not just poles of the associated fractal zeta functions.

Accordingly, the proposed notion of fractality can be potentially applied to a great variety of settings in mathematics, physics, cosmology, chemistry, biology, medicine, geology, computer science, engineering, economics, finance, and the arts, as well as can be illustrated via many different kinds of fractal-type explicit formulas. In mathematics or physics alone, the corresponding fields involved would include, for instance, harmonic analysis, partial differential equations, geometric measure theory, spectral theory, spectral geometry, probability theory, dynamical systems, combinatorics and graph theory, number theory and arithmetic geometry, algebraic geometry, operator algebras and noncommutative geometry, quantum groups, mathematical physics, along with, on the more physical side, condensed matter physics, astronomy and cosmology, quantum theory and its myriad of applications, quantum gravity, classical and quantum chaos, quantum computing and string theory. For the simplicity of exposition, however, and with one single exception,\footnote{This exception will have to do with the spectra of fractal drums; see \S \ref{Sec:3.6.4}.}
we will limit our discussion to the geometric setting, that of RFDs (and in particular, of bounded sets) in an $N$-dimensional Euclidean space $\mbr^N$, as well as briefly illustrate it by means of the fractal tube formulas and the Minkowski measurability results described in \S\S \ref{Sec:3.5.1}--\ref{Sec:3.5.3}.

As in \S \ref{Sec:2.5}, we say that a geometric object (e.g., a fractal drum $(\ao)$ in $\mbr^N$ or, in particular, a bounded subset $A$ of $\mbr^N$) is {\em fractal} if it has at least one nonreal complex dimension (and hence, at least one pair of nonreal complex conjugate complex dimensions). For now, ``complex dimensions'' are interpreted as being the (visible) poles of the associated fractal zeta function $\zao$ or, equivalently, $\tzao$ of $(\ao)$.\footnote{For simplicity, we assume from now on that $D:= \overline{\dim_B} (\ao) < N$, so that the distance and tube zeta functions $\zao$ and $\tzao$ have the same visible poles in a given domain $U$ of $\mbc$ to which one (and hence, both) of these fractal zeta function has a (necessarily unique) meromorphic extension; in particular, $D (\zao) = D(\tzao) =D$. (See \S \ref{Sec:3.3.1} and \S \ref{Sec:3.3.2}.) If one wants to deal with the case when $D=N$, one should then work with $\tzao$ alone.}
However, we will further broaden this notion later on in this subsection; see, especially, \S \ref{Sec:3.6.2} and \S \ref{Sec:3.6.3}. This definition is formally identical to the one proposed in \cite[Ch. 12]{Lap-vF4} (and, prior to that, in [\hyperlinkcite{Lap-vF1}{Lap-vF1--3}]),\footnote{There is one small difference; namely, we no longer require the real part of the nonreal complex dimension to be positive (a condition that was included mainly for aesthetic reasons and is fulfilled, for example, by classic self-similar geometries). Indeed, for an RFD $(\ao)$, even $D = \overline{\dim_B} (\ao)$ can be negative (or equal to $-\infty$); see the two examples provided between \eqref{3.4} and \eqref{3.5} in \S \ref{Sec:3.2}. }
but it is worth pointing out that at the time, there was no suitable general definition of fractal zeta functions, and hence also of complex dimensions, that was compatible with the existence of fractal tube formulas for compact subsets of $\mbr^N$, with $N \geq 2$.

More specifically, given $d \in \mbr$, an RFD $(\ao)$ (and, in particular, a bounded set $A$) in $\mbr^N$ is said to be {\em fractal in dimension} $d$ if it has at least one (and hence, a pair of complex conjugate) complex dimension(s) of real part $d$.\footnote{Then, clearly, we must have $D > -\infty$ and $d \leq D \leq N$, where $D := \overline{\dim}_B (\ao)$ or ($D := \overline{\dim}_B A$). Also, note that $d$ itself need not be a (nonremovable) singularity, let alone a pole, of $\zao$ (or of $\tzao$).}
Therefore, by definition, fractality is equivalent to fractality in some dimension $d \in \mbr$.

Moreover, $(\ao)$ (or $A$) is said to be {\em critically fractal} if it is fractal in dimension $D$, the (upper) Minkowski dimension of $(\ao)$ (or of $A$); in other words, if and only if it has at least one nonreal complex dimension. Otherwise $(\ao)$ (or $A$) is said to be {\em subcritically fractal}; in that case, it is therefore fractal in some dimension $d < D$ but not in dimension $D$.

In addition, the RFD $(\ao)$ is said to be {\em hyperfractal} if the associated fractal zeta function $\zao$ (or equivalently, $\tzao$, since $D < N$ here) cannot be meromorphically extended to a connected open neighborhood of a suitable curve or contour $S$ in $\mbc$ extending in both vertical directions (i.e., $S$ is a screen, in the sense of \cite{Lap-vF4} or \cite{LapRaZu1}), and it is said to be {\em critically hyperfractal} (or {\em strongly hyperfractal}, as in \cite{LapRaZu1}) if $S = \{ Re(s) = D\}$, and {\em maximally hyperfractal} if the critical line $\{ Re(s) = D\}$ consists solely of (nonisolated and nonremovable) singularities of $\zao$ (or equivalently, of $\tzao$).\footnote{Precise definitions of the notion of a singularity (of a complex-valued function on a domain of $\mbc$) can be found in \cite[\S 1.3.2]{LapRaZu1} and the references therein; see also \cite[\S 4.6.3]{LapRaZu1}.} (Naturally, a similar terminology is used in the special case of a bounded subset $A$ of $\mbr^N.)$

Clearly, ``maximally hyperfractal'' implies ``critically hyperfractal'', which itself implies ``hyperfractal''. Furthermore, maximal hyperfractals are such that $\zao$ (or equivalently, $\tzao$) have nonreal (complex conjugate pairs of) singularities (which are not necessarily poles). Moreover, in some sense, they are among the most complicated fractals.

In \cite[\S 4.6]{LapRaZu1} (and \cite{LapRaZu3}) are constructed maximally hyperfractal compact sets as well as RFDs in $\mbr^N$, for $N \geq 1$ arbitrary (and, in particular, when $N=1$, maximally hyperfractal fractal strings), with any prescribed (upper) Minkowski dimension $D \in (0, N)$. In fact, the family of examples constructed in {\em loc. cit.} consists not only of maximally hyperfractal but also of transcendentally $\infty$-quasiperiodic sets or RFDs. Recall from \cite[\S 4.6]{LapRaZu1} that an RFD $(\ao)$ (or a bounded set $A$) in $\mbr^N$, of (upper) Minkowski dimension $D$ is said to be {\em transcendentally $\infty$-quasiperiodic} if its tube function $V = V(\varepsilon)$ (where $V:= V_{\ao}$ or $V:= V_A$) satisfies
\begin{equation}\label{3.204}
V (\varepsilon) = \varepsilon^{N-D} (G (\log \varepsilon^{-1}) + o(1)) \quad \text{as } \varepsilon \rightarrow 0^+,
\end{equation}
where the function $G: \mbr \rightarrow \mbr$ is {\em transcendentally $\infty$-quasiperiodic}; i.e., where $G = G (t)$ is the restriction to the diagonal of a nonconstant  function $H : \mbr^\infty \rightarrow \mbr$ which is separately periodic (of minimal period $T_j$) in each variable $t_j$, for $j = 1,2, \cdots$:
\begin{equation}\label{3.205}
G(t) = H (t,t,t, \cdots), \quad \text{for all } t \in \mbr,
\end{equation}
where (for all $(t_1, t_2, \cdots) \in \mbr^\infty$)
\begin{equation}\label{3.206}
H (t_1, t_2, \cdots, t_{j-1}, t_j + T_j, t_{j+1}, \cdots) = H (t_1, t_2, \cdots, t_{j-1}, t_j, t_{j+1}, \cdots).
\end{equation}
In addition, the use of the adverb ``transcendentally'' in the above definition means that the resulting sequence of {\em quasiperiods} $(T_j)_{j \geq 1}$ is linearly independent over the field of algebraic numbers.\footnote{Recall that the field of {\em algebraic numbers} can be viewed (up to isomorphism) as the algebraic closure $\overline{\mbq}$ of $\mbq$, the field of rational numbers. It is obtained by adjoining to $\mbq$ the complex roots of all of the monic polynomials with coefficients in $\mbq$ (or equivalently, in $\mbz$). By reasoning by absurdum, one can easily check that it is a countable set.}

The aforementioned construction of maximally hyperfractal and transcendentally $\infty$-quasiperiodic RFDs is rather complicated. Its first (and main) step consists in constructing fractal strings with the above properties. In essence, these highly complex fractal strings are obtained by taking the countable disjoint union of a suitable sequence extracted from a two-parameter family of topological Cantor sets (or strings), and then to appropriately apply a key result from transcendental number theory (namely, Baker's theorem \cite{Bak}; see also \cite[Thm. 3.114]{LapRaZu1}).\footnote{Recall that Baker's theorem states that given $n \in \mbn$ with $n \geq 2$, if $m_1, \cdots, m_n$ are positive algebraic numbers such that $\log m_1, \cdots, \log m_n$ are linearly independent over the rationals, then $1, \log m_1, \cdots, \log m_n$ are linearly independent over the field of algebraic numbers (i.e., algebraically independent). In particular, $\log m_1, \cdots, \log m_n$ are transcendental (i.e., not algebraic) numbers, and so are their pairwise quotients.  }
In some sense, this construction can be viewed as a fractal-geometric interpretation of Baker's theorem.\\

We close this discussion by providing several examples of critical and subcritical fractals.

First, we note that under suitable hypotheses (namely, the hypotheses of the Minkowski measurability criterion stated in \S \ref{Sec:3.5.2}), and supposing that the Minkowski dimension $D = \dim_B (\ao)$ exists and is simple,\footnote{If $D$ is multiple (as a pole of $\zao$), then it follows from  \cite[Thm. 5.4.27]{LapRaZu1} briefly discussed towards the end of \S \ref{Sec:3.5.2} that under the hypotheses of that theorem, the RFD $(\ao)$ is $h$-Minkowski measurable with respect to the gauge function $h(t) := (\log t^{-1})^{m-1}$, for all $t \in (0,1),$ where $m \geq 2$ is the multiplicity of $D$. }
the RFD $(\ao)$ in $\mbr^N$ is subcritically fractal if and only if it is Minkowski measurable. Stated another way, and still under the hypotheses of the aforementioned criterion, $(\ao)$ is critically fractal if and only if it is not Minkowski measurable. Naturally, the same statements hold for bounded subsets $A$ of $\mbr^N$.

Moreover, self-similar strings are subcritically fractal if and only if  they are nonlattice, and also if and only if they are Minkowski measurable. (See \cite[\S 8.4]{Lap-vF4} and part ($d$) of Remark \ref{Rem:2.4}.)\footnote{We caution the reader that this statement is not a direct consequence of the general Minkowski measurability criterion (in terms of nonreal principal complex dimensions) used above and stated in \S \ref{Sec:3.5.2}. Indeed, the hypotheses of the corresponding theorem are not satisfied by all nonlattice self-similar strings. However, as was briefly mentioned in Example \ref{Exe:3.31} of \cite[\S 5.4]{LapRaZu1}, the extension of this statement to self-similar sprays with nice generators is proved by using separately the necessary condition and the sufficient condition for Minkowski measurability obtained in {\em loc. cit.}}
Equivalently, self-similar strings are critically fractal if and only if they are lattice, and also if and only if they are not Minkowski measurable (even though they are always Minkowski nondegenerate).

As was mentioned in Example \ref{Exe:3.31}, exactly the same statements (as just above for self-similar strings) hold for self-similar fractal spray RFDs (in $\mbr^N, N \geq 1$) with ``nice'' generators (e.g., monophase or even pluriphase generators, and in particular, with generators that are nontrivial polytopes), in case ($i$) of Example \ref{Exe:3.31}; i.e., when $D_G := \dim_B (\partial G,G) < \sigma_0$, the similarity dimension of $(\ao)$, where $(\partial G,G)$ is the generator (or base) of the self-similar spray RFD $(\ao)$.

In addition, irrespective of whether we are in case ($i$), ($ii$) or ($iii$) of Example \ref{Exe:3.31}, it follows from the results of \cite[Ch. 3]{Lap-vF4} combined with \eqref{3.197} and the text following it, that self-similar sprays (with nice generators) are always fractal, and more specifically, that lattice (respectively, nonlattice) sprays are fractal in dimension $d$ for finitely (respectively, countably infinitely) many values of $d \in \mbr$. Also, such lattice sprays are critically fractal whereas nonlattice sprays are subcritically fractal. Furthermore, still in light of those results, but now also combined with the main theorem in \cite{MorSepVi1} (proving and extending a conjecture in \cite{Lap-vF5}, see also \cite{Lap-vF3}), it is known that the set of $d$'s for which a given nonlattice spray is fractal in dimension $d$ is dense in a single compact (nonempty) interval $[D_{\ell}, D]$ in the generic case, with $D_{\ell} < D$, while we conjecture that an analogous statement  is true in the nongeneric case, but now with finitely many (but more than one) compact nonempty intervals instead of a single one.\footnote{We conjecture that under suitable hypotheses, analogous results hold for the scaling complex dimension of (classic or homogeneous) self-similar sets satisfying the open set condition (as in \cite{Hut} and, e.g., \cite{Fa1}); in particular, all such self-similar sets are fractal. A key open problem in this context is to first obtain a factorization formula of the type \eqref{3.195} (see also \eqref{3.124}), possibly up to the addition of a holomorphic function in an appropriate right half-plane $\{ Re(s) > \beta \}$, for some $\beta < D$. (See also \cite[Pb. 6.2.36 and Rem. 6.2.37]{LapRaZu1}.)}\\

\subsubsection{The $\frac{1}{2}$-square fractal, the $\frac{1}{3}$-square fractal and the relative Sierpinski \\ $N$-gasket}\label{Sec:3.6.1}

The above results can be illustrated by the $\frac{1}{2}$-square fractal RFD, the  $\frac{1}{3}$-square fractal RFD and (for suitable values of $N$) the relative $N$-gasket. \\

($a$) More specifically, for the $\frac{1}{2}$-square fractal RFD $(\ao), D := \dim_B (\ao) =1$ is a complex dimension of multiplicity two; hence, $(\ao)$ is not Minkowski measurable but is also Minkowski degenerate, with Minkowski content $\mcm = +\infty$. However, $(\ao)$ is $h$-Minkowski measurable with respect to the gauge function $h(t) = \log t^{-1}$ (for all $t \in (0,1)$) and with $h$-Minkowski content $\mcm ((\ao), h) = 1/4 \log 2$. (See part ($a$) of \S \ref{Sec:3.4.5} and part ($a$) of Example \ref{Exe:3.27}.) Further, by \eqref{3.67}--\eqref{3.69},
\begin{equation}\label{3.207}
D:=D_{\ao} = 1,
\end{equation}

\begin{equation}\label{3.208}
\mcd_{\ao} = \{ 0,1 \} \cup \Big(1 + i \frac{2 \pi}{\log 2} \mbz \Big)
\end{equation}
and
\begin{equation}\label{3.209}
\dim_{PC} (\ao) = 1 + i \frac{2 \pi}{\log 2} \mbz.
\end{equation}
Therefore, the $\frac{1}{2}$-square fractal $(\ao)$ is critically fractal, and is only fractal in dimension $D=1$, the Minkowski dimension of $(\ao)$. According to {\em loc. cit.}, the same statements are true for the $\frac{1}{2}$-square fractal $A$ (instead of for $(\ao)$).\\

($b$) Next, let $(\ao)$ be the  $\frac{1}{3}$-square fractal, as in part ($b$) of \S \ref{Sec:3.4.5} and in part ($b$) of Example \ref{Exe:3.27}. Then, in light of \eqref{3.76} and the discussion surrounding it, $(\ao)$ is subcritically fractal in dimension $d:= \log_3 2$ and is only fractal in that dimension, the Minkowski dimension of the ternary Cantor set. Furthermore, $D := D_{\ao} =1$ and $1$ is simple (as well as the only principal complex dimension of $(\ao)$). Moreover, according to part ($b$) of Example \ref{Exe:3.27}, the  $\frac{1}{3}$-square fractal RFD $(\ao)$ is Minkowski measurable. In light of {\em loc. cit.}, the exact same results hold for the $\frac{1}{3}$-square fractal $A$ itself.\\

($c$) We now consider the relative $N$-gasket $(A_N, \Omega_N)$, studied for any $N \geq 2$ in \S \ref{Sec:3.4.4}. Recall that with the notation used above for self-similar sprays, $D_G = N-1$ and $\sigma_0 = \log_2 (N+1)$, and in light of \eqref{3.59} and \eqref{3.61}, that
\begin{equation}\label{3.210}
\mcd_{A_N, \Omega_N} = \{ 0,1, \cdots, N-1\} \cup \Big( \log_2 (N+1) + i \frac{2 \pi}{\log 2} \mbz \Big)
\end{equation}
and
\begin{equation}\label{3.211}
D := \dim_B (A_N, \Omega_N) = \max (N-1, \log_2 (N+1)) = 
\begin{cases}
\log_2 3,  &\text{if } N=2, \\
N-1, &\text{if } N \geq 3. 
\end{cases}
\end{equation}
Therefore, as we next explain, we recover the three cases ($i$), ($ii$) and ($iii$) discussed above for general self-similar spray RFDs:\\

Case ($i$) when $D_G < \sigma_0$ corresponds to the $N=2$ case; then, $(A_2, \Omega_2)$ is not Minkowski measurable but is Minkowski  nondegenerate. Also, it is critically fractal and fractal only in dimension $d = \log_2 3$, the dimension of the relative Sierpinski gasket $(A_2, \Omega_2)$. In light of the results of \S \ref{Sec:3.2.1}, exactly the same statements hold for the (classic) Sierpinski gasket $A_2$ itself. \\

Case ($ii$) when $D_G = \sigma_0$ (i.e., $N-1 = \log_2 (N+1)$) corresponds to $N=3$. Then, $(A_3, \Omega_3)$ is not Minkowski measurable (because $D=2$ is of multiplicity two) and is also Minkowski degenerate. However, it is $h$-Minkowski measurable with respect to the gauge function $h(t):= \log t^{-1}$ (for all $t \in (0,1)$). Furthermore, in light of \eqref{3.210} and \eqref{3.211}, $(A_3, \Omega_3)$ is critically fractal (necessarily in dimension $d:= D = \log_2  4 = 2$), and is fractal only in that dimension.\\

Finally, case ($iii$) when $D_G > \sigma_0$ corresponds to every value of $N \geq 4$. In this case, $D = D_G = N-1, (A_N, \Omega_N)$ is subcritically fractal in dimension $d:= \sigma_0 = \log_2 (N+1)$, and is fractal only in that dimension. Also, it is Minkowski measurable but has geometric oscillations of lower order (corresponding to the nonreal complex dimensions of real part $\log_2 (N+1)$). This concludes our discussion of the relative Sierpinski $N$-gasket.\\

\subsubsection{The devil's staircase and fractality}\label{Sec:3.6.2}

We close this part of the discussion by considering the emblematic example of the Cantor graph RFD $(\ao)$ studied in \S \ref{Sec:3.4.9} and in Example \ref{Exe:3.29}. This example is closely related to the Cantor graph $A$, also called the ``devil's staircase'' in \cite{Man}. It is important for a variety of reasons:\\

($a$) First, $(\ao)$ is not a self-similar spray RFD. Indeed, the devil's staircase $A$ is not a self-similar set; instead, it is a self-affine set, which makes the corresponding computation significantly more complicated, if not impossible to carry out.\\

($b$) Secondly, the devil's staircase is {\em not} fractal according to Mandelbrot's definition of fractality (to be recalled just below), even though everyone with an exercised eye (including Benoit Mandelbrot himself)\footnote{See the very explicit and enlightening comments in \cite[p. 82]{Man} and \cite[Plate 83, p. 83]{Man}; see also the very convincing and beautiful figure in {\em loc. cit.} Let us quote from \cite[p. 82]{Man}: ``{\em One would love to call the present curve a fractal, but to achieve this goal, we would have to define fractal less stringently on the basis of notions other than $D$} [the Hausdorff dimension] {\em alone}.''  } 
would  agree that it must be ``fractal''.\footnote{The exact same comment can be made verbatim about the Cantor graph RFD provided, in Mandelbrot's original definition, one substitutes the (relative) Minkowski dimension for the Hausdorff dimension.}
At this point, it may be helpful  to the reader to recall that in \cite{Man} (and in later work), Mandelbrot called ``fractal'' any (bounded) subset $A$ of Euclidean space $\mbr^N$ such that its topological dimension $\dim_T A$ and Hausdorff dimension $\dim_H A$ do not coincide:\footnote{It is noteworthy that if $A$ is nonempty, $\dim_T A$ is always a nonnegative integer. For example, for the ternary Cantor set, it is equal to $0$ while for the classic Sierpinski gasket and for the Koch snowflake curve, it is equal to $1$.}
\begin{equation}\label{3.212}
\dim_H A \neq \dim_T A,
\end{equation} 
or, equivalently, since it is always true that $\dim_T A \leq \dim_H A$,
\begin{equation}\label{3.213}
\dim_H A > \dim_T A.
\end{equation}
In the present case when $A$ is the devil's staircase, this definition clearly fails (as Mandelbrot was well aware of). Indeed, since the Cantor graph is a rectifiable curve (i.e., a curve of finite length), we have that
\begin{equation}\label{3.214}
\dim_H A = \dim_B A = \dim_T A =1,
\end{equation}
according to a well-known result in elementary geometric measure theory (see \cite{Fed1}) and as can also be directly checked here.\\

($c$) Lastly, we mention that the unambiguous and frustrating contradiction between the visual impression one gets when contemplating the devil's staircase and Mandelbrot's definition of fractality (which he only adopted reluctantly, in response to the many queries and criticisms he received after the publication of his earlier French book on fractals) led the present author to wonder how to resolve this paradox and more importantly, how to much better capture the intuition underlying the informal notion of ``fractality''. Combined with the author's  early work on fractal drums [\hyperlinkcite{Lap1}{Lap1--4}] and his joint work on fractal strings and the Riemann zeta function as well as the Riemann hypothesis [\hyperlinkcite{LapPo1}{LapPo1--3}, \hyperlinkcite{LapMa1}{LapMa1--3}], this quest eventually led (first, for fractal strings, in  [\hyperlinkcite{Lap-vF1}{Lap-vF1--4}] and then, in any dimension $N \geq 1$, in \cite{LapRaZu1}) to the notion of complex fractal dimensions and to the present notion of fractality expressed in terms of the irreality (or the ``unreality'') of some of the underlying complex dimensions.\\

Now, let us return to the Cantor graph RFD $(\ao)$ and its complex dimensions. In light of the discussion of that example provided in \S \ref{Sec:3.4.9} and in Example \ref{Exe:3.29},
\begin{equation}\label{3.215}
D := D_{\ao} =D_A =1,
\end{equation}

\begin{equation}\label{3.216}
\{ 1\} \subseteq \mcd_A \subseteq \mcd_{\ao} = \{ 0,1\} \cup \Big( \log_3 2 + i \frac{2 \pi}{\log 3} \mbz \Big)
\end{equation}
and hence,
\begin{equation}\label{3.217}
\dim_{PC} A = \dim_{PC} (\ao) = \{ 1 \},
\end{equation}
where (in \eqref{3.216}) $D_{CS} = \log_3 2$ is the Minkowski dimension of the ternary Cantor set. In addition, it is conjectured (first in [\hyperlinkcite{Lap-vF2}{Lap-vF2--4}], via an approximate computation of the corresponding tube function, and then, with significantly more precise supporting arguments, in [\hyperlinkcite{LapRaZu1}{LapRaZu1,4}]) that
\begin{equation}\label{3.218}
\log_3 2 + i \frac{2 \pi}{\log 3} \mbz \subseteq \mcd_A
\end{equation}
or even that $\mcd_A = \mcd_{\ao}$ in \eqref{3.216}.

The Cantor graph RFD is Minkowski measurable; furthermore, it is fractal in dimension $d:= D_{CS} = \log_3 2$ and is fractal only in that dimension. Consequently, it is critically subfractal. The presence of the nonreal complex (and subcritical) dimensions $s_k = \log_3 2 + i (2 \pi/ \log 3)k$ (with $k \in \mbz \backslash \{0\}$) gives rise to logarithmically (or multiplicatively) periodic oscillations of order $D_{CS} = \log_3 2 < 1$.\footnote{More precisely, they correspond to the term $\varepsilon^{2-D_{CS}} G (\log_3 \varepsilon^{-1}$), which is exactly of order $\varepsilon^{2-D_{CS}}$ since the nonconstant periodic function $G$ in \eqref{3.193} is bounded away from $0$ and $\infty$.}
(See the fractal tube formula \eqref{3.192} along with \eqref{3.193}.) Therefore, it has a real effect on the geometry. According to the aforementioned conjecture (see \eqref{3.218} and the text following it), the devil's staircase (or Cantor graph) $A$ itself should have exactly the same qualitative and quantitative properties as the Cantor graph RFD $(\ao)$. In particular, it is {\em not} critically fractal but is critically fractal in dimension $d:= D_{CS} = \log_3 2$.\\

The following open problem is technically very challenging but also conceptually important, in light of the above discussion.

\begin{op}({\em Complex dimensions of the devil's staircase.})\label{OP:3.38.1/2}
Prove the conjecture stated in \eqref{3.218} or even its stronger form stated just afterward.
\end{op} 

As a first but important step towards that conjecture, try to show that there  are at most finitely many exceptions to the inclusion appearing in \eqref{3.218} or at least, that infinitely many elements of the left-hand side of the inclusion in \eqref{3.218} belong to $\mcd_A.$ \\

\subsubsection{Extended notion of complex dimensions, scaling laws and Riemann surfaces}\label{Sec:3.6.3}

We continue this discussion by broadening the notion of complex dimensions and correspondingly, of fractality. As was alluded to earlier, it is natural to call ``complex dimensions'' of an RFD $(\ao)$ not only the poles but also other types of (nonremovable) singularities of the associated fractal zeta function $\zao$ (or, equivalently, $\tzao$, provided $D:=\dim_B (\ao) < N$). Typically, the (necessarily closed) set of singularities is obtained as the closure of a countable (or finite) set of ``geometric singularities'', which we propose to call the {\em kernel} or simply, the {\em set of ``geometric singularities''} of $(\ao)$.\footnote{Accordingly, the complement of the kernel in the set of all (nonremovable) singularities could be referred to as the {\em set of ``analytic singularities''}. In practice, it can be ignored because it is not expected to contribute (in a significant way) to the fractal tube formula for $(\ao)$, although this remains to be proved in general.}
For example, for each member  $(\ao)$, say, of the family of (transcendentally $\infty$-quasiperiodic) maximally hyperfractal RFDs or compact sets constructed in [\hyperlinkcite{LapRaZu1}{LapRaZu1,3}] and discussed earlier, the kernel consists of a countably infinite set of geometric singularities (which are not poles of $\zao$ or of $\tzao$) and is dense in the whole critical line $\{ Re(s) =D \}$, where $D:= \overline{\dim_B} (\ao)$ can be prescribed a priori in the interval $(0,N)$. However, only the geometric singularities will contribute to the associated fractal tube formulas.\footnote{Simpler examples of this type can be obtained by considering the Cartesian product of two different self-similar (and lattice) Cantor sets with incommensurable oscillatory periods.}

Several other examples are provided in \cite{LapRaZu1}. In one of those examples, the kernel consists of a countable set of essential singularities and is shown to contribute to the fractal tube formulas on the same footing as mere poles of $\zao$. The remaining singularities do not contribute to the main term in the fractal tube formula; they may, however, contribute to the error term (the `noise').

It is then immediate to extend the above notion of fractality by saying that an object (say, an RFD $(\ao)$ in $\mbr^N$) is {\em fractal} if it has at least one nonreal (geometric) singularity. The notions of fractality in dimension $d \in \mbr$, as well as of critical and subcritical fractality, are similarly extended.

In the author's opinion, a deep understanding of the scaling laws in mathematics and physics (among many other fields, including cosmology, computer science, economics, chemistry and biology), as well as of many aspects of dynamics and of fractal, spectral and arithmetic geometry, could be gained by pursuing this venue and extending the theory of complex dimensions by merging it with aspects of the theory of Riemann surfaces\footnote{See, e.g., \cite{Ebe}, \cite{Schl} and beyond, in the spirit of Riemann's original broader intuition of Riemann surface.}
and using the notion of $h$-Minkowski content and related notions for a variety of (admissible) gauge functions. The beginning of such a theory is provided in a joint work under completion, \cite{LapRaZu10}.

As is suggested in \cite{Lap10}, in the long term, a potentially far-reaching further extension of the theory of complex dimensions to multivariable fractal zeta functions and their analytic varieties of singularities, merged with aspects of the theory of complex manifolds and sheaf theory (see, e.g., \cite{Ebe} and \cite{GunRos}), should be even more fruitful in this context; see also the end of \S \ref{Sec:4.4.2} (and of \S \ref{Sec:4}) below.\\

We close this subsection by discussing originally unexpected connections between hyperfractality and the spectra of fractal drums.\\

\subsubsection{Maximal hyperfractals and meromorphic extensions of spectral zeta functions of fractal drums}\label{Sec:3.6.4}

The error estimates obtained in \cite{Lap1} for the spectral asymptotics of fractal drums (i.e., drums with fractal boundary) are, in general, best possible. (See the paragraph following \eqref{2.50}, along with part ($a$) of Remark \ref{Rem:2.10}.) They also imply that the (normalized) spectral zeta function $\zeta_\nu$ associated with a given fractal drum\footnote{For the Dirichlet Laplacian with eigenvalue spectrum $(\lambda_j)_{j=1}^\infty$, where the eigenvalues are repeated according to their multiplicities, we let $\zeta_\nu (s) := \sum_{j=1}^\infty f_j^{-s}$, for all $s \in \mbc$ with $Re(s)$ sufficiently large, where for each $j \geq 1, f_j:= \sqrt{\lambda_j}$ is the $j$-th frequency of the fractal drum (written in nonincreasing order and repeated according to multiplicity). } 
can be meromorphically extended to the open right half-plane $\{ Re(s) > D \}$, where $D$ is the (upper) Minkowski dimension of the boundary of the fractal drum.\footnote{That is, $D:= \dim_B (\ptoo)$ in the notation of relative fractal drums.}

In \cite[\S 4.3.2]{LapRaZu1} and \cite{LapRaZu8}, it is shown  that the construction of (transcendentally $\infty$-quasiperiodic) maximal hyperfractal RFDs $(\ptoo)$, where $\Omega$ is a bounded open set in $\mbr^N$ (and hence, has compact boundary $\partial \Omega$), carried out in \cite[\S 4.6.1]{LapRaZu1} and [\hyperlinkcite{LapRaZu2}{LapRaZu2--4}] (as was briefly discussed earlier in this subsection), implies that the above right half-plane $\{ Re(s) > D \}$ is in general optimal (i.e., as large as possible among all open right half-planes to which $\zeta_\nu$ can be meromorphically continued). This is so for any possible value of the (upper) Minkowski dimension of the RFD $(\partial \Omega, \Omega)$; namely, for any $N \geq 1$ and for every $D \in (0, N).$

This last result establishes an interesting new connection between (maximal) hyperfractality and the vibrations of fractal drums. In particular, for each of the maximally hyperfractal drums or RFDs $(\ptoo)$ constructed in [\hyperlinkcite{LapRaZu1}{LapRaZu1,8}], the half-plane of meromorphic convergence of their spectral zeta function $\zeta_\nu$ coincides with the half-plane of (absolute) convergence of their distance zeta function $\zeta_{\ptoo}$ (or equivalently, since $D < N$ here, of their tube zeta function $\widetilde{\zeta}_{\ptoo}$).
 
\section{Epilogue: From Complex Fractal Dimensions to Quantized Number Theory and Fractal Cohomology}\label{Sec:4}

In [\hyperlinkcite{Lap-vF2}{Lap-vF2--4}], \cite{Lap-vF7} and \cite{Lap7} was proposed a search for a `fractal cohomology theory' that would be naturally associated with the theory of complex fractal dimensions (at the time, as developed for fractal strings and described in part in \S \ref{Sec:2}, but now also extended as in \cite{LapRaZu1} and [\hyperlinkcite{LapRaZu2}{LapRaZu2--10}] to the much broader higher-dimensional setting, as described in \S \ref{Sec:3}), as well as help unify at a deeper level several important aspects of fractal geometry, number theory and arithmetic geometry. (See \cite[\S 12.3 and \S 12.4]{Lap-vF4}, along with \cite[Chs. 4 and 5]{Lap7}.)

In particular, from this perspective (see \cite[\S 12.4]{Lap-vF4}), an analogy was developed between lattice self-similar geometries and finite-dimensional (algebraic) varieties over finite fields, while nonlattices self-similar geometries could be seen as a limiting case corresponding to infinite dimensional varieties (seemingly over a ``field of characteristic one''). Indeed, the zeros and the poles of the scaling zeta functions\footnote{By definition, the scaling zeta function of a self-similar string $\mcl$ coincides with the geometric zeta function of $\mcl$.}
of lattice self-similar strings are periodically distributed along finitely many vertical lines. The same is true of self-similar sprays (as well as, conjecturally, of self-similar sets) with ``nice generators''. (For the case of self-similar strings, see Example \ref{Ex:2.9}, and for the more general case of self-similar sprays, see \S \ref{Sec:3.4.10}.)\footnote{We leave aside here the `integer dimensions' since we use the scaling (rather than the fractal) zeta functions.}

\subsection{Analogy between self-similar geometries and varieties over finite fields}\label{Sec:4.1}

Recall that for a (smooth, projective, finite-dimensional, algebraic) variety $V$ over a finite field $\mbf_q$ (where $q = p^m$, with $m \in \mbn$ and the prime number $p$ being the underlying prime characteristic), the corresponding zeta function $\zeta_V$ is periodic with complex period $i{\bf p}$ (i.e., $\zeta_V (s) = \zeta_V (s + i{\bf p}$), for all $s \in \mbc$), where ${\bf p} := 2 \pi/ \log q = 2 \pi/ m\log p$. Therefore, as $q = \# \mbf_q \rightarrow \infty$ (e.g., if we successively consider $V$ over the finite field extensions $\mbf_{qn}$, with $n \in \mbn$ increasing to $\infty$), then the `oscillatory period' ${\bf p}$ tends to $0$.

Similarly, for a lattice string (or spray), the corresponding scaling zeta function $\zeta_\mfs$ is periodic with complex period $i{\bf p}$ (where ${\bf p} := 2 \pi/ \log r^{-1}$ is the oscillatory period and $r$ in $(0,1)$ is the multiplicative generator of the underlying scaling ratios), in the same sense as above.  Furthermore, in the approximation of a nonlattice string (or spray) by a sequence of lattice strings (or sprays) with increasing oscillatory periods ${\bf p}_n$ (as described in detail in \cite[Ch. 3, esp., \S\S 3.4--3.5]{Lap-vF4} and briefly discussed in Example \ref{Ex:2.9}), so that $\zeta_\mfs (s + iu{\bf p}_n)$ and $\zeta_\mfs (s)$ are close for $u \in \mbc$ with $|u|$ not too large), we have ${\bf p}_n \rightarrow \infty$. In this sense, nonlattice strings (or sprays) satisfy ${\bf p} = \infty$ (or ${\bf p} \rightarrow \infty$) and behave as though $r \rightarrow 1^+$ and hence, as though the ``underlying prime characteristic were equal (or tending) to 1'', as was mentioned above.\footnote{In the case of a self-similar spray with multiple generators, the scaling zeta function has both zeros and poles; in fact, it is of the form $\zeta_\mfs (s) = \sum_{k=1}^K g_k^s/ (1 -\sum_{j=1}^J r_j^s)$, where the positive numbers $g_k$ and scaling ratios $r_j$ are not necessarily assumed to be distinct. By definition, the lattice case then corresponds to the situation when the group generated by the {\em distinct values} of the $g_k$'s and $r_j$'s is of rank 1 (with generator denoted by $r$ and assumed to lie in $(0,1)$); the associated oscillatory period is then ${\bf p} := 2 \pi/\log r^{-1}$. }

Moreover, for a (finite-dimensional) variety $V$ over $\mbf_q$ (really, over the algebraic closure $\overline{\mbf}_q$ of $\mbf_q$), the total number of vertical lines, along which the zeros and the poles of $\zeta_V$ (counted according to their multiplicities) are distributed is equal to $2d +1$, where $d:=\dim V$ is the dimension of the variety. In addition, the zeros (respectively, poles) correspond to the even (respectively, odd) cohomology spaces.\footnote{In other words, the total cohomology space is naturally $\mbz_2$-graded, with $\dot{0}$ corresponding to the zeros and $\dot{1}$ corresponding to the poles, for this choice of grading.}

The aforementioned (Weil-type or \' etale) cohomology spaces played a key role in the proof of the Weil conjectures [\hyperlinkcite{Wei1}{Wei1--3}] (and, in particular, of the counterpart of the Riemann hypothesis) for curves over finite fields by A. Weil in {\em loc. cit.} and then, for higher-dimensional (but still finite-dimensional) varieties over finite fields, by P. Deligne in [\hyperlinkcite{Del1}{Del1--2}].\footnote{For a brief introduction to curves (or, more generally, varieties) over finite fields, as well as to the associated zeta functions and Weil conjectures (including RH in this context), we point out, for example, \cite{Dieu1, Katz, Oort, ParsSh1} and \cite[App. B]{Lap7}, along with the many references therein (or since), including \cite{Art}, \cite{Has}, \cite{Schm}, [\hyperlinkcite{Wei1}{Wei1--3}], [\hyperlinkcite{Gro1}{Gro1--4}], [\hyperlinkcite{Del1}{Del1--2}], [\hyperlinkcite{Den1}{Den1--6}], [\hyperlinkcite{Har1}{Har1--3}], [\hyperlinkcite{CobLap1}{CobLap1--2}] and \cite{Lap10}. Also, for relevant notions from algebraic geometry, we refer, e.g., to \cite{Hart}.}

In the process, a certain map (on the underlying variety $V$ and induced by the self-map $x \mapsto x^q$ of $\mbf_q$), called the {\em Frobenius morphism}, also plays a central role, via the linear endomorphism $F$ it induces on the total cohomology space and called the {\em Frobenius operator}.\footnote{By using  the periodicity of $\zeta_V$ (i.e., by making the changing of variable $z:= q^{-s}$), one formally obtains a reduced (total) cohomology space, which is a {\em finite-dimensional} vector space. In fact, in [\hyperlinkcite{Wei1}{Wei1--3}],  [\hyperlinkcite{Del1}{Del1--2}] or [\hyperlinkcite{Gro1}{Gro1--4}], only finite-dimensional vector spaces (over the underlying field) are considered. }
The (graded or alternating) `{\em characteristic polynomial}' of $F$ is then shown to coincide with the zeta function of $V$ (in the variable $z = q^{-s}$). Symbolically, and with  $Z_V (z) := \zeta_V (q^{-s})$, the {\em Weil zeta function} of the variety $V$, we have (ignoring multiplication by an unessential nowhere vanishing entire function)
\begin{equation}\label{4.1}
Z_V (z) = s\text{-det} (I -zF),
\end{equation}
where ``$s$-det'' stands for the graded  (or super, \cite{Del3, Wein}) determinant of $F$ over the (reduced) total cohomology space. It was later conjectured by C. Deninger in [\hyperlinkcite{Den1}{Den1--6}] (see, especially, [\hyperlinkcite{Den1}{Den1--3}]) that a similar procedure (but now involving, in general, possibly infinite dimensional cohomology spaces) could be used to deal with the Riemann zeta function $\zeta = \zeta (s)$ and other $L$-functions (in characteristic zero, for example, in the case of $\zeta$, over the field $\mbq$).

Since the early to mid-1990s, the author's intuition has been that there should exist a suitable notion of ``fractal cohomology''  such that the nonlattice case (for self-similar geometries) would be analogous to the situation expected to hold for $\zeta$ (the latter being a very special case, however, one typical of arithmetic geometries  and for which RH would hold). In \cite{Lap7}, building on [\hyperlinkcite{Lap-vF1}{Lap-vF1--3}] (see also \cite[\S 12.4]{Lap-vF4}), it was conjectured that a (generalized) Polya--Hilbert operator [having for spectrum (the reciprocals of) the zeros and poles (i.e., the reciprocal of the divisor) of the underlying fractal or arithmetic zeta function] should exist in this context so that the analog of \eqref{4.1} would hold and could be rigorously established. Partially realizing this dream has required significantly building on another semi-heuristic proposal made in  [\hyperlinkcite{Lap-vF3}{Lap-vF3--4}] (regarding a so-called ``spectral operator''), then the development in [\hyperlinkcite{HerLap1}{HerLap1--5}] of quantized number theory (in the ``real case''), followed by the development in  [\hyperlinkcite{CobLap1}{CobLap1--2}] of quantized number theory (in the ``complex case'') and the construction of a corresponding (generalized) Polya--Hilbert operator, along with the foundations of an associated fractal cohomology, in \cite{Lap10}. 

\subsection{Quantized number theory: The real case}\label{Sec:4.2}

We begin by presenting (in a very concise form) aspects of the first version of quantized number theory (in the ``real case''), as developed by Hafedh Herichi and the author in [\hyperlinkcite{HerLap1}{HerLap1--5}] and in \cite{Lap8}, based on a rigorous notion of the infinitesimal shift $\partial$ of the real line and of the corresponding spectral operator (as proposed heuristically in [\hyperlinkcite{Lap-vF3}{Lap-vF3--4}]; see, especially, \cite[\S 6.3]{Lap-vF4} and \cite[Ch. 4]{HerLap1}). In particular, the theory developed in [\hyperlinkcite{HerLap1}{HerLap1--5}] explains how to ``quantize'' the Riemann zeta function $\zeta = \zeta (s)$ in this context in order to view the spectral operator $\mfa$ (which sends the geometry of fractal strings onto their spectrum) as follows (see \cite[\S 7.2]{HerLap1}):
\begin{equation}\label{4.2}
\mfa = \zeta (\partial),
\end{equation}
where $\zeta (\partial)$ is interpreted in the sense of the functional calculus for the unbounded normal operator $\partial = \partial_c$ (with $\partial = d/dt$, the differentiation operator) acting on a suitable complex Hilbert space $\mbh_c := L^2 (\mbr, e^{-2ct} dt)$, the weighted $L^2$-space with respect to the (positive) weight function $w_c (t) := e^{-2ct}$, defined for all $t \in \mbr$. (See \cite[Ch. 5]{HerLap1} for the precise definition of $\partial$, including of its domain of definition, as well as for the proof of the normality of $\partial$: $\partial^* \partial = \partial \partial^*$ since it is shown in {\em loc. cit.} that the adjoint $\partial^*$ of $\partial$ is given by $\partial^* = 2c - \partial$.)

Here, the ``dimensional parameter'' $c \in \mbr$ is {\em fixed}, but enables us, in particular, to sweep out the entire critical strip $0 < Re(s) < 1$ (corresponding to the ``critical interval'' $0 < c < 1$) as well as the half-plane of absolute convergence for $\zeta = \zeta (s)$, namely, the open half-plane $Re(s) > 1$ corresponding to the half-line $c > 1$). 

The (strongly continuous)  semigroup of bounded linear operators of $\mbh_c$ generated by $\partial = \partial_c$ is given by $\{ e^{-h\partial} \}_{h \geq 0}$ and acts as the semigroup of translations (or shifts) of the real line (see \cite[\S 6.3]{HerLap1}):\footnote{For the general theory of strongly continuous semigroups of bounded linear operators and its applications, from different points of view, we refer, e.g., to \cite{Go, HilPh, JoLap, JoLapNie, Kat}, [\hyperlinkcite{ReSi1}{ReSi1--3}] and \cite{Ru2, Sc, Yo}.}
\begin{equation}\label{4.3.1/8}
(e^{-h\partial}) (f) (u) = f (u-h), \quad \text{for all } f \in \mbh_c,
\end{equation}
as well as for almost every $u \in \mbr$ and for every $h \geq 0$. In fact, formula \eqref{4.3.1/8} is also valid for $h \leq 0$ and hence, the {\em real infinitesimal shift} $\partial = \partial_c$ is the infinitesimal generator of the one-parameter group of operators $\{ e^{-h \partial} \}_{h \in \mbr}$.

We note that the semigroup $\{ e^{-h\partial} \}_{h \geq 0}$ is a contractive (respectively, expansive) semigroup for $c \geq 0$ (respectively, $c \leq 0$), while the group $\{ e^{-h\partial} \}_{h \in \mbr}$ is a group of isometries if and only if $c = 0$. Finally, since the adjoint of $\partial$ is given by $\partial^* = 2c - \partial$, as is shown in \cite[\S 5.3]{HerLap1}, the adjoint group  $\{ e^{-h\partial^{*}} \}_{h \in \mbr}$ coincides with $\{ e^{-2ch} e^{h \partial} \}_{h \in \mbr}$.

It is shown in \cite[\S 6.2]{HerLap1} that the spectrum, $\sigma (\partial)$, of the real infinitesimal shift $\partial = \partial_c$ is given by $L_c$, the vertical line going through $c \in \mbr$:\footnote{For the spectral theory of unbounded operators in various settings, we refer, e.g., to \cite{DunSch, ReSi1, JoLap, Kat, Ru2, Sc, Yo}. We note, however, that the references \cite{ReSi1} and \cite{Sc} are mostly limited to the spectral theory of unbounded self-adjoint (rather than normal) operators, as are many of the references focusing on the applications to quantum physics. }
\begin{equation}\label{4.3.1/4}
\sigma (\partial) = L_c := \{ Re(s) = c \}.
\end{equation}
Hence, $\partial$ is an unbounded (normal) operator, for any value of $c \in \mbr$. Further, in light of \eqref{4.2} and \eqref{4.3.1/4}, one deduces from a suitable version of the spectral mapping theorem for unbounded normal operators given in \cite[App. E]{HerLap1} that the spectrum, $\sigma (\mfa)$, of the spectral operator $\mfa = \zeta (\partial)$ coincides with the closure of the range of $\zeta = \zeta (s)$ along the vertical line $L_c$ (see \cite[\S 9.1]{HerLap1}):\footnote{If $c =1$, the unique (and simple) pole of $\zeta$, one must exclude $s =1$ from the line $L_c$ in interpreting \eqref{4.3.1/2}. Alternatively, one can consider the extended spectrum of $\mfa$, denoted by  $\widetilde{\sigma} (\mfa)$ and defined by $\widetilde{\sigma} (\mfa) := \sigma (\mfa)$ if $\mfa$ is bounded and $\widetilde{\sigma} (\mfa) := \sigma (\mfa) \cup \{ \infty \}$ otherwise, and view the meromorphic function $\zeta$ as a continuous  function from $\mbc$ to the Riemann sphere $\widetilde{\mbc} := \mbc \cup \{ \infty\}$. Then, the closure in \eqref{4.3.1/2} is interpreted in the compact space $\widetilde{\mbc}$ instead of in the unbounded complex plane $\mbc$. }
\begin{equation}\label{4.3.1/2}
\sigma (\mfa) = c \ell (L_c) = c \ell (\{ \zeta (s) : Re(s) = c \} ).
\end{equation}

It follows from \eqref{4.3.1/2} that, in light of the Bohr--Courant theorem \cite{BohCou}, itself now viewed as a mere consequence of the universality of $\zeta$ among all (suitable) analytic functions on the right critical strip $\{ 1/2 < Re(s) < 1 \} $ (as established in \cite{Vor}), that
\begin{equation}\label{4.3.3/4}
\sigma (\mfa) = \mbc, \quad \text{for every } c \in (1/2, 1).
\end{equation}
In particular, for any $c \in (1/2, 1)$, the spectral operator is {\em not} invertible  (in the usual sense of unbounded operators); see \cite[\S 9.2]{HerLap1}. 

Next, we examine what happens when $c > 1$ (which, as we recall and in light of the identity \eqref{4.3.1/2}, corresponds to the half-plane of absolute convergence for $\zeta = \zeta (s) = \sum_{n=1}^\infty n^{-s}$, namely, the open half-plane $\{ Re(s) > 1\}$); see \cite[Ch. 7]{HerLap1} for the full details.

For $c > 1$, the spectral operator $\mfa = \zeta (\partial)$ can be represented by a norm convergent {\em quantized} (or operator-valued) {\em Dirichlet series} and a {\em quantized} (also norm convergent) {\em Euler product}: 
\begin{align}\label{4.4}
\mfa = \zeta (\partial) &= \sum_{n=1}^\infty n^{-\partial} \notag \\
&= \prod_{p \in \mcp} (1 - p^{-\partial})^{-1},
\end{align}
where $\mcp$ denotes the set of prime numbers. It follows from \eqref{4.4} that for every $c > 1$, $\mfa$ is invertible (in the strong sense of $\mcb (\mbh_c)$, the space of bounded linear operators on $\mbh_c$), and that its (bounded) inverse $\mfa^{-1}$ is given by
\begin{equation}\label{4.4.1/2}
\mfa^{-1} = \sum_{n=1}^\infty \mu (n) n^{-\partial} = \prod_{p \in \mcp} (1-p^{-\partial}),
\end{equation}
where $\mu = \mu (n)$ is the M\" obius function, defined by $\mu (1) =1, \mu (n) =0$ if the integer $n \geq 2$ is not square-free, and $\mu (n) = k$ if $n \geq 2$ is the product of $k$ distinct primes. We stress that for $c >1$, all of the infinite series and infinite products in \eqref{4.4} and \eqref{4.4.1/2} are convergent in the operator norm (i.e., in the Banach algebra $\mcb (\mbh_c)$).

Moreover, for $0 < c < 1$ (i.e., ``within'' the critical strip $0 < Re(s) < 1$), the spectral operator $\mfa = \mfa_c$ is represented (when applied to a state function $f$ belonging to a suitable dense subspace of $\mbh_c$, which is also an operator core for $\partial = \partial_c$) by the same (but now weakly convergent) operator-valued Dirichlet series and Euler product as in \eqref{4.4}. A similar comment applies to $\mfa^{-1}$, which is then also a possibly unbounded operator when $0 < c < 1$.

Therefore, as was conjectured in \S 6.3.2 of \cite{Lap-vF3} and of \cite{Lap-vF4}, but now in a very precise sense, the spectral operator $\mfa = \zeta (\partial)$, which can be viewed as the {\em quantized Riemann zeta function}, has an operator-valued Euler product representation (as well as a convergent Dirichlet series, which was not conjectured to exist in \cite{Lap-vF3, Lap-vF4}) that is convergent (in a suitable sense) even in the critical strip $0 < Re(s) < 1$ (i.e., even when the dimensional parameter $c$ lies in the critical interval $(0,1)$); see \cite[\S 7.5]{HerLap1}.

One of the key results of \cite{HerLap1} is to provide various characterizations of the ``quasi-invertibility'' of the spectral operator $\mfa = \mfa_c$.\footnote{The possibly unbounded, normal operator $\mfa = \zeta (\partial)$ is said to be {\em quasi-invertible} if for every $T > 0$, the truncated spectral operator $\mfa^{(T)} := \zeta (\partial^{(T)})$ is invertible, where $\partial^{(T)} := \varphi^{(T)} (\partial)$ (defined via the functional calculus for unbounded normal operators) is the {\em truncated infinitesimal shift} and the function $\varphi^{(T)}$ is chosen so that $\sigma (\mfa^{(T)}) = [c - iT, c + iT]$; see \cite[\S 6.4]{HerLap1}. }
More specifically, the Riemann hypothesis (RH) is shown to be equivalent to the fact that for $c \in (0,1)$, the spectral operator $\mfa = \mfa_c$ is quasi-invertible if and only if $c \neq 1/2$, thereby providing an exact operator-theoretic counterpart of the reformulation of RH obtained in \cite{LapMa2} (and briefly discussed in \S \ref{Sec:2.6.2}) in terms of inverse spectral problems for fractal strings. (See \cite[Ch. 8, esp., \S 8.3]{HerLap1}.)

Another, seemingly very different, reformulation of the Riemann hypothesis, obtained by the author in \cite{Lap8} (and also discussed in \cite[\S 9.4]{HerLap1}) is the following statement: The Riemann hypothesis holds true if and only if the spectral operator $\mfa = \mfa_c$ is invertible (in the usual sense of unbounded operators) for every $c \in (0, 1/2)$.\footnote{By using a result in \cite{GarSteu}, this is shown to be equivalent to the non-universality of $\zeta = \zeta(s)$ in the left critical strip $\{ 0 < Re(s) < 1/2 \}$. On the other hand, in light of \eqref{4.3.1/2}, the universality of the Riemann zeta function $\zeta = \zeta (s)$ in the right critical strip $\{ 1/2 < Re(s) < 1 \}$ (due to S.~M. Voronin in \cite{Vor}  and extended by B. Bagchi and A. Reich in \cite{Bag} and [\hyperlinkcite{Rei1}{Rei1--2}]) implies that $\mfa = \mfa_c$ is {\em never} invertible for any $c \in (1/2, 1)$.} This result is referred to in \cite{Lap8} as {\em an asymmetric criterion for RH}.

Many other results are obtained in [\hyperlinkcite{HerLap1}{HerLap1--5}] and \cite{Lap8}, concerning, in particular, a quantized analog of the functional equation for $\zeta = \zeta(s)$ (or for its completion $\xi = \xi (s)$, see \cite[\S 7.6]{HerLap1} where the global spectral operator $\mca = \mca_c := \xi (\partial_c)$ is studied) and of the universality of the Riemann zeta function \cite{KarVo, Lau, Steu} (see \cite[Ch. 10, esp., \S 10.2]{HerLap1}),
as well as about the form of the inverse of the spectral operator, when it exists (see \cite{Lap8} and \cite[\S 9.4]{HerLap1}),
and the nature of the mathematical phase transitions at $c = 1/2$ and at $c =1$ concerning the shape of the spectrum, the quasi-invertibility, the invertibility and the boundedness of $\mfa = \mfa_c$ (see \cite[\S 9.3]{HerLap1}).

We note that the mathematical phase transitions occurring in \cite{HerLap1} at $c=1$ are analogous to (but different from) the one studied from an operator-algebraic point of view by J.-B. Bost  and  A. Connes in [\hyperlinkcite{BosCon1}{BosCon1--2}] and \cite[\S V.II]{Con1} (because the latter also corresponds to the pole at $s =1$ of $\zeta = \zeta(s)$), whereas the phase transitions occurring in \cite{HerLap1} and in \cite{Lap8} in the midfractal  case when $c = 1/2$ are of a very different nature. They were expected to occur in the author's early conjecture (and open problem) formulated in \cite[Quest. 2.6, p. 14]{Lap3} about the existence of a notion of complex fractal dimensions that would enable us to interpret the Riemann hypothesis as a phase transition in the midfractal case, in the sense of Wilson's work \cite{Wils} on critical phenomena in condensed matter physics and quantum field theory. 

For a full exposition of these results, we refer the interested reader to the book \cite{HerLap1}. We note that the above results could likely be extended to a large class of meromorphic functions (instead of just the Riemann zeta function), and especially of the arithmetic $L$-functions (\cite{Sarn}, [\hyperlinkcite{Murt1}{Murt1--3}], \cite[App. C]{Lap7}) for which the counterpart of the Riemann hypothesis or of the universality theorem is expected to hold. In particular, conjecturally, the results concerning RH should have an appropriate counterpart for all the elements of the Selberg class (\cite{Sel1, Sarn}, [\hyperlinkcite{Murt1}{Murt1--3}] and \cite[App. E]{Lap7}).

\subsection{Quantized number theory: The complex case}\label{Sec:4.3}

We continue this epilogue by providing an extremely brief overview of the results of Tim Cobler and the author on ``quantized number theory'' (in the ``complex case'') obtained in [\hyperlinkcite{CobLap1}{CobLap1--2}] and pursued in the book in preparation \cite{Lap10}, in which the real infinitesimal shift $\partial = d/dt$ acting on $\mbh_c = L^2 (\mbr, e^{-2ct} dt)$ (with $c \in \mbr$ and discussed in \S \ref{Sec:4.2}) is replaced by the complex infinitesimal shift $\partial = d/ dz$, now acting on a suitable weighted Bergman space $\mch$ of entire functions (as introduced originally and for completely different  purposes by A. Atzmon and B. Brive in \cite{AtzBri}). Namely, $\mch$ consists of the entire functions belonging to the complex weighted Hilbert space $L^2 (\mbc, e^{-|z|^\alpha} dz$), for some fixed, but otherwise unimportant, parameter $\alpha \in (0,1)$.\footnote{For the theory of (not necessarily weighted) Bergman spaces, see, e.g., \cite{DureSchu} and \cite{HedKonZhu}. }

For these values of $\alpha$, the operator $\partial$ is bounded (but not normal) and its spectrum, $\sigma (\partial)$, is given by a single point, namely, the origin:\footnote{If we were to allow the value $\alpha =1$, then in \eqref{4.4.4/5}, $\sigma (\partial)$, the spectrum of $\partial$, would be the closed unit disk (with center the origin) in $\mbc$ instead of being reduced to $\{0\}$. }
\begin{equation}\label{4.4.4/5}
\sigma (\partial) = \{ 0\}.
\end{equation}
From this, it follows that in the notation introduced below (namely, with $\partial_\tau := \partial + \tau$), we have that $\sigma (\partial_\tau) = \{ \tau \}$, for every $\tau \in \mbc$. We note that $0$ is an eigenvalue of $\partial$ (in fact, its only eigenvalue) and that the associated eigenfunction is the (possibly suitably normalized) constant function $1$, the `vacuum state'. With $0$ replaced by $\tau \in \mbc$, the same statement holds for the shifted differentiation operator $\partial_\tau$. 

The {\em complex infinitesimal shift} $\partial$ is the infinitesimal generator of the group of translations (or shifts) on the complex plane, $\{ e^{-z \partial} \}_{z \in \mbc}$, acting on every $\psi \in \mch$ as follows (for almost every $u \in \mbc$ and every $z \in \mbc$):
\begin{equation}\label{4.9.1/2}
(e^{-z \partial}) (\psi)(u) = \psi (u-z).
\end{equation}

In \cite{CobLap1} are rigorously constructed ``{\em generalized Polya--Hilbert operators}'' (GPOs, in short), in a sense close to that of \cite{Lap7}, and infinite dimensional regularized determinants of the restrictions to their eigenspaces which enable one, in particular, to recover (under appropriate hypotheses) the corresponding meromorphic functions (or zeta functions) as suitable (graded or alternating) ``characteristic polynomials'' of the restrictions of the GPOs to their total eigenspaces, in the spirit of \S \ref{Sec:4.1} (but well beyond).

More specifically, by working with the family of translates $\{ \partial_\tau\}_{\tau \in \mbc}$ of the complex infinitesimal shift $\partial = d/dz$ (i.e., $\partial_\tau := \partial + \tau$, for every $\tau \in \mbc$) and the  associated Bergman spaces, one can then construct (via orthogonal direct sums) a kind of `universal Polya--Hilbert operator' (GPO), $\mbd$, acting on a typically countably infinite direct sum $\mch$ of Bergman spaces, whose spectrum (when it is discrete) consists only of (isolated) eigenvalues with finite multiplicities and coincides with any prescribed discrete subset of $\mbc$, for example, a (multi)set of complex dimensions or (the reciprocal of) the divisor [viewed as a $\mbz_2$-graded (multi)set of zeros and poles] of a suitable meromorphic function, such as (an appropriate version of) the global Riemann zeta function.

We can then write the given meromorphic function  $g = g(s)$ (assumed to be a quotient of two entire functions of finite orders) as a (typically infinite dimensional) regularized $\mbz_2$-graded (or supersymmetric) determinant:
\begin{equation}\label{4.5}
g(s) = s\text{-det} (I - s \mcf),
\end{equation}  
where $\mcf = \mcf_g$ is the restriction of the GPO $\mbd$ to its total eigenspace; so that  
\begin{equation}\label{4.10.1/2}
\sigma (\mbd) = \sigma (\mcf) = \mfD (g)^{-1},
\end{equation}
the  reciprocal of the divisor of $g$, in the sense of Definition \ref{Def:3.15}, and the (graded) spectra of $\mbd$ and $\mcf$ are discrete and consist only of eigenvalues, which are the reciprocals of the zeros and the poles of $g =g(s)$ repeated according to their multiplicities.

The operator $\mcf = \mcf_g$ occurring in \eqref{4.5} and \eqref{4.10.1/2} and defined as the restriction of the GPO $\mbd = \mbd_g$ to its total eigenspace is called (in \cite{Lap10}) the {\em generalized Frobenius operator} (GFO, in short) associated with $g$. Indeed, in the special case of varieties over finite fields $\mbf_q$ discussed in \S \ref{Sec:4.1}, it is clearly a version (in the $s$-variable) of the usual Frobenius operator $F$ (defined by means of the $z$-variable, with $z := q^{-s}$) induced by the Frobenius morphism on the variety and acting on the underlying total cohomology space (which also coincides with the total eigenspace of $F$).

We stress that here and thereafter, $\mbd = \mbd_g$ is the GPO associated with the discrete (and $\mbz_2$-graded) multiset $\mfD (g)^{-1}$ consisting of $0$ (if $0$ is either a zero or a pole of $g$) and of the reciprocals of the (nonzero) zeros and poles of $g$. Under the aforementioned conditions on $g$ (which are relaxed in \cite{CobLap2}), this choice of the GPO $\mbd_g$ guarantees the compactness (and the normality) of the GFO $\mcf_g$ and even, that $\mcf_g$ belongs to some Schatten class (i.e., the sequence of characteristic eigenvalues of $\mcf_g$ belongs to $\ell^n = \ell^n (\mbc)$, for some integer $n \geq 1$); so that the regularized determinant in \eqref{4.5} is well defined (and is of order $n$). 

The regularized superdeterminant (or Berezinian, in the terminology of supersymmetric quantum field theory; see, e.g., \cite{Del3, Wein}) arising in \eqref{4.5} above can be written as a quotient of regularized determinants which are natural generalizations of the well-known (infinite dimensional) Fredholm determinants of trace class compact operators \cite{Fred}. For a detailed exposition and for the genesis of the theory of these generalized Fredholm determinants, we refer to  [\hyperlinkcite{Sim1}{Sim1--3}]; see also [\hyperlinkcite{CobLap1}{CobLap1--2}] and \cite{Lap10}. We simply mention that it relies in part on the work in \cite{Fred}, \cite{Plem}, \cite{Poin5}, \cite{Smit}, \cite{Lids}, \cite{GohKre} and  [\hyperlinkcite{Sim1}{Sim1--3}].

In the special case of the zeta function of a variety $V$ over a finite field $\mbf_q$ and after having made the change of variable $z := q^{-s}$, we obtain a rational function $Z_V = Z_V (z)$ which can be written as an alternating product of (finite-dimensional) determinants, as in \eqref{4.1} in the setting of the Weil conjectures.

Furthermore, in the other very important special case when $g$ is an appropriate version  of the completed (or global) Riemann zeta function, for example, $g(s) = \xi (s) := \pi^{-s/2} \Gamma (s/2) \zeta (s)$ or $g (s) = \Xi (s) := (s-1) \xi (s)$, the corresponding regularized determinant is then truly infinite dimensional  since $\zeta = \zeta (s)$ (and hence also $\xi = \xi (s)$ as well as $\Xi = \Xi (s)$) has infinitely many (critical) zeros . Moreover, the regularized determinant corresponding to the zeros is not a mere Fredholm operator \cite{Fred} but involves a renormalization of second order (that is, in the hierarchy of such regularized determinants indexed by integers $n \geq 1$, as in \cite{Sim2}, it is of level $n=2$).

These results provide, in particular, a partial but completely rigorous and quite general mathematical realization of what was sought for by a number of authors, including A. Weil [\hyperlinkcite{Wei1}{Wei1--6}], A Grothendieck [\hyperlinkcite{Gro1}{Gro1--4}], P. Deligne [\hyperlinkcite{Del1}{Del1--2}], and especially, in the present context, C. Deninger [\hyperlinkcite{Den1}{Den1--6}]. Naturally, we should point out that they do not provide a full realization of those authors and their successors' rich set of ideas and conjectures, especially concerning the existence of suitable cohomology theories ([\hyperlinkcite{Wei1}{Wei1--3}], [\hyperlinkcite{Gro1}{Gro1--4}], [\hyperlinkcite{Den1}{Den1--6}] and, e.g., [\hyperlinkcite{Tha1}{Tha1--2}], along with the relevant references therein) and of an appropriate positivity condition ([\hyperlinkcite{Wei4}{Wei4--6}], [\hyperlinkcite{Har1}{Har1--3}], \cite{Con2}), and therefore cannot (for now) be used to try to prove the Riemann hypothesis, although they may constitute a significant first step, particularly once (and if) one can obtain an appropriate geometric and topological (i.e., cohomological) interpretation (which may, however, be very difficult and take a very long time to achieve; see \S \ref{Sec:4.4}). 

We also mention that the search for suitable Polya--Hilbert operators has been the object of a number of works, both in physics (see, e.g., [\hyperlinkcite{Berr3}{Berr3--4}] and \cite{BerrKea}) and in mathematics (see, e.g., [\hyperlinkcite{Den1}{Den1--5}], \cite{Con2} and \cite[Chs. 4 \& 5]{Lap7}). The abstract (functional analytic) framework for GPOs provided in  [\hyperlinkcite{CobLap1}{CobLap1--2}] and extended in \cite{Lap10} is, in some ways, too general for certain purposes (for example, for proving RH, at least in our current state of knowledge) but presents the advantage of being completely categorical and widely applicable, well beyond the settings of fractal geometry and number theory  which originally motivated it.

The proof of the identity \eqref{4.5} given in \cite{CobLap1} makes use of Hadamard's factorization theorem for entire functions of finite order  ([\hyperlinkcite{Had1}{Had1--2}], \cite{Conw}). 

As is shown in \cite{Lap10}, the above formalism (and, in particular, the identity \eqref{4.5}) can also be applied to all the elements of the Selberg class (\cite{Sel2, Sarn}, [\hyperlinkcite{Murt1}{Murt1--3}] and \cite[App. E]{Lap7}) and hence, at last conjecturally, to all of the arithmetic zeta function or $L$-functions (more specifically, the automorphic $L$-functions) for which the extended Riemann hypothesis (ERH) is expected to hold (see, e.g., \cite{Sarn}, [\hyperlinkcite{Murt1}{Murt1--3}, \hyperlinkcite{ParsSh1}{ParsSh1--2}] and \cite[App. C]{Lap7}). In this case, one also needs to use regularized determinants of order two, as for the various versions of the completed Riemann zeta function discussed above.   

\subsection{Towards a fractal cohomology}\label{Sec:4.4}

Here, we briefly discuss the emerging theory of `fractal cohomology', as expounded upon in \cite{Lap10}. This theory builds on the work from \cite{CobLap1} (and from \cite{CobLap2}) described in \S \ref{Sec:4.3}, itself inspired in part by the theory from \cite{HerLap1} described in \S \ref{Sec:4.2}, as well as on the many references provided in \S \ref{Sec:4.1}, including [\hyperlinkcite{Wei1}{Wei1--3}], \cite{Gro1}, \cite{Den3} and [\hyperlinkcite{Lap-vF2}{Lap-vF2--4}].

The idea underlying the notion of a fractal cohomology is that to every complex dimension (now understood in the extended sense of a zero or a pole of the given fractal zeta function or arithmetic zeta function or else, more generally, of a suitable meromorphic function $g = g(s)$),\footnote{Note that the zeros and poles of the meromorphic function $g=g(s)$, can be viewed as the (necessarily simple) poles of (minus) the logarithmic derivative of $g$, the function $-g'/g$, with the associated residues being equal to the orders of the zeros or of the poles (and with corresponding signs identifying whether the poles of $-g'/g$ comes from a zero or a pole of $g$).} one can associate a finite-dimensional complex Hilbert space $H^\pm (\omega)$, with the plus (respectively, minus) sign corresponding to $\omega$ being a zero (respectively, a pole) of dimension (over $\mbc$) equal to the multiplicity of $\omega$.

More generally, to (the reciprocal $\mfD^{-1} (g)$ of) the divisor $\mfD (g)$ of the meromorphic function $g$, we associate a $\mbz_2$-graded (or supersymmetric) complex (and separable) Hilbert space\footnote{Here and thereafter, for notational simplicity, we use the same symbol $\oplus_s$ to indicate the $\mbz_2$-graded (or supersymmetric) direct sum of Hilbert spaces or of operators acting on them, as well as of (multi)sets, depending on the context.}
\begin{equation}\label{4.6}
H = H^+ \oplus_s H^-,
\end{equation}
where
\begin{equation}\label{4.7}
H^+ := \underset{\omega \in \mcz^+}{\oplus} H^+ (\omega) \quad \text{and} \quad H^- := \underset{\omega \in \mcz^-}{\oplus} H^- (\omega),  
\end{equation}
with $\mcz^+$ (respectively, $\mcz^-$) denoting the set of the reciprocals of the distinct zeros (respectively) poles of $g$ (with the convention according to which one does not have to take the reciprocal of 0, if this value occurs in either $\mcz^+$ or $\mcz^{-}$). The finite-dimensional (complex) Hilbert space $H^\pm (\omega)$ has dimension equal to the multiplicity of the corresponding pole or zero $\omega$ of $g = g(s)$.

As was explained in \S \ref{Sec:4.3}, $H$ is the eigenspace of the {\em generalized Polya--Hilbert space} (GPO) $\mbd = \mbd_g =\mbd_\mcz$, with $\mcz := \mcz^+ \oplus_s \mcz^-$ and $\mbd = \mbd^+ \oplus_s \mbd^-$. Furthermore, $\mcf = \mcf^+ \oplus_s \mcf^-$, defined as the restriction of the GPO $\mbd = \mbd^+ \oplus_s \mbd^-$ to its (total) graded eigenspace $H = H^+ \oplus_s H^-$ (as given by \eqref{4.6} and \eqref{4.7}). The (possibly unbounded) operator $\mcf$ is called the {\em generalized Frobenius operator} (GFO) associated with $g$ (or with $\mcz$). Under suitable assumptions on $g$ (specified in [\hyperlinkcite{CobLap1}{CobLap1--2}] and in \cite{Lap10}),\footnote{More specifically, in light of \cite{CobLap1}, one assumes that the meromorphic function $g = f/h$ is the ratio of two entire functions $f$ and $h$ of finite orders (which is sufficient for all of the potential applications to number theory, as well as for many of the potential applications to fractal geometry), while according to \cite{CobLap2}, one should be able to remove the hypothesis that the entire functions $f$ and $h$ are of finite orders.} the GFO $\mcf$ enables us to recover the meromorphic function $g$ via the determinant formula \eqref{4.5} showing that $g = g(s)$ is the ``characteristic polynomial'' of $\mcf$ (viewed as a supersymmetric regularized determinant).

\begin{remark}\label{Rem:4.1}
By construction, the multiplicity of $\omega \in \mcz^+$ (respectively, $\omega \in \mcz^-$), with $\omega \neq 0$, as an eigenvalue of the GPO $\mbd^+$ (respectively, $\mbd^-$), or equivalently, of the GFO $\mcf^+$ (respectively, $\mcf^-$), coincides with the multiplicity of the corresponding zero (respectively, pole) $\omega^{-1}$ of the meromorphic function $g = g(s)$. Hence, the finite-dimensional complex Hilbert space $H(\omega)$ has dimension over $\mbc$ equal to the multiplicity of the corresponding zero (respectively, pole) $\omega^{-1}$ of $g=g(s)$.\footnote{Of course, the same is true if $\omega = 0$ provided we replace $\omega^{-1}$ by $\omega$, in the latter statement.} Consequently, since finite-dimensional Hilbert spaces of the same dimension are isomorphic, the (total) cohomology space $H$ in \eqref{4.6} could be equivalently defined by letting $\omega$ run through the divisor $\mfD (g)$ of $g$ (rather than through its reciprocal $\mfD^{-1} (g)$, as indicated in \eqref{4.7}): symbolically,
\begin{equation}\label{4.12.1/4}
H = \underset{\omega \in \mfD (g)}{\oplus} H(\omega).
\end{equation}
Accordingly, in \eqref{4.7}, the ``even'' (respectively, ``odd'') cohomology space $H^+$ (respectively, $H^-$) can be defined by replacing $\mcz^+$ (respectively, $\mcz^-$) by the set of {\em distinct poles} (respectively, of {\em distinct zeros}) of $g=g(s)$: symbolically,
\begin{equation}\label{4.12.1/2}
H^+ = \underset{\omega \in \{ \text{zeros of } g \}}{\oplus \ H^+ (\omega)} \quad \text{and} \quad H^- = \underset{\omega \in \{ \text{poles of } g \}}{\oplus \ H^- (\omega)}.
\end{equation}
\end{remark}

When $\mcz$ is infinite (that is, when $g$ has at least infinitely many zeros or poles), then the {\em total cohomology space} $H = H_g =H_\mcz$ is infinite dimensional, and vice versa. Furthermore, since $\mcz$ is at most countable, the cohomology space $H$ is a separable (complex) Hilbert space. This is the case, for example, when $g = g(s)$ is one of the completed Riemann zeta functions, $\xi = \xi (s)$ or $\Xi = \Xi (s)$. Therefore, the (total) cohomology space $H_\xi$ or $H_\Xi$ is {\em intrinsically infinite dimensional} and cannot be reduced to a finite-dimensional one.

On the other hand, in the case of the zeta function attached to a variety $V$ over a finite field $\mbf_q$ (as in \S \ref{Sec:4.1}), the periodicity of the zeta function combined with the change of variable $z := q^{-s}$ yields a rational function $Z_V = Z_V (z)$ to which is associated a finite-dimensional cohomology space (corresponding to the zeros and the poles of the rational function, counted according to their multiplicities). {\em Hence, in this case, the total cohomology space is a priori infinite  dimensional but can also naturally be reduced to a finite-dimensional one.}

We thus obtain a cohomology theory which seems to satisfy several (if not all) of the main properties expected by Grothendieck (in [\hyperlinkcite{Gro1}{Gro1--4}]) and by Deninger (in [\hyperlinkcite{Den1}{Den1--6}]), for example. (See also, e.g., [\hyperlinkcite{Tha1}{Tha1--2}] for an exposition of some of those ideas.) Let us further explain this statement, but without going into the details. The key is that this ``fractal cohomology theory'' satisfies the counterpart (in the present context) of the {\em K\" unneth formula} (for the product of two algebraic varieties or of two differentiable manifolds; see, e.g.,  [\hyperlinkcite{Dieu1}{Dieu1,2}] and \cite{MacL}). This is, in fact, a key motivation for the author's conjecture (Conjecture \ref{Con:3.14}) made in \S \ref{Sec:3.4} about the set of complex dimensions (really, the divisor of the associated fractal zeta function) of the Cartesian product of two bounded sets and, more generally, of two RFDs (in $\mbr^N$). If correct, this conjecture would help provide a geometric interpretation of the K\" unneth formula of fractal cohomology theory in terms of Cartesian products (at least in the geometric setting of bounded sets and RFDs in $\mbr^N$).

A suitable version of Poincar\' e duality (or of an appropriate extension thereof) in the context of fractal cohomology still remains to be found; see, e.g., [\hyperlinkcite{Poin1}{Poin1--4}], [\hyperlinkcite{Dieu1}{Dieu1,2}] and \cite{MacL} for the classic notions of Poincar\' e duality. A natural extension of that notion could clearly be formulated, however, in terms of the fractal cohomology spaces associated with $g(s)$ and with $g(1-s)$.

It is noteworthy that Grothendieck's beautiful dream and elusive notion of a ``motive'' is in fact intimately connected with the search for a ``universal cohomology theory'' satisfying several axioms, including especially, the analog of K\" unneth's formula and of Poincar\' e duality. (See, e.g., \cite{Kah}, \cite{Tha1} and \cite{Tha2}.)\footnote{See also \cite{Cart} for an interesting discussion of the evolution of the notion of a `geometric space'.}
The emerging fractal cohomology theory proposed in \cite{Lap10} (and building, in particular, on \cite{Lap-vF4}, \cite{Lap7} and [\hyperlinkcite{CobLap1}{CobLap1,2}]) is evolving in that spirit. It associates (in a functorial way) to a suitable meromorphic function $g$ (or to its divisor $\mfD (g)$) a cohomology space $H_g$, which is the eigenspace of the generalized Frobenius operator (GFO) $\mcf = \mcf_g$. In addition, the identity \eqref{4.5} which enables one (under appropriate hypotheses) to recover the meromorphic  function from the action of $\mcf$ on $H$ can be viewed as naturally providing an associated inverse functor in this context. \\

\subsubsection{Grading of the fractal cohomology by the real parts}\label{Sec:4.4.1}

It is natural to wonder why, unlike in standard algebraic or differential toplogy (see, e.g., \cite{Dieu2} and \cite{MacL}), the fractal cohomology spaces are no longer (in our context) indexed by integers. In fact, a natural grading of the total fractal cohomology space $H= H^+ \oplus_s H^-$ is provided by the real parts (of the reciprocals) of the zeros and poles of the meromorphic function $g$.\footnote{If $0$ is either a zero or a pole of $g$, then $H^+$ or $H^-$ is also graded by zero. \label{Fn:165}}
Hence, in general, $H$ is graded by real numbers and not just by (nonnegative) integers.\footnote{In light of Remark \ref{Rem:4.1}, we could replace throughout this discussion the real parts of the reciprocals of the zeros and of the poles by the zeros and the poles themselves; see also Remark \ref{Rem:4.2}.}

More specifically, let $\mcz = \mcz^+ \oplus_s \mcz^-$, where $\mcz^+$ and $\mcz^-$ denote, respectively, the reciprocals of the distinct zeros and poles of $g$ (with the same convention concerning $0$ as usual). Furthermore, let $R^\pm$ denote the (at most countable) set of real parts of the elements of $\mcz^\pm$. Moreover, for $\alpha \in R^\pm$, let 
\begin{equation}\label{4.6.5}
\mcz_\alpha^\pm := \{ \omega \in \mcz^\pm : Re(\omega) = \alpha \}
\end{equation}
and
\begin{equation}\label{4.7.5}
H_\alpha^\pm := \underset{\omega \in \mcz_\alpha^\pm}{\oplus} H(\omega),
\end{equation}
where, as earlier, $H(\omega)$ denotes the eigenspace corresponding to the eigenvalue $\omega$ of $\mcf_g$ (and also of $\mbd_g$); so that the dimension (over $\mbc$) of the finite-dimensional Hilbert space $H^\pm (\omega)$ is equal to the multiplicity of $\omega$ as an eigenvalue of $\mcf_g$ (or, equivalently, if $\omega \neq 0$, as the reciprocal of  a zero or a pole of $g$; if $\omega = 0$, then we can omit the word ``reciprocal'' here).

Then, we can now provide the following precise formulation of the statement according to which the (total) fractal cohomology space is graded by the real parts, of the elements of $\mfD (g)^{-1}$:
\begin{equation}\label{4.8}
H^\pm = \underset{\alpha \in R^\pm}{\oplus} H_\alpha^\pm,
\end{equation} 
with $R^\pm$ defined just before \eqref{4.6.5} and $H_\alpha^\pm$ defined by \eqref{4.7.5} for all $\alpha \in R^\pm$.

\begin{remark}\label{Rem:4.2}
For notational simplicity, we have worked above with the grading provided by the real parts of the elements of $\mcz^+$ and $\mcz^-$. However, in light of Remark \ref{Rem:4.1}, we could just as well replace $\mcz^+$ and $\mcz^-$ by the sets of distinct zeros and poles of $g = g(s)$, respectively. And similarly, for the definition of $R^\pm$, $\mcz_\alpha^\pm$, $H_\alpha^\pm$, $H$ and $H^\pm$ in (or around) \eqref{4.6}, \eqref{4.7} and \eqref{4.6.5}--\eqref{4.8}. Consequently, in that setting, the fractal cohomology spaces $H$, $H^+$ and $H^-$ would actually be indexed by the real parts of the elements of the divisor $\mfD (g)$, the zero set and the pole set of the meromorphic function $g = g(s)$, respectively. Accordingly, the corresponding grading index sets $R^+$ and $R^-$ would instead consist of the real parts of the distinct zeros and of the distinct poles of $g = g(s)$, respectively.
\end{remark}

We note that $H^\pm$ is graded by the (possibly infinite, \cite{BruDKPW}) matroid (\cite{Oxl, Wel}) $R^\pm$, of rank (possibly infinite and defined as in \cite{BruDKPW}) equal to the number of vertical lines in $\mcz^\pm$ (or, equivalently, to the cardinality of $R^\pm$, viewed either as the set of distinct zeros/poles of $g = g(s)$).

For example, for a $d$-dimensional variety $V$ over a finite field $\mbf_q$, $R^-$ (respectively, $R^+$) consists of the integers $j$ (respectively, the half-integers $j/2$), where $j \in \{0,1, \cdots, d\}$. This is the case relative to the original $s$-variable. On the other hand, if $z:= q^{-s}$, then the corresponding sets $R^-$ and $R^+$ are $\{ q^{-j} : j=0, 1, \cdots, d \}$ and $\{q^{-j/2} : j=0, 1, \cdots, d\}$, respectively. Thus, in light of the analog of the Riemann hypothesis proved by Deligne in [\hyperlinkcite{Del1}{Del1--2}] in the general case when $d \geq 1$ (and by Weil in [\hyperlinkcite{Wei1}{Wei1--3}] when $d=1$), the determinant formula \eqref{4.5} yields (still in the $z$-variable)
\begin{equation}\label{4.9}
Z_V (z) = \frac{Q_1 (z) \cdots Q_{2d-1} (z)}{Q_0 (z) \cdots Q_{2d} (z)},
\end{equation}
where for $k \in \{0, \cdots, d\}$, the polynomial $Q_{2k}$ is the characteristic polynomial of the linear operator ${\bf F}_{2k} := \mcf_{q-k}^-$ (the restriction of the GFO $\mcf$ to $\bfh_{2k} := H_{q-k}^-$) and for $k \in \{1, \cdots, d\}$, the polynomial $Q_{2k-1}$ is the characteristic polynomial of the linear operator ${\bf F}_{2k-1} := \mcf_{q^{-k/2}}^+$ (the restriction of the GFO $\mcf$ to $\bfh_{2k-1} := H_{q^{-k/2}}^+$). With this notation, \eqref{4.8} becomes
\begin{equation}\label{4.10}
H^+ = \vc{\oplus}{d}{k=1} \bfh_{2k-1} \quad \text{and} \quad H^- = \vc{\oplus}{d}{k=0} \bfh_{2k},
\end{equation}
while the total cohomology space $H$ is given by 
\begin{equation}\label{4.11}
H = H^+ \oplus_s H^- = \vc{\oplus}{2d}{j=0} \bfh_j. 
\end{equation}

In the case of the completed Riemann zeta function $\xi (s)$, {\em and assuming the truth of the classic Riemann hypothesis}, we have that $R^+ := \{1/2\}$ and $R^-:= \{0,1\},$ while $\mcz^- = \{ 0,1\}$ and
\begin{equation}\label{4.12}
\mcz^+ := \{ \text{reciprocals of the critical zeros of } \zeta (s) \}.
\end{equation}
Hence, with the notation\footnote{We are now using the grading of the total cohomology space by the real parts of the zeros and the poles of $\xi = \xi (s)$, as in Remark \ref{Rem:4.2}.}
\begin{equation}\label{4.13}
\bfh_0 := H_0^-, \quad \bfh_2 :=H_1^-, \quad \text{and} \quad \bfh_1 := H_{1/2}^+,
\end{equation}
while $\bm{\mcf}_j$ is defined as the restriction of the generalized Frobenius operator $\mcf$ to $\bfh_j$, for $j \in \{0,1,2\}$. Then, \eqref{4.5} becomes (neglecting an unimportant factor):\footnote{This factor is a nonvanishing entire function involving the number $\pi$ and Euler's constant $\gamma$; see \cite{CobLap1} and \cite{Lap10} for its specific value.}
\begin{equation}\label{4.14}
\xi (s) =\frac{\text{det} (I-s \bm{\mcf}_1)}{\text{det}(I - s \bm{\mcf}_0) \text{det} (I -s \bm{\mcf}_2)},
\end{equation}
with det$(I - s \bm{\mcf}_0) = s$ and det$(I - s \bm{\mcf}_2) = s-1$ corresponding to the pole of $\xi = \xi (s)$ at $s =0$ and at $s=1$, respectively.

Note that with the above notation (which is inspired by the one used for varieties over finite fields as well as by what is expected for Weil-type cohomologies in the present context as well as for more general arithmetic $L$-functions),\footnote{See, e.g., [\hyperlinkcite{Gro1}{Gro1--3}], \cite{Dieu1}, \cite{Den3}, \cite{Con2}, \cite[esp., Ch. 4 \& App. B]{Lap7}, \cite{Kah}, [\hyperlinkcite{Tha1}{Tha1--2}] and the relevant references therein. }
we have (in light of \eqref{4.13})
\begin{equation}\label{4.14.5}
H^+ = \bfh_1 \quad \text{and} \quad H^- = \bfh_0 \oplus \bfh_2.
\end{equation}

The classic Riemann hypothesis (for $\zeta$ or, equivalently, for $\xi$) is equivalent to the fact that the total even cohomology space of $\xi = \xi (s)$ (namely, $H^+ = H_\xi^+ = \bfh_1$) is `{\em monofractal}'; i.e., with the notation of Remark \ref{Rem:4.2}, RH is equivalent to the fact that $R^+$ consists of a single point (of necessity, then, $R^+ = \{ 1/2\}$).\footnote{Observe that by the functional equation for $\xi$ (i.e., $\xi (s) = \xi (1-s)$, for all $s \in \mbc$) and since $\xi$ has zeros on the critical line $\{ Re(s) = 1/2\}$, we must have $R^+ = \{ 1/2\}$ since otherwise, $R^+$ would contain at least three points, $1/2, \rho$ and $1/2 - \rho$, for some $\rho \in (0, 1/2)$. Also, conversely, if $R^+$ had more than one point, then RH would obviously be violated.}

Entirely analogous statements can be made about the zeta functions of number fields [\hyperlinkcite{ParsSh1}{ParsSh1--2}] and of more general arithmetic $L$-functions, such as automorphic $L$-functions or conjecturally equivalently (see, e.g., \cite[App. C and App. E]{Lap7}), the members of the Selberg class.

In contrast, the total odd cohomology  space $H^- = H_{\zeta_\mcl}^-$ (still relative to the underlying $\mbz_2$-grading) of a generic nonlattice self-similar string $\mcl$ is `{\em multifractal}' in a very strong sense; namely, $R^-$, the set of real parts of the complex dimensions of $\mcl$ (in the usual sense of \S \ref{Sec:2} and \S \ref{Sec:3}) is a countably infinite set which is dense in a compact interval $[D_\ell, D]$, where $-\infty < D_\ell <D$ and $D = D_\mcl$ is the Minkowski dimension of $\mcl$.\footnote{We assume here implicitly that there are no (drastic) cancellations between the zeros and the poles of $\zeta_\mcl$. If we work instead with $H^- =H^-_{\ptoo}$, where the RFD $(\ptoo)$ is a geometric realization of $\mcl$, then $R \backslash \{0\}$ is of the above form.} \\

\subsubsection{Open problems and perspectives$:$ Geometric interpretation and beyond}\label{Sec:4.4.2}

We close this section (and paper) by mentioning a few open problems and long-term research directions related to the nascent fractal cohomology theory and its possible geometric and topological interpretations:\\

($i$) {\em Geometric interpretation.} It would be very useful and interesting to obtain a good geometric interpretation of the theory. It is natural to wonder what is the `curve' (or `variety') underlying fractal cohomology. By analogy with the case of curves (or, more generally, varieties) over finite fields, one may guess that the underlying supersymmetric `fractal curve' is the divisor $\mfD (g)$ of the given meromorphic function $g$ (or perhaps, its inverse, $\mfD (g)^{-1}$, as defined earlier).\footnote{For a nonlattice self-similar string $\mcl$ (in the case of a single gap), according to the results obtained in \cite[Ch. 3, esp., \S 3.4]{Lap-vF4}, the (multi)set of complex dimensions $\mfD (g) = \mcd (\zeta_\mcl)$ obeys a quasiperiodic pattern and is conjectured to form a `{\em generalized quasicrystal}' (in a sense to be yet fully specified); see \cite[Problem 3.22]{Lap-vF4}. More broadly, in \cite[esp., Ch. 5 and App. F]{Lap7}, this type of generalized quasicrystal plays a key conceptual role. } Ignoring the distinction between $\mfD (g)$ and its inverse, let us denote symbolically this `curve' by $C = C_g$. Then, abusing geometric language, one can think of the original Hilbert space $\mch_g$, the eigenspace of the GPO $\mbd_g$ (namely, the fractal cohomology space $H_g$) and even the GPO $\mbd_g$ and the GFO $\mcf_g$, as `bundles' or `sheaves' over $C$.

If we let the function $g$ vary in a suitable `moduli space' of meromorphic functions (or zeta functions),\footnote{Compare with the moduli spaces of fractal membranes (and of the associated zeta functions or spectral `partition functions') that play a central role in \cite[Ch. 5]{Lap7}.}
we then obtain geometric, analytic and algebraic structures (such as $C_g, \mch_g, \partial_g, \mbd_g$, $\mcf_g$ and $H_g$) which are naturally defined on that moduli space.

Pursuing the analogy with curves (or, more generally, varieties) over finite fields, one may wonder if the curve $C$ should not be augmented and replaced by a new (and larger) `curve' $\widetilde{C}$ whose (generic) points would coincide with the fixed points of the generalized Frobenius operator (GFO) $\mcf$ and of its iterates $\mcf^n$ (with $n \in \mbn$ arbitrary). In the present discussion, we allow the solutions (i.e., the fixed points of $\mcf^n$, for some $n \in \mbn$) to be distributional solutions (in a suitably defined weighted space of tempered distributions) and not just the Hilbert space solutions.\footnote{If we only allowed Hilbert space solutions, then $\widetilde{C}$ would simply coincide with $C$.} Then, a simple  differential equation computation (see \cite{Lap10}) shows that $\widetilde{C} = \widetilde{C}_g$ consists of the points of $C =C_g$ augmented by a copy of the group $\mu (\infty)$ of all (complex) roots of unity attached to (and acting on) each point of $C$.\footnote{Hence, $\mu (\infty) = \cup_{n \geq 1} \mu (n)$, where for each $n \in \mbn$, $\mu (n)$ is the group of (complex) $n$-th roots of unity; clearly, here, $\mu (n)$ does not denote the value of the M\" obius function at $n$.}
In other words, $\widetilde{C}$ can be viewed as a `principal bundle' over $C$, with structure group $\mu (\infty)$. Alternatively, $\widetilde{C}$ can be viewed as a `sheaf' over $C$.

Amazingly and likely not coincidentally, the group $\mu (\infty)$ (and variants thereof) also plays a key role in the theory of motives over the elusive field of one element, $\mbf_1$ (which originated in Y. Manin's seminal article \cite{Mani}); see, especially, \cite{Kah}, \cite{Tha1}, \cite{Tha2} and the relevant references therein.\footnote{See also the recent work \cite{Har4} in which $\mu (\infty)$ does not seem to play a role but whose proposed geometric and algebraic language and formalism for number theory and `varieties' over $\mbf_1$ is quite appealing and should, in the long term, be connected with aspects of the present theory.}

In the aforementioned references, working over $\mbf_1$ typically leads to adding new points to the underlying curve or variety (or motif) that were formally {\em invisible}. The same is true in the present situation of `fractal motives' and the associated fractal cohomology. Indeed, going from $C$ to $\widetilde{C}$ amounts to adding previously invisible points. Analytically, as was alluded to above, this amounts to going outside the original Hilbert space $\mch_g$ and finding all of the fixed points of the iterates $\mcf_g^n$ (for any $n \geq 1$) that lie in a suitable weighted space of tempered distributions containing $\mch_g$; see \cite{Lap10}.

In this context, for every given $n \in \mbn$, the fixed points of $\mcf^n$, correspond  to the `points' of the curve $\widetilde{C}$ which lie in the $n$-th field extension of $\mbf_1$. Heuristically, they can be thought of as a `principal bundle' over $C$ with structure group $\mu(n)$, the group of $n$-th  complex roots of unity. (Note the  close  analogy with the case of curves (or varieties) over finite fields discussed in \S \ref{4.1}.) 

Therefore, at the most fundamental level, we should aim at interpreting the present theory of fractal cohomology (as developed in [\hyperlinkcite{CobLap1}{CobLap1--2}] and \cite{Lap10}) as a universal cohomology theory of `{\em fractal motives}' over $\mbf_1$, the elusive field of one element. Finding a coherent and completely rigorous geometric and arithmetic interpretation of this type is clearly a long-term open problem.

\begin{remark}\label{Rem:4.3}
The `curve' $C = C_g$ itself, or its augmentation (by invisible points, the distributional fixed points of $\mcf^n$, for any $n \in \mbn$) $\widetilde{C} = \widetilde{C}_g$, can be viewed as the `{\em sheaf of complex dimensions}' (here, visible zeros and poles of $g$) associated with `the' analytic continuation of $g$ (i.e., with `the' meromorphic continuation of $g$ on any given connected open set $U$ of $\mbc$ on which the latter exists). Interestingly, the notion of a sheaf was introduced in the 1940s by J. Leray to solve certain problems in cohomology theory and then used to precisely deal with the a priori ill-defined notion of analytic continuation of a function of several complex variables, in connection with the Cousin problem. (See, e.g., \cite{GunRos}, \cite{Ebe} and \cite{Dieu2}.) The previously mentioned extensions of the notion of complex dimensions (as nonremovable singularities of suitable complex analytic functions on Riemann surfaces, or in higher dimensions, on complex analytic varieties or even, on analytic spaces) and of the associated theory (see \S \ref{Sec:2.5} and \S \ref{Sec:3.6}, as well as part ($iii$) of the present subsection) would, in the long term, provide a natural and significantly broader framework within which to consider the (yet to be precisely defined) sheaves $C$ and $\widetilde{C}$ in this extended context. 
\end{remark}

We next very briefly mention two other key research  directions in this area:\\

($ii$) {\em Fractal homology.} As is well known, standard cohomology theories in topology and in differential geometry, for example) are often the dual theories of suitable homology theories. (See, e.g., \cite{Dieu2} and \cite{MacL}.)  We wish to ask here what could be an appropriate `{\em fractal homology theory}' in the present context. More specifically, in the geometric context of bounded subsets of $\mbr^N$ (and, more generally, of RFDs in $\mbr^N$), can one construct a suitable fractal homology theory whose dual theory would be the corresponding fractal cohomology theory (associated with the divisors of fractal zeta functions, as described above).\footnote{At least for certain self-similar geometries, there seems to be glimpses of such a theory, in which the underlying complex dimensions play a natural role. }
Is there an appropriate generalization of Poincar\' e duality and of K\" unneth's formula in this context? What about the more general (but at present, less geometric) context of general (and suitable) meromorphic functions?\\

($iii$) {\em Singularities of functions of several complex variables.} Finally, we mention that as was alluded to in Remark \ref{Rem:4.3}, it would also be interesting to extend the theory of complex fractal dimensions and the associated fractal cohomology theory to (suitable) meromorphic functions of several complex variables naturally defined on complex manifolds (or, more generally on complex analytic spaces), as well as to develop an associated theory of (generalized) sheaves and sheaf cohomology (see, e.g., \cite{GunRos} and \cite{Ebe} for the classic theory).

Just as in the case of a single complex variable, we need not restrict ourselves to meromorphic functions defined on domains of $\mbc^r,$ with $r \geq 1$. One is therefore naturally led to consider (nonremovable) `{\em singularities}' of functions that go beyond the mere poles, as is already apparent in the present geometric theory of complex fractal dimensions; see, e.g., \S \ref{Sec:2.5} and \S \ref{Sec:3.6} above, along with \cite{LapRaZu10} for a first model in the case when $r=1$ and thus when suitable Riemann surfaces would be associated with the given singularities.\\

For more information on fractal cohomology theory and the many associated  open problems, we refer the interested reader to (the latter part of) the author's forthcoming book, \cite{Lap10}.

\section*{Acknowledgements}
This may be a fitting place to express my deep gratitude to Erin P.~J. Pearse, John A. Rock and Tony Samuel, the organizers of the International Conference and Summer School (supported by the National Science Foundation) on {\em Fractal Geometry and Complex Dimensions} held on the occasion of my sixtieth birthday in the beautiful setting of the campus of the California Polytechnic State University at San Luis Obispo, as well as to the editors of this volume of {\em Contemporary Mathematics} (the proceedings volume of the above SLO meeting), the aforementioned organizers along with Robert G. Niemeyer, for their great patience in waiting (with some anxiety, I presume) for my contribution and for their helpful suggestions for shortening the original text (which has now become the foundation for the author's forthcoming book, \cite{Lap10}).

I wish to thank all of my wonderful collaborators over the past twenty five to thirty years on many aspects of the theory of complex fractal dimensions and its various extensions or related topics, including (in rough chronological order) Jacqueline Fleckinger (\cite{FlLap} and other papers), Carl Pomerance [\hyperlinkcite{LapPo1}{LapPo1--3}], Helmut Maier [\hyperlinkcite{LapMa1}{LapMa1--2}], Jun Kigami [\hyperlinkcite{KiLap1}{KiLap1--2}], Michael M.~H. Pang \cite{LapPa}, Cheryl A. Griffith, John W. Neuberger and Robert J. Renka \cite{LapNRG}, Christina Q. He \cite{HeLap}, Machiel van Frankenhuijsen [\hyperlinkcite{Lap-vF1}{Lap-vF1--7}], Ben M. Hambly \cite{HamLap}, Erin P.~J. Pearse [\hyperlinkcite{LapPe1}{LapPe1--3}], Hung (Tim) Lu ([\hyperlinkcite{LapLu1}{LapLu1--3}], [\hyperlinkcite{LapLu-vF1}{LapLu-vF1--2}]), John A. Rock (\cite{LapRo}, \cite{LapRoZu}, \cite{LapRaRo}), Erik Christensen and Cristina Ivan \cite{ChrIvLap}, Jacques L\' evy Vehel \cite{LapLevyRo}, Steffen Winter [\hyperlinkcite{LapPeWi1}{LapPeWi1--2}], Nishu Lal [\hyperlinkcite{LalLap1}{LalLap1--2}], Hafedh Herichi [\hyperlinkcite{HerLap1}{HerLap1--5}], Rolando de Santiago and Scott A. Roby \cite{deSLapRRo}, Robyn L. Miller and/or Robert Niemeyer (in papers on fractal billiards), Jonathan J. Sarhad \cite{LapSar}, Goran Radunovi\'c and Darko \v{Z}ubrini\'c [\hyperlinkcite{LapRaZu1}{LapRaZu1--10}], Kate E. Ellis and Michael Mackenzie \cite{ElLapMcRo}, Tim Cobler [\hyperlinkcite{CobLap1}{CobLap1--2}] and Ryszard Nest \cite{LapNes}.

I also wish to thank my many diverse and talented Ph.D. students (soon twenty of them, this year, in 2018) and postdocs, a number of whom have been mentioned among my collaborators just above.

Furthermore, I am most grateful to the two very conscientious and thorough referees of this paper for their very helpful suggestions, questions and comments.

Moreover, I am very indebted to my personal Administrative Assistant, Ms. Yu-Tzu Tsai, without whose invaluable help this paper would never have been completed (almost) on time.

Finally, I would like to gratefully acknowledge the long-term support of the U.S. National Science Foundation (NSF) over the past thirty years, since the beginning of the theory of complex dimensions and throughout much of its prehistory, under the research grants DMS-8703138, DMS-8904389, DMS-9207098, DMS-9623002, DMS-0070497, DMS-0707524 and DMS-1107750.

\section*{Glossary}

\begin{itemize}
\setlength\itemsep{.3em}
\item[] $a_k \sim b_k$ as $k \rightarrow \infty$, asymptotically equivalent sequences\dotfill 13
\item[] $\mfa = \mfa_c$, spectral operator\dotfill 94
\item[] $\mfa^{(T)} = \mfa_c^{(T)}$, truncated spectral operator\dotfill 96
\item[] $A_\varepsilon$, $\varepsilon$-neighborhood of $A \subseteq \mbr^N$\dotfill 31
\item[] $\mca = \mca_c$, global spectral operator\dotfill 97
\item[] $\mcb (\mbh_c)$, Banach algebra of bounded linear operators on the Hilbert space $\mbh_c$\dotfill 96
\item[] CS, Cantor string\dotfill 9
\item[] $\mbc$, the field of complex numbers\dotfill 10
\item[] $\widetilde{\mbc} := \mbc \cup \{\infty \}$, the Riemann sphere\dotfill 95 
\item[] $c \ell (A)$ or $\overline{A}$, closure of a subset $A$ of $\mbc$ or of $\mbr^N$\dotfill 35, 95
\item[] $\underline{D}, \overline{D}$, lower, upper Minkowski dimension (of a RFD $(\ao)$ in $\mbr^N$)\dotfill 31
\item[] $D$, Minkowski dimension (of a fractal string $\mcl$ or of $A \subseteq \mbr^N$, or else of a RFD $(\ao)$ in $\mbr^N$)\dotfill 10, 31, 34
\item[] $\mcd= \mcd_\mcl$, the set of complex dimensions of a fractal string $\mcl$ (or of its boundary)\dotfill 10, 11
\item[] $\dim_\mbc \mcl = \mcd_\mcl (\mbc)$, the set of all complex dimensions (in $\mbc$) of the fractal string $\mcl$\dotfill 11
\item[] $\dim_{PC} \mcl = \mcd_{PC} (\mcl)$, the set of principal complex dimensions of the fractal string $\mcl$\dotfill 11
\item[] $\dim_\mbc (\ao) = \mcd_{\ao} (\mbc)$, the set of all complex dimensions (in $\mbc$) of the RFD $(\ao)$ in $\mbr^N$\dotfill 38
\item[] $\dim_{PC} (\ao) = \mcd_{PC} (\ao)$, the set of principal complex dimensions of the RFD $(\ao)$ in $\mbr^N$\dotfill 38
\item[] $\mcd (\zeta_\mcl) = \mcd (\zeta_\mcl; U)$, the set of (visible) poles (in $U$) of the geometric zeta function $\zeta_\mcl$ of a fractal string $\mcl$\dotfill 10
\item[] $\mcd_\mcl = \mcd_\mcl (U)$, the set of (visible) complex dimensions of a fractal string $\mcl \ (\mcd_\mcl = \mcd (\zeta_\mcl))$\dotfill 10
\item[] $\mcd (\zeta_{\ao}) = \mcd (\zao; U)$ or $\mcd (\tzao; U)$, the set of (visible) poles  (in $U$) of the distance or tube zeta function $\zao$ or $\tzao$ of an RFD $(\ao)$ in $\mbr^N$ $\ldots$\dotfill 34
\item[] $\mcd (\zeta_A) = \mcd (\zeta_A ; U)$ or $\mcd (\widetilde{\zeta}_A) = \mcd (\widetilde{\zeta}_A; U)$, the set of (visible) poles (in $U$) of the distance or tube zeta function $\zeta_A$ or $\widetilde{\zeta}_A$ of a bounded set $A$ in $\mbr^N$ $\ldots$\dotfill 34
\item[] $\mcd = \mcd_A$ or $\mcd = \mcd_{\ao}$, the set of (visible) complex dimensions of a bounded set $A$ or of an RFD in $\mbr^N (\mcd_A = \mcd (\zeta_A) = \mcd (\widetilde{\zeta}_A)$ and $\mcd_{\ao} = \mcd (\zeta_{\ao}) = \mcd (\tzao)$)\dotfill 34
\item[] $\mcd_\mfs$, the set of scaling complex dimensions  (of a self-similar spray or set) $\ldots$\dotfill 63, 64 
\item[] $\mbd_\mcz = \mbd_h$, generalized Polya--Hilbert operator (GPO or GPH) associated with the divisor $\mcz$ of a meromorphic function $h$\dotfill 101
\item[] $\mfD (f)$, divisor of a meromorphic function $f$\dotfill 59, 60
\item[] $d(x, A)= $ dist$(x, A) = \inf \{|x-a|: a \in A \}$, Euclidean distance from $x$ to $A \subseteq \mbr^N$\dotfill 4
\item[] $\partial = \partial_c$, infinitesimal shift of the real line ($\partial = d/dt$, acting on $\mbh_c$)\dotfill 94 
\item[] $\partial = d/dz$, infinitesimal shift of the complex plane (acting on a suitable weighted Bergman space)\dotfill 97, 98
\item[] $\partial^{(T)} = \partial_c^{(T)}$, truncated infinitesimal shift (of the real line)\dotfill 96 
\item[] $\partial \Omega$, boundary of a subset $\Omega$ of $\mbr^N$ (or of a fractal string)\dotfill 33
\item[] $|E| = |E|_N, N$-dimensional Lebesgue measure of $E \subseteq \mbr^N, 10(N=1),$\dotfill 31
\item[] $F$, Frobenius operator\dotfill 93 
\item[] $\mcf_\mcz = \mcf_h$, generalized Frobenius operator, associated with the divisor $\mcz$ of a meromorphic function $h$\dotfill 99
\item[] $\mbf_1$, field with one element\dotfill 106
\item[] $\mbf_q$, finite field (with $q$ elements)\dotfill 92
\item[] GFO, generalized Frobenius operator\dotfill 99
\item[] GPO (or GPH), generalized Polya--Hilbert operator\dotfill 99, 101
\item[] $\Gamma (t) := \int_0^{+ \infty} x^{t-1} e^{-x} dx$, the gamma function\dotfill 57, 99
\item[] $H_\mcz = H_h$, total eigenspace of the GPO $\mbd_\mcz$ associated with the divisor $\mcz$ of a meromorphic function $h$\dotfill 101, 105
\item[] $\mch_\mcz = \mch_h$, weighted Bergman space associated with the divisor $\mcz$ of a meromorphic function $h$\dotfill 98, 105
\item[] $\mbh_c := L^2 (\mbr, e^{-2ct}dt)$, weighted Hilbert space\dotfill 94
\item[] $i = \sqrt{-1}$, imaginary unit\dotfill 9
\item[] (ISP)$_D$, inverse spectral problem for the fractal strings of Minkowski dimension $D$\dotfill 26
\item[] $\xi = \xi (s)$, global (or completed) Riemann zeta function\dotfill 99
\item[] $\log_a x$, the logarithm of $x > 0$ with base $a > 0$; $y = \log_a x \Leftrightarrow x = a^y$\dotfill 9
\item[] $\log x := \log_e x$, the natural logarithm of $x; \, y = \log x \Leftrightarrow x = e^y$\dotfill 9
\item[] $\mcl = \{\ell_j \}_{j=1}^\infty$, a fractal string with lengths $\ell_j$\dotfill 8
\item[] $Li = Li(x)$, integral logarithm\dotfill 17
\item[] $\mcm, \mcm_*$ and $\mcm^*$, Minkowski content, lower and upper Minkowski content (of $\mcl$ or of $A \subseteq \mbr^N$ or else of a RFD $(\ao)$ in $\mbr^N$)\dotfill 12, 13, 32 
\item[] Moran's equation (complexified)\dotfill 64
\item[] $\mu (n)$, M\" obius function\dotfill 96
\item[] $N_\mcl$, geometric counting function of a fractal string $\mcl$\dotfill 14, 15
\item[] $N_\nu$, spectral (or frequency) counting function of a fractal string or drum $\ldots$\dotfill 14, 15, 24
\item[] $\Omega_\varepsilon$, $\varepsilon$-neighborhood of a fractal string (or of a fractal drum)\dotfill 10, 31
\item[] $\mathcal{P}$, the set of prime numbers\dotfill 16 
\item[] $\Pi_\mcp$, prime number counting function\dotfill 16
\item[] RFD, relative fractal drum\dotfill 31
\item[] RH, Riemann hypothesis\dotfill 26
\item[] (RH)$_D$ or RH$_D$, partial Riemann hypothesis (in dimension $D$)\dotfill 27
\item[] $\res (f, \omega)$, residue of the meromorphic function $f$ at $\omega \in \mbc$ \dotfill 5, 12
\item[] $\sigma (T)$, spectrum of the operator $T$\dotfill 95
\item[] $\widetilde{\sigma} (T)$, extended spectrum of the operator $T$\dotfill 95
\item[] $V(\varepsilon) = V_\mcl (\varepsilon)$ (or $V_{\ao} (\varepsilon)$ or $V_A (\varepsilon)$), volume of the $\varepsilon$-neighborhood of a fractal string $\mcl$ (or of a RFD $(\ao)$ or of $A \subseteq \mbr^N$)\dotfill 11, 31
\item[] $W (x)$, Weyl term\dotfill 24
\item[] $\Xi (s)$, completed Riemann zeta function (second version)\dotfill 102
\item[] $\zeta (s) = \sum_{j=1}^\infty j^{-s}$, Riemann zeta function\dotfill 15
\item[] $\zeta_\mcl (s) = \sum_{j=1}^\infty (\ell_j)^s$, geometric zeta function of a fractal string $\mcl = \{\ell_j \}_{j=1}^\infty$ $\ldots$\dotfill 9
\item[] $\zeta_\nu = \zeta_{\nu, \mcl}$, spectral zeta function of a fractal string $\mcl$\dotfill 15
\item[] $\zao$ or $\zeta_A$, distance zeta function (of a RFD $(\ao)$ or of a bounded set $A$ in $\mbr^N$)\dotfill 33, 34
\item[] $\tzao$ or $\widetilde{\zeta}_A$, tube zeta function (of a RFD $(\ao)$ or of a bounded set $A$ in $\mbr^N$)\dotfill 33, 34 
\end{itemize}

\end{document}